\documentclass[UTF8,10pt,a4paper,english]{book}
\usepackage{ctex}
\usepackage{babel}
\usepackage{geometry}
\usepackage{authblk}
\usepackage{booktabs}
\usepackage{url}
\usepackage{multirow}
\usepackage{graphicx}
\usepackage{subfigure}
\usepackage{amsmath,amsthm}
\usepackage{amssymb,amsfonts}
\usepackage{hyperref}
\usepackage{cite}
\usepackage{appendix}
\usepackage{textgreek}
\usepackage{fontspec}
\usepackage{fontspec-xetex}
\usepackage[nottoc,notlot,notlof]{tocbibind}
\usepackage{tikz}
\usetikzlibrary{shapes.geometric}
\usepackage{tikz-qtree}
\usetikzlibrary{trees} 

\title{Copula Entropy: Theory and Applications}
\author{Jian MA\thanks{Email: majian03@gmail.com}}

\newtheorem{definition}{Definition}
\newtheorem{theorem}{Theorem}

\newtheorem{corollary}[theorem]{Corollary}
\newtheorem{property}{Property}

\newlength{\drop}

\begin{document}


\begin{titlepage}
	\drop=0.1\textheight
	\centering
	\vspace*{\baselineskip}
	\rule{\textwidth}{1.6pt}\vspace*{-\baselineskip}\vspace*{2pt}
	\rule{\textwidth}{0.4pt}\\[\baselineskip]
	{\Huge  \textbf{COPULA ENTROPY} \\[0.3\baselineskip] {\LARGE \textbf{Theory and Applications}}}\\[\baselineskip]
	\rule{\textwidth}{0.4pt}\vspace*{-\baselineskip}\vspace{3.2pt}
	\rule{\textwidth}{1.6pt}\\[\baselineskip]
	\scshape
	\vspace*{2\baselineskip}
	{\Large Jian Ma\par} 
	\vfill
	{\small\scshape \today}\par
	\vspace*{\drop}	
\end{titlepage}

\frontmatter

\chapter{Preface}
Copula theory is of fundamental importance in probability and become mature after decades of development. As the core of it, Sklar's theorem presents the universal representation of dependence, i.e., copula, which can be used for dependence inference and measurements. Many types of copulas have been proposed and copula based bivariate dependence measures, such as Spearman's $\rho$ and Kendall's $\tau$, have also been given. Copula based multivariate dependence measure remains a problem in the field.

This monograph originated from the author's PhD study at Tsinghua University, during which the author defined the concept of copula entropy (CE), proved its equivalence to mutual information in information theory, proposed a nonparametric CE estimator, and applied CE to structure learning and copula component analysis. Since then, the author continued his PhD research and applied CE to eleven fundamental problems in statistics. Particularly, he proved the CE representation of transfer entropy / conditional mutual information and therefore built a theoretical framework that unified the concepts of correlation/dependence and causality/conditional independence based on CE. Additionally, he proposed five hypothesis testing methodologies based on CE. Meanwhile, researchers have introduced many mathematical generalizations of CE and applied CE to real problems in every branch of science. CE has become a mature theory and an integrated area of copula theory. This monograph aims to present the theoretical framework of CE and its latest developments.

CE theory presents a multivariate dependence measure based on copula and has many theoretical meanings, including
\begin{itemize}
	\item developed copula theory by introducing a perfect copula based multivariate dependence measure;
	\item built a bridge between copula theory and information theory with the proof of the equivalence of CE and mutual information and the CE representation of conditional mutual information, and hence made probability and information theory a much more integrated mathematical field;
	\item built a unified framework of correlation and causality based on CE, proposed the system of statistical methodologies including independence test, conditional independence test, multivariate normality test, copula hypothesis test, two-sample test, change point detection, and symmetry test and made mathematical statistics much more mature.
\end{itemize}

This monograph is not finished yet, and any comments and suggestions are welcome.

\begin{flushright}
	Jian Ma\\
	Haidian, Beijing\\
	December 2025
\end{flushright}

\tableofcontents


\mainmatter

\chapter*{Introduction}
\addcontentsline{toc}{chapter}{Introduction}
\label{chp:introduction}
Statistical independence is a core concept in mathematical statistics and how to represent and measure it is a basic problem. In the early days of statistics, Pearson\cite{pearson1896mathematical} proposed his correlation coefficient $r$ as a measure of correlation. Many scholars also proposed other measures of correlation, such as Spearman's $\rho$\cite{Spearman1987} and Kendall's $\tau$\cite{Kendall1938}.

In the last century, copula theory was proposed as a tool for universal representation of dependence between random variables \cite{nelsen2007introduction,joe2014dependence}. According to Sklar's theorem\cite{sklar1959fonctions}, the core of copula theory, any dependence relations between random variables correspond to a function for dependence representation, called copula, which is irrelevant to margins. Based on copula, many traditional correlation measures have their copula version, such as Spearman's $\rho$ and Kendall's $\tau$. However, these measures are all for bivariate relations, copula based multivariate dependence measure remains a problem yet to be solved.

Information theory is a theory about information measuring and processing \cite{Shannon1948}, and entropy and mutual information (MI) is its two core concepts \cite{cover1999elements}. Entropy measures information of random variables and MI measures information shared by two random variables. As a nonlinear measure of dependence, MI is considered to contain all the information of dependence between two random variables.

In 2008, Ma and Sun defined the concept of copula entropy (CE) \cite{ma2011mutual,Ma2009phd}. CE is defined as a type of Shannon entropy with copula. They also proved that negative CE is equivalent to MI. A nonparametric CE estimator was also proposed by them. In fact, the definition of CE is inspired by this thinking that since copula represents the dependence relation between random variables and MI measures all the information of dependence, there must be certain relationship between copula and MI. CE is the fruit of study on this relationship. By defining CE, we build a bridge between copula theory and information theory.

CE is the theory on measures of multivariate dependence while copulas is that of representation of dependence relations. CE measures all the information in copula representation. CE is an ideal measure of statistical independence and has many good properties, such as continuity, symmetry, additivity, nonpositivity, invariance to monotonic transformation, and equivalence to correlation coefficients in Gaussian assumptions.

As a basic statistical tool, CE can be applied to many statistical problems. In 2008, we have applied it to structure learning \cite{ma2008dependence}. Recently, we have applied it to more problems successfully, including association discovery \cite{jian2019discovering}, variable selection \cite{jian2019variable}, causal discovery \cite{jian2019estimating}, domain adaptation \cite{ma2022causal}, multivariate normality test \cite{Ma2022}, copula hypothesis testing \cite{Ma2025c}, time lag estimation \cite{Ma2023}, system identification \cite{Ma2023a}, two-sample test \cite{Ma2023b}, change point detection \cite{Ma2024a}, and symmetry test \cite{Ma2025a}. These work has become a system of methodologies based on CE.

As an ideal independence measure, CE can also be used to measure another important related concept -- conditional independence (CI). Transfer Entropy (TE) \cite{schreiber2000measuring} is a model-free measure of causality, and is essentially conditional mutual information (CMI), a measure of CI in information theory. We proved that TE / CMI can be represented with only CE and based on this, proposed a CE based nonparametric estimator of TE / CMI\cite{jian2019estimating}. Therefore, we derived a CE based theoretical framework on measuring both independence and CI and unified the concepts of correlation and causality together. This framework laid the foundation for CE based system of methodologies.

We derived a CE based system of methodologies, including independence test, CI test, multivariate normality test, two-sample test, change point detection, and symmetry test. There are many existing similar methods for these hypothesis testing problems while CE has theoretical advantages over them due to its rigorous definition and nonparametric estimator. To evaluate CE based methodologies, we surveyed the existing methods on the above six problems, and conducted comparative experiments based on simulations \cite{Ma2022b}, to verify the real advantages of CE based methodologies over their counterparts.

As a basic probabilistic concept, CE has been generalized by researchers to define new concepts, including Tsallis CE\cite{Mortezanejad2019}, survival CE\cite{Sunoj2023}, cumulative CE\cite{Arshad2024}, copula extropy\cite{Saha2023}, cumulative copula Tsallis entropy\cite{Zachariah2025}, copula R\'enyi entropy\cite{Saha2025}, copula R\'enyi divergence and copula Tsallis divergence\cite{Mohammadi2023a}, copula Jeffreys divergence and copula Hellinger divergence\cite{Mohammadi2021}, copula fractal inaccuracy\cite{Pandey2025}, and copula information measures\cite{Borzadaran2010}. These new concepts were defined by generalizing traditional information theoretical concepts with copulas, and therefore became a CE based system of concepts as an integrated part of generalized entropies \cite{Amigo2018,Ebrahimi2010,Arndt2001}.

As a statistical tool, CE has been applied to every branches of sciences, including theoretical physics\cite{ma2021thermodynamic}, astrophysics\cite{Ma2023c}, space science\cite{Akerele2024}, geology\cite{Ma2025}, geophysics\cite{Shan2024,Shan2025master,Mack2025}, fluid mechanics\cite{Chen2025}, thermology\cite{Li2025sst,Zhang2024master2,Ma2025b}, theoretical chemistry\cite{Cuendet2016}, cheminformatics\cite{wieser2020inverse}, material science\cite{Tian2023b,Liu2025,Zhang2024a,Jin2025}, hydrology\cite{chen2013measure,chen2014copula,Chen2021,Li2022a,Mo2023,Chen2023b,Wang2018master,wen2019,Wen2021master,Huang2019,huang2021phd,Niu2023,ni2020vine,Kanthavel2022,Mohammadi2021,Xu2025master,Liu2024master,Xing2024,xu2017a,Xu2018phd,Wang2019,Wang2022book,li2020developing,li2021developing,Xu2022,Xu2022a,Yang2019master,chen2019copulas,Chen2013book,Wang2023a,Qian2022,Wang2024a,liu2022water,porto2021a,Porto2023,huang2021,Jiang2023}, climatology\cite{Hao2015,Condino2009}, meteorology\cite{jian2019estimating,Wang2022,Wu2022b,Chen2023a,Guo2023,chen2019}, environmental science\cite{Wu2022d,Jin2023,Yang2024master,Qiao2024master,Xue2025}, ecology\cite{Hodel2021,Li2025a,Li2025b}, animal morphology\cite{Escolano2017,Purkayastha2022}, botany\cite{Wang2023master2,Li2025d}, agronomy\cite{Zhang2023,Zhang2023c,Zhang2024}, immunology\cite{Abe2025}, neurology\cite{Redondo2025,Li2025,Ciezobka2025}, cognitive neuroscience\cite{ince2016the,ince2017a,kayser2015irregular,Combrisson2022,Wang2023,DeClercq2022,Pospelov2023,Belloli2025,Walden2024,Walden2024a,Kaufmann2017}, motor neurology\cite{wu2021she,Wu2022,Wang2021master,Zhu2022,Jin2022,Reilly2022,OReilly2024}, computational neuroscience\cite{leugering2018a,Leugering2021phd,pakman2021estimating,Coroian2024}, psychology\cite{Ravijts2019}, system biology\cite{charzyńska2015improvement,farhangmehr2013an}, bioinformatics\cite{wieczorek2016causal,Wu2022a,Wu2022c,Shang2024,Shang2024master,Yan2025,Pan2023b,Zhong2024master,Zhu2023,Li2024phd,Lacalmita2025}, clinical diagnosis\cite{jian2019variable,Mesiar2021,Ma2022a,Fu2023b,Luo2023,Luo2023master,Sunoj2023,Pan2023,Tang2023}, geriatrics\cite{jian2019predicting,Li2023,ma2020predicting,ma2020associations}, psychiatrics\cite{Zhang2022,zhang2022master,Han2025}, forensics\cite{Wang2024phd}, pharmacy\cite{Gao2024}, public health\cite{Mesiar2021,Purkayastha2022}, economics\cite{Shan2020,luo2022,Zhang2023b,Bossemeyer2021,wei2021,Han2022,Ardakani2024}, management\cite{An2023,Flores2025master,Tian2023,Wang2022b,Wang2025b}, sociology\cite{ma2022causal}, pedagogy\cite{liu2018master}, computational linguistics\cite{Chen2023patent,Chen2025b}, media science\cite{Zhang2022a}, law\cite{Wieser2020}, political science\cite{Card2011}, military science\cite{Zhang2022patent}, informatics\cite{Xu2023b}, and energy\cite{fu2017uncertainty,Zhu2022a,Yang2025b,Wang2024p2,Cui2022,Ma2023,Yan2024,Kan2023,Hu2022master,Wu2024p,Wang2024master,Tang2024master,Dong2022,Liu2022a,Liu2023m,Feng2022,Feng2023,Lu2024,Sun2023,Wang2024p,Yang2024,Hu2022,Qin2023,Xiong2022,He2023,Zhao2025master,Ma2024}, textile engineering\cite{Mortezanejad2025}, food engineering\cite{Lasserre2021a,Lasserre2022}, civil engineering\cite{Li2022,Ma2024b,Chi2024master,Chang2025,Liu2025a,Liu2024phd,Lin2023master}, transportation\cite{Huang2021a,Xu2023,Ji2022phd,Wang2022a,Chang2024,Chang2024master,Zhou2024master}, manufacturing\cite{Sun2021,Cai2023,wang2015phd,Dong2023,Liu2023,Wang2024,Hu2023,Li2023master2}, reliability\cite{sun2019a,Cheng2023,Cheng2025,Meng2025,Geng2024master}, petroleum engineering\cite{Luo2024,Yuan2025}, mining engineering\cite{Gao2024phd,Jian2024ccc}, metallurgy\cite{Tian2024,Lin2024}, chemical engineering\cite{Yin2022,Wei2022,Wei2023,Pan2023a,Pan2024master,Bi2023,Wu2023,Shi2025master}, medical engineering\cite{Tang2024,Tang2025a}, aeronautics and astronautics\cite{Krishnankutty2020,Liu2022,Liu2023master,Zeng2022,Jia2023,Sun2024master,Wu2020,Tang2025}, arms\cite{Chen2024p}, automobile\cite{Gao2023,Xu2025}control engineering\cite{Li2025c}, electronics engineering\cite{Liu2022b}, communication\cite{wang2016physical,Wang2016,Fu2023}, high performance computing\cite{Gocht-Zech2022}, information security\cite{Liu2024,Wang2024p3,Wang2025,Baawi2025,Wang2025master}, geomatics\cite{Zeng2009,Cao2023,Zhang2022master2}, ocean engineering\cite{Zhao2022phd}, and finance\cite{mlfinlab,arbitragelab,Wang2015,Liao2023,Zhu2022b,Zhu2024,Xu2024,Calsaverini2009,Calsaverini2013,Alanazi2021,Wang2023b,Xiong2020,Ding2024master,Chen2024,Ardakani2024a,Kong2021,Peng2022m,Wang2023master,Zhang2023a,Gurgul2024,Gurgul2024a,Li2023master,Uddin2025,Mahmutovic2025}. In these applications, CE is used to analysis and measure correlation or causality in different fields for better understanding and modeling. CE provides theoretical support or practical tool and also bring computational efficiency and reliability.

In these real applications, researchers also proposed new methodologies by combining CE theory and other theory and methods, listed as following. The new methodologies in classical information theory include
\begin{itemize}
	\item CE based MI estimation, such as GCMI\cite{ince2017a}, Vector Copula based MI estimation\cite{Chen2025b}, R-Vine copula based MI and CMI estimation\cite{Wang2021master}, semi-parametric MI estimation\cite{Mohammadi2023}, asymmetric MI estimation\cite{Purkayastha2022}, $CE^2$\cite{Liu2024}, Bayesian dependence measure\cite{Marrelec2025}, and tail dependence measure\cite{Ardakani2024a};
	\item CE based information decomposition\cite{wieczorek2016causal,Kaufmann2017};
	\item CE based maximal entropy\cite{Bubak2025,Mortezanejad2019,Mortezanejad2025} and maximal Tsallis CE\cite{Mortezanejad2019}.
\end{itemize}
The methodology that combines CE with graph theory includes
\begin{itemize}
	\item MI based graph similarity\cite{Escolano2017}.
\end{itemize}
The methodology that combine CE with information bottleneck\cite{Tishby2000} includes:
\begin{itemize}
	\item CE based information bottleneck computation\cite{wieser2020inverse}.
\end{itemize}
Those that combine CE with partial information decomposition\cite{Williams2010} include:
\begin{itemize}
	\item unique information estimation\cite{pakman2021estimating}, synergy estimation\cite{Coroian2024}, and$\Omega$ information estimation\cite{Belloli2024}.
\end{itemize} 
The new methodologies in copula theory include:
\begin{itemize}
	\item copula parameter estimation\cite{Qian2022} and Vine copula model selection \cite{Alanazi2021,Calsaverini2009,Wang2023b,Liu2025a,Kovacs2011}.
\end{itemize}
CE based causal analysis include:
\begin{itemize}
	\item causal compression\cite{wieczorek2016causal}, causal structure learning\cite{Bossemeyer2021,Lasserre2021a,Lasserre2022}, LiNGAM-MMI\cite{Suzuki2024a}, and time series causal network construction\cite{He2023master,Yang2025}.
\end{itemize}
CE based generalizations of machine learning methods include:
\begin{itemize}
	\item clustering\cite{liu2022water,Condino2009}, nonlinear PCA\cite{Wei2022}, decision trees\cite{Shan2020,luo2022}.
\end{itemize}
Those that combine CE with neural networks include:
\begin{itemize}
	\item graph neural networks\cite{Chen2023b,Tang2023}, IEGAIN\cite{Wu2023}, neural network pruning\cite{Jiang2024}.
\end{itemize}
Those that combine CE with controllers include:
\begin{itemize}
	\item model prediction controller\cite{Li2025c} and sliding mode controller\cite{Yang2024}.
\end{itemize}
CE based optimization include:
\begin{itemize}
	\item CE based evolutionary algorithm\cite{Card2011}, CE based estimation of vine copula distribution algorithm\cite{Wang2013phd}, and CE based grey wolf optimization\cite{Pan2023b}.
\end{itemize}
CE based function approximation includes:
\begin{itemize}
	\item B-spline approximation\cite{Zhang2024master}.
\end{itemize}
As a universal theoretical tool that deals with correlation and causality, CE can be combined with other methods to make more methodologies possible.

Cover and Thomas \cite{cover1999elements} is a classical textbook in information theory, a main reference of this monograph, including a chapter discussing the relationship between information theory and statistics. Kullback \cite{Kullback1968} first introduced information theory to statistics and proposed to use Kullback-Leibler (KL) divergence for hypothesis testing. Csisz\'ar and Shields\cite{Csiszar2004} introduce a series of problems that apply information theory to statistics, including hypothesis testing. Pardo\cite{Pardo2018} introduces how to apply entropy and divergence to statistical inference, such as Goodness of Fit test, independence test, and symmetry test. Arndt\cite{Arndt2001} introduces information measures, including Shannon entropy and its generalizations. \cite{Aczel1975} introduces information measures and their characterizations. \cite{Gell-Mann2004} introduces Tsallis entropy and its interdisciplinary applications.

This monograph will apply CE to hypothesis testing problems, some of which can be referred to \cite{Rice2007}. \cite{Thode2002} is a book about normality tests. Bonnini et al. \cite{Bonnini2014} presents nonparametric hypothesis testing based on rank and permutation, including one-sample test, two-sample test, and independence test, etc. Lehmann and Romano\cite{Lehmann2022} is about hypothesis testing, including symmetry test and Goodness of fit test. Tartakovsky et al.\cite{Tartakovsky2014} introduced sequential analysis, including hypothesis testing and change point detection. \cite{Horvath2024} is a book focusing on change point detection. \cite{Kallenberg2005} is a monograph discussing probabilistic symmetry that provides a special perspective of all hypothesis testing. This monograph will define copula likelihood and present its relation with CE. Pawitan\cite{Pawitan2001} is a comprehensive book on likelihood. The book \cite{Parzen1998} is the selected papers of Hirotugu Akaike, including Akaike's paper on AIC that connects likelihood with information theory.

This monograph will investigate problems on causality and dynamical systems. Pearl\cite{Pearl2009} is an important book on causality. \cite{Atmanspacher1991} includes the papers on the workshop on ``information dynamics" held at Munich in 1990, some of which are about causality and information flow in dynamical systems. Spirtes et al.\cite{Spirtes2000} introduces the theory and algorithms of causal structure discovery. Bossomaier et al.\cite{Bossomaier2016} is a book focusing on TE and its applications, in which entropy based TE estimations are presented. \cite{Emmert-Streib2009} includes the papers that solve statistical problems with information theory, one of which is about causality detection with CMI and TE and introduced parametric and nonparametric entropy estimators for TE estimation. Peters et al.\cite{Peters2017} introduces foundations and algorithms on causal inference.

There are several books on copula theory, including those by Nelsen\cite{nelsen2007introduction} and by Joe\cite{joe2014dependence}, two main references of this monograph, which introduce the definition and properties of copulas, copula families and their construction, and copula based dependence measures (Kendall's $\tau$ and Spearman's $\rho$). Mari and Kotz\cite{Mari2001} focus on correlation and dependence, introduce both copulas and MI but did not connect them together. \cite{Jaworski2010a} includes the papers of the workshop on ``Copula theory and its applications" held at Warsaw in 2009, one of which is on copula based dependence measures\cite{Schmid2010}. \cite{Gorecki2024} is on hierarchical Archimedean copulas and their applications in finance. \cite{Kurowicka2010} is a paper collection on Vine copula. \cite{Czado2019} is a recent monograph on Vine copula.  Mai and Scherer\cite{Mai2017} introduce how to simulate copulas and multivariate distributions. \cite{Sun2020} and \cite{Cherubini2011} introduce copula based Markov processes and their applications in finance. Singh\cite{Singh2013} focuses on applications of entropies to hydrology, including an introduction of CE. Chen and Guo \cite{chen2019copulas,Chen2013book} focus on copulas in hydrology and water resource management, which applied CE to flood forecasting and river correlation measurement. Wang et al.\cite{Wang2022book} focus on applications of entropy and copulas to hydrological station network design and optimization, including the applications of CE to this problem.

This monograph will present the theory and applications of CE, which is organized as follows:

Chapter \ref{chap:ce} first introduces background, and then the theory of CE, including definitions, theorems and corollary, properties of CE. The parametric and nonparametric methods for CE estimation and CMI estimation based on CE are also presented.

Chapter \ref{chap:theoapp} introduces the theoretical applications of CE to twelve statistical problems, including variable selection, causal discovery, etc. Particularly, the applications of CE to seven main hypothesis testing problems, including independence test, CI test, multivariate normality test, copula hypothesis testing, two-sample test, change point detection, and symmetry test, are presented.

Chapter \ref{chap:discussion} first discusses the relationships between theoretical applications of CE, and then discusses their connections to correlation and causality, followed by the comparison between the framework of CE and those of the other two important independence measures, i.e., kernel method and distance correlation.

Chapter \ref{chap:bench} evaluates CE based hypothesis testing methods with simulation experiments and compares them to their counterparts with their \textsf{R} implementations.

Chapter \ref{chap:math} summarizes the systematic generalizations of CE, including CE based generalizations of R\'enyi entropy and Tsallis entropy, generalizations of information distances, divergences, and inaccuracy.

Chapter \ref{chap:realapp} briefly introduces the applications of CE to every branch of science and engineering.

Appendix \ref{chap:impl} lists the official and third-party implementations of CE based methods in \textsf{R}, \textsf{Python}, \textsf{Julia}, \textsf{Matlab}, and \textsf{C++} for reader's practical uses.

\chapter{Copula Entropy}
\label{chap:ce}
\section{Background}
\subsection{Copula theory}
Copula theory is about representation of dependence relations between random variables\cite{nelsen2007introduction,joe2014dependence}. It defines a type of probability function for such representation, called copula, as follows:
\begin{definition}[Copula]
	Given $n$ random variables $\mathbf{X}=(X_1,\ldots,X_n)$. Let $\mathbf{u}$ be margins ${u_i=F_i(x_i),i=1,\dots,n}$ of $\mathbf{X}$, then $n$ dimensional copula $C:I^n \rightarrow I, I=[0,1]$ of $\mathbf{X}$ should satisfy the following properties:
	\begin{enumerate}
		\item $C$ is non-decreasing on any subset of $I^n$;
		\item $0 \leq C(\mathbf{u}) \leq 1$;
		\item $C(1,\dots,1,u_i,1,\ldots,1)=u_i$.
	\end{enumerate}	
\end{definition}

Intuitively, copula is a distribution function on $I^n$ with margins being uniform distribution. Copula density function can be derived accordingly from copula as $c(\mathbf{u})=\frac{\partial C(\mathbf{u})}{\partial u_1 \cdots \partial u_n}$.

Sklar's theorem is the core of copula theory, which states that
\begin{theorem}[Sklar's theorem]\cite{sklar1959fonctions}
\label{thm:sklar}
	Given any $n$ dimensional random variables $\mathbf{X}$ with joint distribution $F(\mathbf{x})$ and margins $F_i(x_i)$. Then there exists a copula $C(\mathbf{u})$ such that 
	\begin{equation}
		F(\mathbf{x})=C(F_1(x_1),\ldots,F_n(x_n)).
		\label{eq:sklar}
	\end{equation}
	If $F_i$ are continuous, then $C$ is unique. Conversely, if $C$ is a copula and $F_i$ are distribution functions, then the function $F$ defined by \eqref{eq:sklar} is a joint distribution with margins $F_i$.
\end{theorem}

Copula separates joint distribution from margins and represents dependence relations as a function. So, dependence relation is irrelevant to properties of individual variables and copula contains all the information of dependence between random variables. By differentiating \eqref{eq:sklar}, one can derive the pdf version of Sklar's theorem:
\begin{equation}
	p(\mathbf{x})=c(\mathbf{u})\prod_{i}{p(x_i)},
\end{equation}
where $p(\cdot)$ are joint or marginal density functions.

\subsection{Information theory}
Information theory is a theory about measuring and processing information\cite{Shannon1948}. Entropy is one of its fundamental concepts for measuring information content of random variables, defined as follows:
\begin{definition}[Shannon entropy]\cite{cover1999elements}
\label{def:shannonentropy}
	Given random variables $X\in R^n$ and their probability density function $p(\mathbf{x})$, Shannon entropy is defined as
	\begin{equation}
		H(\mathbf{x})=-\int_{\mathbf{x}}p(\mathbf{x})\log p(\mathbf{x})d\mathbf{x}.
		\label{eq:entropy}
	\end{equation}	
\end{definition}

MI is another fundamental concept in information theory for measuring information content shared by two random variables, which can also be understood as the information of a random variable contained in another random variable, defined as follows:
\begin{definition}[Mutual Information]\cite{cover1999elements}
\label{def:mi}
	Given a pair of random variables $(X,Y)$ with joint density distribution $p(x,y)$ and marginal density distribution $p(x),p(y)$, MI of $X$ and $Y$ is defined as
	\begin{equation}
		I(x;y)=\int_{x}\int_{y}p(x,y)\log \frac{p(x,y)}{p(x)p(y)}dxdy.
	\end{equation}
\end{definition}

The following theorem can be derived from the definition of MI:
\begin{theorem}
	\label{thm:midecomposition}
	MI equals the difference between joint entropy and marginal entropy, i.e.,
	\begin{equation}
		I(x;y)=H(x)+H(y)-H(x,y).
	\end{equation}
\end{theorem}

MI is proved to be non-negative, stated as the following theorem:
\begin{theorem}\cite{cover1999elements}
	 \label{thm:mipositive}
	 Given a pair of random variables $(X,Y)$, then
	 \begin{equation}
	 	I(x;y)\geq 0
	 \end{equation}
	 with equality if and only if $X,Y$ are independent.
\end{theorem}

By replacing density function with conditional density function, one can define CMI:
\begin{definition}[Conditional Mutual Information]
	Given random variables $(X,Y,Z)$, CMI between $(X,Y)$ conditional on $Z$ is defined as
	\begin{equation}
		I(x;y|z)=\int_{x}\int_{y}\int_{z}p(x,y,z)\log \frac{p(x,y|z)}{p(x|z)p(y|z)}dxdydz.
	\end{equation}
\end{definition}

\section{Theory}
\label{sec:theory}
By replacing probability density function in Shannon entropy with copula density function, one can define copula entropy:
\begin{definition}[Copula Entropy]\cite{ma2011mutual,Ma2009phd}
	\label{def:ce}
	Given random variables $\mathbf{X}$ with margins $\mathbf{u}$ and copula density function $c(\mathbf{u})$, copula entropy is defined as
	\begin{equation}
		H_c(\mathbf{x})=-\int_{\mathbf{u}}{c(\mathbf{u})\log c(\mathbf{u})d\mathbf{u}}.
		\label{eq:ce}
	\end{equation}
\end{definition}
CE is a special type of Shannon entropy that measures information content in copula, and therefore provides a tool for measuring multivariate dependence.

In information theory, entropy and MI are two different concepts\cite{cover1999elements}. In \cite{ma2011mutual}, we proved that they are essentially same, i.e., MI is equivalent to negative CE and hence also a type of Shannon entropy, which is stated in the following theorem:
\begin{theorem}
	\label{thm:hc}
	MI is equivalent to negative CE:
	\begin{equation}
		I(\mathbf{x})=-H_c(\mathbf{x}).
		\label{eq:theorem}
	\end{equation}
\end{theorem}
\begin{proof}
\begin{align}
	I(\mathbf{x})&=\int_{\mathbf{x}}p(\mathbf{x})\log \frac{p(\mathbf{x})}{\prod_{i}p(x_i)}d\mathbf{x}\\
	&=\int_{\mathbf{x}}c(\mathbf{u})\prod_{i}p(x_i)\log c(\mathbf{u})d\mathbf{x}\\
	&=\int_{\mathbf{u}}c(\mathbf{u})\log c(\mathbf{u})d\mathbf{u}\\
	&=-H_c(\mathbf{x})
\end{align}
\end{proof}
From Theorem \ref{thm:midecomposition} and Theorem \ref{thm:hc}, one can instantly derive a corollary about the relationship between joint entropy, marginal entropies and CE:
\begin{corollary}
	\label{coro1}
	Joint entropy is the sum of marginal entropies and CE:
	\begin{equation}
		H(\mathbf{x})=\sum_{i}{H(x_i)}+H_c(\mathbf{x}).
		\label{eq:corollary}
	\end{equation}
\end{corollary}

CMI is proved to be represented with only CE\cite{jian2019estimating}:
\begin{theorem}\cite{jian2019estimating}
	\label{thm:cmi}
	Given random variables $(X,Y,Z)$, CMI can be represented with only CE:
	\begin{equation}
		I(x;y|z)=H_c(x,z)+H_c(y,z)-H_c(x,y,z).
		\label{eq:cmithm}
	\end{equation}
\end{theorem}
\begin{proof}
\begin{align}
	I(x;y|z)&= \int_{x}\int_{y}\int_{z}p(x,y,z)\log{\frac{p(x,y|z)}{p(x|z)p(y|z)}}dxdydz\\
	&=\int_{x}\int_{y}\int_{z}p(x,y,z)\log{\frac{p(x,y,z)p(z)}{p(x,z)p(y,z)}}dxdydx\\
	&=I(x,y,z)-I(x,z)-I(y,z)\\
	&=H_c(x,z)+H_c(y,z)-H_c(x,y,z).
\end{align}
\end{proof}

With Definition \ref{def:ce} of CE, Theorem \ref{thm:hc} and its Corollary \ref{coro1}, and Theorem \ref{thm:cmi}, we build a theoretical framework on the concepts of information theory, which deepens our understanding on them and builds a bridge between copula theory and information theory.

\section{Properties}
\label{s:properties}
\paragraph{Multivariate}
MI is defined for bivariate cases while CE is not confined to this, but for any multivariate cases. According to Sklar's theorem, any multivariate random variables have their corresponding copula, and therefore their CE.
\begin{property}[Multivariate]
	For any $n$ random variables $\mathbf{X}\in R^n, n>1$, there exists $H_c(\mathbf{x})$.
\end{property}

\paragraph{Properties of entropy}
Since copulas are actually probability density functions on $I^n$, CE defined on them is a special type of Shannon entropy and then shares the same properties of entropy, such as continuity, symmetry, and additivity, etc.\cite{Khinchin1957,Csiszar2008}.

\begin{property}[Continuity]
Given random variables $\mathbf{X}$, if $\mathbf{x_n}\rightarrow \mathbf{x}$, then
\begin{equation}
	\lim_{\mathbf{x_n}\rightarrow \mathbf{x}}H_c(\mathbf{x_n})=H_c(\mathbf{x}).
\end{equation}
\end{property}
 
\begin{property}[Symmetry]
Given random variables $X$ and $Y$, then
\begin{equation}
	H_c(x,y)=H_c(y,x).
\end{equation}
\end{property}

\begin{property}[Additivity]
Given independent random variables $\mathbf{X}_1,\ldots,\mathbf{X}_n$, then
\begin{equation}
	H_c(\mathbf{x}_1,\ldots,\mathbf{x}_n)=\sum_{i=1}^{n}H_c(\mathbf{x}_i).
\end{equation}
\end{property}

\paragraph{Nonpositivity}
Theorem \ref{thm:mipositive} shows that MI is nonnegative. Hence, based on Theorem \ref{thm:hc}, CE is then non-positive, with 0 if and only if random variables are independent. 
\begin{property}[Nonpositivity]
Given random variables $\mathbf{X}$, then
\begin{equation}
	H_c(\mathbf{x})\leq 0,
	\label{eq:nonpositive}
\end{equation}
with equality if and only if $\mathbf{X}$ are independent.
\end{property}
This property implies that dependence between random variables makes random variables containing others information and hence reduces total information, i.e., joint entropy is no more than the sum of marginal entropies. Generally speaking, entropy that measures uncertainty of random variables is nonnegative, while CE that reduces uncertainty due to dependence is non-positive.

\paragraph{Invariance to monotonic transformation}
Copula is invariant to monotonic transformation\cite{nelsen2007introduction}, so CE inherits this property from copulas.
\begin{property}[Invariance to monotonic transformation]
Given random variables $(X,Y)$, and monotonic increasing functions $f$ and $g$ defined on them, then
\begin{equation}
	H_c(x,y)=H_c(f(x),g(y)).
\end{equation}
\end{property}

\paragraph{Marginal Free}
As Sklar's theorem states, joint distribution is separated into margins and copula. Accordingly, joint entropy is decomposed into marginal entropies and CE. Due to equivalence between MI and CE, MI is actually margin free, so the original definition of MI, Definition \ref{def:mi}, is not complete mathematically.
\begin{property}[Marginal free]
	$H_c(\mathbf{x})$ of random variables $\mathbf{X}\in R^n$ is margin free.
\end{property}

\paragraph{All orders}
Copula is considered to contain all order dependence relations of random variables, so CE measures all the information of all order dependence. As functional of random variables, joint entropy measures all order information of random variables which is decomposed into marginal entropies for all order information of individual variables and CE for all order information of dependence according to Corollary \ref{coro1}.
\begin{property}[All order]
	\label{p:allorder}
	$H_c(\mathbf{x})$ measures information of all order dependence between random variables.
\end{property}

\paragraph{Equivalence to correlation coefficient under Gaussian assumption}
Correlation coefficient measures linear correlation between random variables. Gaussian distributions contain only second order correlation and these correlation is uniquely determined by variance/covariance. We prove that under Gaussian assumption, CE is equivalent to correlation coefficient matrix. This is a special case of Property \ref{p:allorder}.
\begin{property}[Equivalence to correlation coefficient matrix under Gaussian assumption]
\label{p:cov}
Given random variables $\mathbf{X}\sim N(\mathbf{\mu},\Sigma)$ governed  by Gaussian distribution with correlation coefficient matrix $\Sigma_\rho$, then
\begin{equation}
	H_c(\mathbf{x})=\frac{1}{2}\log |\Sigma_\rho|.
	\label{eq:cevx}
\end{equation}
\end{property}
\begin{proof}
Given $n$ random variables $\mathbf{X}\sim N(\mu,\Sigma)$, then
\begin{align}
	H_c(\mathbf{x})&=H(\mathbf{x})-\sum_{i}H(x_i)\\
				   &=\frac{1}{2}\log (2\pi e)^n |\Sigma|-\sum_{i=1}^n \frac{1}{2}\log 2\pi e \delta_i^2\\
				   &=\frac{1}{2}\log|\Delta\Sigma\Delta|\\
				   &=\frac{1}{2}\log |\Sigma_\rho|.
\end{align}
$\delta_i^2$ and $\Delta$ denote variance of individual variables and matrix with diagonal elements as inverse standard variances $\delta_i^{-1}$.
\end{proof}

\section{Estimation}
As a fundamental concept in information theory, MI enjoys wide applications in different fields. However, it is widely considered that estimating MI is notoriously difficult. Based on Definition \ref{def:ce} and Theorem \ref{thm:hc}, we can estimate MI/CE parametrically or non-parametrically. Furthermore, based on Theorem \ref{thm:cmi}, we can estimate CMI with CE estimators.

\paragraph{Parametric} If copula density function $c(\mathbf{u};\alpha)$is known with $\alpha$ as its parameters, then we can compute CE according to its original definition, Definition \ref{def:ce}, as follows:
\begin{equation}
	H_c(\mathbf{x})=-E(\log c(\mathbf{u})).
	\label{eq:paraest}
\end{equation}

When using \eqref{eq:paraest}, one needs to derive $\mathbf{u}$ first in two situations: if margins are known,  then $\mathbf{u}$ can be computed directly; if margins are unknown, then empirical copula density function $\hat{\mathbf{u}}$ can be estimated instead.

If $\alpha$ is unknown, then one can estimate it with the likelihood methods \cite{joe2014dependence} first.

\paragraph{Nonparametric} Based on Theorem \ref{thm:hc}, Ma and Sum\cite{ma2011mutual} proposed an easy and elegant nonparametric CE estimation method \footnote{The method has been implemented in the \texttt{copent} package\cite{ma2021copent} in  \textsf{R} and \textsf{Python}, and shared on \href{https://cran.r-project.org/package=copent}{CRAN} and \href{https://pypi.org/project/copent/}{PyPI}.}. The proposed method composes of two steps:
\begin{enumerate}
	\item estimate empirical copula density functions;
	\item estimate CE from empirical copula density functions with entropy estimators.
\end{enumerate}

Given random variables $\mathbf{X}\in R^n$ and a sample from them $\{\mathbf{x}_1,\ldots,\mathbf{x}_T \}$, then Step 1 can be achieved with rank statistic:
\begin{equation}
	F_i(x_i)=\frac{1}{T}{ \sum_{t=1}^{T}{\mathbf{1}(x_t^i \leq x_i)} },
\end{equation}
where $i=1,\ldots,n$, $\mathbf{1}(\cdot)$ is indication function.

After empirical copula density function is estimated in Step 1, Step 2 is actually entropy estimation problem for which many methods exist. We proposed to use kNN entropy estimation by Kraskov et al.\cite{kraskov2004estimating}.

Since two steps are based on nonparametric methods, we then present a nonparametric CE estimator. It is easy to understand and implement and with low computation burden. This CE estimator is based on rank statistic, essentially to estimate entropy of normalized rank statistic.

\paragraph{CMI estimation} 
Based on Theorem \ref{thm:cmi}, CMI can be represented as \eqref{eq:cmithm}. So, one can easily estimate CMI with CE estimators accordingly.

\chapter{Theoretical Applications}
\label{chap:theoapp}
\section{Structure learning}
Dependence structure between random variables is important in data analysis because it can help us understand the underlying mechanism behind data. In statistics and machine learning, graphical models\cite{Jordan1999,Koller2009} are powerful tools to use graphs to represent dependence relationships between random variables with nodes in graph for random variables and edges for dependence relationships. Graphical models is a kind of approximation to the underlying probability density distribution. Learning dependence structure from data is a common problem in this area, also call structure learning\cite{Heckerman1995,Kitson2023,Drton2017}.

Structure learning is essentially a problem of model selection. Akaike\cite{akaike1974a,Akaike1998} proposed AIC for model selection that connects likelihood principle with KL divergence in information theory. Gr\o{}nneberg and Hjort\cite{Groenneberg2014} applied AIC to copula model selection and proposed copula information criteria (CIC) through improving AIC.

Given random variables $\mathbf{X}=(X_1,\ldots,X_m)$ and their corresponding graph structure $G$, then their joint density function can be represented as the following product form:
\begin{equation}
	p(\mathbf{x}|G)=\prod_{i=1}^{m}p(x_i|x_{\pi_i}),
	\label{eq:chowliu1}
\end{equation}
where $\pi_i$ is the predecessor nodes of $x_i$ in $G$.

Graph representation has its corresponding copula, since graph represents dependence relationships and so does copula. In this sense, \eqref{eq:chowliu1} can be transformed into the following copula form:
\begin{equation}
	p(\mathbf{x}|G)=\prod_{i=1}^{m}c_i(x_i,x_{\pi_i})\prod_{i=1}^{m}p(x_i),
	\label{eq:chowliu2}
\end{equation}
where the first term approximates copula of $X$ and the second term is about margins which is irrelevant to graph structure. Based on this connection between graph structure and copula, we can learn graph structure with copula.

Ma\cite{ma2008dependence,Ma2009phd} proposed a copula based framework for structure learning: first estimate empirical copula density functions, and then learn graph structure based on copula likelihood principle or CE. Here, he defined copula likelihood from traditional likelihood concept\cite{Myung2003}. Given random variables $\mathbf{X}\sim P(x)$ and their samples $\mathbf{X}_T$, likelihood function is defined as
\begin{equation}
	\mathcal{L}(\theta;\mathbf{X})=\sum_{t=1}^T\log p(\mathbf{x}_t;\theta).
	\label{eq:likelihood}
\end{equation}
According to Sklar's theorem, \eqref{eq:likelihood} can be decomposed into two terms for copula density function and margins respectively, i.e.,
\begin{equation}
	\mathcal{L}(\theta;\mathbf{X})=\sum_{t=1}^T\log c(\mathbf{u}_t;\theta_c)+\sum_{t=1}^{T}\sum_{i=1}^m\log p_i(\mathbf{x}_t;\theta_i).
	\label{eq:likelihood-decomposition}
\end{equation}
The first term in \eqref{eq:likelihood-decomposition} is for copula density function and corresponds to graph to be learned, which is irrelevant to the second term for individual variables.
\begin{definition}[Copula Likelihood] 
	\label{def:copulalikelihood}
	Given random variables$\mathbf{X}\sim P(x)$ and their samples $\mathbf{X}_T$, let copula density function be $c(\mathbf{u};\theta_c)$, copula likelihood is defined as
	\begin{equation}
		\mathcal{L}_c(\theta_c;\mathbf{X})=\sum_{t=1}^T\log c(\mathbf{u}_t;\theta_c).
		\label{eq:copulalikelihood}
	\end{equation}
\end{definition}
Copula likelihood is related to CE, i.e., negative copula likelihood is equivalent to empirical CE so maximum copula likelihood means minimal empirical CE.
\begin{theorem}
	Expectation of negative copula likelihood is equivalent to empirical CE: 
	\begin{equation}
		H_c(\mathbf{X}_T)=-\frac{1}{T}\mathcal{L}_c(\theta_c;\mathbf{X}_T).
	\end{equation}
\end{theorem}
\begin{proof}
	This can be derived directly from Definition \ref{def:ce} and Definition \ref{def:copulalikelihood}.
\end{proof}

With maximum copula likelihood or minimal CE as principle, we proposed a copula based framework for structure learning which can get copula approximation to graph in \eqref{eq:chowliu2}.

There are many traditional structure learning algorithms, and the Chow-Liu method\cite{chow1968approximating} is a typical one among them. It generates optimal tree structure to approximate density function by maximizing MI sum of tree structure with minimal-spanning-tree (MST) algorithm. The following sum of MI is used in Chow-Liu algorithm:
\begin{equation}
	\max \sum_{i=1}^{m}I(x_i,x_{\pi_i}).
	\label{eq:chowliusum}
\end{equation}

Through the equivalence between MI and CE, we can transform the maximal sum of MI \eqref{eq:chowliusum} into the minimal sum of CE:
\begin{equation}
	\min \sum_{i=1}^{m}H_c(x_i,x_{\pi_i}).
\end{equation}

With this optimizing objective, Ma proposed the CE version of Chow-Liu method \cite{ma2008dependence}, which composed of two steps:
\begin{enumerate}
\item get negative CE matrix from sample with CE estimator;
\item generate tree structure from CE matrix with MST algorithm.
\end{enumerate}

This method is with copula based framework for structure learning: CE estimator needs to estimate empirical copula density function first and then derives negative CE matrix. Since the nonparametric CE estimator is efficient and robust, the proposed method is more advantageous than those based on MI estimation.

Vine Copula is another copula based graph structure learning method \cite{Czado2019,Kurowicka2010}. Given random variable $\mathbf{X}=(X_1,\ldots,X_m)$ and their density function $p(\cdot)$, also a vine structure $V=(T_1,\cdots,T_{m-1})$, $T_k=(\mathbf{V}_k,\mathbf{E}_k),k=1,\ldots,m-1$, then the Vine copula based approximation to joint density function is\cite{Bedford2001}:
\begin{equation}
	p(\mathbf{x}|V)=\prod_{k=1}^{m-1}\prod_{e\in E_k}{c_{u_e,v_e|\pi_e}(u_e,v_e|\mathbf{x_{\pi_e}})}\prod_{i=1}^{m}p(x_i),
	\label{eq:vinecopula}
\end{equation}
where $\pi_e$ and $u_e,v_e$ represent conditional variables of $e$ and margins of two nodes.

It can be learned from \eqref{eq:vinecopula} that in the Vine copula approximation, each edge corresponds to a bivariate conditional copula. Vine copula is usually estimated with partial correlation \cite{Bedford2002}. Simplified assumptions is required for such estimation \cite{HobaekHaff2010} and therefore limit the real applications of Vine copula \cite{Nagler2025}.

Ma\cite{ma2008dependence} verified the proposed method with simulation experiments and applied it to two real datasets: abalone data and Boston housing dataset\footnote{The code is available at \url{https://github.com/majianthu/dse}}. Experimental results (see Figure \ref{fig:slresult}) show that the proposed method can give meaningful dependence structure which can deepen our understanding of the target datasets.

\begin{figure}
	\centering
	\subfigure[Abalone data]{\includegraphics[width=0.45\textwidth]{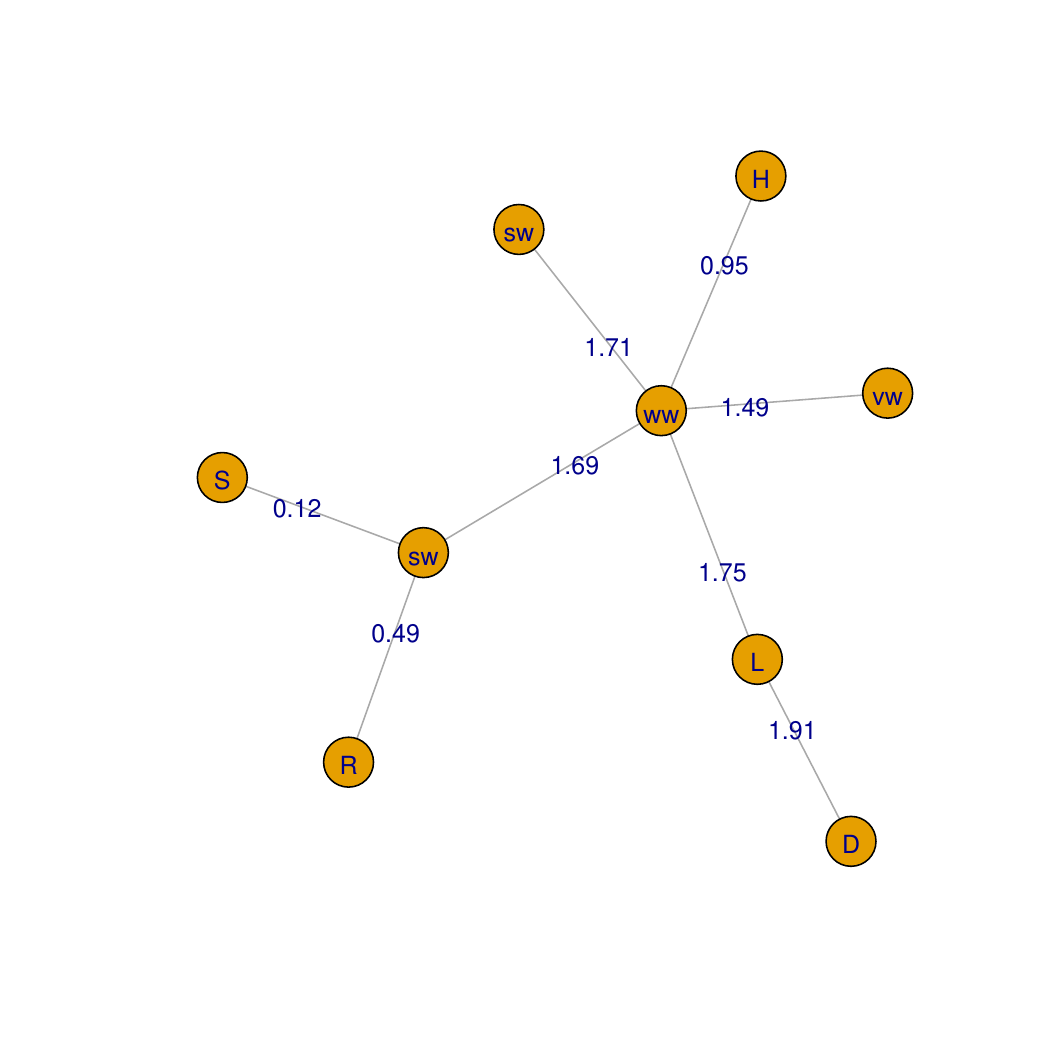}}
	\subfigure[Boston housing data]{\includegraphics[width=0.45\textwidth]{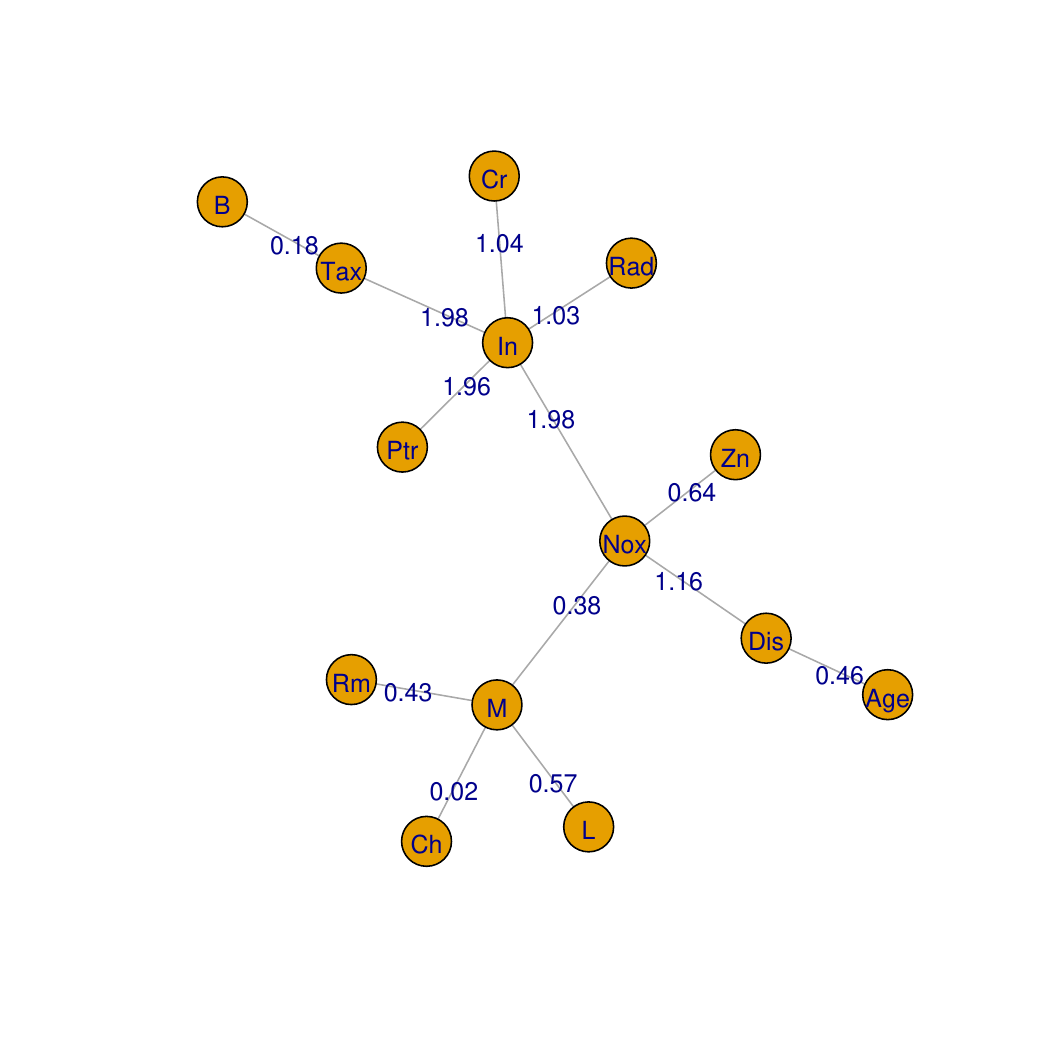}}
	\caption{Results of structure learning experiments on real datasets.}
	\label{fig:slresult}
\end{figure}

\section{Association discovery}
\label{s:ad}
Empirical science is about drawing scientific conclusions by data analysis. Association between random variables is an important relationship in statistics. Association discovery is a key problem in many sciences \cite{Liebetrau1983}.

Pearson correlation coefficient (PCC)\cite{pearson1896mathematical} is a commonly used association measure and has its history date back to the early days of statistics. However, it has many well-known limitations: only applicable to linear correlation and with implicit Gaussian assumption. It is only for bivariate cases.

Compared to PCC, CE has many advantages (see Table \ref{tb:ccce}). It is universally applicable to any case without any assumption. It measures nonlinear dependence, not only linear correlation. It also has properties, such as invariance to monotonic transformation, equivalence to correlation coefficient under Gaussian assumption.

\begin{table}
	\centering
	\caption{Comparison of PCC and CE.}
	\begin{tabular}{l|c|c}
		\toprule
		&PCC&CE\\
		\midrule
		variables&bivariate&multivariate\\
		linearity&linear&nonlinear\\
		order&second&all\\
		assumption&Gaussanity&none\\
		type&correlation&dependence\\
		\bottomrule
	\end{tabular}
	\label{tb:ccce}
\end{table}

Besides PCC, there are two common nonparametric correlation coefficients: Spearman's $\rho$\cite{Spearman1987} and Kendall's $\tau$\cite{Kendall1938}. Similar to CE, these two measures can also be represented with copula.
\begin{theorem}\cite{nelsen2007introduction}
	Given random variable $(X,Y)$ and their copula $C(u,v)$, then Spearman's $\rho$ can be represented with copula:
	\begin{equation}
		\rho_{X,Y}=12\int_{u}\int_{v}C(u,v)dudv - 3.
		\label{eq:spearman}
	\end{equation}	
\end{theorem}
\begin{theorem}\cite{nelsen2007introduction}
	Given random variables $(X,Y)$ and their copula $C(u,v)$, then Kendall's $\tau$ can be represented with copula:
	\begin{equation}
		\tau_{X,Y}=4\int_{u}\int_{v}C(u,v)dC(u,v) - 1.
		\label{eq:kendall}
	\end{equation}	
\end{theorem}
These two measures are also bivariate but can measure nonlinear correlation and without Gaussian assumption.
Barbe et al.\cite{Barbe1996} and Joe\cite{Joe1990} proposed multivariate versions of Kendall's $\tau$, and Spearman's $\rho$ also has two multivariate versions\cite{Wolff1980,Joe1997,nelsen2007introduction}. Given random variables $\mathbf{X}=(X_1,\ldots,X_n)$ and their copula $C(\mathbf{u})$, Joe's multivariate version of Kendall's $\tau$ is defined as\cite{Joe1990}
\begin{equation}
	\tau_J = \frac{1}{2^{n-1}-1}\left\{2^n\int_{\mathbf{u}}C(\mathbf{u})dC(\mathbf{u})-1\right\}.
\end{equation}
The two multivariate versions of Spearman's $\rho$ is based on copula: Given random variables $\mathbf{X}=(X_1,\ldots,X_n)$ and their copula $C(\mathbf{u})$, Wolff's version is defined as\cite{Wolff1980}
\begin{equation}
	\rho_{W} = \frac{n+1}{2^n-(n+1)}\left\{2^n \int_{\mathbf{u}}C(\mathbf{u})d\mathbf{u}-1\right\},
\end{equation}
and Joe's \cite{Joe1997} and Nelson's \cite{nelsen2007introduction} are defined as
\begin{equation}
	\rho_{JN} = \frac{n+1}{2^n-(n+1)}\left\{2^n \int_{\mathbf{u}}u_1\cdots u_n dC(\mathbf{u})-1\right\}.
\end{equation}

Another important correlation coefficient is Gini's $\gamma$\cite{Gini1914} which is defined with rank statistic and has its copula representation\cite{nelsen2007introduction}.
\begin{theorem}\cite{nelsen2007introduction}
	Given random variables $(X,Y)$ and their copula $C(u,v)$, then Gini's $\gamma$ can be represented with copula as
	\begin{equation}
		\gamma_{X,Y}=2\int_{u}\int_{v}(|u+v-1|-|u-v|)dC(u,v).
		\label{eq:gini}
	\end{equation}	
\end{theorem}
Behboodian et al.\cite{Behboodian2007} presented a multivariate version of Gini's $\gamma$: given random variables $\mathbf{X}=(X_1,\ldots,X_n)$ and their copula $C(\mathbf{u})$, the definition is
\begin{equation}
	\gamma_{B} = \frac{1}{b(n)-a(n)} \left\{\int_{\mathbf{u}}(A(\mathbf{u})+\bar{A}(\mathbf{u}))dC(\mathbf{u})-a(n)\right\},
\end{equation}
where $A(\mathbf{u})=\{M(\mathbf{u})+W(\mathbf{u})\}/2$, $\bar{A}$ is survival function of $A$, $M,W$ are Fr\'echet upper and lower bound functions, $b(n),a(n)$ are normalized terms.

Schweizer and Wolff\cite{Schweizer1981} defined a copula based dependence measure $\sigma$ which is defined as follows:
\begin{definition}[Schweizer and Wolff's $\sigma$]
	\label{def:sigma}
	Given random variables $(X,Y)$ and their copula $C(u,v)$, Schweizer and Wolff's $\sigma$ is defined as
	\begin{equation}
		\sigma_{X,Y}=12\int_u \int_v |C(u,v)-uv| dudv.
	\end{equation}
\end{definition}
This measure can be generalized to multivariate cases: given random variables $\mathbf{X}=(X_1,\ldots,X_n)$ and their copula  $C(\mathbf{u})$, multivariate version of Schweizer and Wolff's $\sigma$ \cite{Wolff1980} is defined as
\begin{equation}
	\sigma_W = \frac{2^n(n+1)}{2^n-(n+1)}\int_{\mathbf{u}}|C(\mathbf{u})-\prod_{i=1}^{n}{u_i}|d\mathbf{u}.
\end{equation}

R\'enyi\cite{Renyi1959} proposed the famous axioms for an ideal dependence measure $\delta$, including
\begin{enumerate}
	\item $\delta$ is defined on a pair of random variables $(X,Y)$;
	\item $\delta_{X,Y}=\delta_{Y,X}$;
	\item $0\leq \delta_{X,Y} \leq 1$;
	\item if $X,Y$ are independent, then $\delta_{X,Y}=0$;
	\item if $X,Y$ has monotonic functional relationship, then $\delta_{X,Y}=1$;
	\item if $f,g$ are monotonic functions on $X,Y$, then $\delta_{X,Y}=\delta_{f(X),g(Y)}$;
	\item if $(X,Y)$ are governed by Gaussian distribution, then $\delta_{X,Y}=|\rho_{X,Y}|$, where $\rho_{X,Y}$ is correlation coefficient.
\end{enumerate}
Schweizer and Wolff\cite{Schweizer1981} revised R\'enyi's axioms and proved that their $\sigma$ satisfies the revised axioms. Schmid et al.\cite{Schmid2010} discussed several common dependence measures' axiomatic properties.

By comparing the properties in Section \ref{s:properties} with R\'enyi's axioms, one can learn that CE satisfies 2, 4, 6 out of R\'enyi's seven axioms, and the differences include 1) CE is multivariate; 2) CE does not satisfy 3, 5 but is non-positive with equality if and only if independent; 3) CE does not satisfy 7 but is equivalent to correlation coefficient matrix as \eqref{eq:cevx}. Additionally, CE has the properties of Shannon entropy, such as additivity. In a word, CE is an ideal association measure with good axiomatic properties. 

Association between random vectors is a common problem in real applications. There has been many work on this, one can refer to \cite{Josse2016} for more. This problem can be solved easily with CE as stated in the following theorem:
\begin{theorem}
	Given a group of random vectors $\mathbf{X}=\{\mathbf{X}_1,\ldots,\mathbf{X}_n\}$, $\mathbf{X}_1\in R^{d_1},\ldots,\mathbf{X}_n\in R^{d_n}$, $d_1,\ldots,d_n \geq 1$, then CE based association measure for random vectors is
	\begin{equation}
		H_c(\mathbf{x}_1;\ldots;\mathbf{x}_n)=H_c(\mathbf{x}) - \sum_{i=1}^{n}{H_c(\mathbf{x}_i)}.
	\end{equation}
\end{theorem}
\begin{proof}
	\begin{align}
		H_c(\mathbf{x}_1;\ldots;\mathbf{x}_n)&=-I(\mathbf{x}_1;\ldots;\mathbf{x}_n)\\
		&=-\int_{\mathbf{x}}p(\mathbf{x})\log\frac{p(\mathbf{x})}{p(\mathbf{x}_1)\cdots p(\mathbf{x}_n)}d\mathbf{x}\\
		&=-I(\mathbf{x})+\sum_{i=1}^{n}{I(\mathbf{x}_i)}\\
		&=H_c(\mathbf{x}) - \sum_{i=1}^{n}{H_c(\mathbf{x}_i)}.
	\end{align}
Here if $d_i=1$, then $H_c(\mathbf{x}_i)=0$.
\end{proof}
This theorem can be intuitively understood as that association between random vectors equals the difference between all the information and information in each vector. It is not only easy to understand, but also easy to implement with CE estimators.

In summary, CE is a perfect association measure for any cases, not only for random variables but for random vectors. It is natural and unique. It enjoys theoretical advantages by satisfying axiomatic properties and can be easily estimated for practical uses.

Due to CE's advantages, Ma \cite{jian2019discovering} proposed to use it for association discovery. We verified its effectiveness and compared it with other measures on the famous NHANES data \footnote{The code is available at \url{https://github.com/majianthu/nhanes}}. Experimental results (see Figure \ref{fig:nhanes}) show that CE can discover meaningful nonlinear relationships within data and present better results than the others. For more comparisons between independence measures, please see Section \ref{s:indbench}.

\begin{figure}
	\centering
	\subfigure[Pearson's $r$]{\includegraphics[width=0.45\textwidth]{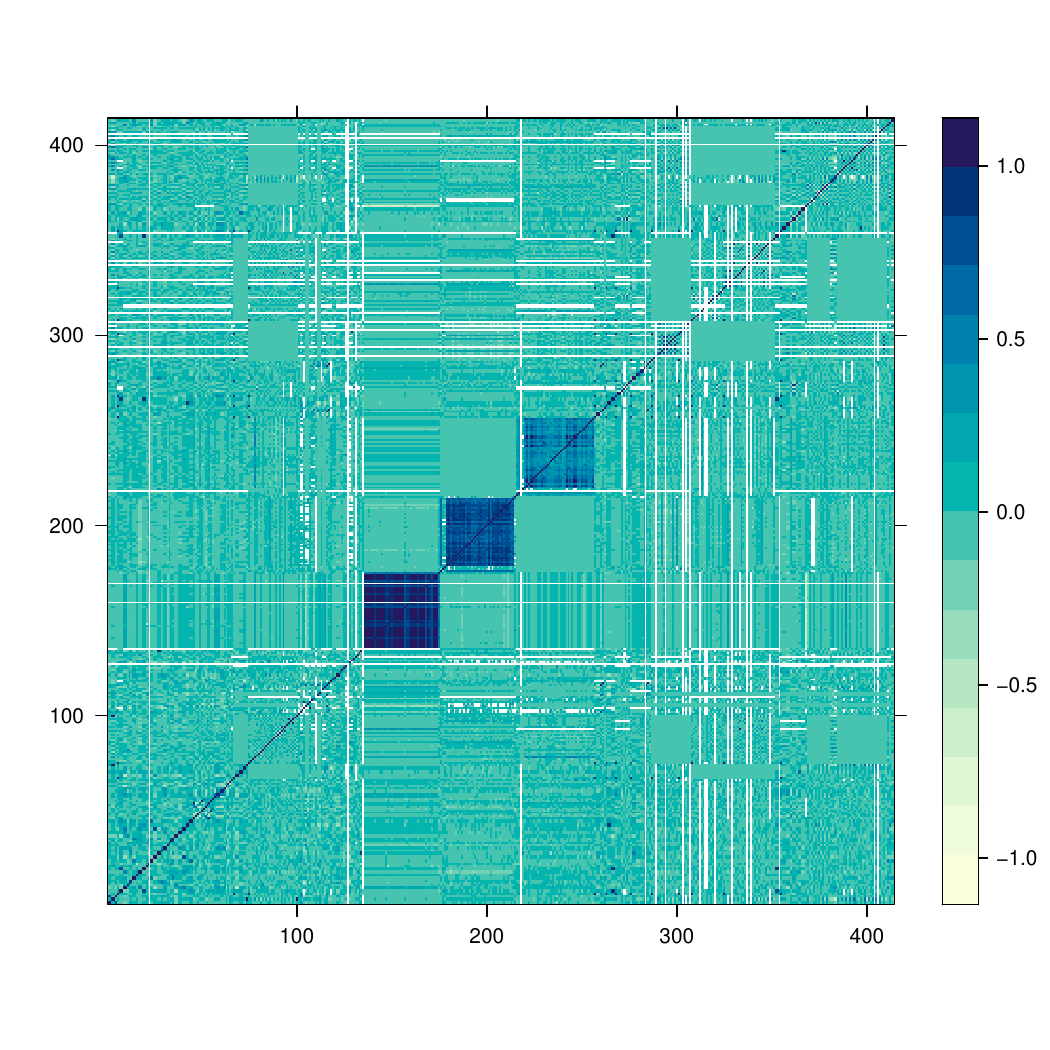}}
	\subfigure[Spearman's $\rho$]{\includegraphics[width=0.45\textwidth]{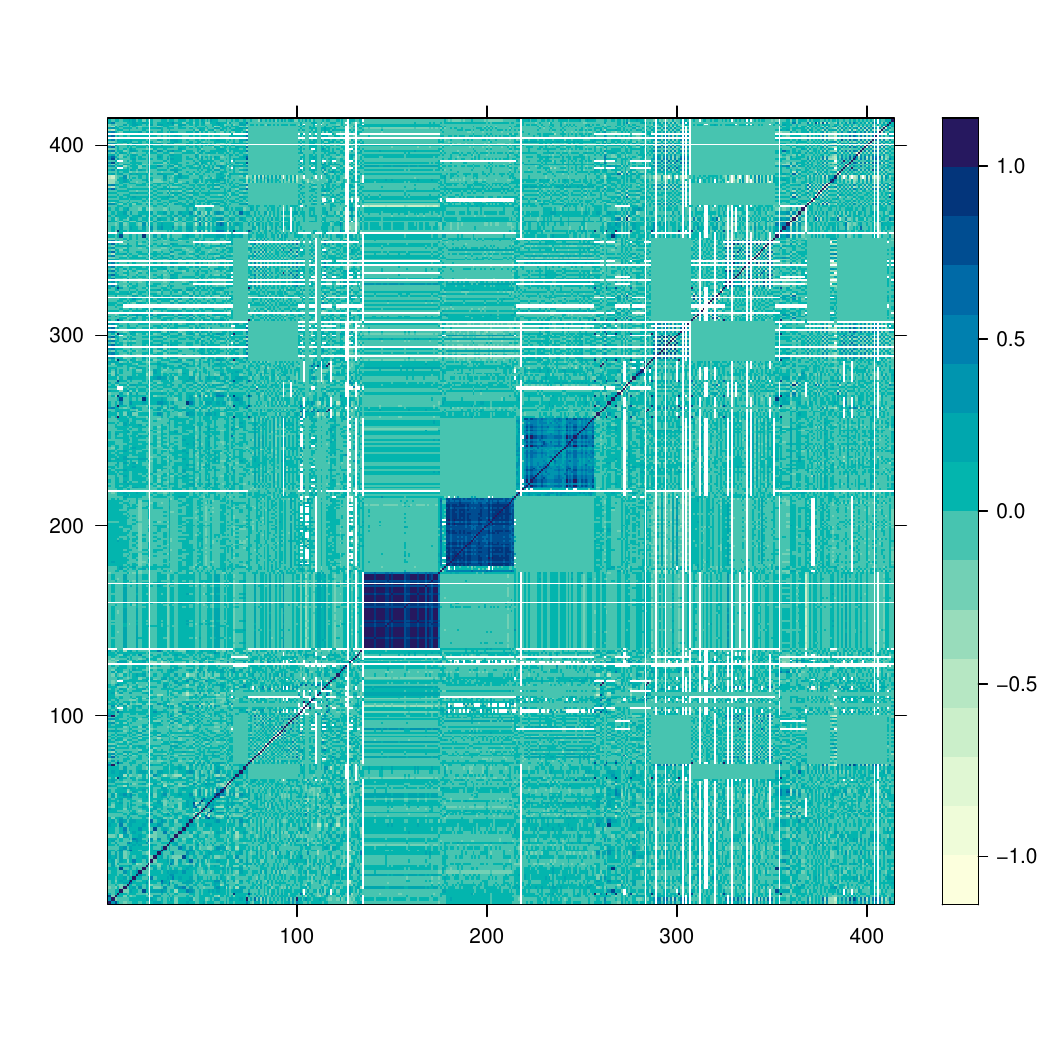}}
	\subfigure[Kendall's $\tau$]{\includegraphics[width=0.45\textwidth]{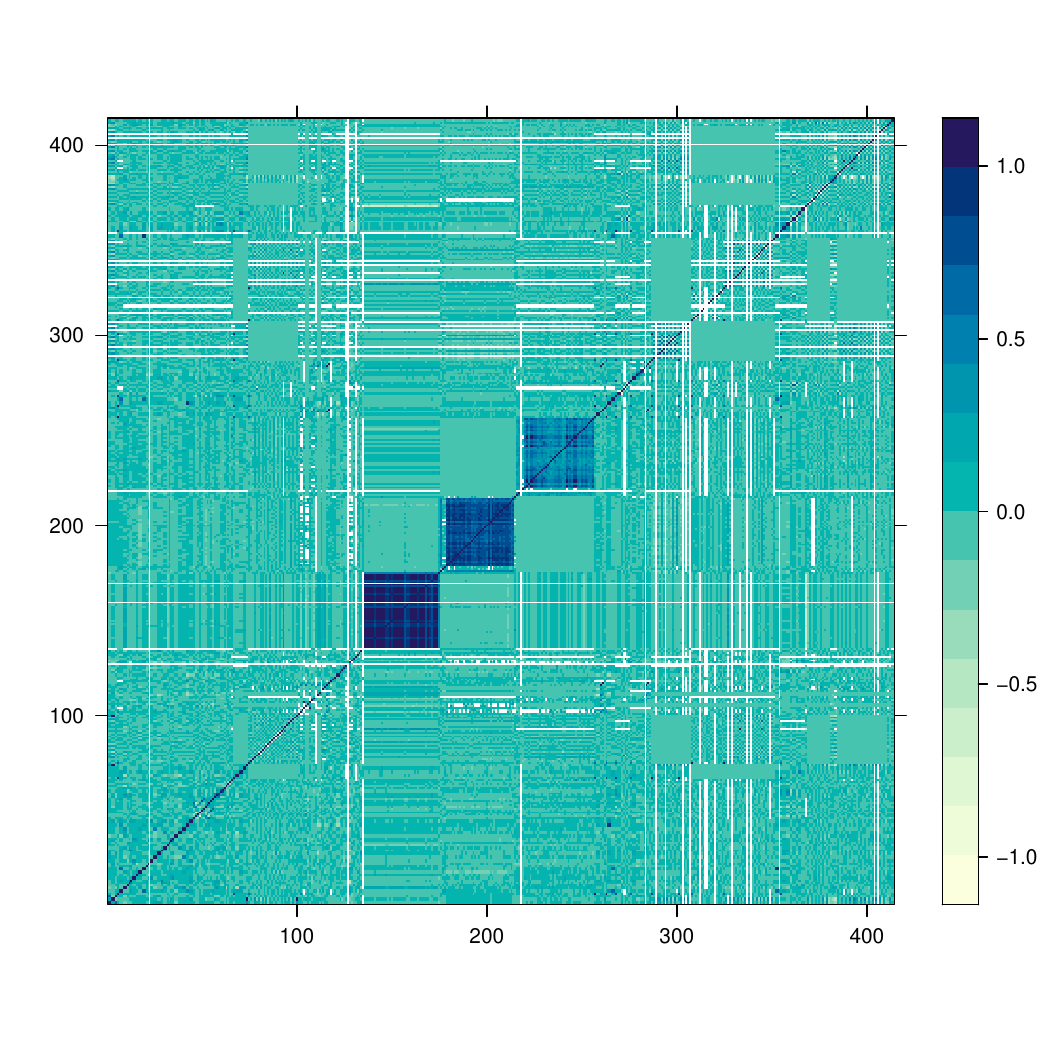}}
	\subfigure[Gini's $\gamma$]{\includegraphics[width=0.45\textwidth]{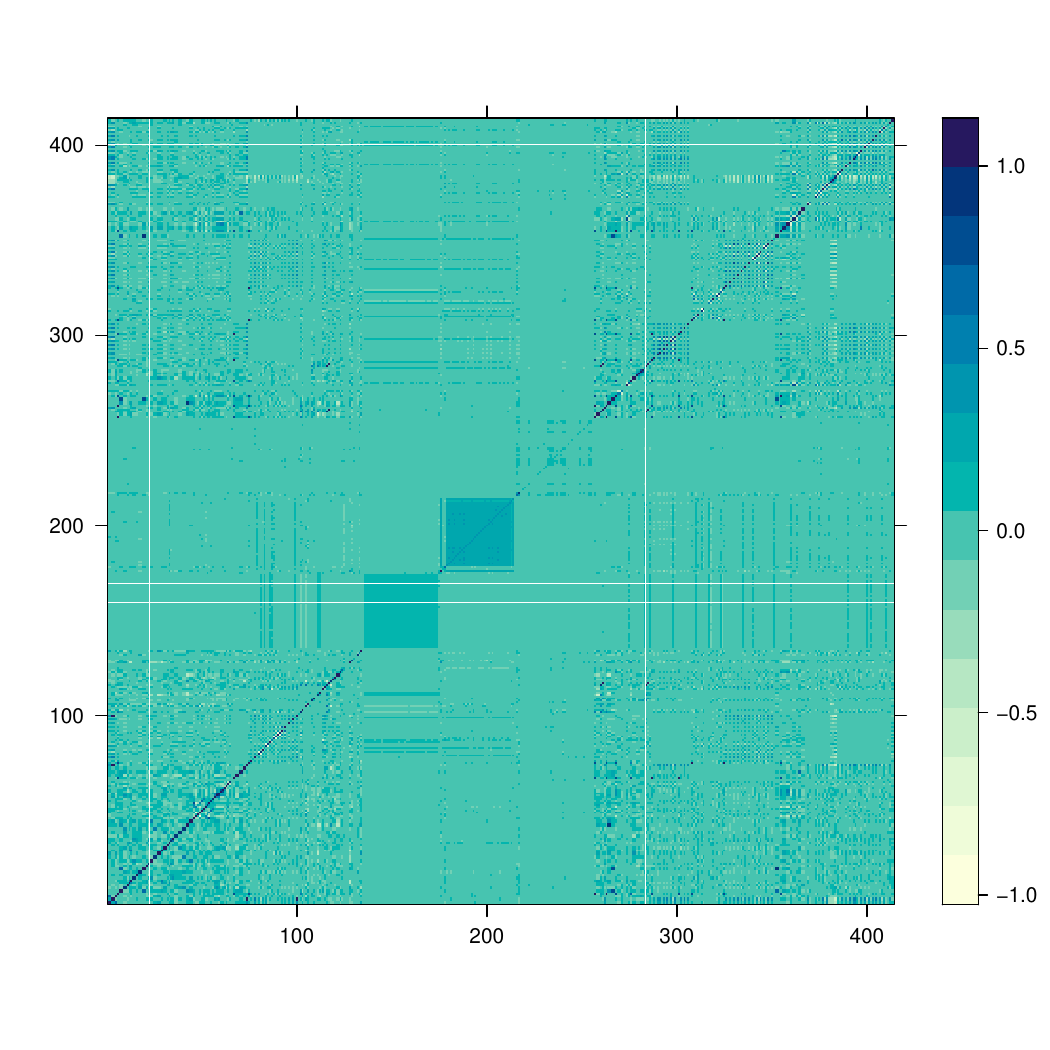}}
	\subfigure[Schweizer \& Wolff's $\sigma$]{\includegraphics[width=0.45\textwidth]{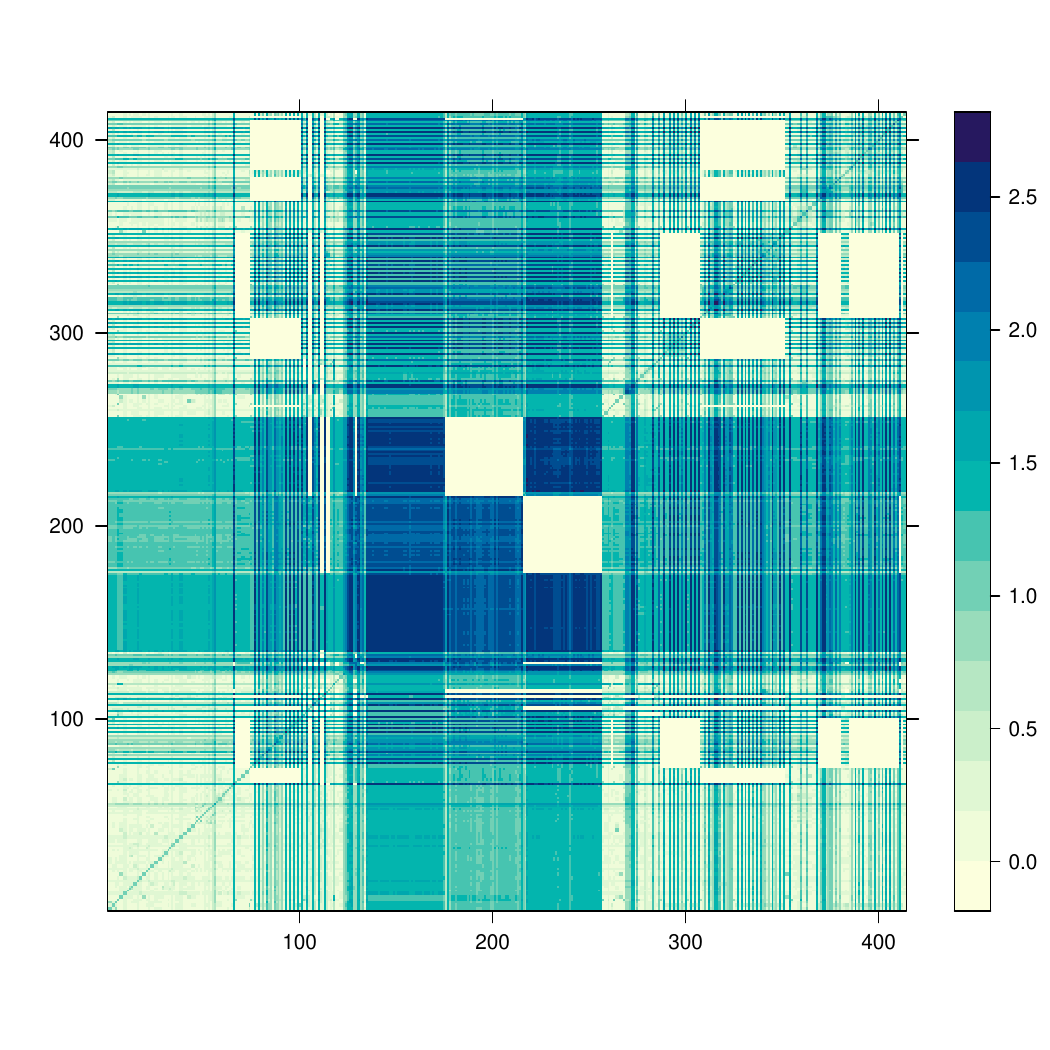}}
	\subfigure[CE]{\includegraphics[width=0.45\textwidth]{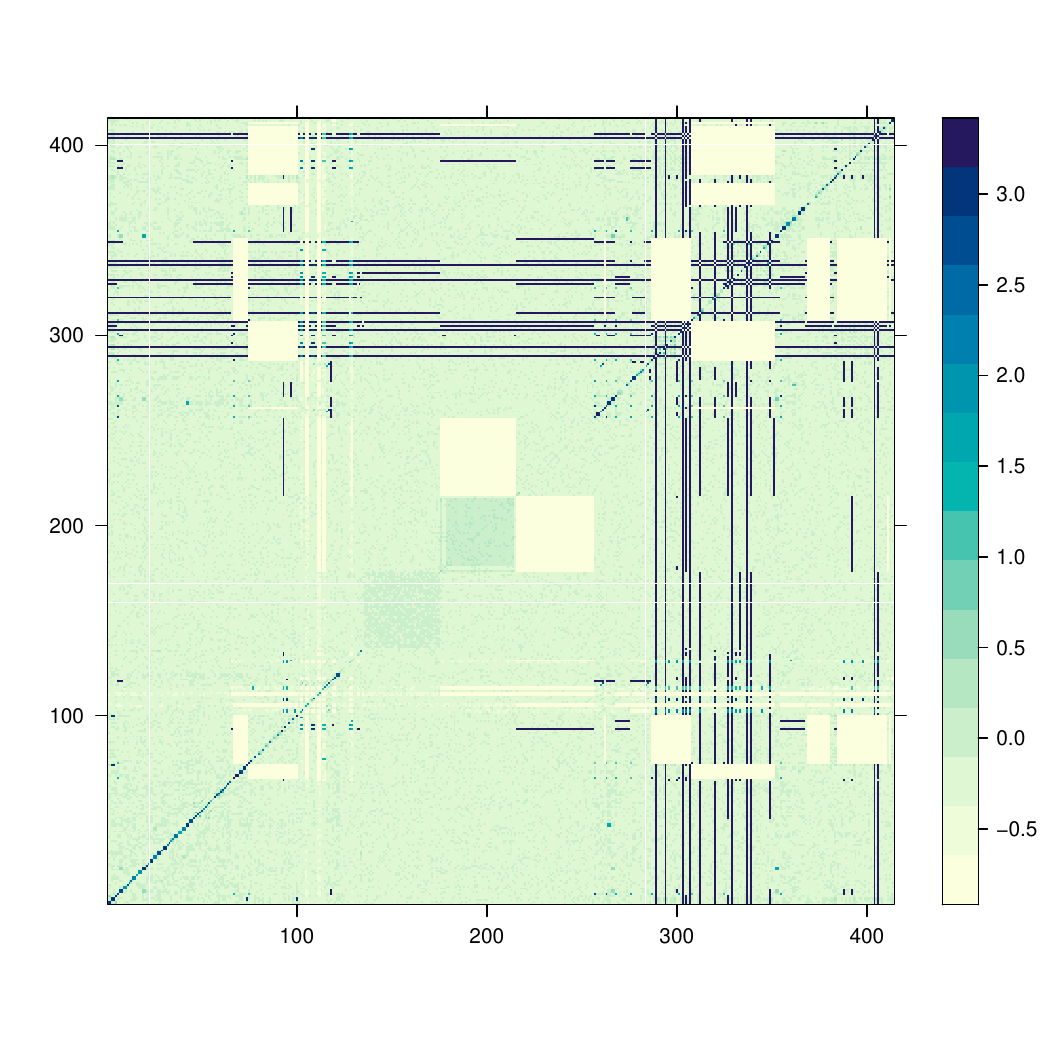}}
	\caption{Experimental results on the NHANES data with the six association measures.}
	\label{fig:nhanes}
\end{figure}

\section{Variable selection}
\label{s:vs}

Variable selection, also called feature selection, is an important problem in statistics and machine learning\cite{george2000the,guyon2003an}. When building functional relations between variables and target variables, one tries to select those variables related to target variables as inputs of function to make the functional model more sound or interpretable and to lower model complexity. This problem is called variable selection.

There are many traditional variable selection methods, including criteria based, regularization based, and measure based ones. The criteria based ones include AIC\cite{akaike1974a,Akaike1998}, BIC\cite{schwarz1978estimating} and  their variants which are derived by adding likelihood with penalty on model complexity. The regularization based ones are mainly based generalized linear models and derived by adding likelihood with L1 or L2 norms penalties, including Lasso\cite{tibshirani1996regression}, ridge regression\cite{hoerl2000ridge}, and elastic net\cite{zou2005regularization}. These two types of methods can based on penalized likelihoods and are based on models. The measure based ones select variables based on the association measures between variables and target variables and therefore are model-free, including linear measures, such as PCC, and nonlinear measures, such as Hilbert-Schmidt independence criteria (HSIC)\cite{gretton2007a,pfister2018kernel} and distance correlation (dCor)\cite{Szekely2007,Szekely2009}, etc.

Ma\cite{jian2019variable} proposed to use CE for variable selection by selecting variables according to association measures between variables and target variables. The larger absolute CE, the better the variable. It has been demonstrated better to the following variable selection methods:
\begin{itemize}
\item LASSO / Ridge Regression / Elastic Net \cite{tibshirani1996regression,hoerl2000ridge,zou2005regularization},
\item AIC / BIC \cite{akaike1974a,Akaike1998,schwarz1978estimating},
\item Adaptive LASSO \cite{zou2006the},
\item Hilbert-Schimdt Independence Criterion (HSIC) \cite{gretton2007a,pfister2018kernel},
\item dCor \cite{Szekely2007,Szekely2009},
\item Heller-Heller-Gorfine tests of independence \cite{heller2016consistent},
\item Hoeffding's D test \cite{hoeffding1948a},
\item Bergsma-Dassios T* sign covariance \cite{bergsma2014a},
\item Ball correlation \cite{pan2020ball}.
\end{itemize}

Experiments\footnote{The code is available at \url{https://github.com/majianthu/aps2020}} with the proposed method was conducted on the UCI heart disease dataset\cite{asuncion2007uci}. The dataset contains clinical biomedical measurement and diagnosis results on the patients from four sites worldwide for heart disease research. Some of the biomedical measurements have been identified by clinical experts as relevant to heart disease, which can be considered as gold standard for benchmarking variable selection methods. Experimental results show that, CE based variable selection method selects more biomedical variables identified by experts than the other comparative methods and presents better prediction performance and much interpretable results than others. Some of the results are shown in Figure \ref{fig:varsel}.

\begin{figure}
	\centering
	\subfigure[CE]{\includegraphics[width=0.8\textwidth]{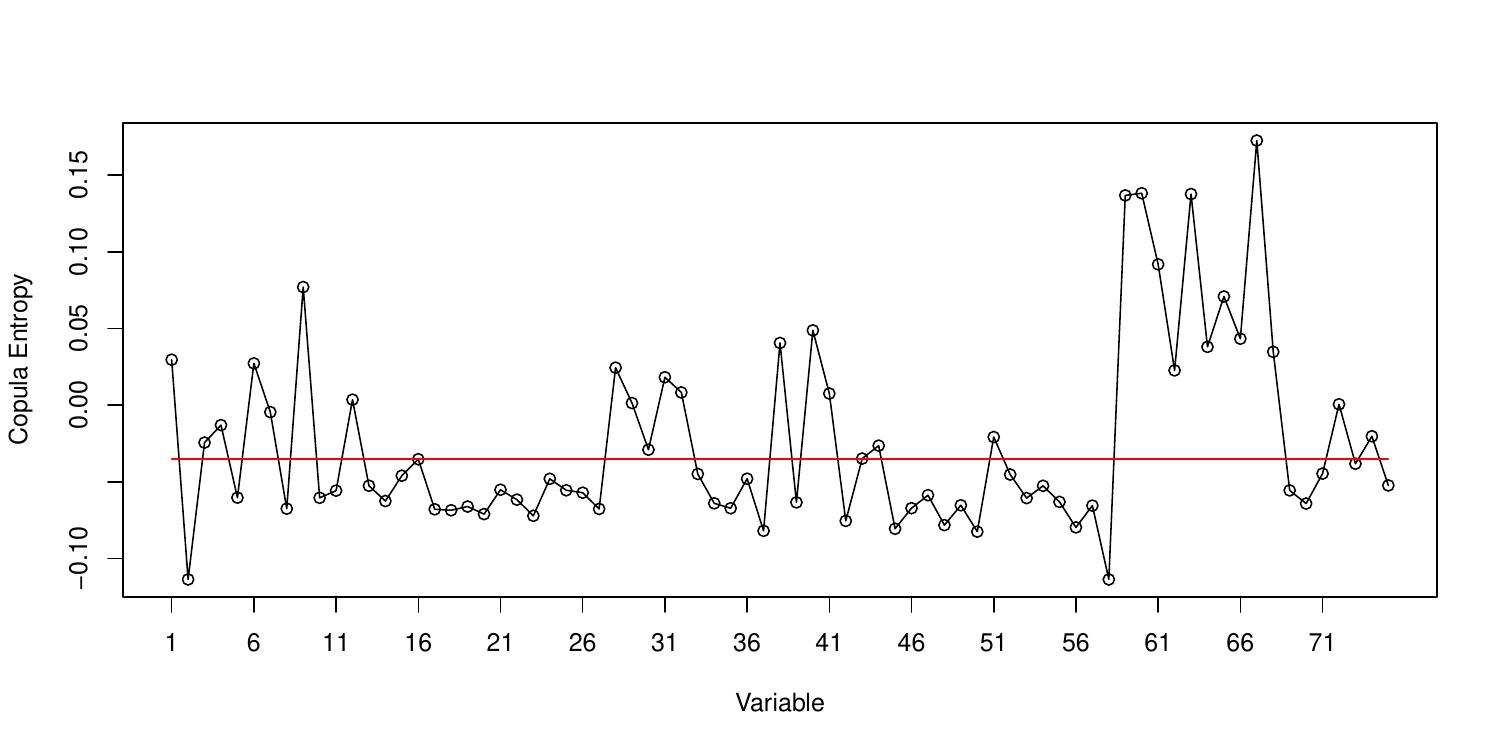}}	
	\subfigure[dCor]{\includegraphics[width=0.8\textwidth]{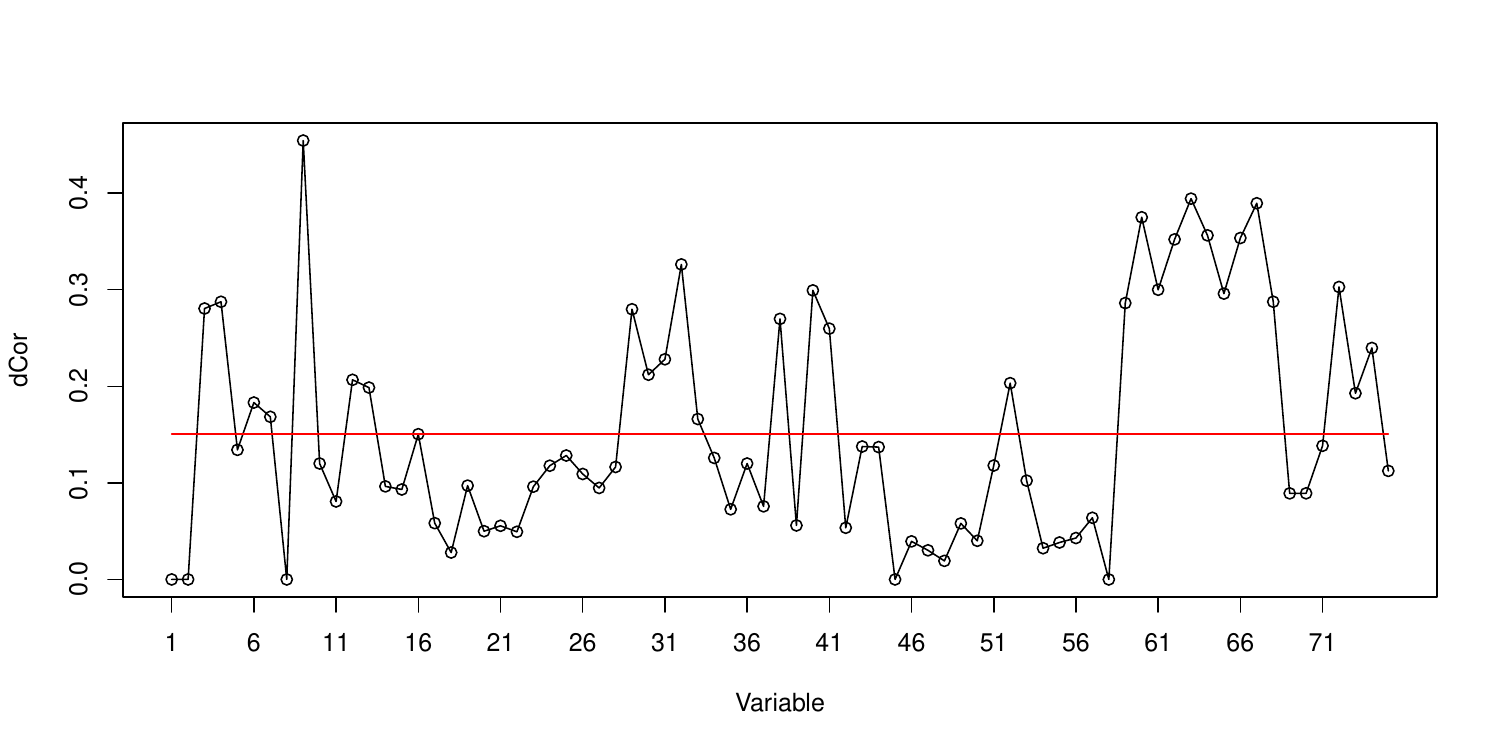}}	
	\subfigure[dHSIC]{\includegraphics[width=0.8\textwidth]{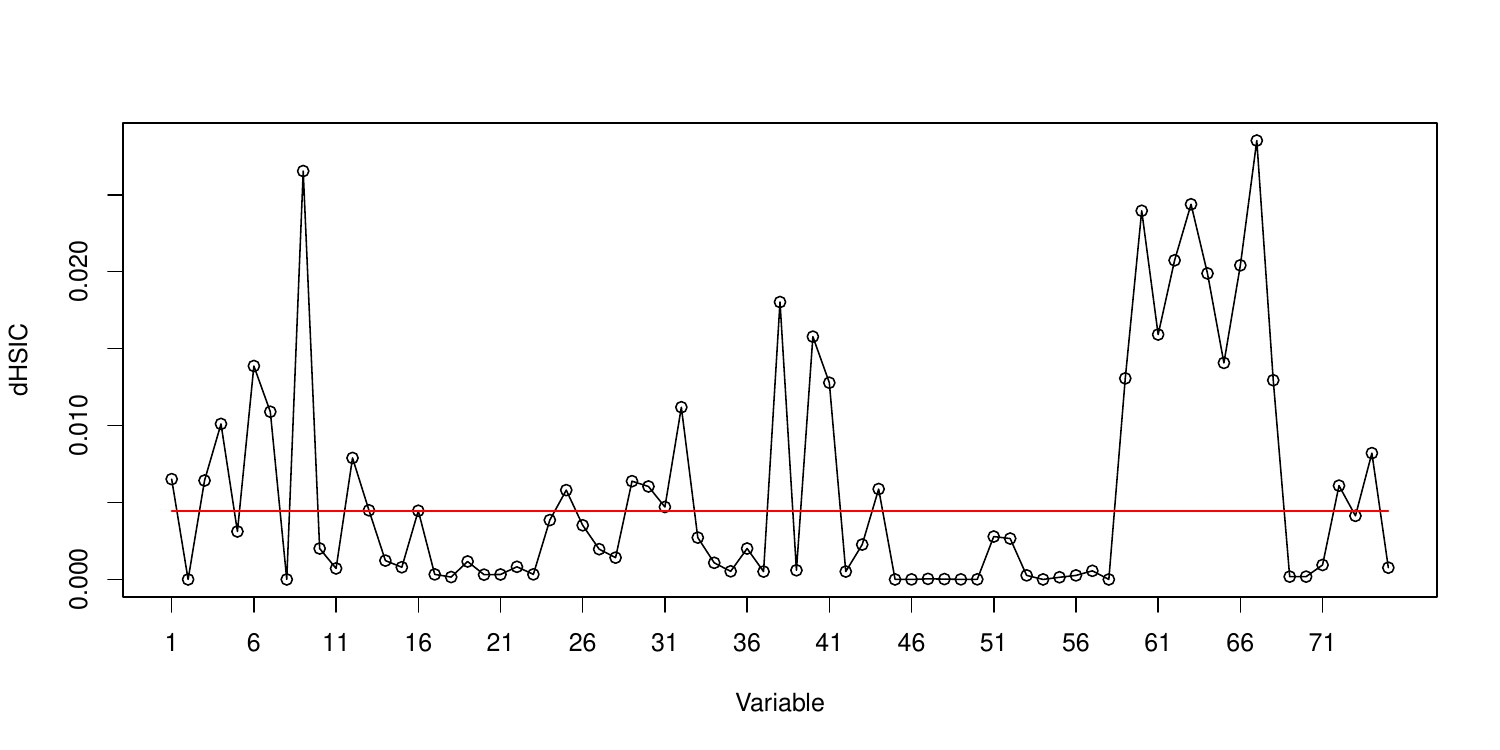}}	
	\caption{Variables selected with three main dependence measures.}
	\label{fig:varsel}
\end{figure}

CE provides a unified framework for variable selection with the following advantages:
\begin{itemize}
\item model-free;
\item mathematically rigorous and sound;
\item physically interpretable;
\item with nonparametric estimator;
\item tuning free.
\end{itemize}

In contrast to the likelihood based ones, the CE based method is model-free, and therefore universally applicable. Compared with other measures, CE has a rigorous definition and many good properties and therefore enjoys theoretical advantages. Meanwhile, entropy is a concept with physical meanings, CE measures the information content of dependence between variables and target variables so is interpretable. CE also has a good nonparametric estimator which makes the proposed method applicable to real problems without assumption. The CE estimator is almost tuning-free, which makes it more advantageous than the methods need tuning, such as Lasso.

Survival analysis \cite{Kleinbaum2012} is a special type of regression problem aiming to predict time-to-event, i.e., time needed for an event happens in the future. This problem is special also due to the censorship mechanism that is used when events does not happen in the observation time window. Survival analysis has many applications in medicine, reliability, and social sciences, etc. Variable selection is also needed for building survival models. Ma\cite{Ma2022a} proposed to use CE for variable selection in survival models. He applied the proposed method to two open cancer datasets, and compared it with random survival forest and Lasso-Cox. Experimental results show that the proposed method presents better prediction performance and much interpretable models\footnote{The code is available at \url{https://github.com/majianthu/survival}}.

\section{Causal discovery}
\label{s:ci}
Causality is ubiquitous in natural and social world and causal discovery is a core philosophical problem with a long history\cite{Beebee2009,Mumford2013,Illari2014} and also a fundamental problem in every branch of sciences\cite{Pearl2009}. Causal discovery is to find causal relationships between random variables from time series data, and is a main problem in time series analysis\cite{Glymour2019,Nogueira2022,Vowels2022,Zanga2022}. The methods for finding causal relationships are usually inferring causal structure with causality principles\cite{Spirtes2000}. 

How to judge causal relations is of fundamental importance in causal discovery and measures of causality is the core components of causal learning algorithms. Wiener proposed a philosophical principle for judging causal relations which states that causal variable should be able to improve the prediction of effects\cite{wiener1956theory}. Further on this, Granger\cite{granger1969investigating,granger1980testing} proposed his principle of causality, called Granger Causality (GC), which is define as follows:
\begin{definition}[Granger Causality]
	\label{def:gc}
	Given time series variables $X_t,Y_t$, when predicting $Y_{t+1}$ with function $f_p$, if adding $X_t$ into inputs of $f_p$ can make the variance of prediction become smaller, i.e.,
	\begin{equation}
		var[y_{t+1}-f_p(y_{t+1}|y_t)] > var[y_{t+1}-f_p(y_{t+1}|y_t,x_t)],
		\label{eq:gc}
	\end{equation}
	then $X$ and $Y$ has causal relationship in Granger's sense, and $X$ is the cause of the effect $Y$.
\end{definition}
GC test is a classical tool for causality. However, it is easy to learn from Definition \ref{def:gc} that GC is only fitful for linear cases with Gaussian assumptions due to variance in it \cite{Shojaie2022}, and therefore limit its wide applications.

Schreiber\cite{schreiber2000measuring} defined transfer entropy (TE) for discovering causality in stationary time series data, as follows:
\begin{definition}[Transfer Entropy]
	Given time series variables $X_t,Y_t,t=1,\ldots,T$, transfer entropy from $X$ to $Y$ is defined as
	\begin{equation}
		TE_{X\rightarrow Y}=\sum_t{ p(y_{t+1},y_t,x_t) \log \frac{p(y_{t+1}|y_t,x_t)}{p(y_{t+1}|y_t)}}.
		\label{eq:tedef}
	\end{equation}
\end{definition}
As a measure of causality, TE is essentially CMI in information theory and can be used to test and measure CI, since \eqref{eq:tedef} can be transformed into CMI:
\begin{equation}
	TE_{X\rightarrow Y}=I(y_{t+1};x_t|y_t).
	\label{eq:tecmi}
\end{equation}

TE is considered as the nonlinear generalization of GC and is proved to be equivalent to GC in Gaussian cases\cite{Barnett2009}. TE measures the uncertainty reduction of effect due to adding causal variables and is model-free and applicable to any cases of causal discovery. Amblard and Michel\cite{Amblard2013} discussed the relationships between GC and directed information and presented the MI representation of directed information.

A key issue when using TE for causal discovery is how to estimate TE from time series data. Traditional TE estimations are based on entropy based expansion of TE and done with parametric or nonparametric entropy estimation methods \cite{Hlavackova-Schindler2007}. Rozo et al.\cite{Rozo2021} evaluated those entropy based TE estimation methods. Barnett and Bossomaier \cite{Barnett2012} proposed to estimate TE with log-likelihood ratio statistic.

CE is a kind of independence measure while TE is a measure of CI. Ma \cite{jian2019estimating} proved that TE can be represented with only CE which is stated as the following theorem:

\begin{theorem}
	\label{thm:tece}
	TE can be represented with only CE, i.e.,
	\begin{equation}
		TE_{X\rightarrow Y}=H_c(y_{t+1},y_t)+H_c(y_t,x_t)-H_c(y_{t+1},y_t,x_t).
		\label{eq:te}
	\end{equation}
\end{theorem}
\begin{proof}
\begin{align}
	TE_{X\rightarrow Y} &= \sum_t{p(y_{t+1},y_t,x_t)\log{\frac{p(y_{t+1}|y_t,x_t)}{p(y_{t+1}|y_t)}}}\\
	&=\sum_t{p(y_{t+1},y_t,x_t)\log{\frac{p(y_{t+1},y_t,x_t)p(y_t)}{p(y_{t+1},y_t)p(y_t,x_t)}}}\\
	&=I(y_{t+1},y_t,x_t) - I(y_{t+1},y_t) - I(y_t,x_t)\\
	&=-H_c(y_{t+1},y_t,x_t) + H_c(y_{t+1},y_t) + H_c(y_t,x_t). \label{eq:three}
\end{align}
\end{proof}
Since TE is essentially a measure of CI relation $Y_{t+1}\perp\!\!\!\perp X_t|Y_t$, Theorem \ref{thm:tece} is same as Theorem \ref{thm:cmi}.

Based on CE representation of TE, Ma proposed a nonparametric TE estimation method based on the CE estimators. This makes causal discovery without assumption possible. The proposed TE estimator composed of two steps\footnote{This method has been implemented in the \texttt{copent} packages\cite{ma2021copent} in \textsf{R} and \textsf{Python}.}:
\begin{enumerate}
\item estimate three CE terms in \ref{eq:te} with nonparametric CE estimator;
\item compute TE from the estimated CEs.
\end{enumerate}

To verify the effectiveness of the nonparametric TE estimator, we applied it to the Beijing air pollution problem. We try to estimate causal relations between meteorological factors and PM2.5 based on the UCI Beijing PM2.5 air dataset \footnote{The code is available at \url{https://github.com/majianthu/transferentropy}}. The dataset used contains hourly meteorological observation and PM2.5 observations in Beijing from 2010 to 2014. We select a period of data without missing values and infer causal strength from meteorological factors to PM2.5 with 1-24 hours time lags. With such experiment design, we implicitly assume the stationarity of observation data and also Markovianity. Experimental analysis results show that causal effects from meteorological factors, such as temperature and pressure, to PM2.5 is increasing sharply at first and slow gradually. 

On the same data, we compared our TE method with other two CI measures, i.e., kernel base CI (KCI)\cite{zhang2011kernel} and conditional distance correlation (CDC) \cite{wang2015conditional}. The comparison results are shown in Figure \ref{fig:causality} which demonstrate the advantage of the TE estimator over the other two methods. For more benchmarking on CI measures, please refer to Section \ref{s:cibench}.

\begin{figure}
	\centering
	\subfigure[TE]{\includegraphics[width=0.8\textwidth]{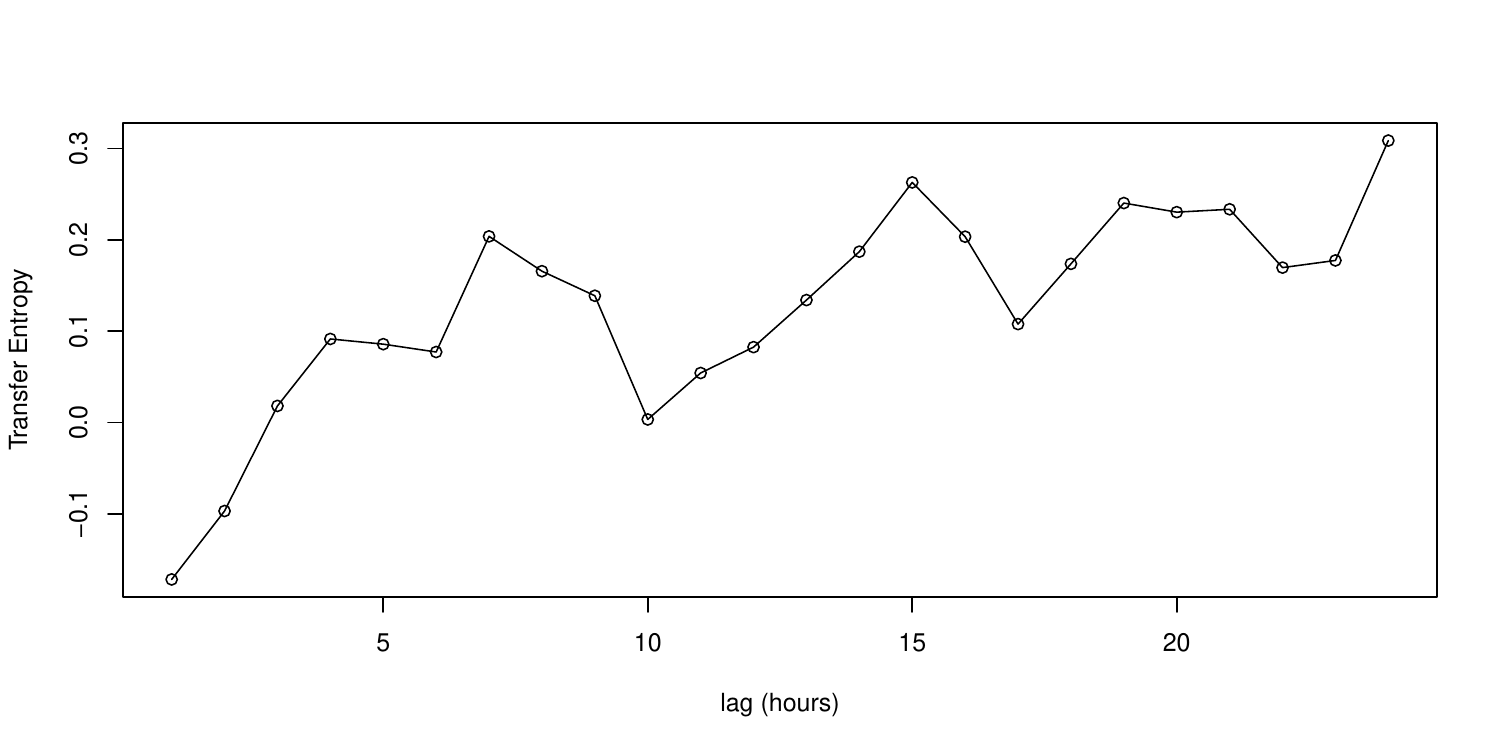}}	
	\subfigure[KCI]{\includegraphics[width=0.8\textwidth]{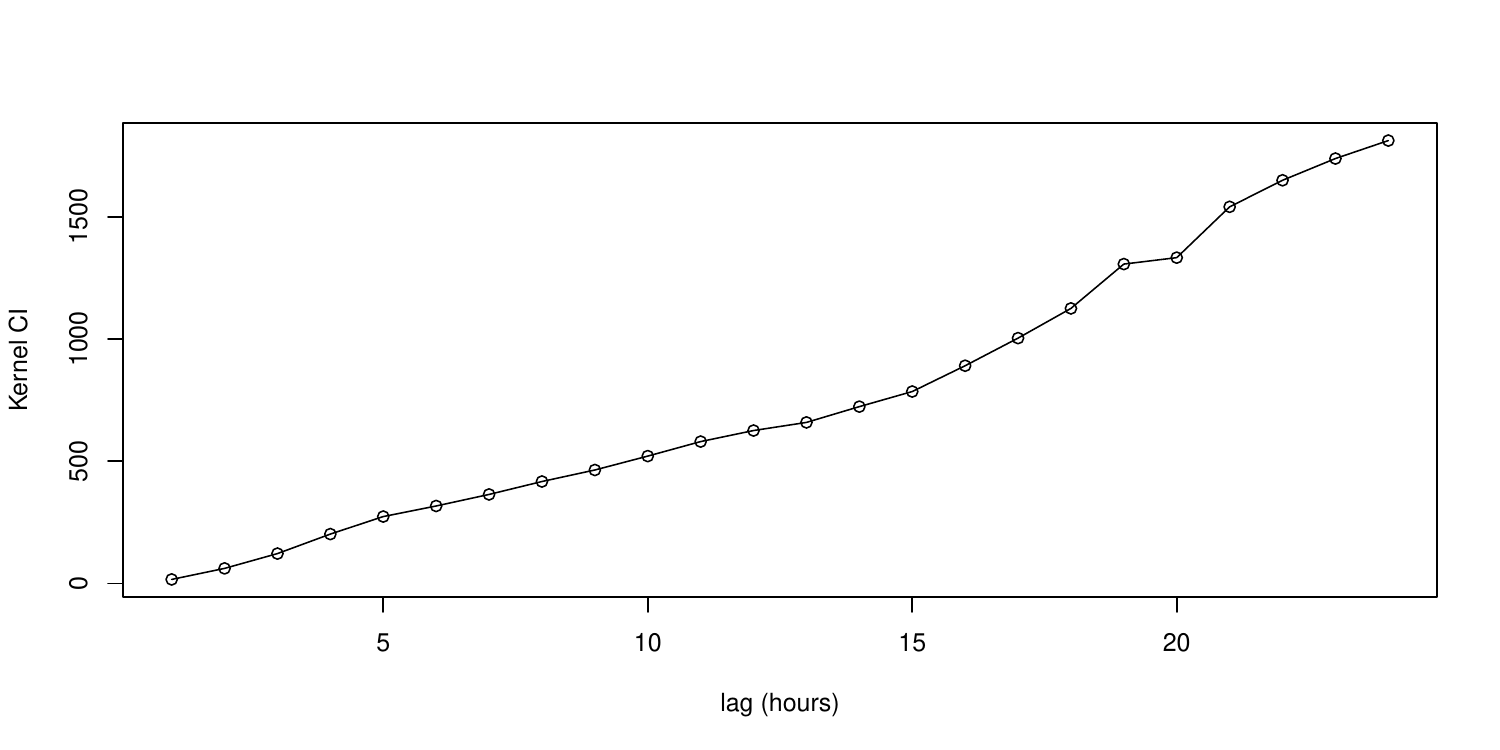}}	
	\subfigure[CDC]{\includegraphics[width=0.8\textwidth]{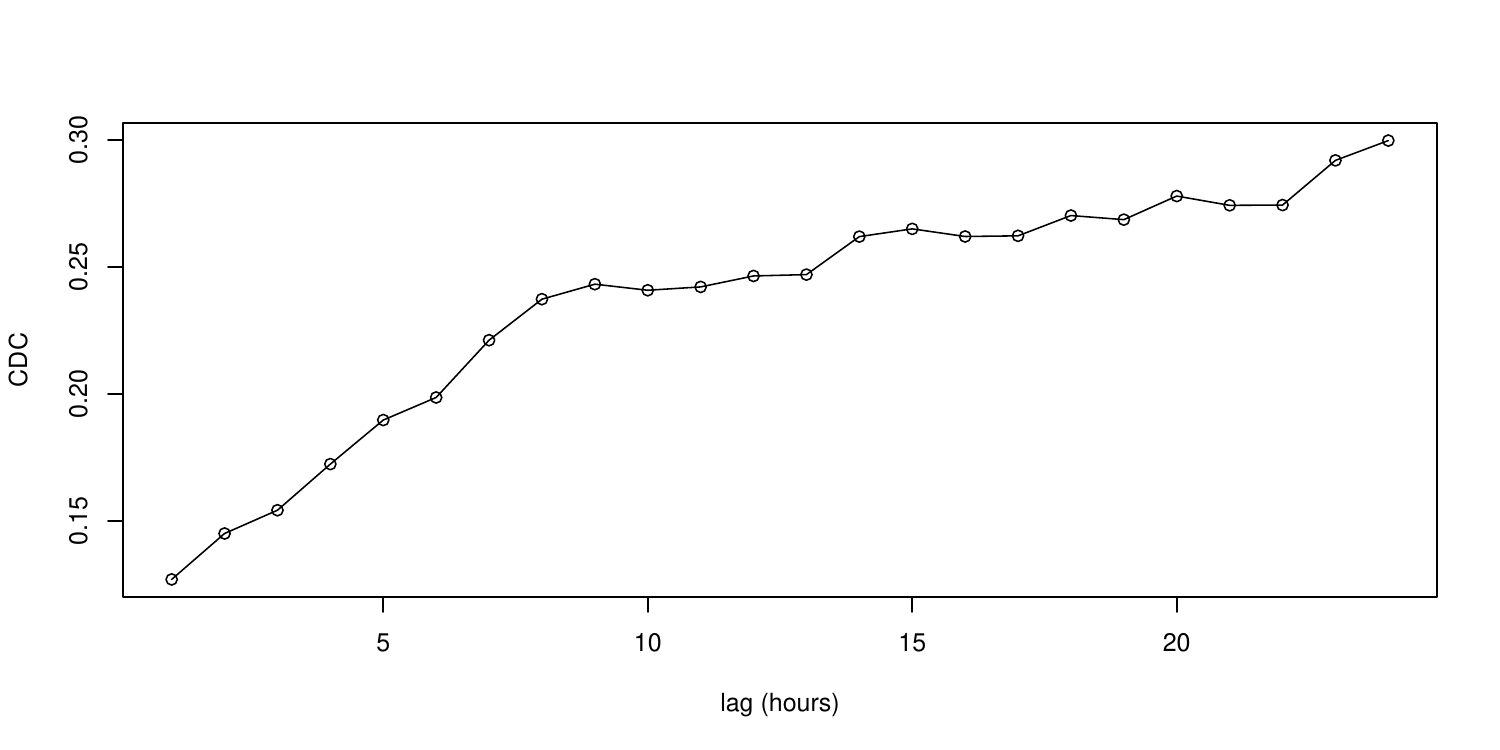}}	
	\caption{Causal strength from pressure to PM2.5 estimated with three measures.}
	\label{fig:causality}
\end{figure}

\section{System identification}
Differential equations are the mathematical tools for modeling dynamical systems and are widely applied to different scientific fields. Discovering differential equations from data is an important problem in this area, also called system identification or equation discovery\cite{Chiuso2019,Camps-Valls2023,North2023}, has gained a lot of attention recently\cite{Kaptanoglu2023a,Champneys2024a}.

Equation discovery can be treated as a regression problem, i.e., building a regression model from system states to derivative. Given a general differential equation:
\begin{equation}
	\frac{dx_i}{dt}=f_i(\mathbf{x},t),
\end{equation}
where $x_i,i=1,\ldots,n$ is system state variables, then equation discovery is to identify $f_i$. Such identification needs to choose variables of equation, which is a typical variable selection problem. There have been many existing regression methods proposed for this problem, such as sparsity based SINDy\cite{Brunton2016} and kernel based method\cite{Pillonetto2014}.

Entropy has a long history in differential equation area\cite{Evans2004}. It is used to measure ``randomness" in chaotic systems, such as the famous Kolmogorov-Sinai entropy\cite{Sinai2009}. Nardone and Sonnino\cite{Nardone2024} proposed an entropy of difference for describing time series complexity. MI has also been applied to system identification. For instance, Chernyshov and Jharko\cite{Chernyshov2018} proposed to use Tsallis MI for identifying the correlation between inputs and outputs of dynamical systems. Stoorvogel and van Schuppen\cite{Stoorvogel1995} proposed to identify linear dynamical systems by estimating MI rate between error and white noise.

Ma\cite{Ma2023a} proposed a CE based method for equation discovery which treats the problem as variable selection and solves it with the CE based method in Section \ref{s:vs}. The proposed method is essentially using entropy to measure the ``randomness" between states and derivative to identify dynamical systems, which composed of two steps:
\begin{enumerate}
	\item estimate derivative with difference operator;
	\item select state variable by estimating the CEs between states and derivative.
\end{enumerate}
The derivative can be estimated in such a nonparametric way:
\begin{equation}
	\frac{dx}{dt}\bigg|_{t=t_0}\approx \frac{x_{t_1}-x_{t_0}}{t_1-t_0}.
\end{equation}
Meanwhile, CE can also be estimated with the nonparametric estimator. So the proposed method is nonparametric without assumptions and can be applied to any system.

We applied the proposed method to two typical dynamical systems: Lorenz system\cite{Lorenz1960} and R\"ossler system\cite{Roessler1976,Rossler1979}, both of which have 3 equations with first and second order state variables as follows:
\paragraph{Lorenz system}
\begin{equation}
	\begin{aligned}
		\frac{dx}{dt} &= \sigma (y-x),\\
		\frac{dy}{dt} &= \rho x - y - xz,\\
		\frac{dz}{dt}&= -\beta z + xy,\\
	\end{aligned}
	\label{eq:lorenz}
\end{equation}
where $\sigma,\rho,\beta$ are Prandtl number,  Rayleigh number, and geometrical factor respectively and are set as $10,28,8/3$ in the experiments.

\paragraph{R\"ossler system}
\begin{equation}
	\begin{aligned}
		\frac{dx}{dt} &= -(y+z),\\
		\frac{dy}{dt} &= x + ay, \\
		\frac{dz}{dt}&= b + z(x-c),\\
	\end{aligned}
	\label{eq:rossler}
\end{equation}
where $a,b,c$ are system parameters and set as $0.38,0.2,5.7$ in simulation experiments.

We generated the simulated data of these two systems and applied the proposed method to the generated data to estimate CEs between state variables and the estimated derivatives. Large absolute CEs mean the variable is selected. Experimental results are shown in Figure \ref{fig:lorenz} and Figure \ref{fig:rossler}, from which one can learn that the variables identified correspond to the original equations\footnote{The code is available at \url{https://github.com/majianthu/sysid}}.

\begin{figure}
	\centering
	\subfigure[Simulated data]{\includegraphics[width=0.6\textwidth]{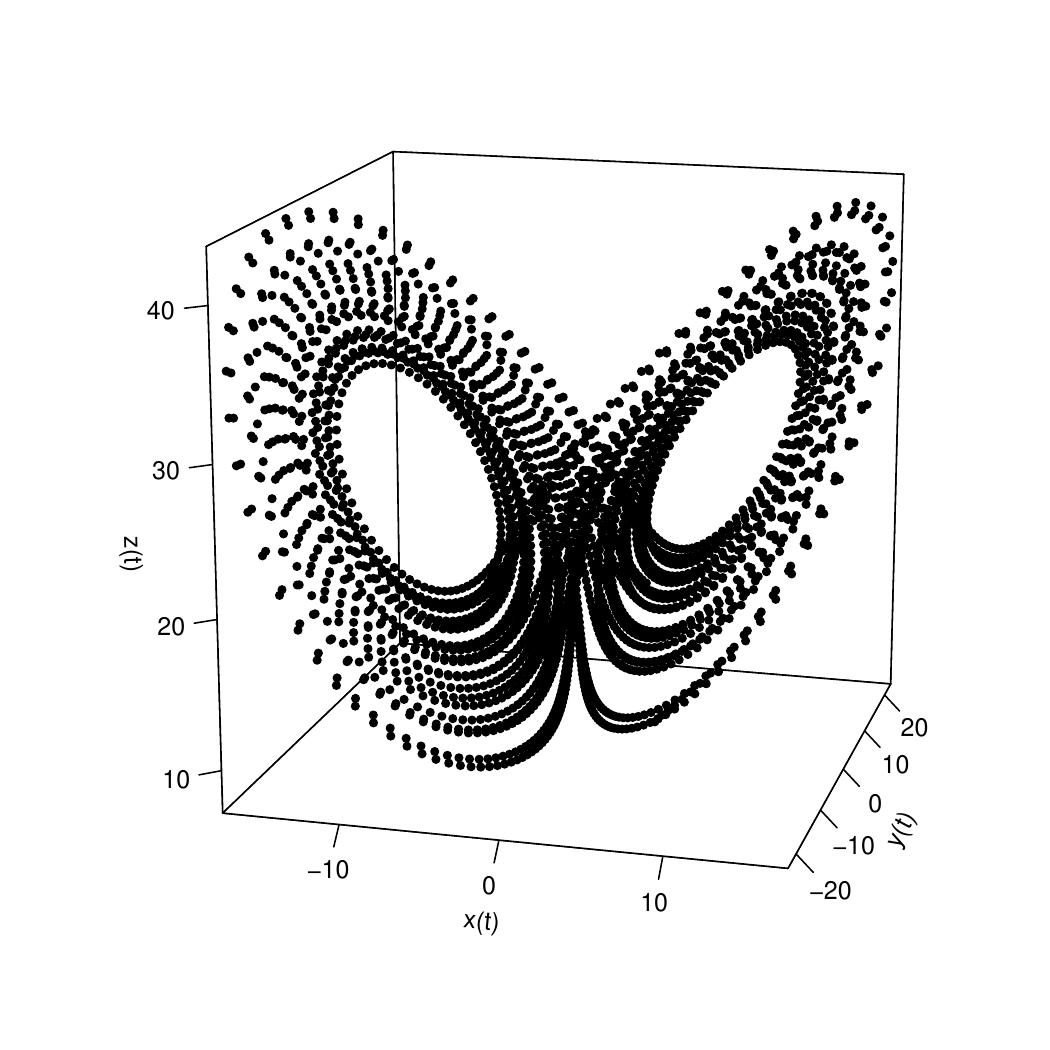}}
	\subfigure[Identification results]{\includegraphics[width=0.75\textwidth]{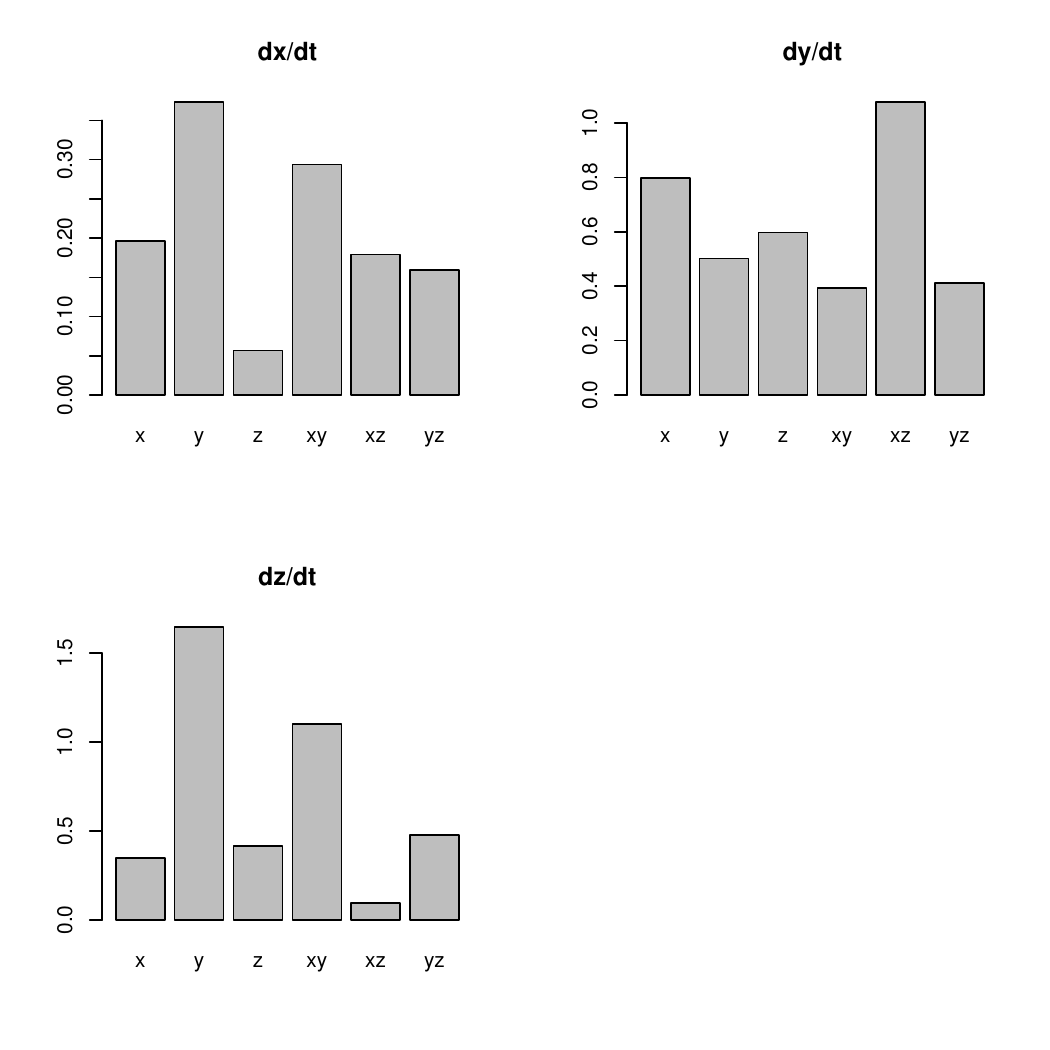}}
	\caption{Results of the experiment for identifying Lorenz system.}
	\label{fig:lorenz}
\end{figure}

\begin{figure}
	\centering
	\subfigure[Simulated data]{\includegraphics[width=0.6\textwidth]{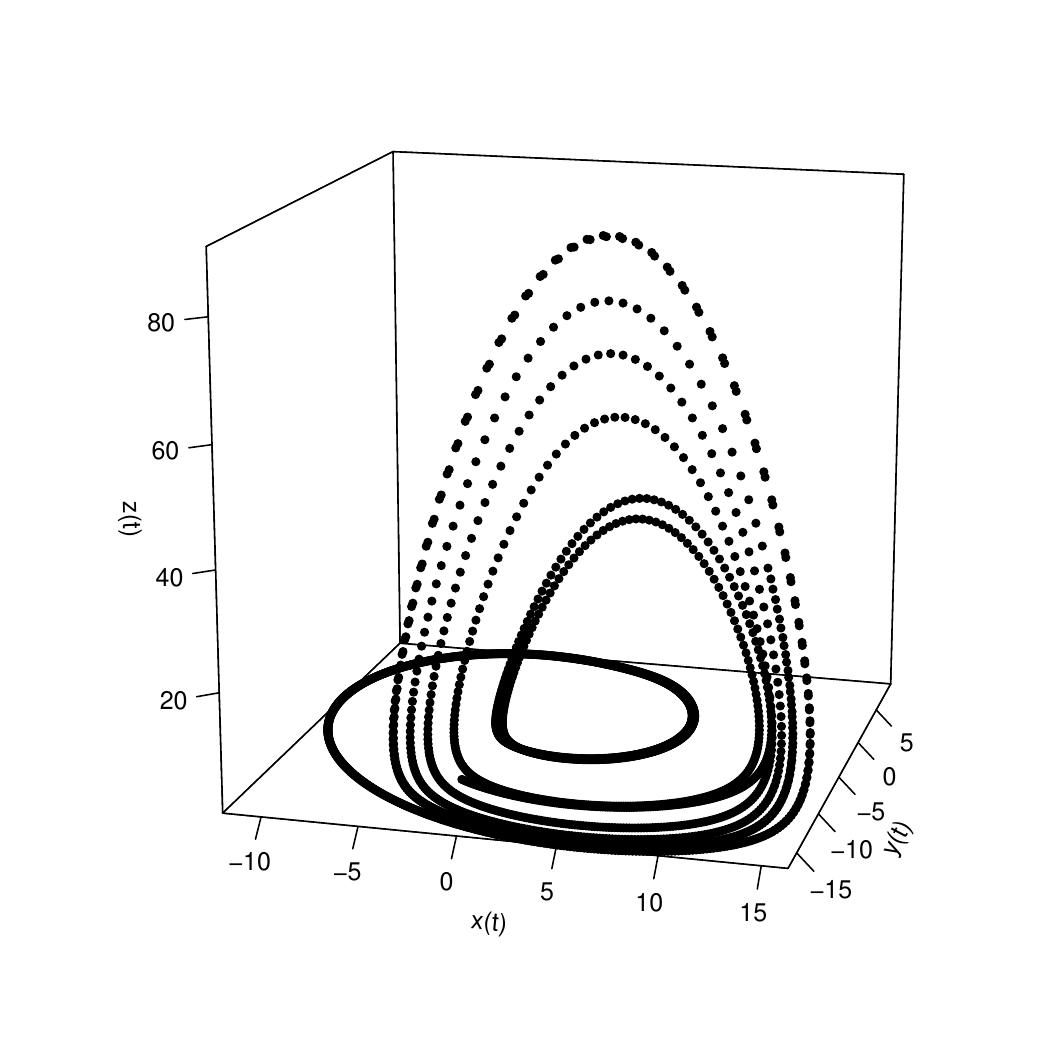}}
	\subfigure[Identification results]{\includegraphics[width=0.75\textwidth]{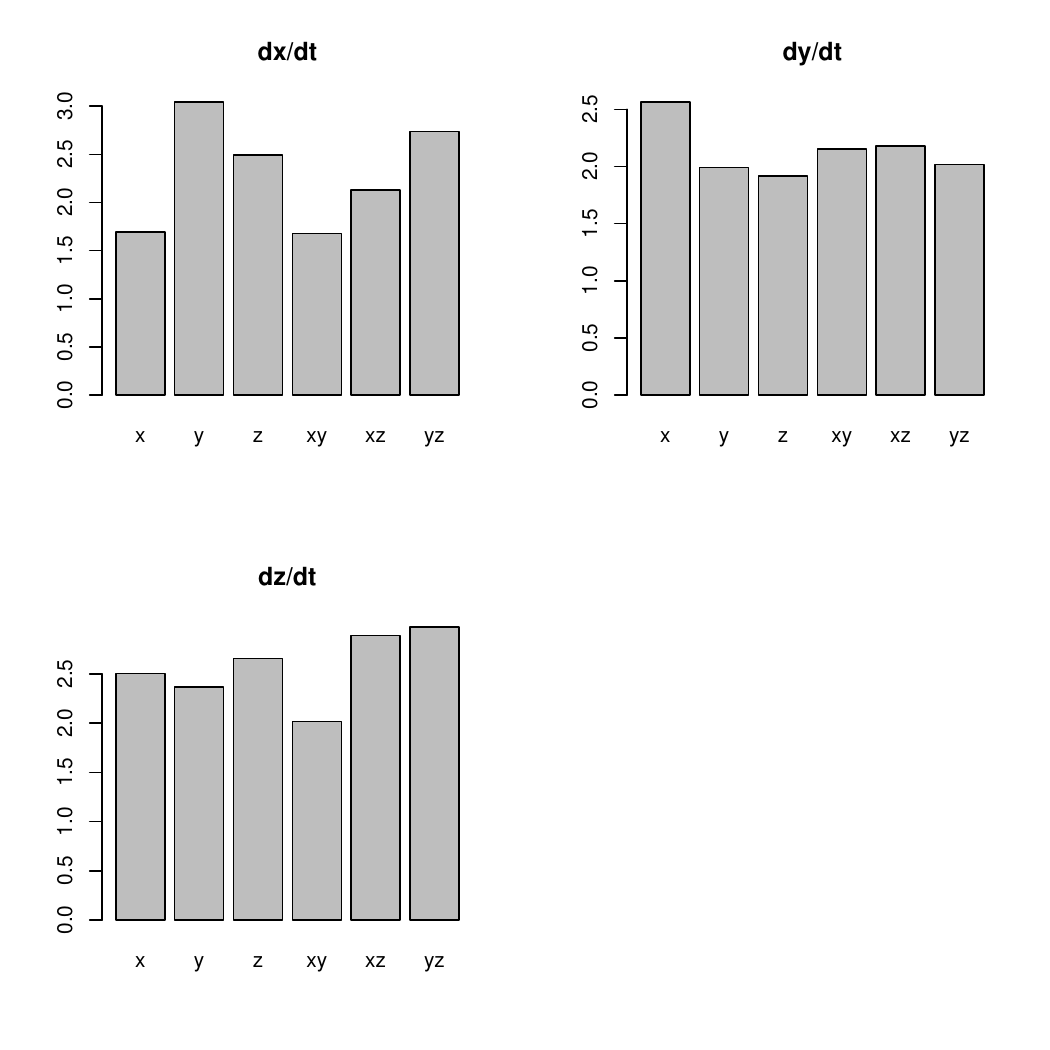}}
	\caption{Results of the experiment for identifying R\"ossler system.}
	\label{fig:rossler}
\end{figure}

\section{Time lag estimation}
Time lag is a property of many dynamical systems due to the time needed to transfer matter, energy, or information in all the physical, social and biological systems. So time lag estimation is a basic problem in time series analysis\cite{Abarbanel1993,Carter1987} and has wide applications in different areas of sciences and engineering, such as traffic in transportation, Solar's effect on earth, or social policy evaluation, etc.

Time lag estimation is to estimate time delay from sources signals to target signals from their measurements. Given source signals $x_t$ and target signals $y_t$, $t=1,\ldots,T$, suppose the functional relation $f(\cdot)$ between them, and time lag $l$ from source $x$ to target $y$, the general model for time lag estimation is
\begin{equation}
	y_t = f_L(x_{t-l},\mathbf{z}_t) + \epsilon,
	\label{eq:lag}
\end{equation}
where $\mathbf{z}_t$ are other variables, $\epsilon$ is noise. The problem is to estimate $l$ from measurement $x_t,y_t$. To do so, some assumptions  for $f_L$ are needed, usually is linearity even though real systems are much more complex than linear.

Traditional methods for time lag estimation is based autocorrelation\cite{Bjorklund2003} which is very limited due to linear assumption. Another common method is time-delayed MI\cite{Mars1982} which is fitful for nonlinear time series. However, these two methods are both symmetric measures while time lag is essentially asymmetric due to causal relationship between source and target.

Time lag is actually the time of causal effect from source to target and therefore can be estimated with causality tools. As a measure of causality, TE can measure asymmetric causal relations so is a right tool for the problem. Given the above model \ref{eq:lag}, one can estimate a groups of TEs $TE_{X\rightarrow Y}(l)$ from $X_{t-l}$ to $Y_t$ within time lag window $l=1,\ldots,L$, and then take the lag corresponding to the maximum of TEs as the estimate:
\begin{equation}
	\hat{l} = \mathop{\arg\max}\limits_{l} TE_{X\rightarrow Y}(l).
\end{equation}

Ma \cite{Ma2023} proposed to estimate time lag with the TE estimator in Section \ref{s:ci}, composed of two steps:
\begin{enumerate}
	\item estimate TEs from source to target in time lag window with the CE based TE estimator;
	\item take the lag corresponding to maximum of TEs as estimate.
\end{enumerate}
Since the CE based TE estimator is nonparametric without assumptions on the underlying system, whether linear or nonlinear, the proposed method is universally applicable.

To verify the effectiveness of the proposed method, we simulate five time-delayed dynamical system with different characteristics of dynamics:
\paragraph{System with random input}
\begin{equation}
	\begin{aligned}
		x_i &= \xi_1,\\
		y_{i+l} &= x_i + \xi_2;
	\end{aligned}
\end{equation}
\paragraph{System with nonlinear random input}
\begin{equation}
	\begin{aligned}
		x_i &= \sin(2\pi i/m) + \xi_1,\\
		y_{i+l} &= x_i + \xi_2;
	\end{aligned}
\end{equation}
\paragraph{Wiener processes}
\begin{equation}
	\begin{aligned}
		x_i &= x_{i-1} + \xi_1,\\
		y_{i+l} &= x_i + \xi_2;
	\end{aligned}
\end{equation}
\paragraph{Second order Wiener processes}
\begin{equation}
	\begin{aligned}
		x_i &= \alpha x_{i-1} + \beta x_{i-l} + \xi_1,\\
		y_i &= x_i + \xi_2;
	\end{aligned}
\end{equation}

\paragraph{Second order nonlinear Wiener processes}
\begin{equation}
	\begin{aligned}
		x_i &= \alpha x_{i-1} + \beta x_{i-l} + \xi_1,\\
		y_i &= x_i^2 + \sin(x_i) + \xi_2.
	\end{aligned}
\end{equation}
where $x_i$ is input variable, $y_i$ is output variable, $\xi_1 \sim N(0,0.1), \xi_2 \sim N(0,0.001)$ are white noise, $\alpha=0.2,\beta=0.8$ are system parameters, $l=1,2,3,4$ are time lags in simulation. Simulated trajectories are shown in Figure \ref{fig:sim}. After applying the proposed method to the simulated trajectories, we estimated the time lags from $X$ to $Y$ in the simulated systems, as shown in Figure \ref{fig:sysidres}, from which one can learn that the pre-specified time lags are estimated correctly. Particularly, in second order Wiener processes, the estimation results not only present time lag estimates but show the memory decay of the simulated systems.

\begin{figure}
	\centering
	\subfigure[System with random input]{
		\includegraphics[width=0.78\textwidth]{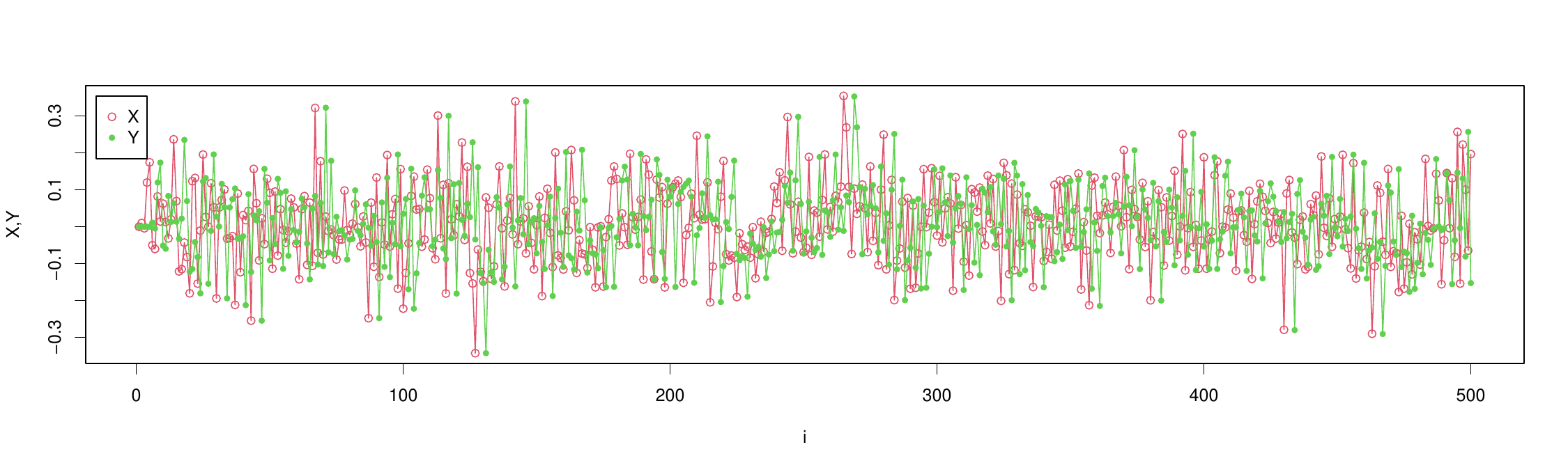}
		\label{fig:sim1data1}}
	\subfigure[System with nonlinear random input]{
		\includegraphics[width=0.78\textwidth]{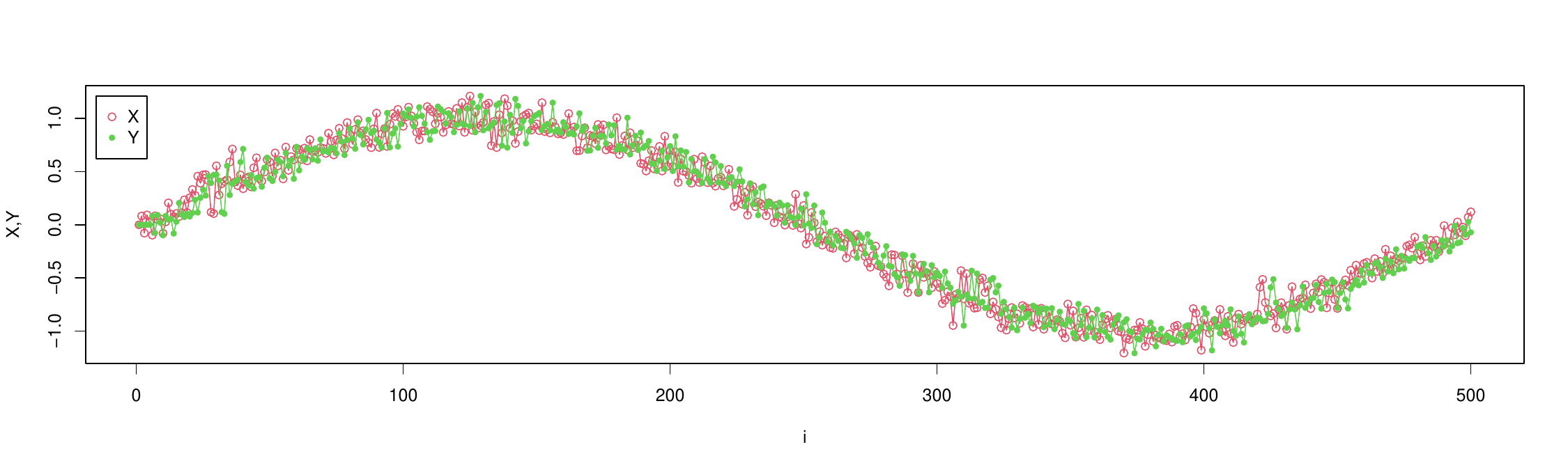}
		\label{fig:sim1data2}}
	\subfigure[Wiener processes]{
		\includegraphics[width=0.78\textwidth]{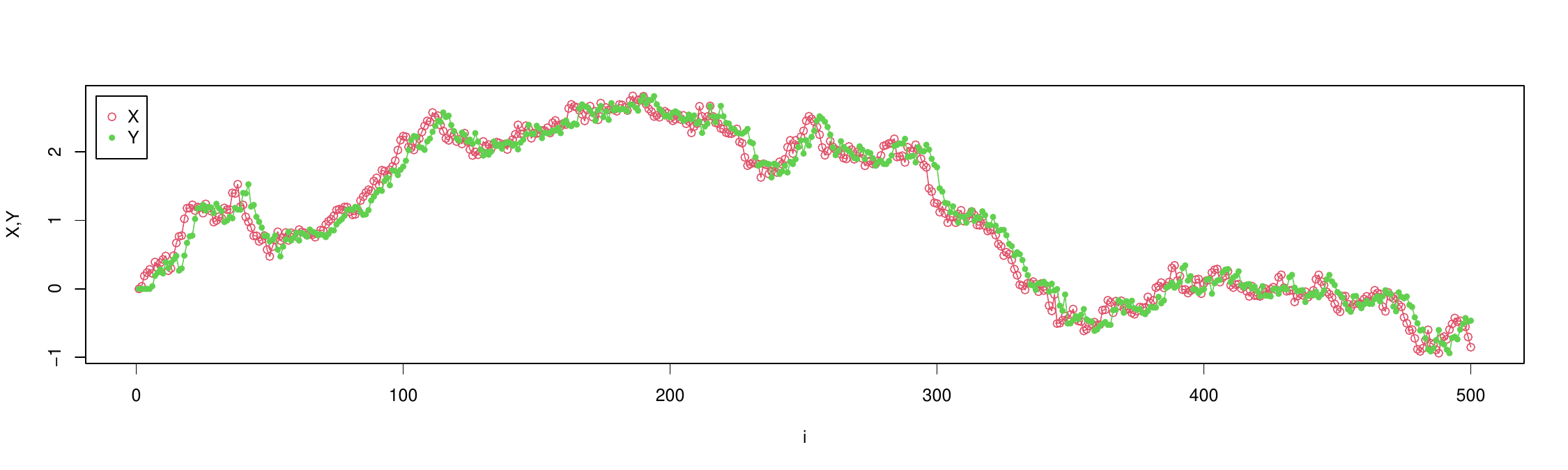}
		\label{fig:sim1data3}}
	\subfigure[Second order Wiener processes]{
		\includegraphics[width=0.78\textwidth]{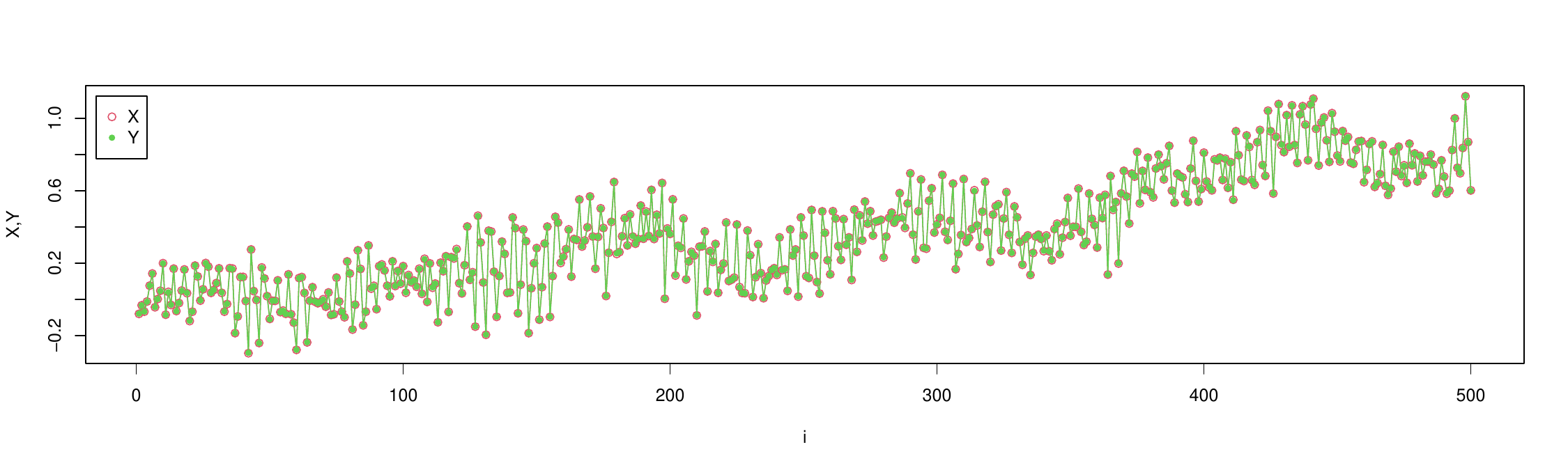}
		\label{fig:sim2}}
	\subfigure[Second order nonlinear Wiener processes]{
		\includegraphics[width=0.78\textwidth]{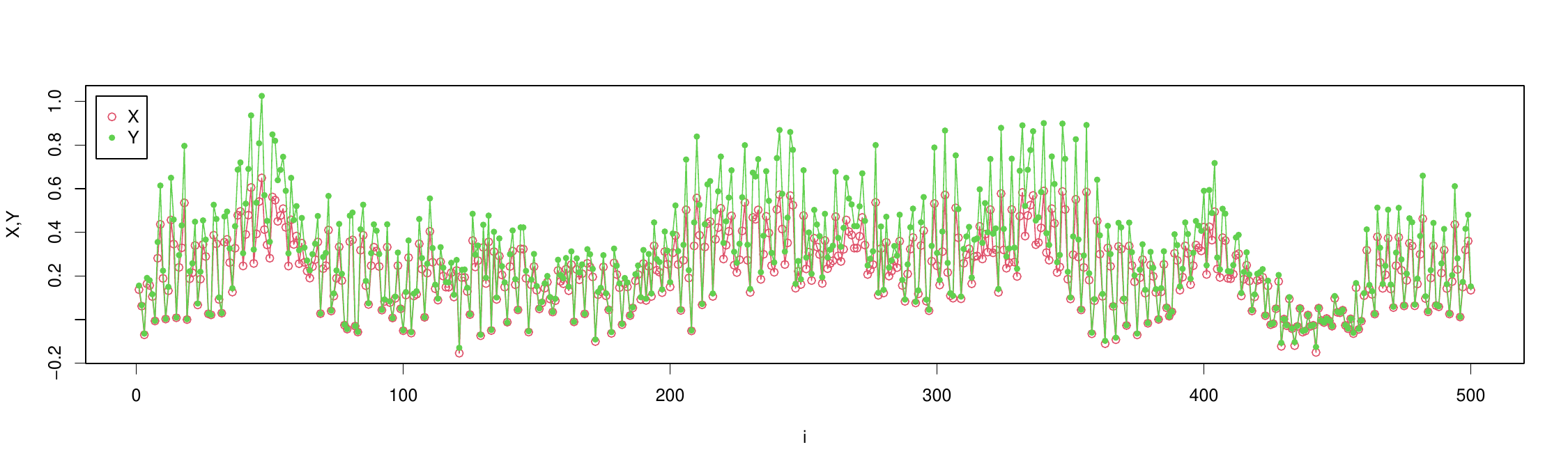}
		\label{fig:sim2n}}	
	\caption{Simulated trajectories of the simulated time-delayed systems ($l=4$).}
	\label{fig:sim}
\end{figure}

\begin{figure}
	\centering
	\subfigure[System with random input]{
		\includegraphics[width=0.4\textwidth]{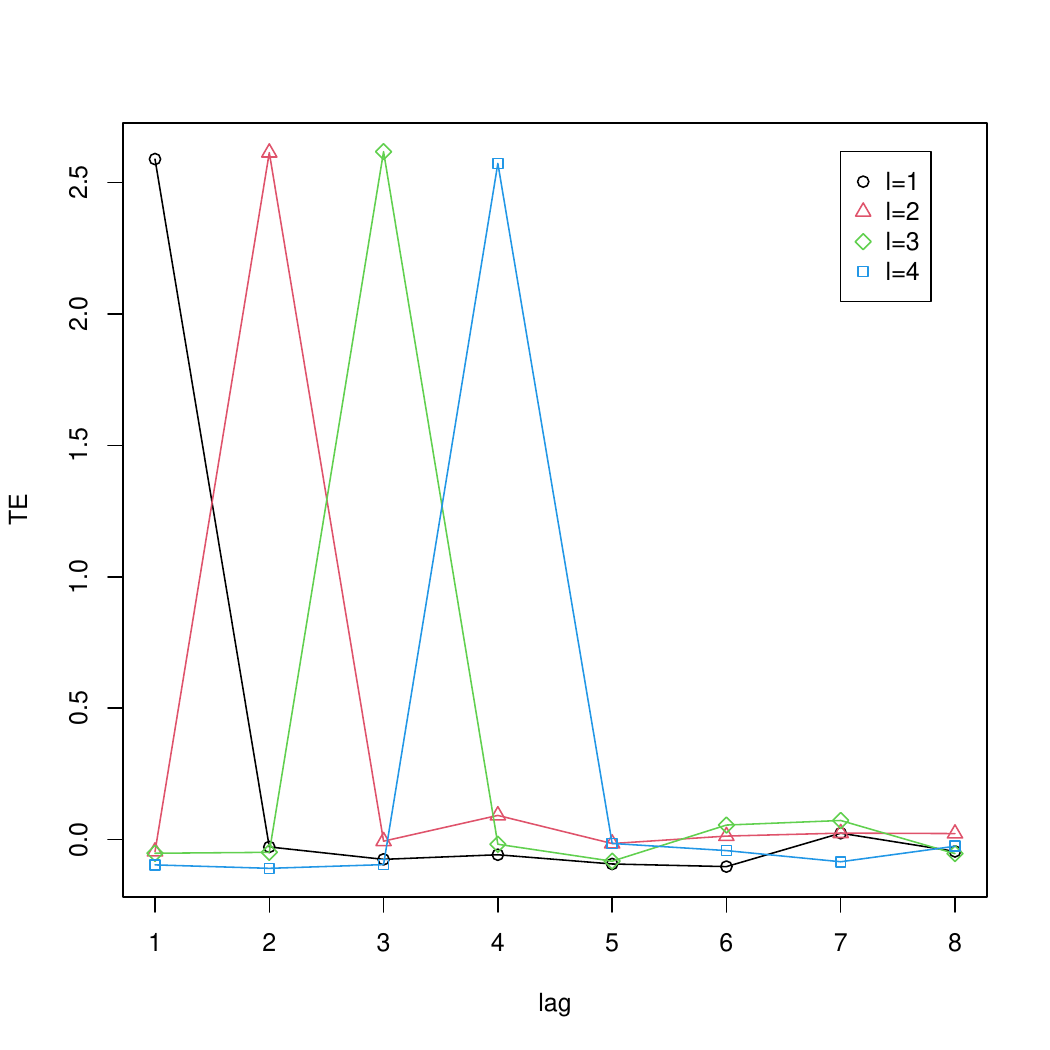}
		\label{fig:sim1res1}}
	\subfigure[System with nonlinear random input]{
		\includegraphics[width=0.4\textwidth]{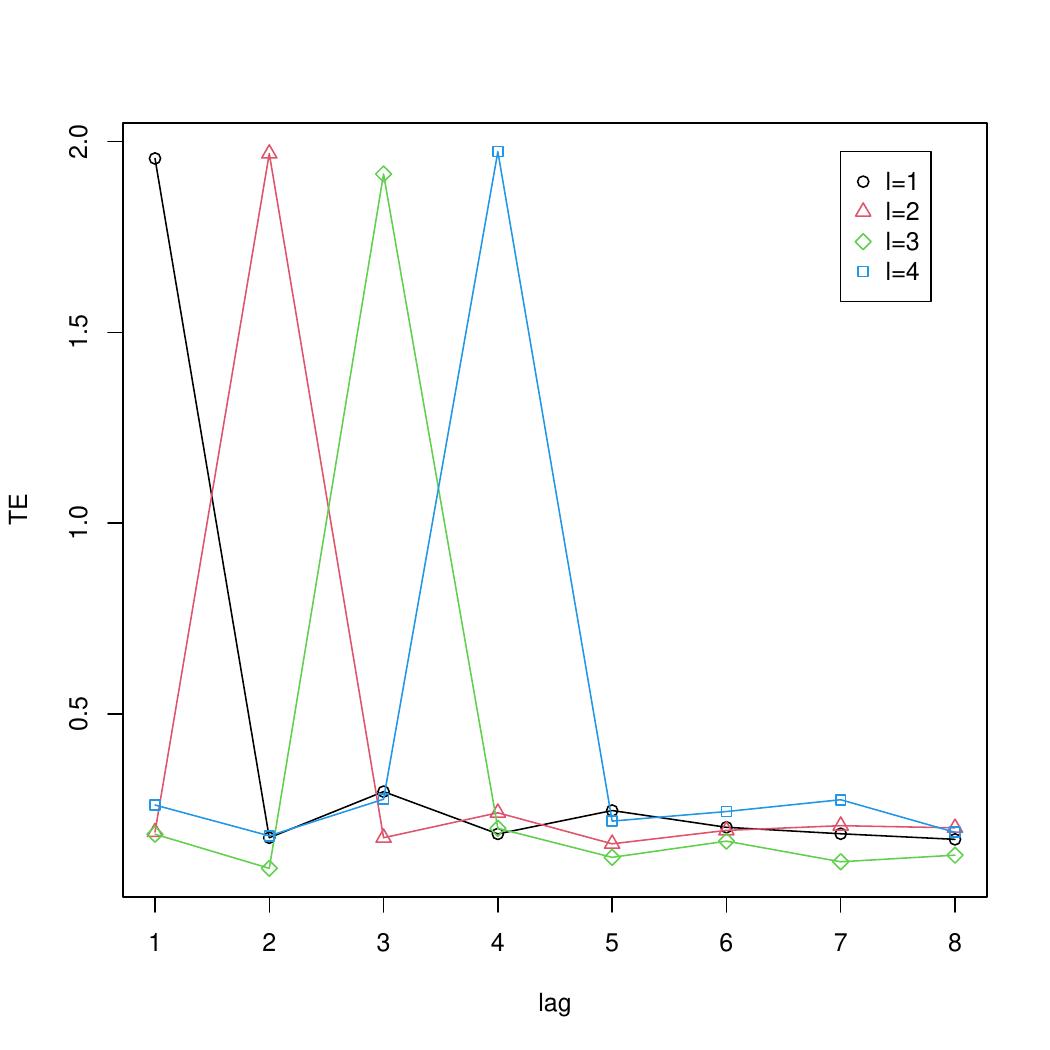}
		\label{fig:sim1res2}}
	\subfigure[Wiener processes]{
		\includegraphics[width=0.4\textwidth]{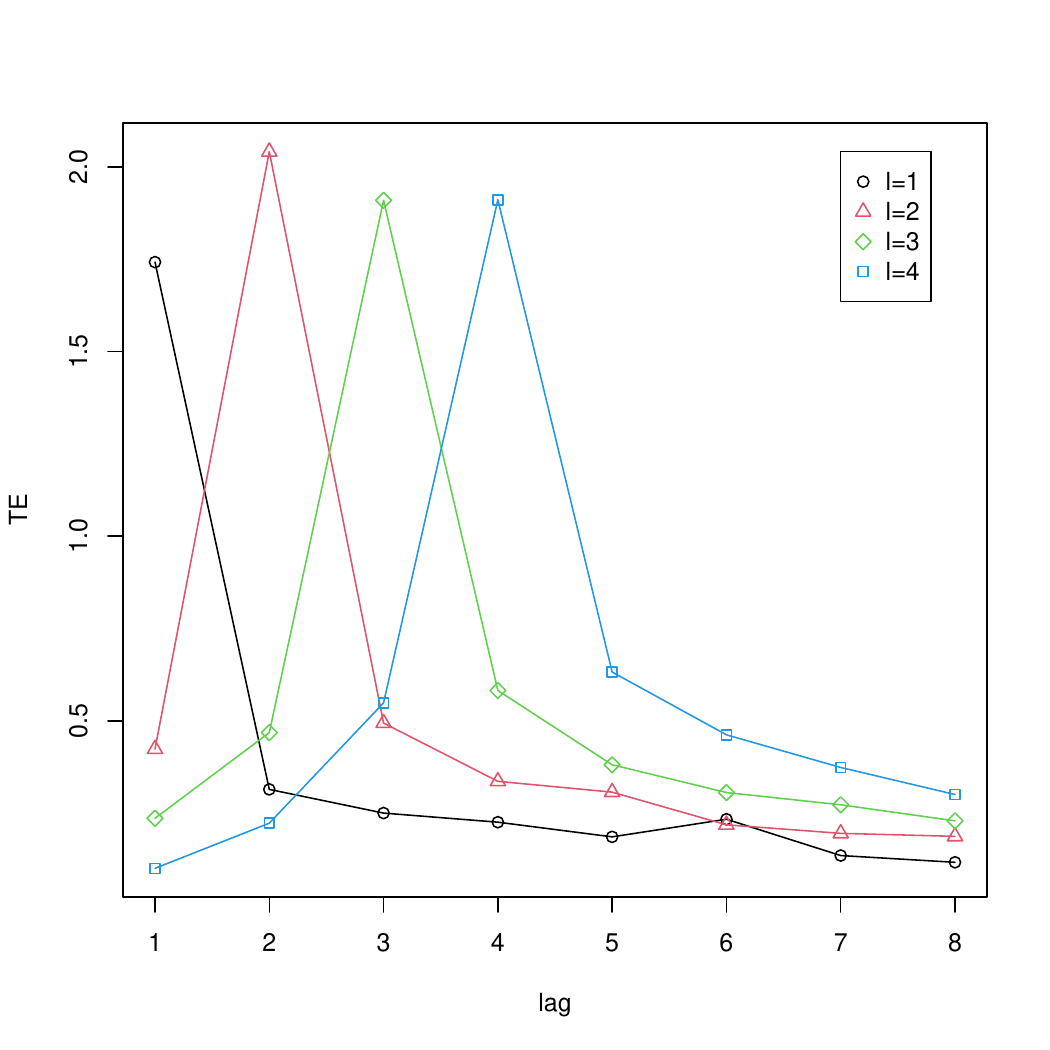}
		\label{fig:sim1res3}}
	\subfigure[Second order Wiener processes]{
		\includegraphics[width=0.4\textwidth]{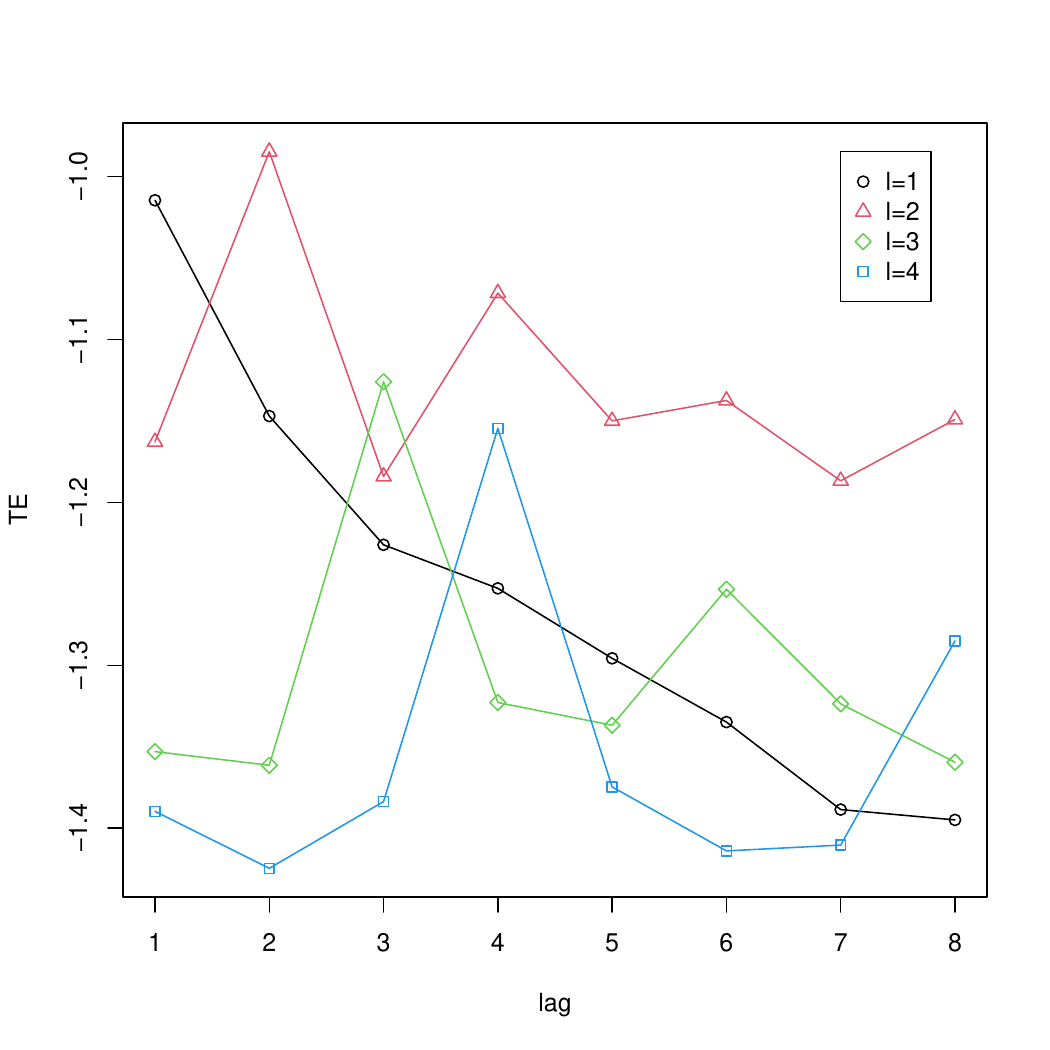}
		\label{fig:sim2res1}}
	\subfigure[Second order nonlinear Wiener processes]{
		\includegraphics[width=0.4\textwidth]{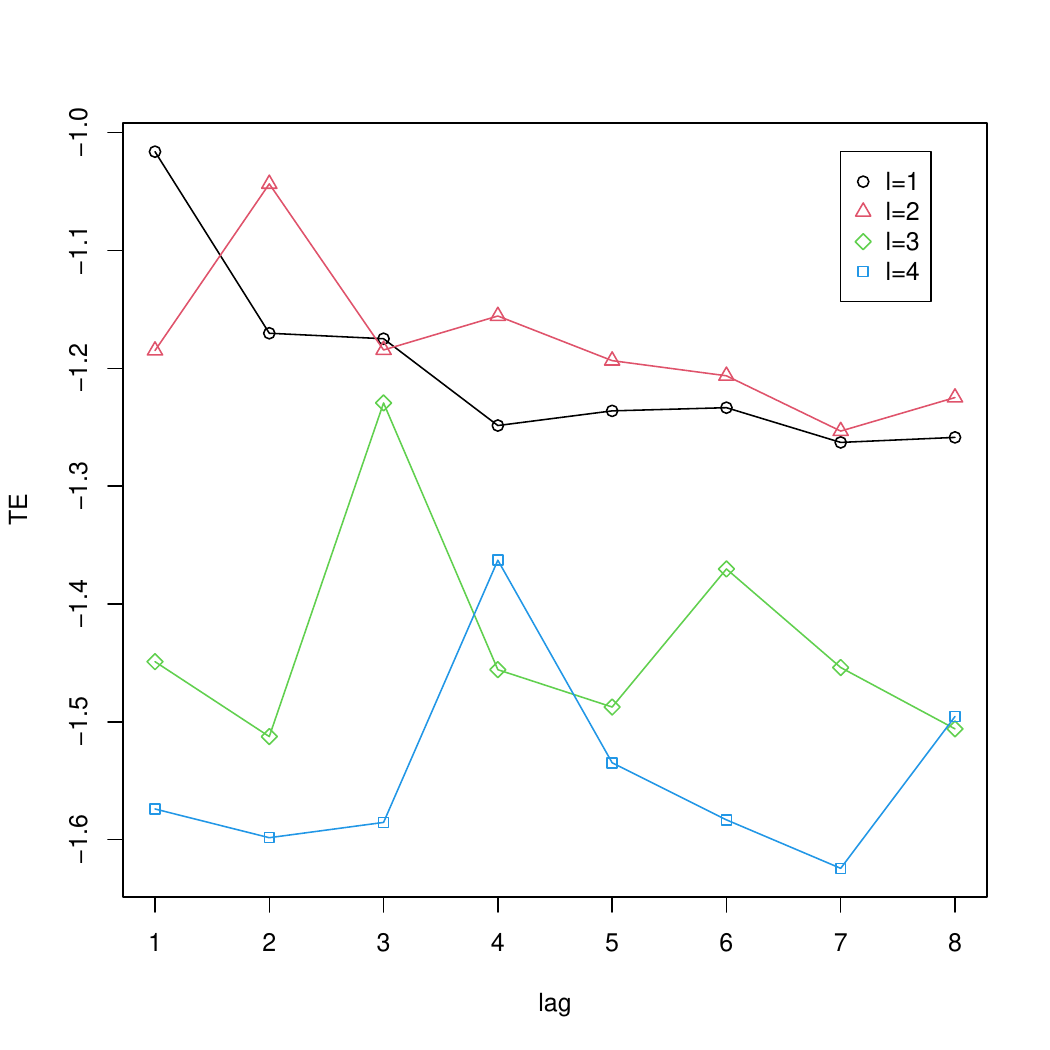}
		\label{fig:sim2res2}}
	\caption{Estimation results of the simulation experiments for time lag estimation.}
	\label{fig:sysidres}
\end{figure}

We also applied the proposed method to the data of electric load of the T\'etouan city to analyze time lag from weather factors to electric loads of three areas of the city. The analyzing shows how weather factors affect electric load with time lags\footnote{The code is available at \url{https://github.com/majianthu/timelag}}.

\section{Domain adaptation}
Domain Adaptation (DA) is a common problem in machine learning which refers to distribution in model deployment shifting away from distribution in model training due to certain outside factors. To make the deployed model work in the right ways as in training, DA tries to do model transfer between different situations. There are many ways of such model transfer proposed \cite{Pan2010}. 

DA is meaningful to real applications. For example, when a medical model was trained on the data from one hospital and wants to be deployed in another hospital, there may be problematic due to different medical equipments used in these hospitals. The same situation happens in social science when using a social model got from one population to another different population.

Tackling DA from the causal viewpoint is an important direction in DA research. It considers distribution shift as an effect of outside cause and learns such a causal model for the DA problem. In this sense, CE based CI measure can be applied to identify such causal relations in DA problem. Ma \cite{ma2022causal} proposed a causal DA method with CE based CI measure. He suppose the invariance of the functional relation from variable $X$ to predictive variables $Y$ on different domains $D_i$ and associate a domain variable $Z$ to different domains and then the DA problem is transformed into one that estimates CI relation between $X$ and $Y$ conditional on $Z$, i.e.,
\begin{equation}
	X \not \perp\!\!\!\perp Y|Z.
	\label{eq:cda}
\end{equation}
The CE based CI test in Section \ref{s:ci} can be used to identify such CI relations for causal DA problems.

We verified the effectiveness with simulation experiments and a real data on gender inequality on income\footnote{The code is available at \url{https://github.com/majianthu/cda}}.

\section{Multivariate normality test}
Normal distribution, also named Gaussian distribution, is a very important probability distribution since it is at the center in probability due to central limit theorem and commonly exists in natural and social worlds \cite{Ross2020,Durrett2019}. Normality is a common assumption of many statistical models and methods and therefore Testing on it is necessary for applications. Normality test is such a methodology in hypothesis testing, including univariate and multivariate ones. There are many traditional tests, including those based moments, characteristic functions, entropy, and optimal transport, etc\cite{Ebner2020,Chen2023,Yap2011,Yazici2007,Shapiro1968}.

According to maximum entropy principle\cite{Jaynes1957}, normal distributions have maximum entropy among all the distributions with same second order moment. Based on this, scholars proposed entropy based univariate normality test by the ratio of entropies of unknown distribution and normal distribution with same variance. Since entropy of normal distribution with variance $\delta^2$ is $\log \sqrt{2\pi e}\delta$, the test statistic for normality is defined as
\begin{equation}
	T_{un}=\frac{H(x)}{\log \sqrt{2\pi e}\delta}.
\end{equation}

As an entropy measure of all order dependence, CE can be used to test multivariate normality. Given $n$ random variables $\mathbf{X}\sim N(\mu,\Sigma)$, the entropy of them is 
\begin{equation}
	H(\mathbf{x})=\frac{1}{2}\log (2\pi e)^n |\Sigma|.
\end{equation}
According to Property \ref{p:cov}, the equivalence between CE and correlation coefficient matrix $\Sigma_\rho$ is
\begin{equation}
	H_c(\mathbf{x})=\frac{1}{2}\log |\Sigma_\rho|.
\end{equation}
For Gaussian distributions, the difference between $\Sigma$ and $\Sigma_\rho$ is that the latter contains no information of individual variances, and hence $H_c(\mathbf{x})$, compared to $H(\mathbf{x})$, has no information of individual variables and only information of second order correlation. For non-Gaussian distribution, there are not only second order correlation, but also higher order correlations, and hence CE measures information of more than second order correlation. The greater non-Gaussianity, the more information in CE. So we can define a test statistic for multivariate normality based on the difference of CE between Gaussian distributions and non-Gaussian distributions. 

Ma \cite{Ma2022} proposed a CE based methodology for multivariate normality test based on the equivalence between CE and correlation coefficient matrix and test multivariate normality by the difference between CE of unknown distribution and entropy of normal distribution with same covariance matrix.

Given random variables $\mathbf{X}$, the null hypothesis of multivariate normality test is
\begin{equation}
	H_0: \mathbf{X}\sim N(\mathbf{\mu},\Sigma);
\end{equation}
the alternative is
\begin{equation}
	H_1: \mathbf{X}\nsim N(\mathbf{\mu},\Sigma),
\end{equation}
where $N(\mathbf{\mu},\Sigma)$ is normal distribution.

Given random variables $\mathbf{X}$ and its sample $\mathbf{X}_T$, the test statistic for multivariate normality test based on CE is defined as
\begin{equation}
	T_{mvn}(\mathbf{X}_T)=H_c(\mathbf{x})-\frac{1}{2}\log |\Sigma|,
	\label{eq:tce}
\end{equation}
where $\Sigma$ is covariance matrix of $\mathbf{X}$. Based on this definition, if $\mathbf{X}$ is governed by normal distribution, $H_0$ is true, then $T_{mvn}$ is small; if $\mathbf{X}$ is non-normal, $H_1$ is true, then $T_{mvn}$ is large.

We proposed the estimation method of the test statistic\footnote{The estimation method is implemented in the \texttt{copent} package in \textsf{R} and \textsf{Python} \cite{ma2021copent}.}, which composed of two part: the first term in \eqref{eq:tce} can be estimated with the nonparametric CE estimator from $\mathbf{X}_T$ and the second term can be computed analytically from covariance matrix $\Sigma$ estimated from $\mathbf{X}_T$$\Sigma$.

We designed two simulation experiments to verify the effectiveness of the proposed test and compare the proposed test with five other tests\footnote{The code is available at \url{https://github.com/majianthu/mvnt}}. For benchmarking on more tests for multivariate normality, please see Section \ref{s:mvntbench}.

\section{Copula hypothesis testing}
Model selection is a common task in scientific research and hypothesis testing is to test the fitness of hypothesis model to data. Copula is a basic method for probabilistic models and is widely used in many scientific disciplines\cite{Genest2009,Berg2013,Patton2012,Fan2014,Tootoonchi2022,NazeriTahroudi2023}. There are many types of copulas, such as Gaussian copula, Archimedean copulas, t copulas\cite{Demarta2005}, Archimax copula\cite{Mesiar2013}, Sibuya copula\cite{Hofert2013}, etc which makes testing copula hypothesis a common problem in applications. 

There are some work on copula hypothesis testing, such as Gaussian copula hypothesis testing\cite{Malevergne2003,Amengual2020}, Archimedeanity test\cite{Jaworski2010,Buecher2012}, Test for copula symmetry\cite{Genest2012}, etc. However, those work are mainly focusing on a particular type of copula. Kole et al.\cite{Kole2007} proposed to select copula with Kolmogorov-Smirmov test or Anderson-Darling test. Genest and R\'emillard\cite{Genest2008} proposed a goodness of fit test for copula based on Cram\'er-von Mises test and Kolmogorov-Smirnov test. Gr\o{}nneberg and Hjort\cite{Groenneberg2014} proposed a AIC like criterion for copula model selection, called copula information criteria.

Given random variables $\mathbf{X}$ and their sample $\mathbf{X}_T$, and their unknown copula density function $c_\mathbf{x}(\mathbf{u})$, suppose the candidate copula density function be $c(\mathbf{u})$, the null hypothesis of copula hypothesis testing is
\begin{equation}
	H_0: c_\mathbf{x}(\mathbf{u}) = c(\mathbf{u});
\end{equation}
the alternative is
\begin{equation}
	H_1: c_\mathbf{x}(\mathbf{u}) \neq c(\mathbf{u}).
\end{equation}

Ma \cite{Ma2025c} proposed a CE based method for copula hypothesis testing. The test statistic is defined from the difference between the CE of the null hypothesis and the CE estimated from data:
\begin{equation}
	T_c(\mathbf{X}_T|c)=H_c(\mathbf{X}_T|c)-H_c(\mathbf{X}_T|c_\mathbf{x}).
	\label{eq:tchstats}
\end{equation}
The first term in \eqref{eq:tchstats} is the estimated CE under the copula hypothesis $c$ and the second term is the CE estimated from data. If $H_0$ is true, then $T_c=0$; otherwise, if $H_1$ is true, then $T_c$ becomes larger.

Since CE is a general theory, the proposed method based on it is then a common methodology for any type of copulas.

We propose the method for estimating the test statistic: the second term can be done with the nonparametric CE estimator and the first term can be estimated with the following parametric estimation method:
\begin{enumerate}
	\item estimate empirical copula density function $\hat{\mathbf{u}}$,
	\item estimate the parameters $\alpha$ of copula hypothesis $c$ based on $\hat{\mathbf{u}}$,
	\item compute the CE of copula hypothesis based on the following equation:
	\begin{equation}
		H_c(\mathbf{X}_T|c)=-E(\log c(\mathbf{\hat{u}};\alpha)).
		\label{eq:celikelihood}
	\end{equation}	
\end{enumerate}

We present the parametric estimation of CE of two common copula hypothesis:
\paragraph{Gaussian Copula}
Gaussian copula density function can be write as the following form with correlation coefficient matrix $\Sigma_\rho$ as parameter\cite{Song2000,Arbenz2013}:
\begin{equation}
	c_n(\mathbf{u})=|\Sigma_\rho|^{-\frac{1}{2}}exp\left\{-\frac{1}{2}\Phi(\mathbf{u})(\Sigma_\rho^{-1}-I)\Phi^T(\mathbf{u})\right\},
	\label{eq:igc}
\end{equation}
where $\Phi$is quintile normal function and $I$ is identity matrix.

Estimating CE of Gaussian copula hypothesis needs to estimate correlation coefficient matrix $\Sigma_\rho$ first, then compute the values of Gaussian copula density function with \eqref{eq:igc}, finally get the result with \eqref{eq:celikelihood}.

\paragraph{Archimedean Copula}
Gumbel copula, Frank copula, and Clayton copula are three typical Archimedean copulas. Bivariate Gumbel copula density function is
\begin{equation}
	c_{g}(\mathbf{u})=exp\left\{-\left[\sum_{i=1}^{2}(-\ln u_i)^\alpha\right]^{\frac{1}{\alpha}} \right\}\left[\left( \left[\sum_{i=1}^{2}(-\ln u_i)^\alpha\right]^{\frac{1}{\alpha}-1} \right) \left(\sum_{i=1}^2 \frac{(-\ln u_i)^{\alpha}}{u_i}\right) \right],
	\label{eq:gumbelcopula}
\end{equation}
bivariate Frank copula density function is
\begin{equation}
	c_f(\mathbf{u})=\frac{\alpha(1-e^{-\alpha})e^{-\alpha(u_1+u_2)}}{\left\{1-e^{-\alpha}-(1-e^{-\alpha u_1})(1-e^{-\alpha u_2})\right\}^2},
	\label{eq:frankcopula}
\end{equation}
bivariate Clayton copula density function is
\begin{equation}
	c_c(\mathbf{u})=(\alpha+1)(u_1 u_2)^{-\alpha-1}(u_1^{-\alpha}+u_2^{-\alpha}+1)^{-2},
	\label{eq:claytoncopula}
\end{equation}
where $\alpha$ is parameter.

Estimating CE of Archimedean copula hypothesis needs to estimate $\alpha$ with the likelihood method first, then compute the values of Archimedean copulas with \eqref{eq:gumbelcopula}, \eqref{eq:frankcopula}, or \eqref{eq:claytoncopula}, finally get the result with \eqref{eq:celikelihood}.

We verified the proposed method for copula hypothesis testing with two simulation experiments\footnote{The code is available at \url{http://github.com/majianthu/tch}}. The first experiment simulates a group of bivariate normal distributions with covariance ranging from 0.1 to 0.9 by step 0.1; the second experiment simulated three groups of bivariate Archimedean copulas with parameter changing from 2 to 10 and margins being a normal distribution and an exponential distribution. The size of dataset is 300.

We then applied the above methods for Gaussian copula hypothesis testing and Archimedean copula hypothesis testing to the simulated datasets and derived four statistics for each dataset.

Experimental results are shown in Figure \ref{fig:tgc1}, Figure \ref{fig:tgc2}, Figure \ref{fig:tgc3}, and Figure \ref{fig:tgc4}. One can learn from them that for the first experiment, the estimated test statistic of Gaussian copula hypothesis is the smallest which means Gaussian copula hypothesis is preferred; for the second experiment, the estimated test statistics of Gumbel copula, Frank copula, and Clayton copula hypothesis are the smallest respectively which means each hypothesis is true respectively.

\begin{figure}
	\centering
	\includegraphics[width=0.62\textwidth]{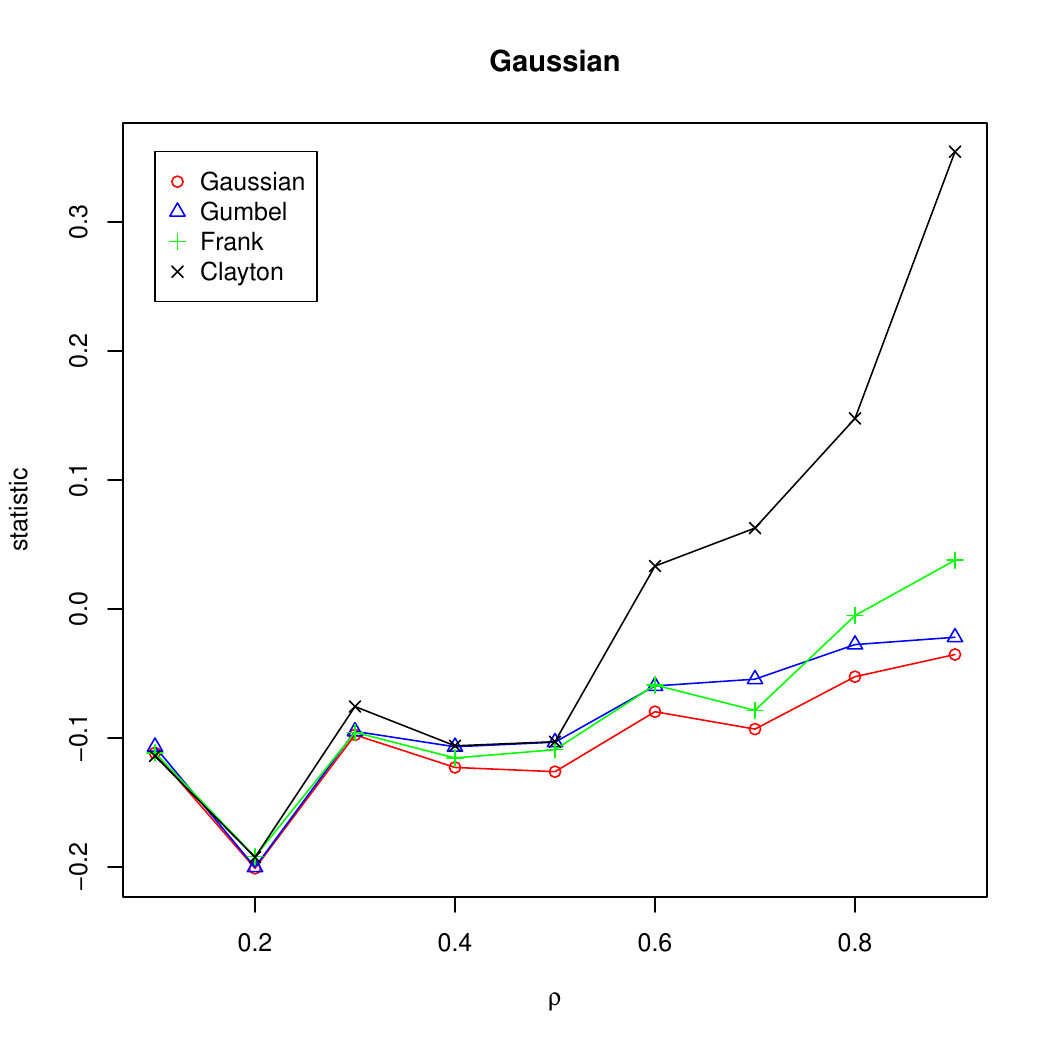}
	\caption{Results of simulation experiment on Gaussian copula hypothesis.}
	\label{fig:tgc1}
\end{figure}
\begin{figure}
	\centering
	\includegraphics[width=0.62\textwidth]{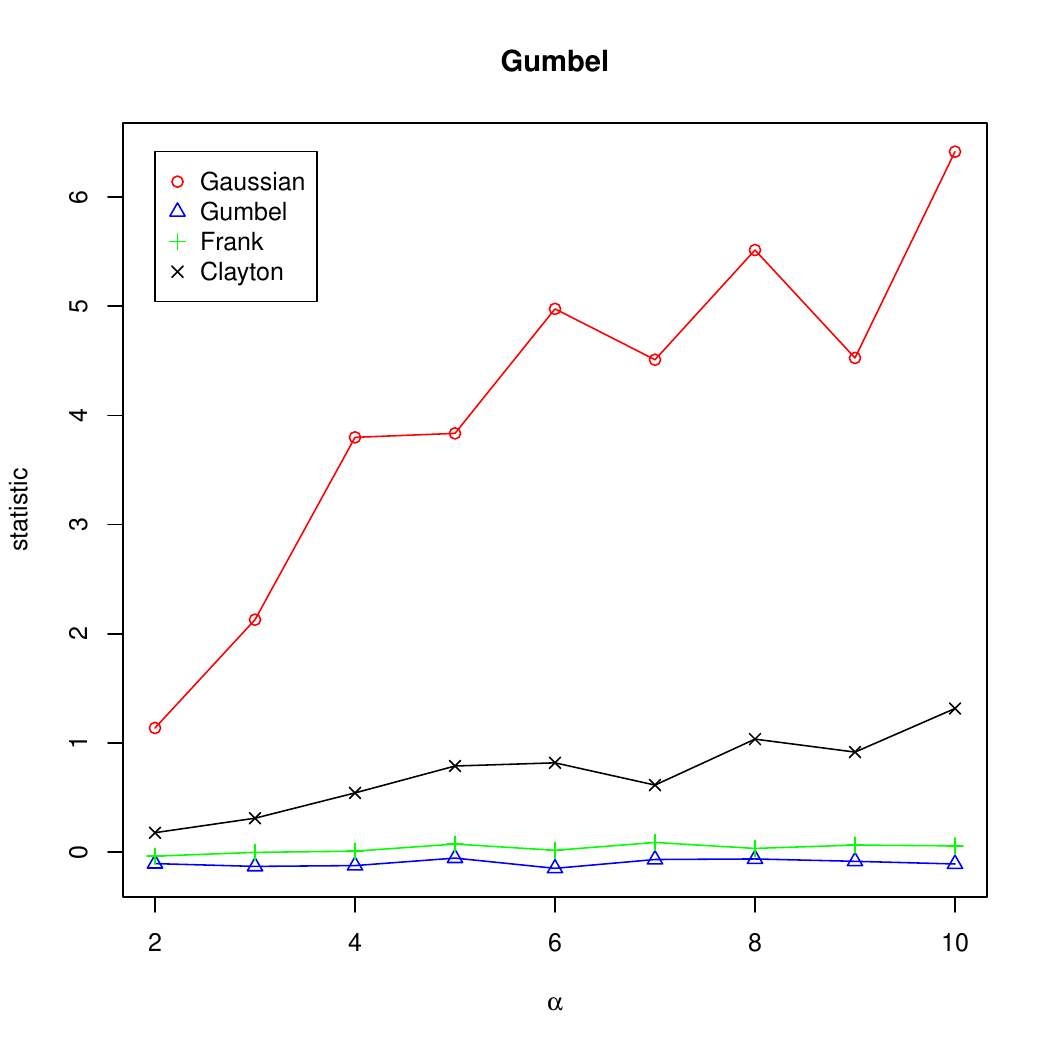}
	\caption{Results of simulation experiment on Gumbel copula hypothesis.}
	\label{fig:tgc2}
\end{figure}
\begin{figure}
	\centering
	\includegraphics[width=0.62\textwidth]{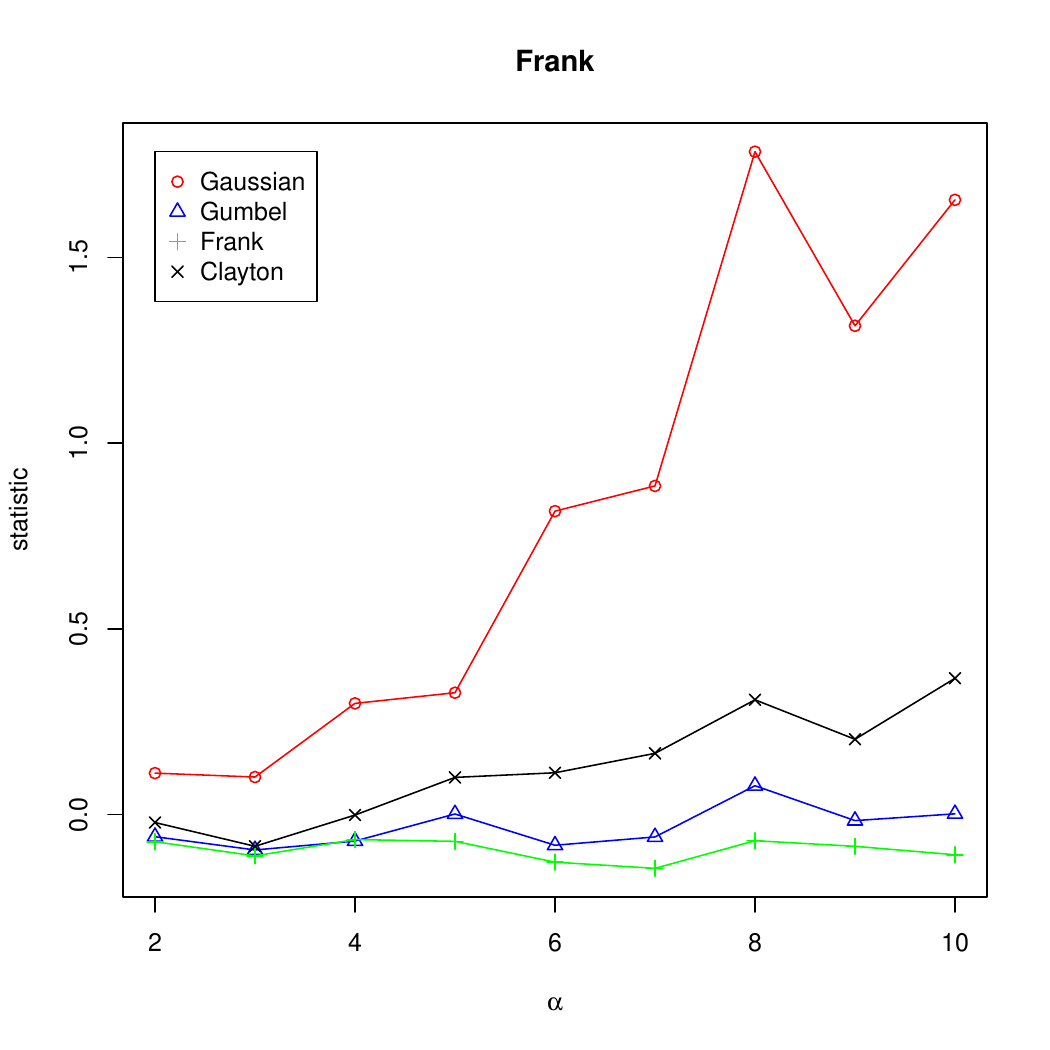}
	\caption{Results of simulation experiment on Frank copula hypothesis.}
	\label{fig:tgc3}
\end{figure}
\begin{figure}
	\centering
	\includegraphics[width=0.62\textwidth]{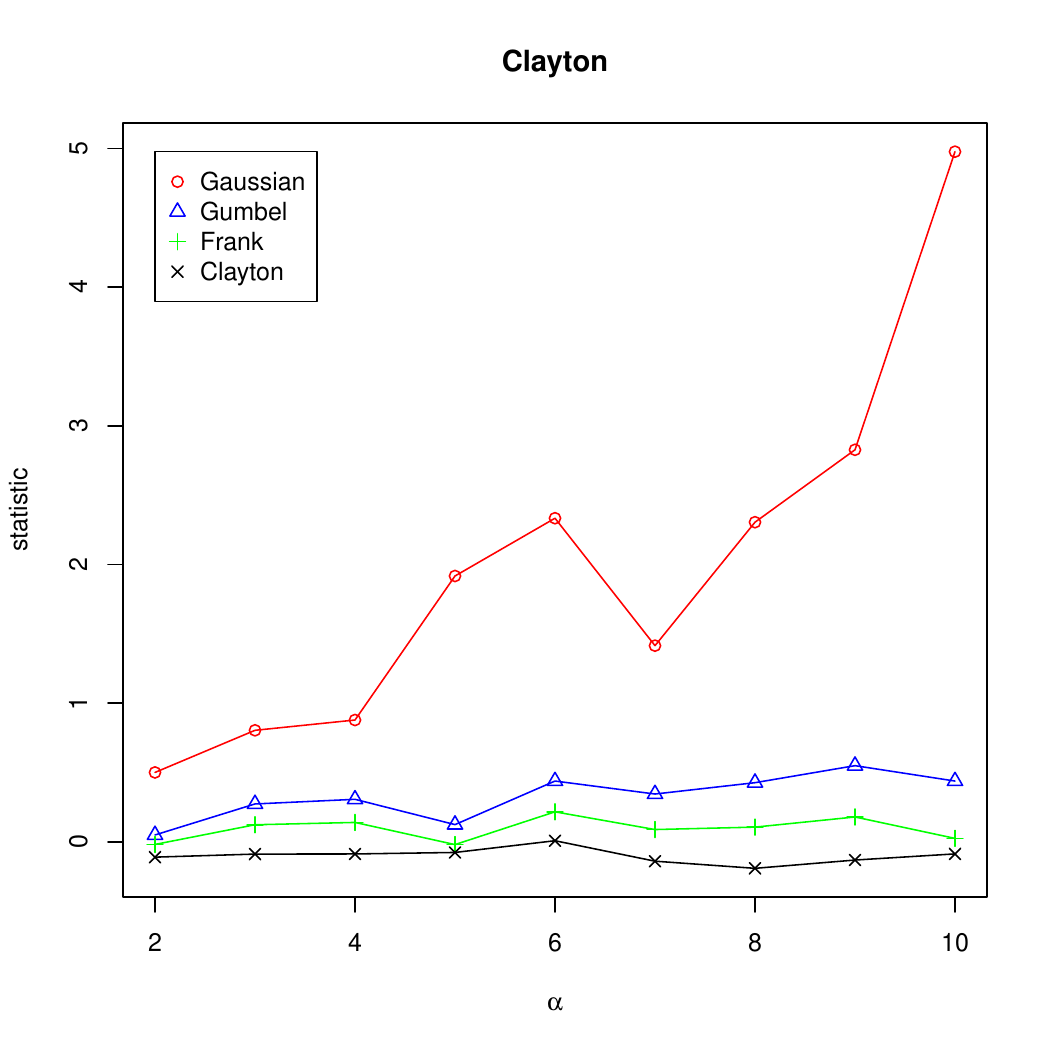}
	\caption{Results of simulation experiment on Clayton copula hypothesis.}
	\label{fig:tgc4}
\end{figure}

\section{Two-sample test}
\label{sec:tst}
Two-sample test is a main methodology for testing whether two samples are from a same distribution. Many problems in statistics can be transformed as two-sample test, such as symmetry test which tests whether a sample and its symmetric sample are from same distribution, change point detection which find a change point where two samples before and after it are different. Two-sample test has also many applications in different fields, including checking the effect of clinical treatment, the effect of a policy, etc.

There are many traditional two-sample tests, such as t-test\cite{Cressie1986}, Kolmogorov–Smirnov test\cite{Berger2014,Kolmogorov1933,Smirnov1948}, kernel base test\cite{Gretton2012}, etc. However, these tests have their own limitations. For example, t-test is with normal assumption; Kolmogorov–Smirnov test cannot be used to multivariate cases; and kernel based one needs hyperparameter tuning.

Given two samples $\mathbf{X}_0=\{X_{01},\cdots,X_{0m}\}\sim P_0$ and $\mathbf{X}_1=\{X_{11},\cdots,X_{1n}\}\sim P_1$, the null hypothesis of two-sample test is
\begin{equation}
	H_0: P_0=P_1,
\end{equation}
the alternative is
\begin{equation}
H_1: P_0\neq P_1.
\end{equation}

Previously, the test statistic was usually defined by a  certain distance between $P_0$ and $P_1$:
\begin{equation}
	T = d(P_0,P_1).
\end{equation}
If $d$ is large, then $H_1$ is true; otherwise, $H_0$ is true. Some typical distances $d$ include MMD\cite{Gretton2012}, dCor\cite{Szekely2004}, Wasserstein distance\cite{Ramdas2017}, etc. Kullback\cite{Kullback1968} proposed a two-sample test based on KL divergence in information theory.

Ma \cite{Ma2023b} proposed a method for two-sample test based on CE\footnote{The method is implemented in the \texttt{copent} package in \textsf{R} and \textsf{Python}\cite{ma2021copent}.}. The idea is based on the dependence between samples and the labels for hypotheses.

Let $\mathbf{X}=(\mathbf{X}_0,\mathbf{X}_1)$, $H_0,H_1$be the labels for $Y_0,Y_1$, the proposed test defines its statistic based on the following criteria:
\begin{equation}
	T=d(\mathbf{X},Y_1)-d(\mathbf{X},Y_0).
\end{equation}
Let the labels for the null hypothesis and the alternative be $Y_0=(1_1,\cdots,1_{m+n})$ and $Y_1=(1_1,\cdots,1_m,2_1,\cdots,2_n)$, which mean two samples from same distribution and different distributions respectively, then the distance $d$ between samples and labels is defined with CE as
\begin{align}
	H_0: &~ H_c(\mathbf{x};y_0)=H_c(\mathbf{x},y_0)-H_c(\mathbf{x}),\\
	H_1: &~ H_c(\mathbf{x};y_1)=H_c(\mathbf{x},y_1)-H_c(\mathbf{x}),
\end{align}
which measure the fitfulness of samples to $H_0$ and $H_1$. The test statistic of the proposed method for two-sample test is defined as the difference between CEs of the null hypothesis and the alternative:
\begin{equation}
	T_{tst}(\mathbf{X}_0,\mathbf{X}_1)=H_c(\mathbf{x},y_0)-H_c(\mathbf{x},y_1).
\end{equation}
It can learn that if $H_0$ is true, then $T_{tst}$ is small; if $H_1$ is true, then $T_{tst}$ is large.

We proposed the estimation method of the test statistic $T_{tst}$ of the proposed method with the nonparametric CE estimator.

The proposed test is nonparametric and distribution-free. It can be used in both univariate cases and multivariate cases. It is also tuning free due to the CE estimator. $T_{tst}$ is based on information theory and hence it is scale-free.

Here, $Y_0,Y_1$ is designed for the general case of two-sample hypothesis. It can be tailored to particular situations.

We verified the effectiveness of the proposed test on three simulated data with normal distribution and Gaussian copula, and compared it with the tests based on MI, kernel function, and dCor\footnote{The code is available at \url{https://github.com/majianthu/tst}}. For benchmarking on more two-sample tests, please see Section \ref{s:tstbench}.

\section{Change point detection}
Change point detection\cite{Basseville1993,Brodsky2013} is a basic task in time series analysis and is to check abrupt points of system states. It can be offline or online, single point or multiple points, univariate or multivariate. It has wide applications in natural, social, biological, and industrial systems.

Since its first proposal in 50s last century\cite{Page1954,Page1955}, it has been intensively studied and there are many existing methods for it\cite{Truong2020,Aminikhanghahi2017}. Particularly, scholars contributed some methods based on entropy and copulas. For example, Gao et al.\cite{Gao2024a} summarized information criteria based change point detection; Sofronov et al.\cite{Sofronov2012} proposed a method based on cross entropy; Wang et al.\cite{Wang2013} proposed a entropy based method for bird song segmentation; Song and Xia\cite{Song2025} proposed a method based on normalized entropy; Zhu et al.\cite{Zhu2013} propose a copula based method. Dehling et al.\cite{Dehling2022} proposed a method based on weighted two-sample U-statistic. Meanwhile, there are work on detecting change points on copulas. Quessy et al.\cite{Quessy2013} proposed a copula structure change detection based on Kendall's $\tau$; Na et al.\cite{Na2012} proposed a CUSUM change point detection based on copula ARMA-GARCH model; Stark and Otto\cite{Stark2020} proposed a CUSUM style structure change detection based on copula based Spearman's $\rho$ and quintile correlation.

Change point detection can be transformed as two-sample test problem, i.e., performing two-sample test on two samples before and after a group of points and taking the point associated with maximum of the test statistics as the estimate. Given a time series $\mathbf{X}_t,t=1,\ldots,T$, to judge the existence of a change point $t_c$ where the series $\mathbf{X}_b,\mathbf{X}_a$ before and after it are different, a two-sample test for this is to verify the null hypothesis
\begin{equation}
	H_0: P_b(\mathbf{x_b}) = P_a(\mathbf{x_a}),
\end{equation}
and the alternative 
\begin{equation}
	H_1: P_b(\mathbf{x_b}) \neq P_a(\mathbf{x_a}),
\end{equation}
where $P_b,P_a$ is the probability distributions before and after the change point.

Based on this idea, Ma \cite{Ma2024a} proposed a nonparametric multivariate single change point detection based on the CE based two-sample test in Section \ref{sec:tst}. The proposed test is to do two-sample test on each point of $X_t$ to derive a group of the corresponding test statistics $T_{tst}(t),t=1,\ldots,T$, and take the point associated with maximum of the statistics as change point:
\begin{equation}
	t_c = \mathop{\arg\max}\limits_{t} T_{tst}(\mathbf{X}_b,\mathbf{X}_a;t).
\end{equation}
If $T_{tst}(t_c)$ is greater than a pre-specified threshold, then $H_1$ is true, $t_c$ is considered as change point; otherwise, $H_0$ is true, no change point exists.

We then proposed a method for multiple change point detection\footnote{The method for change point detection is implemented in the \texttt{copent} package in \textsf{R} and \textsf{Python}\cite{ma2021copent}}, which composed of the following steps:
\begin{enumerate}
	\item add the series to be checked to the waiting list;
	\item do single change point detection to the next series in the waiting list;
	\item if there exists a change point, then add the two series before and after the detected change point to the waiting list;
	\item loop step 2 and 3 till there is no series in the waiting list.
\end{enumerate}

The proposed method judges the existence of change point by a pre-specified threshold which makes the detection automatic. Since the test statistic is a CE based information theoretical measure and problem-free, the threshold can then be set as a value for any situation and therefore makes the proposed method tuning-free.

We verified the proposed method on a group of simulated data and compared it with its competitors. For more on benchmarking the methods for change point detection, please see Section \ref{s:cpdbench}.

We also verified the proposed method on two real data widely used for change point detection: the Nile data and the British coal mine disasters data. The Nile data contains the measurements of the annual flow of the Nile River at Aswan during 1871-1970. After the construction of the Aswan dam in 1898, the annual flow became smaller. The British coal mine disasters data records the dates of British coal mining explosions with fetal consequences during 1851-1962. It is widely believed that due to legislation in 1887, the disasters decreased sharply. Experiments on these two datasets show that the proposed method derived a change point (1898) for the Nile data and a change point (1892) for the British coal mine disasters data.

\begin{figure}
	\centering
	\includegraphics[width=0.7\textwidth]{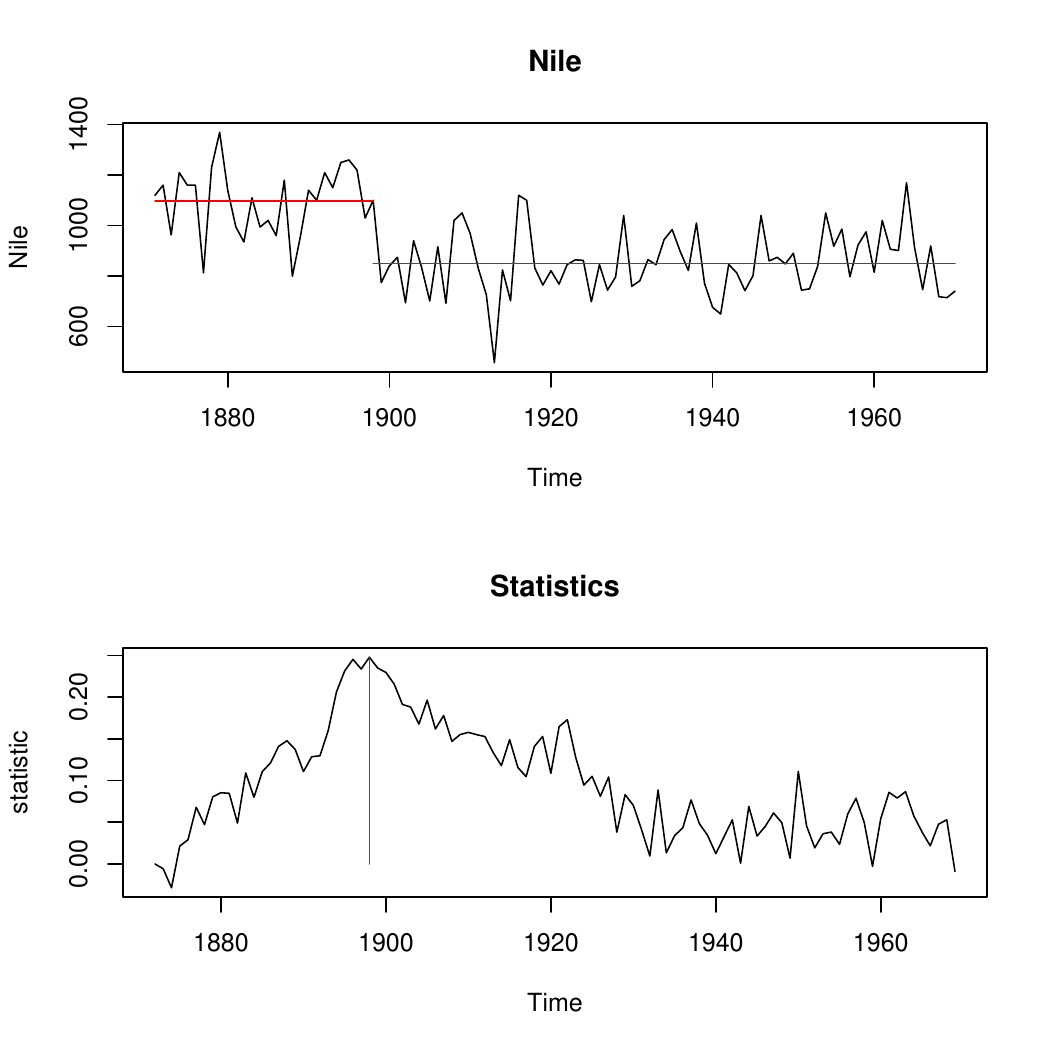}
	\caption{Results of single change point detection on the Nile data.}
	\label{fig:nile}
\end{figure}

\begin{figure}
	\centering
	\includegraphics[width=0.7\textwidth]{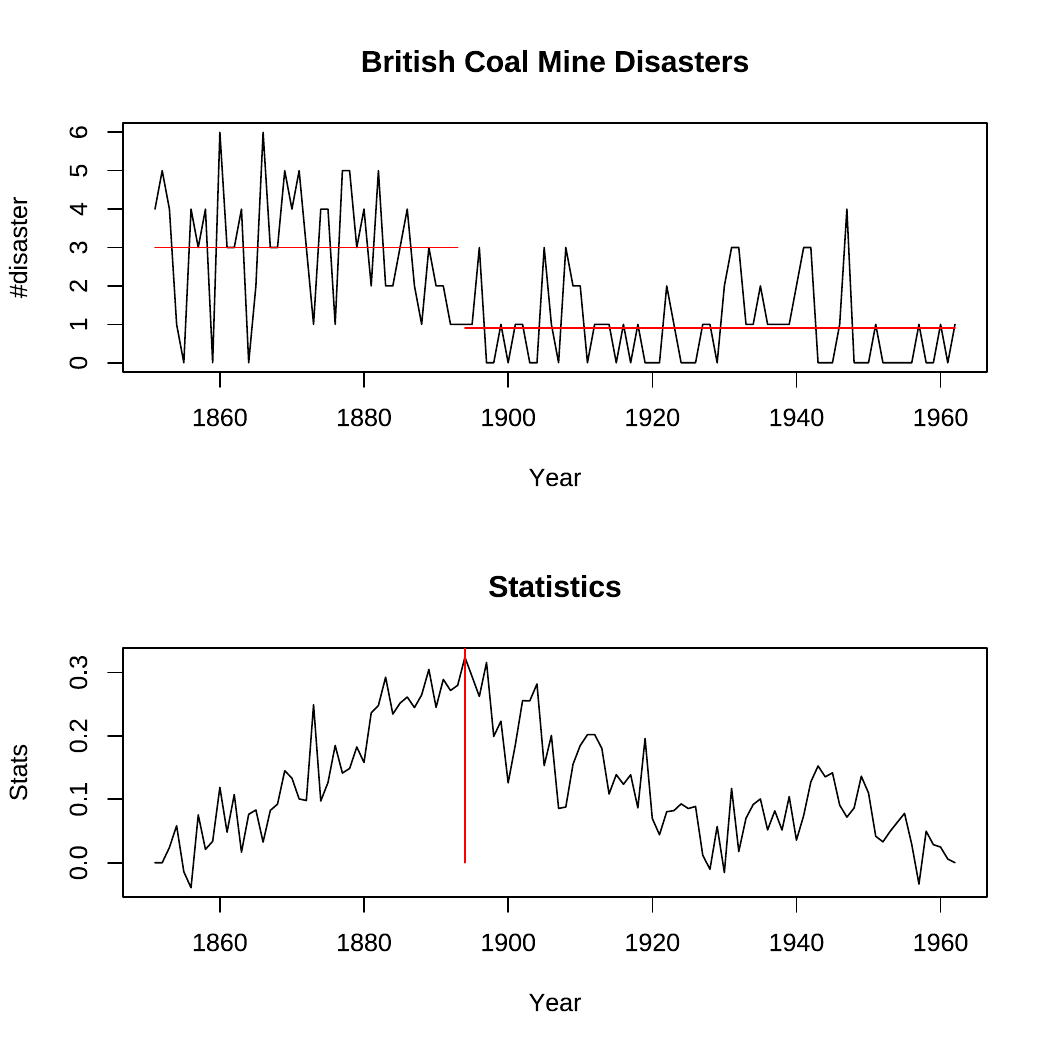}
	\caption{Results of single change point detection on the British coal mining disasters data.}
	\label{fig:coalmine}
\end{figure}

\section{Symmetry test}
Symmetry is one of the fundamental principles in sciences and plays a critical role in theoretical physics\cite{Gross1996,Sundermeyer2014}. Mathematically, symmetry is associated with invariance on groups\cite{Armstrong1997}. Symmetry is also a basic assumption in probability and mathematical statistics\cite{Laplace1820,Bartha2001}. There are four types of probabilistic symmetries: stationarity, contractability, exchangeability, and rotatability\cite{Kallenberg2005}. Symmetry provides a unified perspective for hypothesis testing: normality test is to test normal distribution which is invariant to rotation; two-sample test is to test symmetry of distributions about change point, or invariant to shift transformation. Many traditional models and methods are based on the symmetry assumption\cite{Dawid1985}, such as t-test\cite{Prescott2006}, Wilcoxon signed-rank test\cite{Wilcoxon1945,Randles2006}. So, symmetry test is a basic problem in hypothesis testing.

There are many symmetry tests\cite{Allison2017,Baringhaus1992,Bozin2020,Miao2006,Milosevic2016,Milosevic2019,Mira1999,Cabilio1996,Feuerverger1977,Gaigall2020,Litvinova2001,Nikitin2015}. Particularly, scholars contributed some methods based on copulas.

For example, Bouzeboda and Cherfi\cite{Bouzebda2012a} proposed to test symmetry with copulas; Kato et al.\cite{Kato2020} proposed a copula based measure for asymmetry of the lower and upper tail probabilities. Mokhtari et al.\cite{Mokhtari2023} proposed to measure and test conditional asymmetry with copulas. There are some work on entropy based symmetry tests. For instance, Bormashenko et al.\cite{Bormashenko2021} proposed entropy based measures for symmetry of tiling; Noughabi and jarrehiferiz\cite{Noughabi2022} proposed to test symmetry with extropy of order statistics; Maasoumi and Racine\cite{Maasoumi2008} proposed to test asymmetry of processes based on entropy; Avhad et al.\cite{Avhad2025} proposed to test distribution symmetry based on entropy.

Symmetry of distributions is defined as invariance to certain transformations. Given random variables $X\sim P$, to test its symmetry to transformation $T$ is to judge the equivalence of the distributions of $X$ and $TX$, i.e., to check the following hypothesis:
\begin{equation}
	p(x)=p(Tx).
\end{equation}
So, symmetry test is transformed into a two-sample test problem.

Given a sample $X=\{X_i,i=1,\cdots,N\}$ from probability density function $p(x)\in R$, $u$ is the expectation of $p(x)$, symmetry test is to check the invariance of $p$ with respect to $u$ from $X$. So, the null hypothesis of symmetry test is
\begin{equation}
	H_0:p(u-x)=p(u+x);
\end{equation}
the alternative is
\begin{equation}
	H_1:p(u-x)\neq p(u+x).
\end{equation}

Therefore, symmetry test is transformed as a two-sample test: if $u-x$ and $u+x$ are from same distribution, $H_0$ is true, then $p$ is symmetric; otherwise, $H_1$ is true, then $p$ is asymmetric.

Based on this idea, Ma\cite{Ma2025a} proposed a method for symmetry test based on CE based two-sample test. Suppose $\tilde{X}$ be $X$ reduced its mean $\tilde{u}$, the problem is to judge whether $\tilde{X}$ and $-\tilde{X}$ are from same distribution. The test statistic is then defined as
\begin{equation}
	T_{sym}(X)=T_{tst}(\tilde{X},-\tilde{X}).
	\label{eq:symstats}
\end{equation}
If $p$ is symmetric, then $T_{sym}=0$; otherwise, $T_{sym}>0$, $p$ is asymmetric. The larger $T_{sym}$, the more asymmetric $p$.

The test such proposed is composed of two steps:
\begin{enumerate}
	\item derive $\tilde{X}=X-\tilde{u}$ and $-\tilde{X}$;
	\item estimate $T_{sym}$ according to \eqref{eq:symstats}.
\end{enumerate}

Since CE based $T_{tst}$ can be estimated non-parametrically, $T_{sym}$ can also be estimated so.

We verified the effectiveness of the proposed method with simulation experiments and compared it with the other methods\footnote{The code is available at \url{https://github.com/majianthu/symmetry}}. Please see Section \ref{s:symbench} for more.

The proposed method for symmetry test can be generalized to multivariate test by replacing $p$ with multivariate distributions\cite{Serfling2014}. By replacing mirror symmetry transformation in the proposed method with more complex symmetry transformations\cite{Martin1982}, one can derive tests for other types of symmetries, such as exchangeability test by permutations.

\chapter{Discussion}
\label{chap:discussion}
\section{Relations between theoretical applications}
The first four applications in Chapter \ref{chap:theoapp} have inner relationships with each other. Theoretically, they are all based on the CE based framework on independence and conditional independence with the common goal of discovering intrinsic dependence relations. The difference between them is in the way that dependence relation is represented: associations between pairs of variables are discovered in the application of association discovery; in structure learning, graph is used as the representation of dependence relations between random variables; the aim of variable selection is to build a functional model based on dependence between variables and target variables; causal relations between time series variables are found and can be represented as directed graph for causal variable selection or building functions of causal relationships.

Both independence test and conditional independence test are the theoretical foundations of the other applications and therefore a system of methodologies based on CE is formed (see Figure \ref{fig:ceapprelation}) as a versatile toolbox for various applications. Particularly, independence test can be used for variable selection, whether in static functions, including survival functions, or in dynamical functions, such as differential equations. Conditional independence test can be used for causal discovery by estimating TE via CE; furthermore, time lag can be estimated with TE estimation. CE based hypothesis testings include multivariate normality test, copula hypothesis testing, and two-sample test which can be further used for change point detection and symmetry test.

\tikzstyle{ceapp} = [rectangle, rounded corners, minimum width=1cm, minimum height=0.7cm, text centered, draw=black]
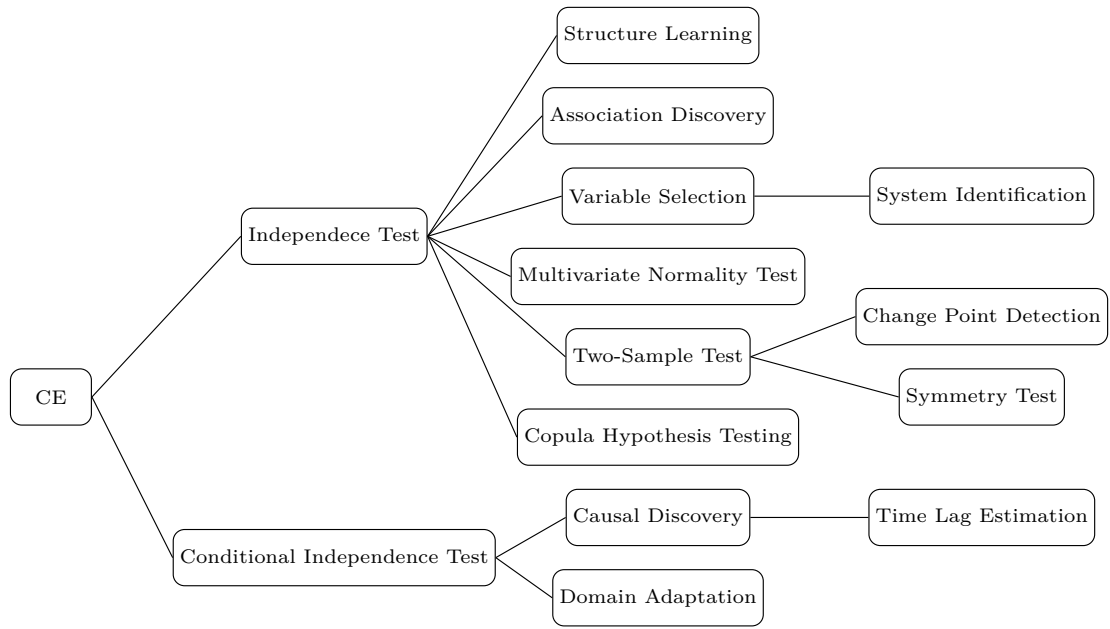
\begin{figure}
\centering
\resizebox{\textwidth}{!}{
\begin{tikzpicture}
[
level 1/.style = {sibling distance = 4cm, level distance = 3.5cm,},
level 2/.style = {sibling distance = 1cm, level distance = 4cm},
grow' = right,
]
\scriptsize
	\node [ceapp]{CE}
		child {node [ceapp] {Independece Test}
			child {node [ceapp]{Structure Learning}}
			child {node [ceapp]{Association Discovery}}
			child {node [ceapp]{Variable Selection}
				child {node [ceapp]{System Identification}}
				}
			child{node[ceapp]{Multivariate Normality Test}}
			child{node(tst)[ceapp]{Two-Sample Test}
				child{node[ceapp]{Change Point Detection}}
				child{node[ceapp]{Symmetry Test}}
			}
			child{node[ceapp]{Copula Hypothesis Testing}}
		}		
		child {node [ceapp] {Conditional Independence Test}
			child{node[ceapp]{Causal Discovery}
				child{node[ceapp]{Time Lag Estimation}}
			}
			child{node[ceapp]{Domain Adaptation}}
			};

\end{tikzpicture}
}
\caption{Relations between theoretical applications of CE.}
\label{fig:ceapprelation}
\end{figure}

\section{Correlation and causality}
Correlation and causality are two closely related concepts in statistics which correspond to independence and conditional independence respectively. The former can be measured with CE directly while the latter can be measured with CE based TE. And hence we built a theoretical framework for both based on CE and proposed the estimators of them.

\section{Comparison to the other frameworks}
Besides our CE based framework for independence and conditional independence, there are two other similar frameworks: one based on kernel functions\cite{pfister2018kernel,zhang2011kernel} and the other based on dCor\cite{Szekely2009,wang2015conditional}. However, our CE based framework has theoretical advantages due to rigorous definition of CE and additionally, the CE estimator has merits, such as simple and elegant, easy to understand and implement, low computational complexity.

Table \ref{t:3measures} compares these three independence measures where one can see the advantages of CE clearly. For example, CE is a natural and unique measure of multivariate independence, while the other two need to be generalized to multivariate cases; CE has the properties including invariance to monotonic transformation and equivalence to second order statistic in Gaussian cases while dCor has similar properties and HSIC has not. The computational complexity of the CE estimator is low while that of the other two is high.

All three frameworks have developed their own system of methodologies\cite{Szekely2023,Muandet2017}, including independence test, conditional independence test, normality test, two-sample test, and change point detection. In the applications on variable selection and causal discovery, we compared independence tests and conditional independence tests of these three frameworks on real data and experimental results have demonstrated the advantages of CE over the other two measures. There are normality test and two-sample test in these three frameworks and CE based ones is much more rigorous theoretically and therefore presents better performance in simulation comparisons. There are change point detections based on both CE and kernels and simulation experiments showed that CE based one presents better performance than kernel based one. Additionally, CE can also be used for copula hypothesis testing.

Note that HSIC and dCor are equivalent on independence test\cite{Sejdinovic2013} and conditional independence test\cite{Sheng2023} and that they both have copula based versions\cite{Lin2017,Boettcher2020,Poczos2012}.

\begin{table}
	\setlength{\abovecaptionskip}{0pt}
	\setlength{\belowcaptionskip}{5pt}
	\centering
	\caption{Comparison between the frameworks based on three independene measures.}
	\begin{tabular}{l|c|c|c} 
		\toprule
\textbf{Framework}&\textbf{CE}&\textbf{dCor}&\textbf{HSIC}\\
\midrule
Definition&copula based&generalized correlation&kernel based\\
\hline
Multivariate&by nature&\cite{Boettcher2019}&\cite{pfister2018kernel}\\
\hline
invariance&monotonic transformation&linear transformation&\\
\hline
CC in Gaussianity&equivalence&equivalence&\\
\hline
Comp. complexity&$O(n^2)$&$O(n^4)$&$O(n^4)$\\
\hline
Independence test&\cite{ma2011mutual}&\cite{Szekely2007,Boettcher2019}&\cite{gretton2007a,pfister2018kernel}\\
\hline
CI test&\cite{jian2019estimating}&\cite{wang2015conditional}&\cite{zhang2011kernel}\\
\hline
Normality test&\cite{Ma2022}&\cite{Szekely2005}&\cite{Kellner2019}\\
\hline
Copula hypothesis test&\cite{Ma2025c}&--&--\\
\hline
Two-sample test&\cite{Ma2023b}&\cite{Szekely2004}&\cite{Gretton2012}\\
\hline
Change point detection&\cite{Ma2024a}&--&\cite{Harchaoui2008}\\
\hline
Symmetry test&\cite{Ma2025a}&\cite{Chen2025a}&--\\
\bottomrule
\end{tabular}
\label{t:3measures}
\end{table}

\chapter{Evaluation}
\label{chap:bench}
\section{Introduction}
We proposed a CE based system of methodologies for hypothesis testing, including independence test, conditional independence test, multivariate normality test, two-sample test, change point detection, and symmetry test. Meanwhile, there are similar methods for each of these methodologies. To benchmark on them, we designed simulation experiments for each type of methodology. We implements the CE based methodologies in the \texttt{copent} package in \textsf{R}\cite{ma2021copent}, and compare them to the similar methods implemented in both \textsf{R} and \textsf{Python}. This chapter presents the results of simulation experiments for benchmarking.

Simulation experiments show that CE based methodologies perform best among all the methods in the experiments. Based on these, we believe that CE based system of methodologies is the most scientific and effective answer to these problems of hypothesis testing.

\section{Independence tests}
\label{s:indbench}
Independence is a fundamental concept in statistics. Since PCC, measure of independence is a core problem of the discipline and there are many measures proposed already \cite{Tjostheim2022}, including the three main measures mentioned in the previous chapter. The problem is which one is best among all? To answer this, R\'enyi\cite{Renyi1959} proposed the famous seven axioms for dependence measures. Schweizer and Wolff\cite{Schweizer1981} revised these axioms later.

Evaluation of these measures is a practical problem. Ma \cite{Ma2022b} designed a group of simulation experiments toward this goal\footnote{The code is available at \url{https://github.com/majianthu/eval}} and compared sixteen measures (see Table \ref{tb:benchindep}) in them, including

\paragraph{Kendall's $\tau$} 
Kendall's $\tau$\cite{Kendall1938} is a commonly used bivariate independence measure based on rank. Given  sign function $s$, Kendall's $\tau$ is defined as
\begin{equation}
	\tau=\frac{1}{n^2}\sum_{i,j=1}^{n}s(x_i-x_j)s(y_i-y_j),
\end{equation}
where $n$ is sample size.

\paragraph{Hoeffding's D}
Hoeffding's D\cite{hoeffding1948a}is another nonparametric bivariate independence measure.To test independence is to test the following equality: 
\begin{equation}
	F_{XY}=F_X F_Y,
\end{equation}
where $F_{XY}$ and $F_X,F_Y$ are joint and margins of $X,Y$ respectively. Hoeffding defined the functional $D$ as test statistic:
\begin{equation}
	D(x,y)=\int \{F_{XY}(x,y)-F_X(x)F_Y(y)\}^2 dF_{XY}(x,y).
	\label{eq:hoeffding}
\end{equation}

\paragraph{Bergsma-Dassios's $\tau^*$}
Inspired by Kendall's $\tau$, Bergsma and Dassios defined a new bivariate independence measure $\tau^*$\cite{bergsma2014a}. Given random variables $X,Y$, $\tau$ is defined as
\begin{equation}
	\tau^*(X,Y)=\frac{1}{n^4}\sum_{i,j,k,l=1}^{n}a(x_i,x_j,x_k,x_l)a(y_i,y_j,y_k,y_l),
\end{equation}
where sign function $a$ is defined as
\begin{equation}
	a(z_1,z_2,z_3,z_4)=s(|z_1-z_2|+|z_3-z_4|-|z_1-z_3|-|z_2-z_4|).
\end{equation}

\paragraph{HHG}
Heller et al.\cite{heller2016consistent} proposed a bivariate independence measure based on Hoeffding's $D$ \eqref{eq:hoeffding}. 

\paragraph{Ball}
Pan et al.\cite{Pan2018} proposed an measure of independence, called Ball Covariance. 

\paragraph{BET}
Zhang\cite{Zhang2019} proposed a framework for independence test, called Binary Expansion Testing (BET). 

\paragraph{QAD}
Trutschnig\cite{Trutschnig2011} proposed a measure of asymmetric dependence (QAD) $\zeta$ which is defined as the distance between copula and independence copula:
\begin{equation}
	\zeta(C,\Pi)=D_1(C,\Pi),
\end{equation}
where $D_1$ is a type of distance between copulas defined with Markov kernels.

\paragraph{CODEC}
If random variables $X,Y$ are independent, then
\begin{equation}
	E(Y|X)=E(Y).
	\label{eq:eyxey}
\end{equation}
Based on this, Chatterjee\cite{Chatterjee2021} defined a simple correlation coefficient (CODEC). Given random variables $(X,Y)$, let $X_1\leq \cdots\leq X_n$, $r_i$ is the rank of $Y_i$, CODEC is defined as
\begin{equation}
	\xi_n(X,Y)=1-\frac{3\sum_{i=1}^{n-1}|r_{i+1}-r_i|}{n^2-1}.
\end{equation}

\paragraph{Mixed}
Genest et al.\cite{Genest2019} proposed a measure of independence based on empirical copula:
\begin{equation}
	G_c(\hat{C},\Pi)=\int_{I^d}n(\hat{C}-\Pi)^2 d\mathbf{u},
\end{equation}
where $\hat{C}$ is empirical copula, $\prod$ is independence copula.

\paragraph{Subcopula}
Erdely\cite{Erdely2017} proposed a subcopula based measure of independence which is defined as
\begin{equation}
	d(S)=\sup\{S-\Pi\}-\sup\{\Pi-S\},
\end{equation}
where $S$ is bivariate subcopula.

\paragraph{dCor}
dCor is an important measure of bivariate dependence proposed by Sz\'ekely et al.\cite{Szekely2007,Szekely2009}. Given random variables $(X,Y)$, dCor is defined as
\begin{equation}
	dCor(X,Y)=\frac{\nu^2(X,Y)}{\sqrt{\nu^(X)\nu^2(Y)}},
\end{equation}
where $\nu^2(X,Y)$ is distance covariance defined as
\begin{equation}
	\nu^2(X,Y;w)=||f_{X,Y}(t,s)-f_X(t)f_Y(s)||_w^2,
	\label{eq:dcornu}
\end{equation}
$||\cdot||_w$ is L2 norm.

\paragraph{MDC}
Shao and Zhang\cite{Shao2014} proposed a variant of dCor, called Martingle Difference Correlation (MDC), which is essentially to test \eqref{eq:eyxey}.

\paragraph{HSIC}
HSIC\cite{gretton2007a} is a widely studied measure of bivariate independence based on kernel function. HSIC has its multivariate version, called dHSIC\cite{pfister2018kernel}.

\paragraph{NNS}
Voile and Nawrocki\cite{Viole2012} proposed a nonlinear correlation coefficient based on partial moments, called  Nonlinear Nonparametric Statistic (NNS).

~\

We designed six simulation experiments, three of whom is for bivariate independence, two for trivariate independence, and one for quadvariate independence, as follows:
\begin{enumerate}
	\item Random variables $\mathbf{X}=(X_1,X_2)$ governed by bivariate normal distribution $\mathbf{X}\sim N(\mathbf{u},\rho)$, whose covariance $\rho$ changes from 0 to 0.9 by step 0.1;
	\item random variables $\mathbf{X}=(X_1,X_2)$ governed by bivariate normal copula $C_N(u,v;\rho)$ with margins as normal distribution $u\sim N(0,1)$ and exponential distribution $v\sim E(\lambda=2)$, the parameter $\rho$ of normal copula changes from 0 to 0.9 by step 0.1;
	\item random variables $\mathbf{X}=(X_1,X_2)$ governed by Archimedean copulas, including Clayton copula, Frank copula, and Gumbel copula as follows:
	\begin{equation}
		C_\alpha^{Clayton}(u,v)=max\left([u^{\alpha} + v^{\alpha} -1]^{-\frac{1}{\alpha}},0\right),
	\end{equation}
	\begin{equation}
		C_\alpha^{Frank}(u,v)=-\frac{1}{\alpha}\ln\left(1+\frac{(e^{-\alpha u}-1)(e^{-\alpha v}-1)}{e^{-\alpha}-1}\right),
	\end{equation}
	\begin{equation}
		C_\alpha^{Gumbel}(u,v)=\exp \left\lbrace  -[(-\ln u)^{\alpha}+(-\ln v)^{\alpha}]^{\frac{1}{\alpha}}\right\rbrace ,
	\end{equation}
	margins are same as above, the parameter $\alpha$ of the copulas change from 1 to 10;
	\item random variables $\mathbf{X}=(X_1,X_2,X_3)$ governed by trivariate normal distribution $\mathbf{X}\sim N(\mathbf{u},\Sigma)$, where covariance matrix $\Sigma$ is
	\begin{equation}
		\begin{vmatrix}
			1 & \rho & \rho \\
			\rho & 1 & \rho \\
			\rho & \rho & 1
		\end{vmatrix},
	\end{equation}
	$\rho$ changes from 0 to 0.9 by step 0.1;
	\item random variables $\mathbf{X}=(X_1,X_2,X_3)$ governed by trivariate Gumbel copulas as follows:	
	\begin{equation}
				C_\alpha^{Gumbel}(u,v,w)=\exp \left\lbrace  -[(-\ln u)^{\alpha}+(-\ln v)^{\alpha}+(-\ln w)^{\alpha}]^{\frac{1}{\alpha}}\right\rbrace ,
	\end{equation}	
	the parameter $\alpha$ change from 1 to 10, margins are a normal distribution $u\sim N(0,2)$ and two exponential distributions $v\sim E(\lambda=2),w\sim E(\lambda=0.5)$;
	\item random variables $\mathbf{X}=(X_1,X_2,X_3,X_4)$ governed by quadvariate normal distribution $\mathbf{X}\sim N(\mathbf{u},\Sigma)$ for simulating dependence relations between two bivariate random vectors, the covariance matrix $\Sigma$ is
		\begin{equation}
		\begin{vmatrix}
			1 & \rho_{12} & \rho & \rho \\
			\rho_{12} & 1 & \rho & \rho \\
			\rho & \rho & 1 & \rho_{34} \\
			\rho & \rho & \rho_{34} & 1
		\end{vmatrix},
	\end{equation}
	$\rho_{12}=0.8,\rho_{34}=0.75$, and $\rho$ changes from 0 to 0.8 by step 0.1.
\end{enumerate}
Experimental results of the six simulations are shown in Figure \ref{fig:sim1bi1normal}, Figure \ref{fig:sim1bi1normalcop1}, Figure \ref{fig:sim1bi1cop}, Figure \ref{fig:sim1multi1}, and Figure \ref{fig:sim1multi2}. Additionally, we also evaluated these measures on two real dataset (heart disease data and wine data), please refer to \cite{Ma2022b} for details. It can be learned from the results that CE presents the best performance among all.

\begin{table}
	\centering
	\caption{Implementations of the measures of independence evaluated.}
	\begin{tabular}{l|c|c}
		\toprule
		\textbf{Package}&\textbf{Measure}&\textbf{Language}\\
		\midrule
		\texttt{copent}&CE\cite{ma2011mutual}&\textsf{R}\\
		\texttt{stats}&Ktau\cite{Kendall1938}&\textsf{R}\\		
		\texttt{energy}&dCor\cite{Szekely2007}&\textsf{R}\\
		\texttt{dHSIC}&dHSIC\cite{pfister2018kernel}&\textsf{R}\\
		\texttt{HHG}&HHG.chisq, HHG.lr\cite{heller2016consistent}&\textsf{R}\\
		\texttt{independence}&Hoeff\cite{hoeffding1948a}, BDtau\cite{bergsma2014a}&\textsf{R}\\
		\texttt{Ball}&Ball\cite{pan2020ball}&\textsf{R}\\
		\texttt{qad}&QAD\cite{Trutschnig2011}&\textsf{R}\\
		\texttt{BET}&BET\cite{Zhang2019}&\textsf{R}\\
		\texttt{MixedIndTests}&Mixed\cite{Genest2019}&\textsf{R}\\
		\texttt{subcopem2D}&subcopula\cite{Erdely2017}&\textsf{R}\\
		\texttt{EDMeasure}&MDM\cite{Shao2014}&\textsf{R}\\
		\texttt{FOCI}&CODEC\cite{Chatterjee2021}&\textsf{R}\\
		\texttt{NNS}&NNS\cite{Viole2012}&\textsf{R}\\		
		\bottomrule
	\end{tabular}
	\label{tb:benchindep}
\end{table}

\begin{figure}
	\centering
	\includegraphics[width=\textwidth]{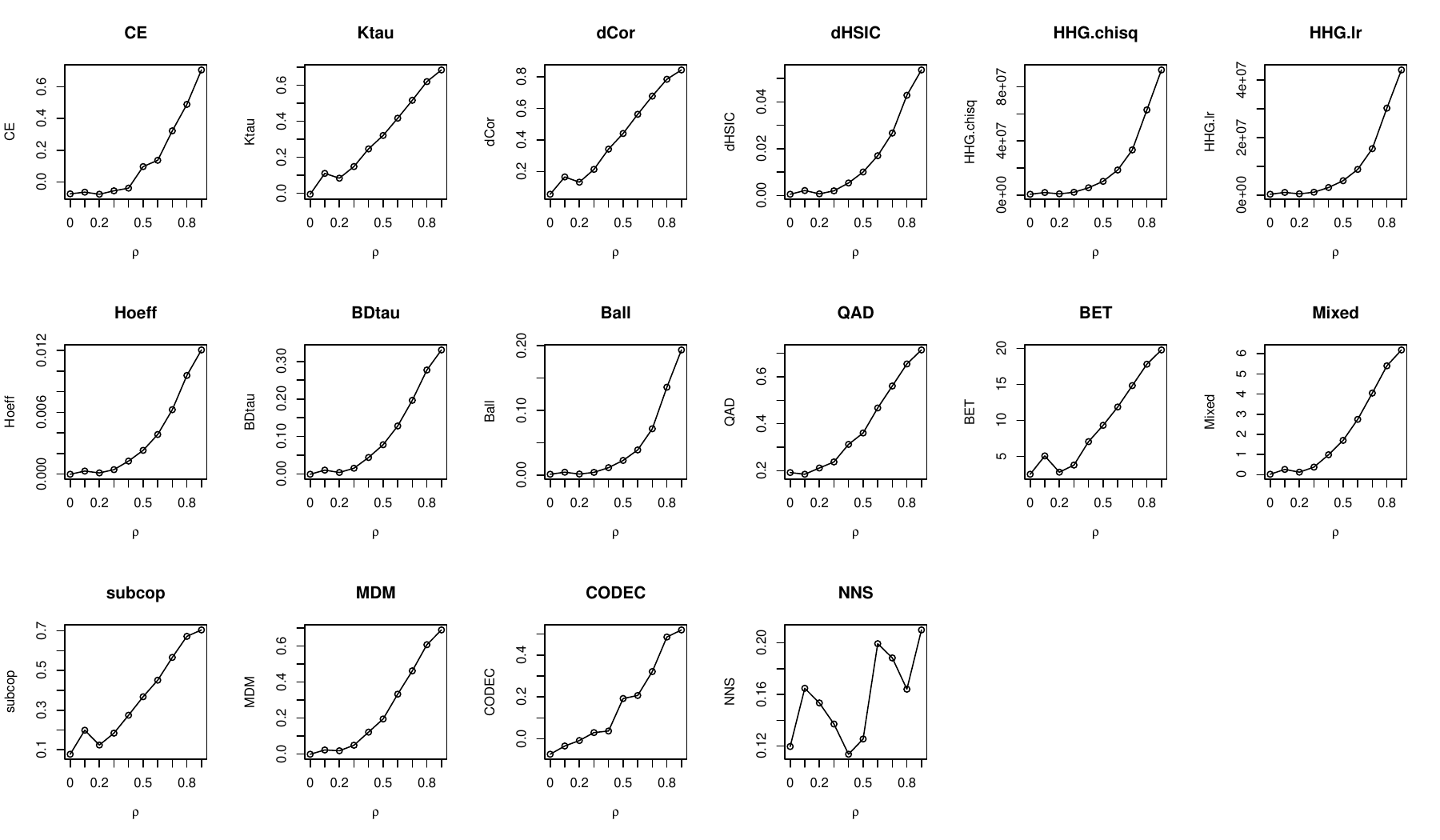}
	\caption{Results of bivariate independence test on bivariate normal distribution.}
	\label{fig:sim1bi1normal}
\end{figure}

\begin{figure}
	\centering
·	\includegraphics[width=\textwidth]{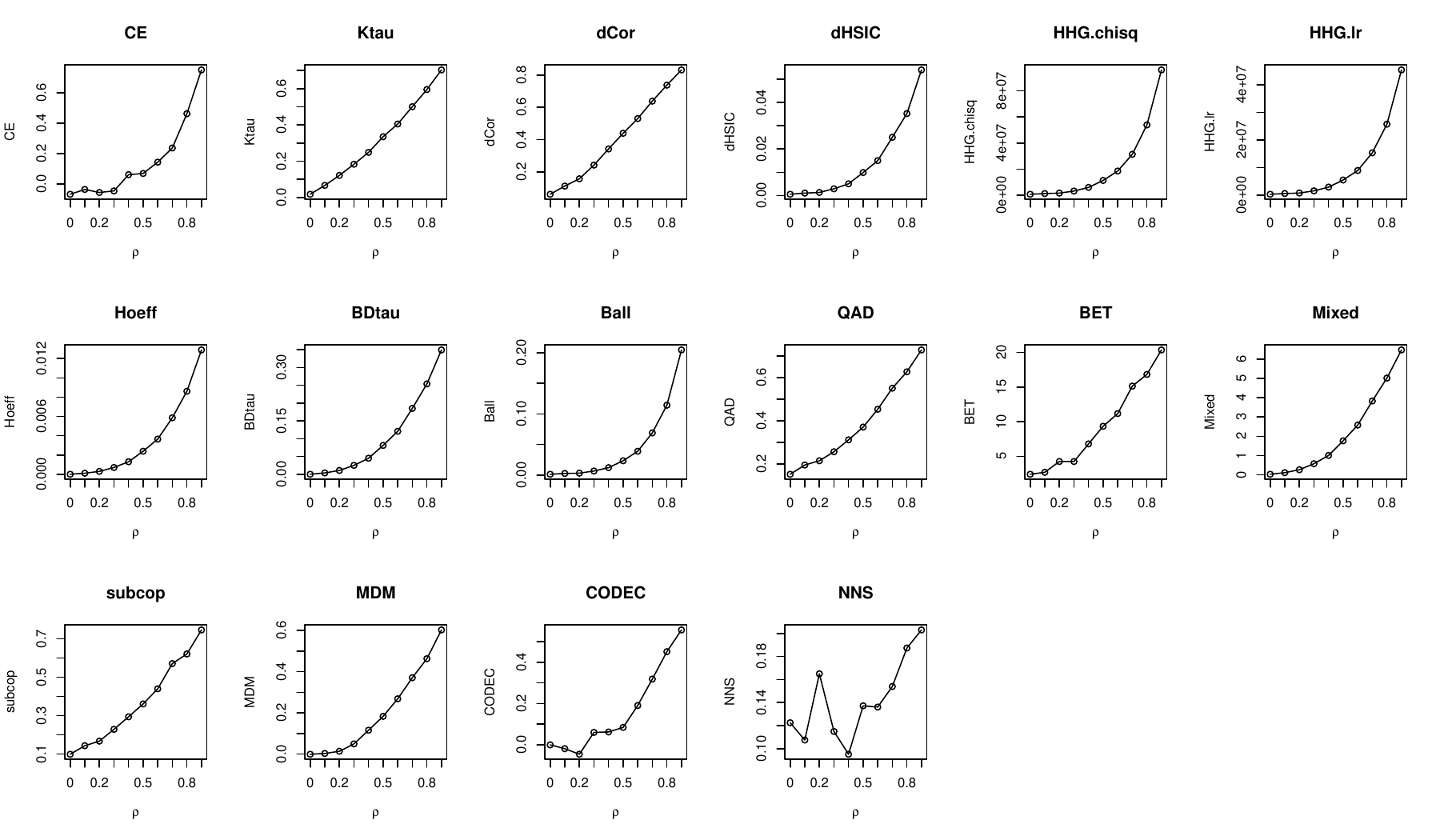}
	\caption{Results of bivariate measures of independence on bivariate normal copulas.}
	\label{fig:sim1bi1normalcop1}
\end{figure}

\begin{figure}
	\centering
	\subfigure[Clayton Copula]{\includegraphics[width=\textwidth]{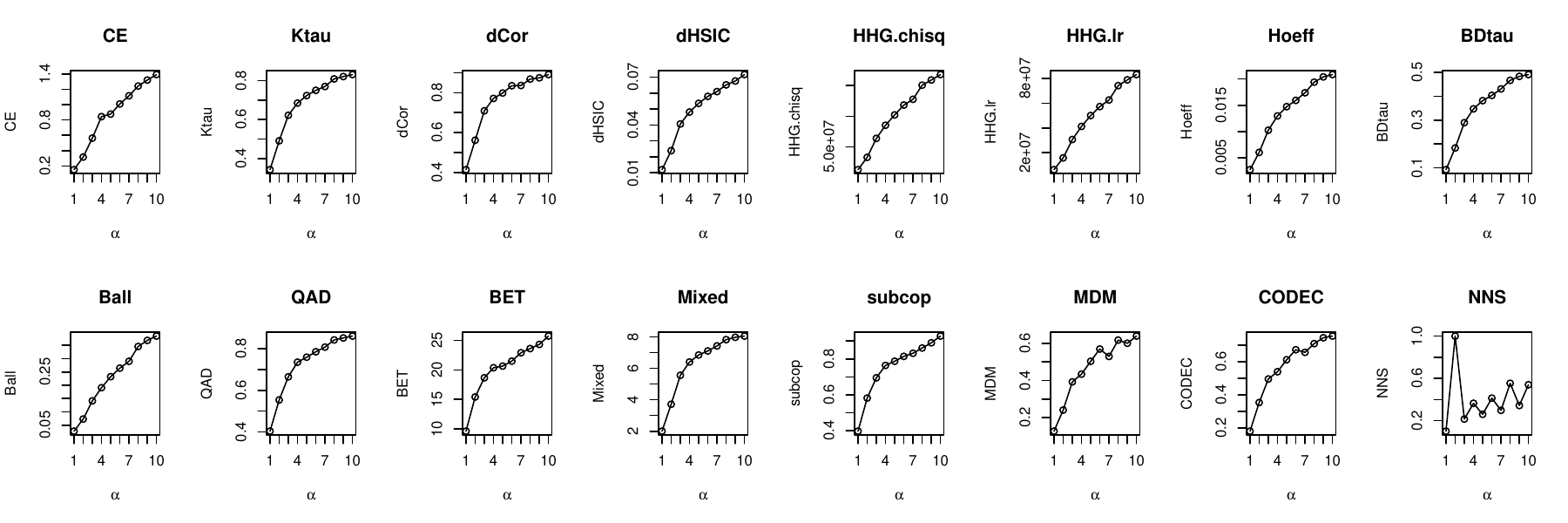}}
	\subfigure[Frank Copula]{\includegraphics[width=\textwidth]{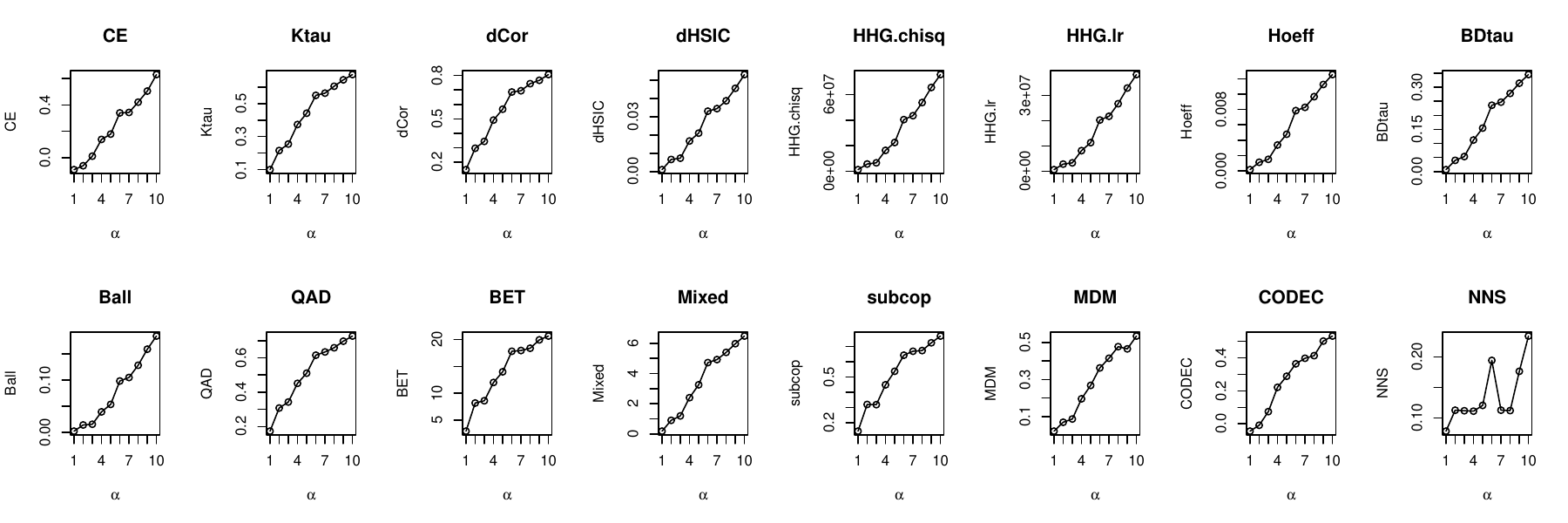}}
	\subfigure[Gumbel Copula]{\includegraphics[width=\textwidth]{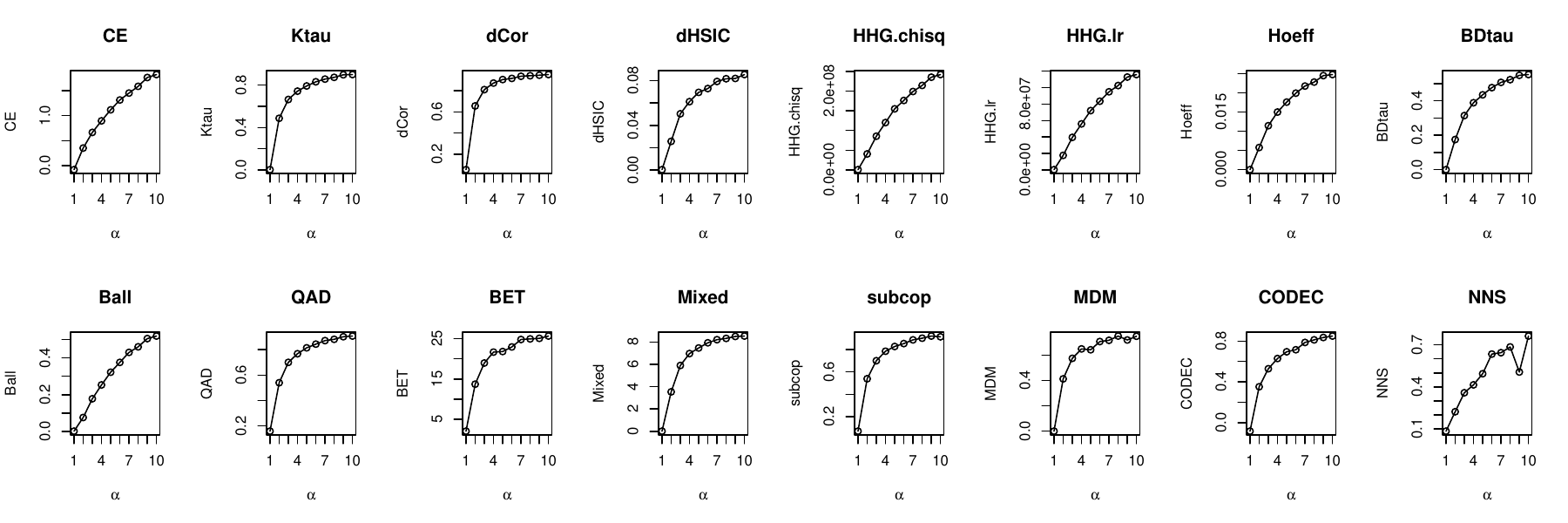}}
	\caption{Results of bivariate independence test on bivariate Archimedean copulas.}
	\label{fig:sim1bi1cop}
\end{figure}

\begin{figure}
	\centering
	\subfigure[trivariate normal distribution]{\includegraphics[width=\textwidth]{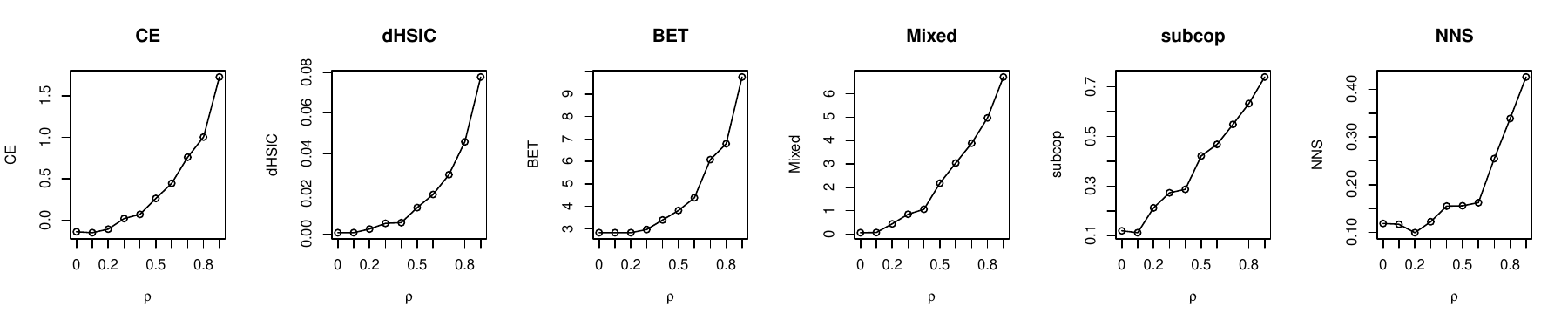}}
	\subfigure[trivariate Clayton copulas]{\includegraphics[width=\textwidth]{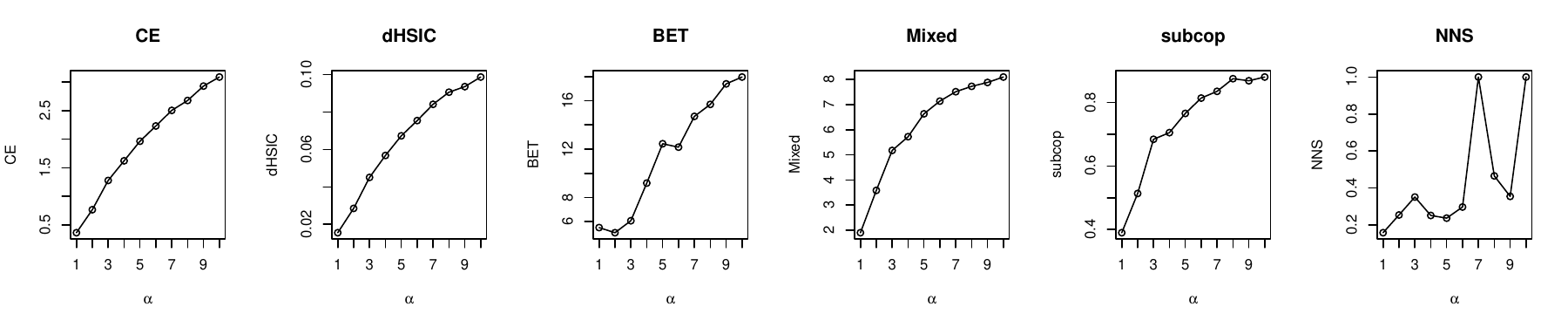}}
	\subfigure[trivariate Frank copulas]{\includegraphics[width=\textwidth]{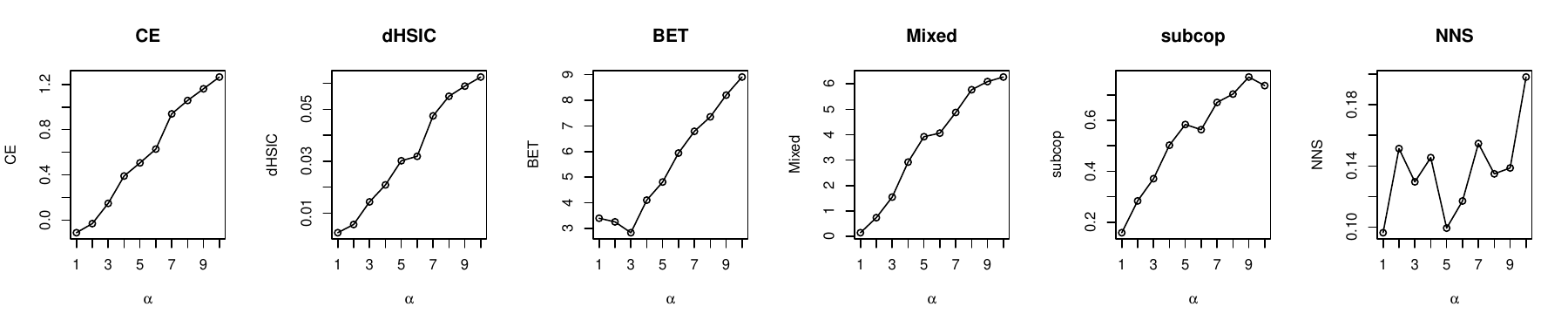}}
	\subfigure[trivariate Gumbel copulas]{\includegraphics[width=\textwidth]{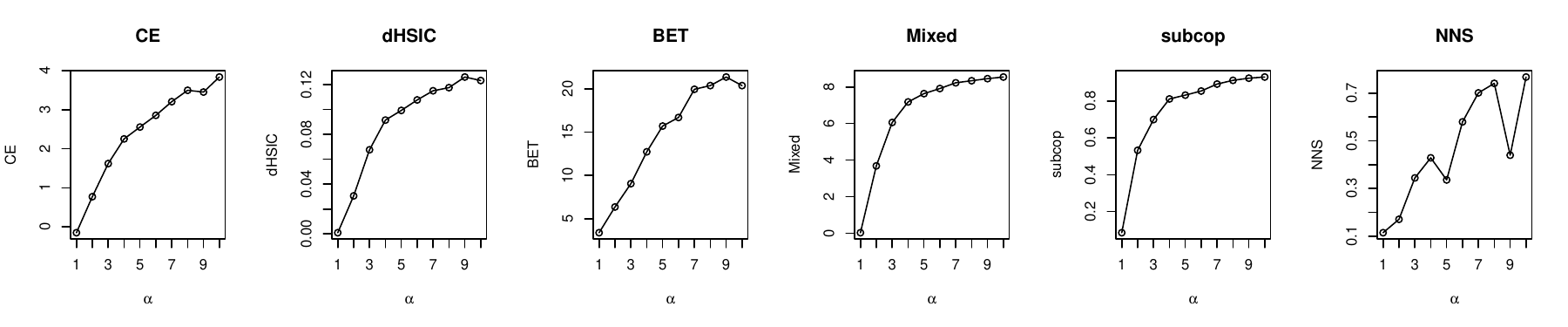}}
	\caption{Experiments on multivariate dependence test.}
	\label{fig:sim1multi1}
\end{figure}

\begin{figure}
	\centering
	\includegraphics[width=\textwidth]{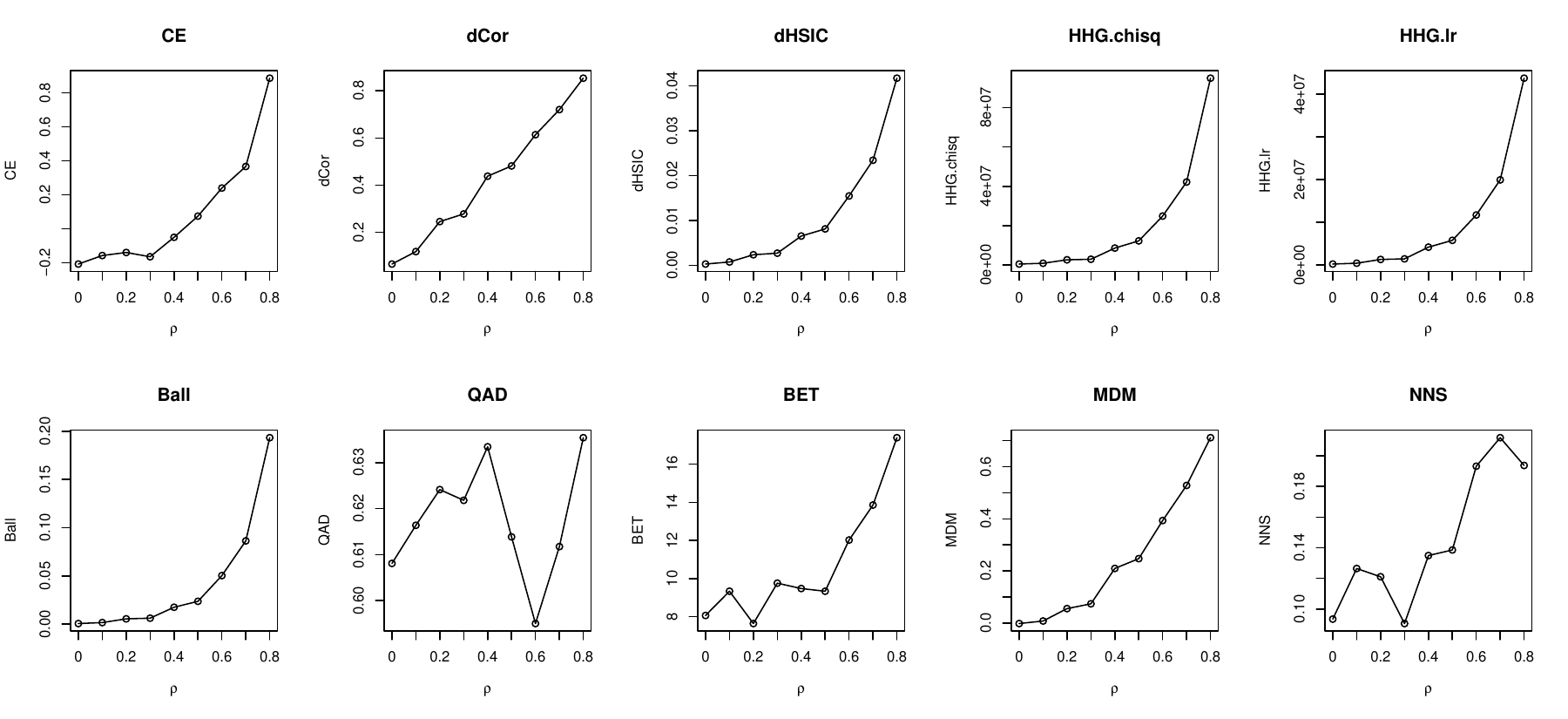}
	\caption{Results of multivariate measures of independence on quadvariate normal distribution.}
	\label{fig:sim1multi2}
\end{figure}

\section{Conditional independence tests}
\label{s:cibench}
Conditional independence is another fundamental concept in statistics and is related to many problems\cite{Dawid1979}. We proved that TE can be represented with only CE and proposed a TE estimator based on this proof\cite{jian2019estimating}. Since TE is essentially CI, we have actually presented a information theoretical measure of CI. There are many other measures of CI currently\cite{Li2020}, (see Table \ref{tb:benchci}), such as conditional dCor\cite{wang2015conditional}, kernel based CI\cite{zhang2011kernel,Strobl2019}, and copula based one\cite{Petersen2021}, etc.

To evaluate these CI measures (see Table \ref{tb:benchci}), we designed experiments with two simulations and a real data \footnote{The code is available at \url{https://github.com/majianthu/eval}}. The measures of CI evaluated include:

\paragraph{CDC}
Wang et al.\cite{wang2015conditional} proposed a measure of CI, CDC, which is a generalization of dCor. CDC is defined by  extending \eqref{eq:dcornu} to the cases of CI, i.e.,
\begin{equation}
	\nu^2(X,Y|Z;w)=||f_{X,Y|Z}(t,s)-f_{X|Z}(t)f_{Y|Z}(s)||_w^2.
\end{equation}

\paragraph{CMD}
Part et al.\cite{Park2015} proposed a dCor based measure of CI by generalizing MDC.

\paragraph{CODEC}
\eqref{eq:eyxey} can be extended to the cases of CI. Given random variables $(X,Y,Z)$, CI of $X,Y$ conditional on $Z$ can be tested by 
\begin{equation}
	E(Y|X,Z)=E(Y,Z).
	\label{eq:eyxzeyz}
\end{equation}
Based on this, Azadkia and Chatterjee\cite{Azadkia2021} proposed a measure of CI that is a generalization of CODEC.

\paragraph{KPC}
Huang et al.\cite{Huang2022a} proposed a kernel based measure of CI, called Kernel Partial Correlation (KPC), to check \eqref{eq:eyxzeyz}. KPC can be considered as the counterpart of CODEC in RKHS.

\paragraph{FCIT}
Chalupka et al.\cite{Chalupka2018} proposed a measure of CI based on testing \eqref{eq:eyxzeyz} through prediction, called Fast Conditional Independence Test (FCIT).

\paragraph{PCIT}
Burkart and Kir\'aly\cite{Burkart2017} proposed a measure of CI by testing \eqref{eq:eyxzeyz} with prediction, called Predictive Conditional Independence Test (PCIT).

\paragraph{CCIT}
Sen et al.\cite{Sen2017} proposed a classifier based measure of CI, called Classifier Conditional Independence Test (CCIT). Given random variables $(X,Y,Z)$, this measure is to check
\begin{equation}
	p_{X,Y,Z}(x,y,z)=p_{X|Z}(x|z)p_{Y|Z}(y|z)p_Z(z),
\end{equation}
where $p$ is pdf.

\paragraph{CMI1}
CMI is a measure of CI in information theory. Runge\cite{Runge2017} proposed a kNN based estimator of CMI which is based on the following expansion of CMI:
\begin{equation}
	I_{X,Y|Z}=H_{XZ}+H_{YZ}-H_Z-H_{XYZ}.
	\label{eq:cmi1}
\end{equation}

\paragraph{CMI2}
Mesner and Shalizi\cite{Mesner2021} proposed another CMI estimator based on \eqref{eq:cmi1} for the cases with discrete and continuous variables.

\paragraph{PCor}
Partial Correlation is a traditional measure of CI\cite{Yule1907,Baba2004}. It can be estimated by the correlation between residuals of regression.

\paragraph{GCM}
Shah and Peters\cite{Shah2020} proposed a measure of CI, called Generalized Covariance Measure (GCM), which is defined as covariance of regression residuals.

\paragraph{wGCM}
Scheidegger et al.\cite{Scheidegger2022} proposed a weighted GCM, called wGCM.

\paragraph{KCIT}
Zhang et al.\cite{zhang2011kernel} proposed a kernel based measure of CI (KCIT). KCIT can be considered as partial correlation in RKHS.

\paragraph{RCoT}
Strobl et al.\cite{Strobl2019} proposed two approximations of KCIT, called RCIT and RCoT, to tackle the computational issue of KCIT.

\paragraph{PCop}
Given random variables $(X,Y,Z)$, Partial Copula (PCop) is joint conditional margins $U_X,U_Y$, where $U_X,U_Y$ are defined as
\begin{equation}
	U_X=F_{X|Z}(X|Z), U_Y=F_{Y|Z}(Y|Z).
\end{equation}
Petersen and Hansen\cite{Petersen2021} proposed to test CI by testing independence of $U_X,U_Y$.

~\

In simulation experiments, data are generated from the following distributions:
\begin{enumerate}
	\item random variables $\mathbf(X,Y,Z)$ governed by trivariate normal distribution $N(\mathbf{u},\Sigma)$ with covariance matrix $\Sigma$ as
		\begin{equation}
		\begin{vmatrix}
			1 & \rho_{xy} & \rho_{xz} \\
			\rho_{xy} & 1 & \rho_{yz} \\
			\rho_{xz} & \rho_{yz} & 1
		\end{vmatrix},
	\end{equation}	
	where $\rho_{xy}=0.7,\rho_{yz}=0.6$, $\rho_{xz}$ changes from 0 to 0.9 by step 0.1;
	\item random variables $\mathbf(X,Y,Z)$ governed by trivariate normal copula $C_N(u_x,u_y,u_z;\Sigma)$ with covariance matrix $\Sigma$ same as above, $\rho_{xz}$ changes from 0 to 0.9 by step 0.1.
\end{enumerate}
In these simulations, the strength of CI increases with parameters. We applied the measures in Table \ref{tb:benchci} to the simulated data. Experimental results (see Figure \ref{fig:trinormci} and Figure \ref{fig:trinormcopci}) shows that CE increases with CI strength and presents better results than others.

\begin{table}
	\centering
	\caption{Implementations of the CI measures evaluated.}
	\begin{tabular}{l|c|c}
		\toprule
		\textbf{Package}&\textbf{CI measure}&\textbf{Language}\\
		\midrule
		\texttt{copent}&CE\cite{jian2019estimating}&\textsf{R}\\
		\texttt{EDMeasure}&CMDM\cite{Park2015}&\textsf{R}\\
		\texttt{FOCI}&CODEC\cite{Azadkia2021}&\textsf{R}\\
		\texttt{RCIT}&RCoT\cite{Strobl2019}&\textsf{R}\\		
		\texttt{cdcsis}&CDC\cite{wang2015conditional}&\textsf{R}\\		
		\texttt{GeneralisedCovarianceMeasure}&GCM\cite{Shah2020}&\textsf{R}\\		
		\texttt{weightedGCM}&wGCM\cite{Scheidegger2022}&\textsf{R}\\
		\texttt{comets}&PCM\cite{Lundborg2024}&\textsf{R}\\	
		\texttt{KPC}&KPC\cite{Huang2022a}&\textsf{R}\\		
		\texttt{ppcor}&PCor\cite{Yule1907}&\textsf{R}\\		
		\texttt{parCopCITest}&PCop\cite{Petersen2021}&\textsf{R}\\		
		\texttt{causallearn}&KCI\cite{zhang2011kernel}&\textsf{Python}\\
		\texttt{pycit}&CMI1\cite{Runge2017}&\textsf{Python}\\		
		\texttt{knncmi}&CMI2\cite{Mesner2021}&\textsf{Python}\\		
		\texttt{fcit}&FCIT\cite{Chalupka2018}&\textsf{Python}\\		
		\texttt{CCIT}&CCIT\cite{Sen2017}&\textsf{Python}\\		
		\texttt{pcit}&PCIT\cite{Burkart2017}&\textsf{Python}\\		
		\bottomrule
	\end{tabular}
	\label{tb:benchci}
\end{table}

\begin{figure}
	\centering
	\includegraphics[width=\textwidth]{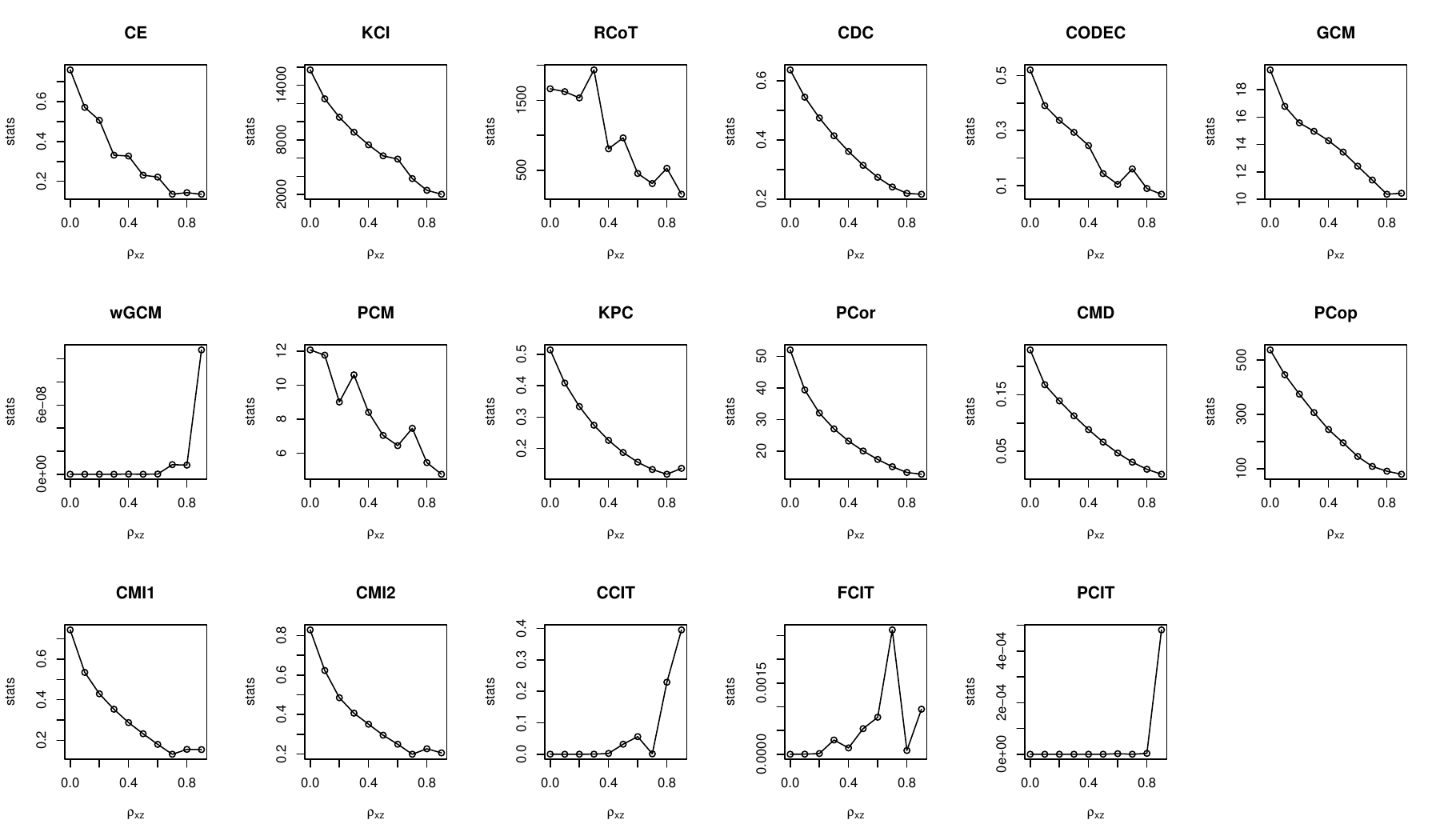}
	\caption{Results of simulation experiments on trivariate normal distributions.}
	\label{fig:trinormci}
\end{figure}

\begin{figure}
	\centering
	\includegraphics[width=\textwidth]{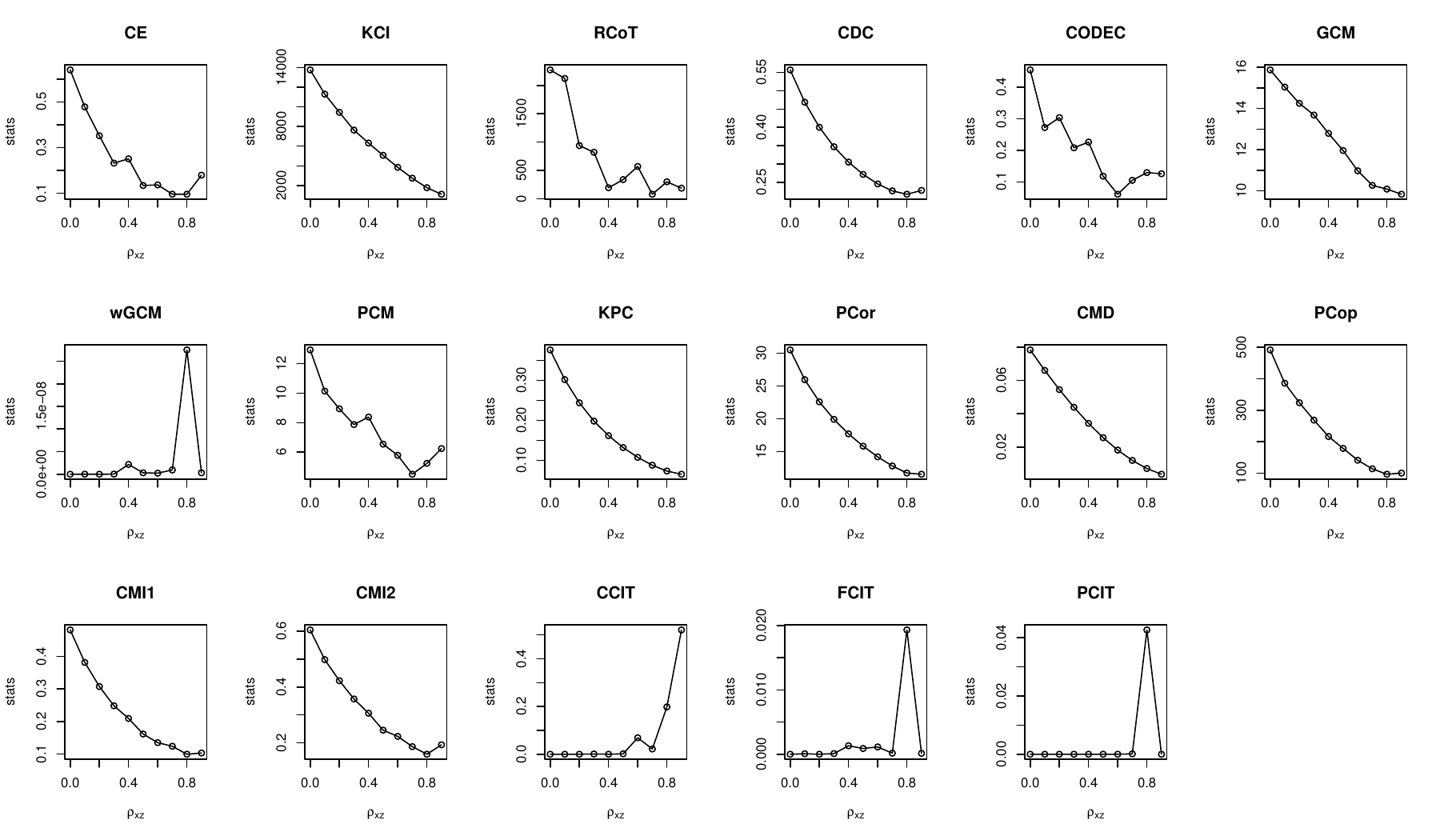}
	\caption{Results of simulation experiments on trivariate normal copulas.}
	\label{fig:trinormcopci}
\end{figure}

\section{Multivariate normality tests}
\label{s:mvntbench}
Normality test is a basic problem of hypothesis testing and has been contributed for decades. There are many tests, such as BHEP, kurtosis and skewness, and dCor based one, etc. \cite{Ebner2020}. We proposed a CE based method for multivariate normality test and its estimator \cite{Ma2022}. To compare it with others, we designed simulation experiments with the implementations of CE and 29 other important measures in \textsf{R} (see Table \ref{tb:benchmvn}). The compared measures include:

\paragraph{Mardia}
Mardia proposed a skewness and kurtosis based method for multivariate normality test (MVNT)\cite{Mardia1970,Mardia1974}. Skewness and kurtosis are defined as
\begin{align}
	S&=\frac{1}{n^2}\sum_{i=1}^{n}\sum_{j=1}^{n}m_{ij}^3,\\
	K&=\frac{1}{n}\sum_{i=1}^{n}m_{ii}^2,
\end{align}
where $m_{ij}$ is the square of Mahalanobis distance.

\paragraph{Henze-Zirkler}
Henze and Zirkler\cite{Henze1990} proposed a method for MVNT based on nonnegative functional distance between pdfs.

\paragraph{Royston}
Royston\cite{Royston1982,Royston1983,Royston1992} proposed a method for MVNT based on Shapiro-Wilk statistic.

\paragraph{Doornik–Hansen}
Doornik and Hansen\cite{Doornik2008} proposed a method for MVNT based on skewness and kurtosis.

\paragraph{Energy}
Sz\'ekely and Rizzo\cite{Szekely2005} proposed a method for MVNT based on energy statistics.

\paragraph{Anderson-Darling}
Anderson-Darling method\cite{Anderson1954} is a method of Goodness of fit that can be used for MVNT\cite{Anderson2014}.

\paragraph{Cram\'er-von Mises}
Koziol\cite{Koziol1982} proposed a method for MVNT based on Cram\'er-von Mises statistic.

\paragraph{Nikulin-Rao-Robson}
Vionov et al.\cite{Voinov2016} proposed a method for MVNT based on Nikulin-Rao-Robson statistic.

\paragraph{McCulloch}
McCulloch\cite{McCulloch1985} proposed a $\chi^2$ goodness of fit statistic for MVNT.

\paragraph{Dzhaparidze-Nikulin}
Dzhaparidze and Nikulin\cite{Dzaparidze1975} proposed also a $\chi^2$ goodness of fit statistic for MVNT.

\paragraph{BHEP}
Henze and Wagner\cite{Henze1997} proposed a distance based statistic for MVNT, called BHEP.

\paragraph{Cox-Small}
Cox and Small\cite{Cox1978} proposed a method for MVNT based on regression.

\paragraph{DEHT and DEHU}
D\"orr et al.\cite{Doerr2021a,Doerr2021} proposed two statistics for MVNT (DEHT and DEHU) inspired by the solution of partial differential equation.

\paragraph{EHS}
Ebner et al.\cite{Ebner2022} proposed a statistic for MVNT based on solving partial differential equation.

\paragraph{HJG}
Henze and Jim\'enez-Gamero\cite{Henze2019} proposed a statistic based on weighted $L^2$ distance.

\paragraph{HV}
Henze and Visagie\cite{Henze2020} proposed a statistic for MVNT based on solving partial differential equation.

\paragraph{HZ}
See Henze-Zirkler.

\paragraph{KKurt}
Kozoil\cite{Koziol1989} proposed a multivariate kurtosis statistic.

\paragraph{MAKurt and MASkew}
Malkovich and Afifi\cite{Malkovich1973} proposed a multivariate skewness and kurtosis statistics.

\paragraph{MKurt and MSkew}
See Mardia.

\paragraph{MQ1 and MQ2}
Manzotti and Quiroz\cite{Manzotti2001} proposed two statistics for MVNT.

\paragraph{MRSSkew}
M\'ori et al.\cite{Mori1994} proposed a multivariate skewness and kurtosis.

\paragraph{PU}
Pudelko\cite{Pudelko2005} proposed a statistic for MVNT based on empirical characteristic function.

\paragraph{SR} 
See Energy.

\paragraph{Shapiro-Wilk}
See Royston.

~\

We degined three simulation experiments \footnote{The code is available at \url{https://github.com/majianthu/mvnt}} to simulate different types of non-normality:
\begin{enumerate}
	\item The first simulation is based on two random variables $\mathbf{X}=(X_1,X_2)$ governed by bivariate normal copula $C_N(u,v)$ with margins as normal distribution $u\sim N(0,1)$ and exponential distribution $v\sim E(\lambda)$. The parameter of exponential distribution $\lambda=1,\dots,10$, when $\lambda=2$, the simulated distribution is close to normal distribution and become non-normal as $\lambda$ increases;
	\item the second simulation is based on two random variables $\mathbf{X}=(X_1,X_2)$ governed by bivariate $t$ distribution $\mathbf{X}\sim t_{\nu}(\mathbf{0},\Sigma)$, the non-diagonal element of $\Sigma$ $\rho=0.5$, non-normality is gradually weakened as $\nu=1,\ldots,10$;
	\item the third simulation is based on two bivariate normal distributions $X_1\sim N_1(\mathbf{\mu}_1,\rho_1)$ and $X_2\sim N_2(\mathbf{\mu}_2,\rho_2)$ with $\mu_1=\mathbf{0},\mu_2=\mathbf{3},\rho_1=0.5,\rho_2=0.8$and simulated data is derived from two samples from these two normal distributions as $\{(\beta-1)X_1 + (10-\beta)X_2\}/9$, where $\beta = 1,\ldots,10$ so non-normality becomes stronger first and then weakened.
\end{enumerate}
We applied CE based test and the other tests to these three simulated data. Experimental results are shown in Figure \ref{fig:benchmvnt1}, Figure \ref{fig:benchmvnt2}, and Figure \ref{fig:benchmvnt3}, from which one can learn that the CE based test can measure the change of simulated non-normality correctly and presents better results than others.

\begin{table}
	\centering
	\caption{Implementations of the MVNT evaluated in \textsf{R}.}
	\begin{tabular}{l|c}
		\toprule
		\textbf{Package}&\textbf{Test}\\
		\midrule
		\texttt{copent}&CE\cite{Ma2022}\\
		\hline
		\multirow{2}{*}{\texttt{MVN}\cite{Korkmaz2014}}&Mardia, Royston, Henze-Zirkler\\
		&Dornik-Haansen, Energy\\
		\hline
		\multirow{3}{*}{\texttt{mvnTest}\cite{Pya2016}}&Anderson-Darling, Cram\'er-von Mises\\
		&Nikulin-Rao-Robson, McCulloch\\
		&Dzhaparidze-Nikulin\\
		\hline
		\multirow{3}{*}{\texttt{mnt}\cite{Ebner2020}}&BHEP, Cox-Small, DEHT, DEHU, EHS, HJG\\
		&HV, HZ, KKurt, MAKurt, MASkew, MKurt\\
		&MQ1, MQ2, MRSSkew, MSkew, PU, SR\\
		\hline
		\texttt{mvnormtest}&Shapiro-Wilk\cite{Royston1982}\\
		\bottomrule
	\end{tabular}
	\label{tb:benchmvn}
\end{table}

\begin{figure}
	\centering
	\includegraphics[width=\textwidth]{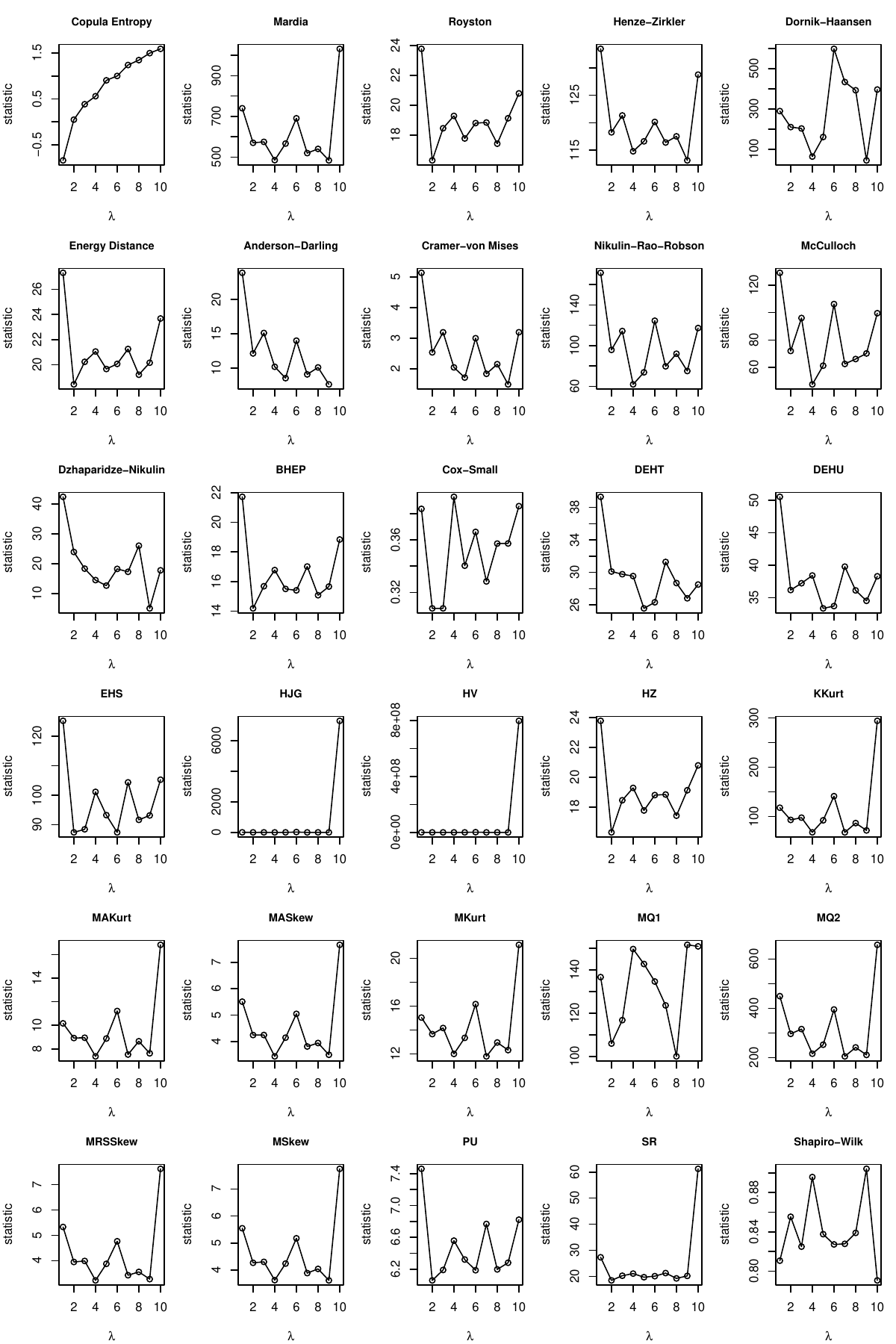}
	\caption{Results of the simulation experiment based on margins.}
	\label{fig:benchmvnt1}
\end{figure}

\begin{figure}
	\centering
	\includegraphics[width=\textwidth]{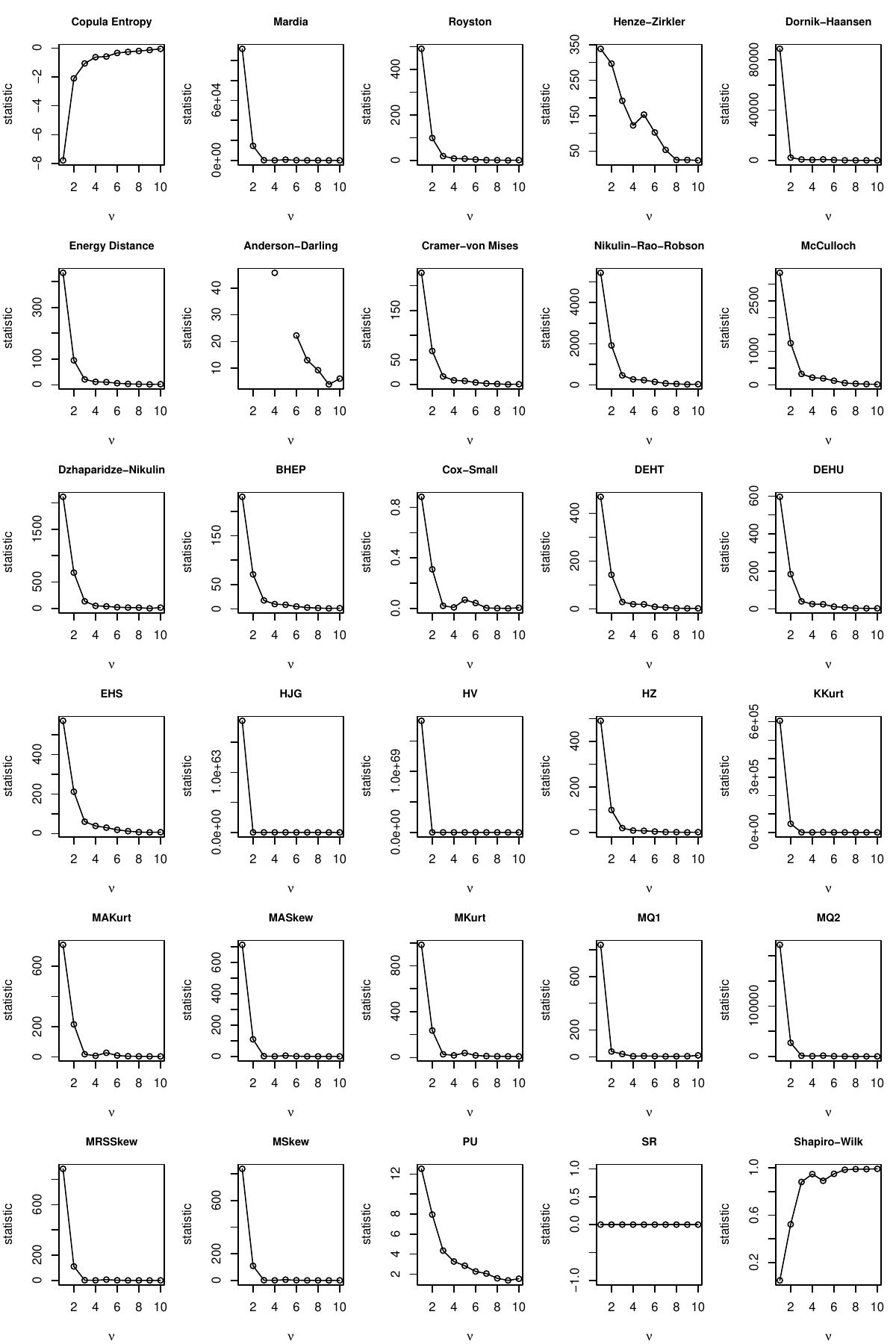}
	\caption{Results of the simulation experiment based on copulas.}
	\label{fig:benchmvnt2}
\end{figure}

\begin{figure}
	\centering
	\includegraphics[width=\textwidth]{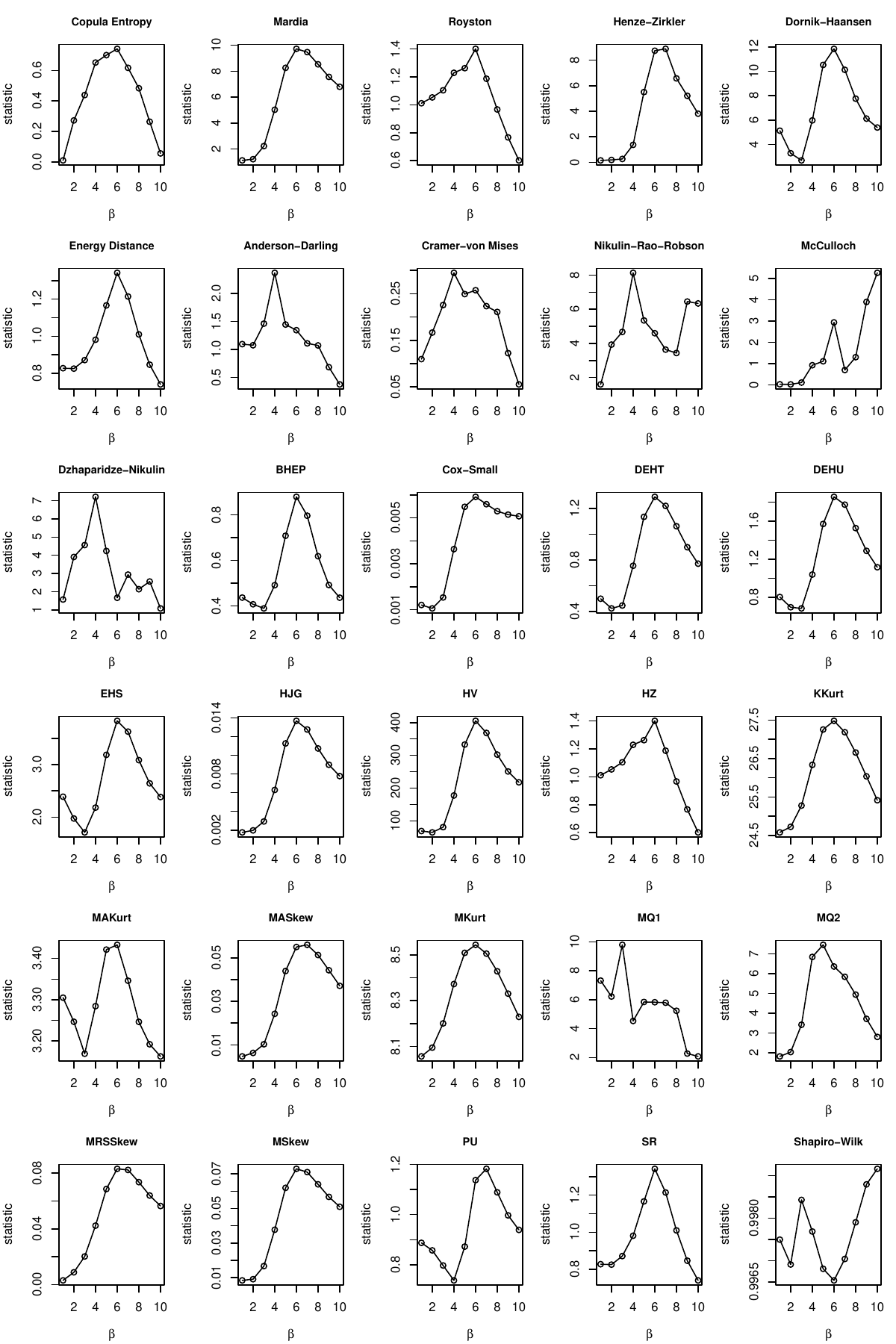}
	\caption{Results of the simulation experiment based on mixtures of distributions.}
	\label{fig:benchmvnt3}
\end{figure}

\section{Two-sample tests}
\label{s:tstbench}
Two-sample test is a basic problem in hypothesis testing and there are many existing methods for it, such as Wilcoxon test, Kruskal-Wallis test, and Kolmogorov-Smirnov test \cite{Hollander1973}, dCor base test\cite{Szekely2004}, kernel based test\cite{Gretton2012}, HHG statistic \cite{heller2016consistent}, Ball statistic\cite{Pan2018}, learning based test \cite{Hediger2022}, etc. We proposed a CE based two-sample test and its estimator \cite{Ma2023b}. To evaluate it, we compared it with the tests implemented in \textsf{R} for univariate cases (see Table \ref{tb:benchtstuni}) and multivariate cases (see Table \ref{tb:benchtstmulti}).

The two-sample tests for univariate cases evaluated include:
\paragraph{Wilcoxon1}
Wilcoxon test\cite{Wilcoxon1945}, also called Mann-Whitney test\cite{Mann1947}, is a commonly used method for univariate two-sample test.

\paragraph{Kruskal-Wallis}
Kruskal-Wallis test is proposed by Kruskal and Wallis in 1952\cite{Kruskal1952a}.

\paragraph{Kolmogorov-Smirnov}
Kolmogorov-Smirnov test\cite{Kolmogorov1933,Smirnov1948} is another common method for univarite two-sample test which is based on empirical distributions.

\paragraph{Cram\'er-von Mises}
Cram\'er-von Mises test\cite{Cramer1927,Mises1931}is proposed by Cram\'er and von Mises in 1927 and 1931 respectively. 

\paragraph{Kuiper}
Kuiper test \cite{Kuiper1960} is proposed by Kuiper in 1960.

\paragraph{WASS}
Wasserstein test is a univariate nonparametric method for two-sample test based on Wasserstein distance\cite{Ramdas2017}.

\paragraph{DTS}
Dowd\cite{Dowd2020} proposed a univariate nonparametric two-sample test.

\paragraph{AD}
Pettitt\cite{Pettitt1976} proposed a two-sample test based on Anderson-Darling rank statistic.

\paragraph{Wilcoxon2}
See Wilcoxon1.

\paragraph{Vartest}
Ammous et al.\cite{Ammous2024} proposed a two-sample test based on James-Welch ANOVA.

\paragraph{LR}
LR test is proposed by Lehmann and Rosenblatt in 1951 and 1952 respectively\cite{Lehmann1951,Rosenblatt1952}.

\paragraph{ZA,ZK and ZC}
Zhang\cite{Zhang2006} proposed two tests based on likelihood ratio.

\paragraph{TNL}
Aliev et al.\cite{Aliev2024} proposed a two-sample test based rank.

~\

The two-sample tests for multivariate cases include:
\paragraph{MI}
MI based two-sample test\cite{Guha2014,Drake2014} is based on the MI between samples and labels.

\paragraph{Kernel}
Gretton et al.\cite{Gretton2012} proposed a kernel based two-sample test.

\paragraph{Energy}
Sz\'ekely and Rizzo\cite{Szekely2004} proposed a two-sample test based energy statistics.

\paragraph{Ball}
Pan et al.\cite{Pan2018} proposed a two-sample test based on Ball divergence.

\paragraph{Random Forest}
Hediger et al.\cite{Hediger2022} proposed a two-sample test based on random forests.

\paragraph{HHG}
Heller et al.\cite{heller2016consistent} proposed a k-sample test based on HHG statistic.

\paragraph{Cramer}
Baringhaus and Franz\cite{Baringhaus2004} proposed a k-sample test.

\paragraph{TST.HD}
Cousido-Rocha et al.\cite{CousidoRocha2019} proposed a k-sample test based on distance between empirical characteristic functions.

\paragraph{F-F}
Fasano and Franceschini\cite{Fasano1987} proposed a multivariate version of Kolmogorov-Smirnov test.

\paragraph{Peacock}
Peacock\cite{Peacock1983} proposed a Kolmogorov-Smirnov test for bivariate cases.

\paragraph{RPT}
Lopes et al.\cite{Lopes2011} proposed a two-sample test based on Hotelling $T^2$ test.

\paragraph{Depth}
Liu and Singh\cite{Liu1993} proposed a depth based two-sample test.

\paragraph{AD}
See AD above.

\paragraph{QN}
Kruskal\cite{Kruskal1952} proposed a nonparametric k-sample test based on rank.

\paragraph{BM}
Neubert and Brunner\cite{Neubert2007} proposed a two-sample test based on a studentized permutation test.

\paragraph{WASS}
See WASS above.

\paragraph{SWD}
Wasserstein distance is proposed for two-sample test\cite{Rabin2012}.

\paragraph{Graph}
Bai and Chu\cite{Bai2023} proposed a graph-based multivariate two-sample test.

\paragraph{KMD}
Huang and Sen\cite{Huang2024} proposed a nonparametric k-sample test based on kernels.

~\

We designed simulation experiments for univariate tests and multivariate tests\footnote{The code is available at \url{https://github.com/majianthu/tst}}:

\paragraph{Univariate tests}
Two simulation experiments for univariate tests include:
\begin{enumerate}
	\item The first simulation includes a random variable governed by normal distribution $x_0\sim N(\mu_0,\delta_0)$ with $\mu_0=0$ and $\delta_0=1$ and another random variable governed by normal distribution $x_2\sim N(\mu_1,\delta_1)$ with $\mu_1$ changes from 0 to 9 and $\delta_1=1$;
	\item The second simulation includes the first random variable $x_0$ same to the first simulation, and the second random variable $x_1\sim N(\mu_1,\delta_1)$ with $\delta_1$ changes from 1 to 10 and $\mu_1=0$.
\end{enumerate}
We applied the selected two-sample tests to the simulated data. Experimental results are shown in Figure \ref{fig:benchtst1} from which one can learn that CE based test performs better than others.

\paragraph{Multivariate tests}
Three simulation experiments for multivariate tests include:
\begin{enumerate}
	\item The first simulation includes two random variables governed by normal distributions $\mathbf{x}_0\sim N(u_0,\rho_0)$ and $\mathbf{x}_1\sim N(u_1,\rho_1)$ with $u_0=\mathbf{0}$ and $u_1=(i,i)$, $i$ changes from 0 to 9 by step 1, and $\rho_0,\rho_1=0.5$;
	\item The second simulation includes a bivariate random variable governed by normal distribution $\mathbf{x}_0\sim N(u_0,\rho_0)$ with $u_0=\mathbf{0},\rho_0=0$ and another random variable $\mathbf{x}_1$ governed by normal distribution $N(u_1,\rho_1)$ with $\rho_1$ changes from 0 to 0.9 by step 0.1 and $u_1=\mathbf{0}$;
	\item The third simulation includes a random variable governed by bivariate normal distribution $\mathbf{x}_0\sim N(u_0,\rho_0)$ with $\rho_0=0$ and another random variable $\mathbf{x}_1$ governed by bivariate normal copula $c(u,v)$ with margins as normal distribution $u\sim N(0,2)$ and exponential distribution $v\sim E(\lambda)$ with $\lambda$ changes from 1 to 10.
\end{enumerate}
We applied the selected multivariate two-sample tests to the simulated data. Experimental results are shown in Figure  \ref{fig:benchtst2a}, Figure \ref{fig:benchtst2b}, and Figure \ref{fig:benchtst2c} from which one can learn that CE based test can measure the relations between two simulated distributions correctly and performs better than others.

\begin{table}
	\centering
	\caption{Implementations of the univariate two-sample tests evaluated in \textsf{R}.}
	\begin{tabular}{l|c}
		\toprule
		\textbf{Package}&\textbf{Test}\\
		\midrule
		\texttt{copent}&CE\cite{Ma2023b}\\
		\hline
		\multirow{2}{*}{\texttt{stat}}&Wilcoxon1\cite{Hollander1973}\\
		&Kruskal-Wallis\cite{Hollander1973}\\
		\hline
		\multirow{2}{*}{\texttt{twosamples}\cite{Dowd2020}}&CVM, KS, Kuiper\\
		&WASS, DTS, AD\\
		\hline
		\texttt{robusTest}\cite{Ammous2024}&Wilcoxon2, Vartest\\
		\hline
		\multirow{2}{*}{\texttt{R2sample}}&LR\cite{Lehmann1951,Rosenblatt1952},\\&ZA,ZK,ZC\cite{Zhang2006}\\
		\hline
		\texttt{tnl.Test}&TNL\cite{Aliev2024}\\
		\bottomrule
	\end{tabular}
	\label{tb:benchtstuni}
\end{table}

\begin{table}
	\centering
	\caption{Implementations of the multivariate two-sample tests evaluated in \textsf{R}.}
	\begin{tabular}{l|c}
		\toprule
		\textbf{Package}&\textbf{Test}\\
		\midrule
		\multirow{2}{*}{\texttt{copent}}&CE\cite{Ma2023b}\\
		&MI\cite{Drake2014,Guha2014}\\
		\hline
		\texttt{kernlab}\cite{Karatzoglou2004}&Kernel\cite{Gretton2012}\\
		\hline
		\texttt{energy}&Energy\cite{Szekely2004}\\
		\hline
		\texttt{Ball}&Ball divergence\cite{Pan2018}\\
		\hline
		\texttt{hypoRF}&Random Forest\cite{Hediger2022}\\
		\hline
		\multirow{4}{*}{\texttt{HHG}\cite{heller2016consistent}}&HHG sum.chisq\\
		&HHG sum.lr\\
		&HHG max.chisq\\
		&HHG max.lr\\
		\hline
		\texttt{cramer}&Cramer\cite{Baringhaus2004}\\
		\hline
		\texttt{TwoSampleTest.HD}\cite{Cousido-Rocha2023}&TST.HD\\
		\hline
		\texttt{fasano.franceschini.test}\cite{Puritz2023}&F-F\\
		\hline
		\texttt{Peacock.test}&Peacock\cite{Peacock1983}\\
		\hline
		\texttt{RandomProjectionTest}&RPT\cite{Lopes2011}\\
		\hline
		\texttt{DepthProc}&Depth\cite{Liu1993}\\
		\hline
		\multirow{2}{*}{\texttt{kSamples}}&AD\cite{Pettitt1976}\\&QN\cite{Kruskal1952}\\
		\hline
		\texttt{lawstat}&BM\cite{Neubert2007}\\
		\hline
		\multirow{2}{*}{\texttt{T4transport}}&WASS\cite{Peyre2019}\\&SWD\cite{Rabin2012}\\
		\hline
		\texttt{rgTest}&Graph\cite{Bai2023}\\
		\hline
		\texttt{KMD}&KMD\cite{Huang2024}\\
		\bottomrule
	\end{tabular}
	\label{tb:benchtstmulti}
\end{table}

\begin{figure}
	\centering
	\subfigure[Simulations based on changing mean]{\includegraphics[width=\textwidth]{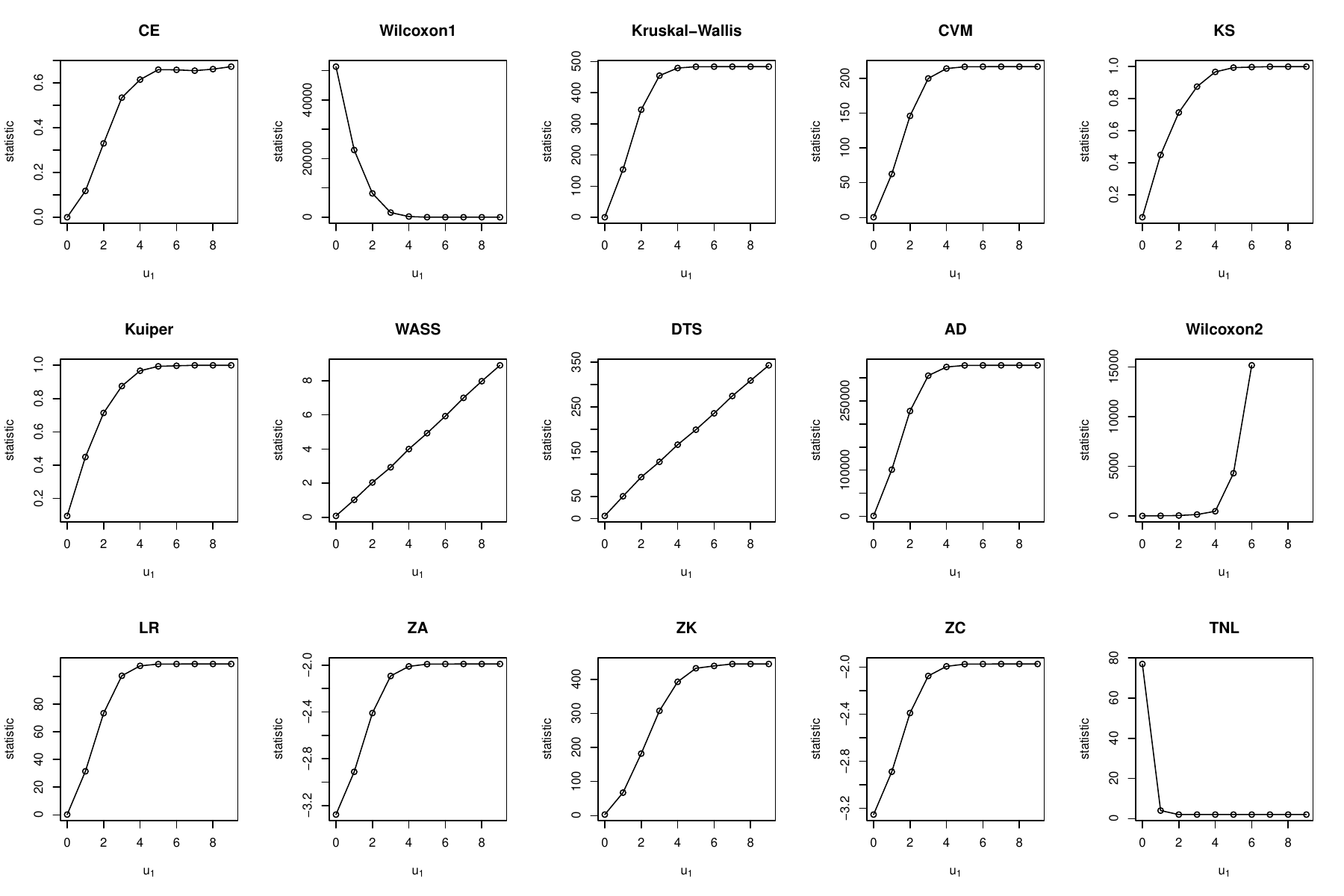}}
	\subfigure[Simulations based on changing variance]{\includegraphics[width=\textwidth]{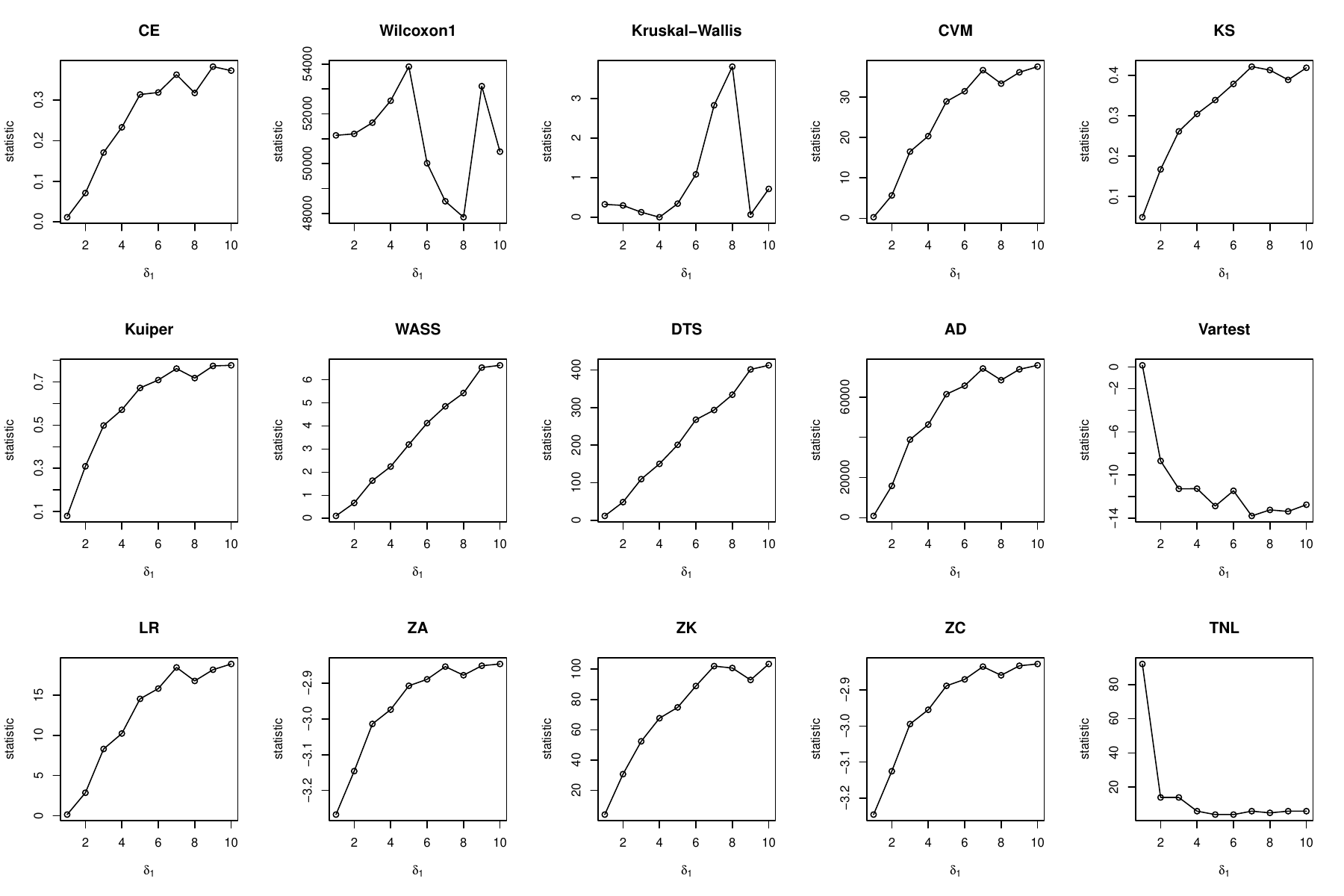}}	
	\caption{Results of the simulation experiments on univariate two-sample tests.}
	\label{fig:benchtst1}
\end{figure}

\begin{figure}
	\centering
	\includegraphics[width=\textwidth]{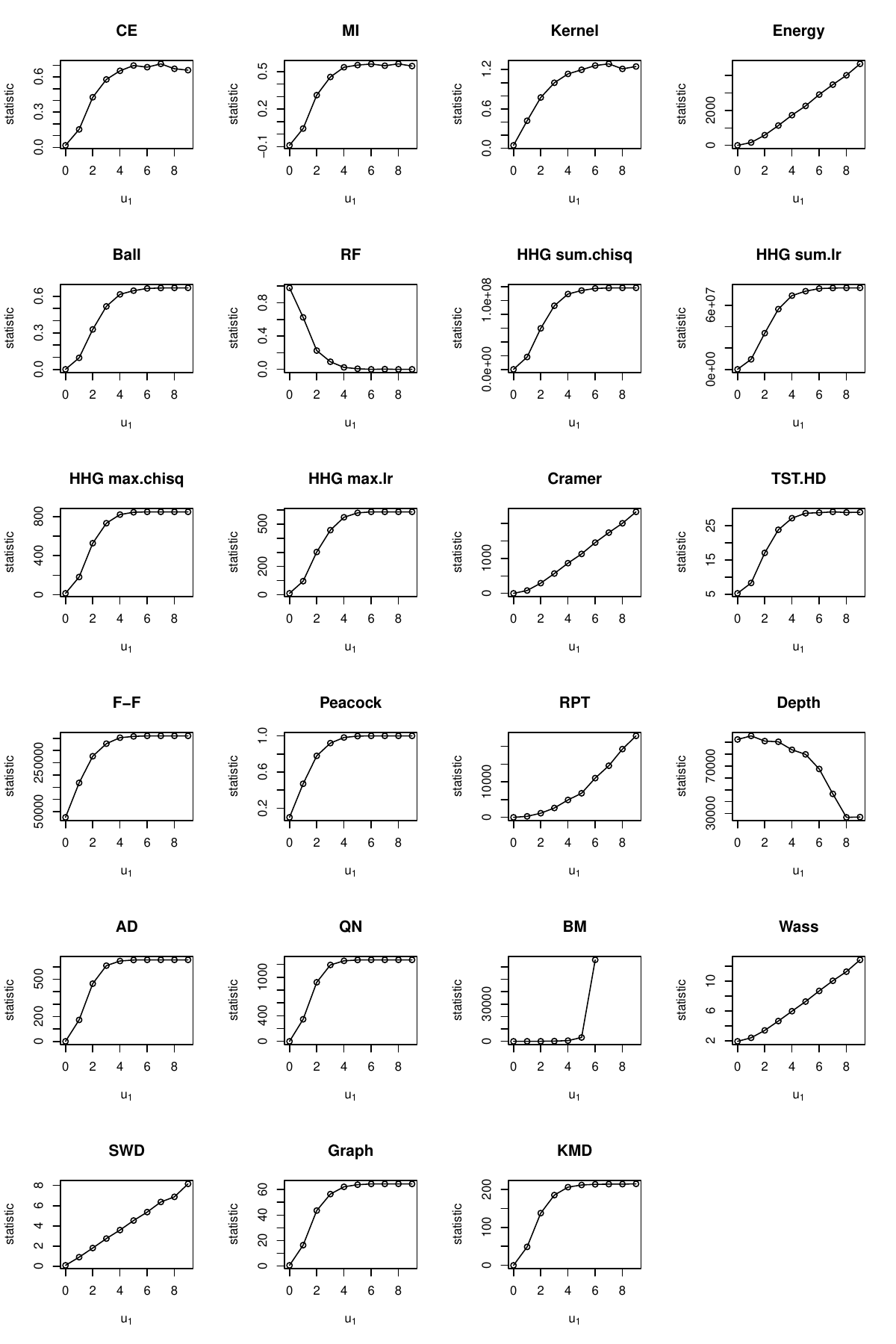}
	\caption{Results of the simulation experiments with changing means for multivariate two-sample tests.}
	\label{fig:benchtst2a}
\end{figure}

\begin{figure}
	\centering
	\includegraphics[width=\textwidth]{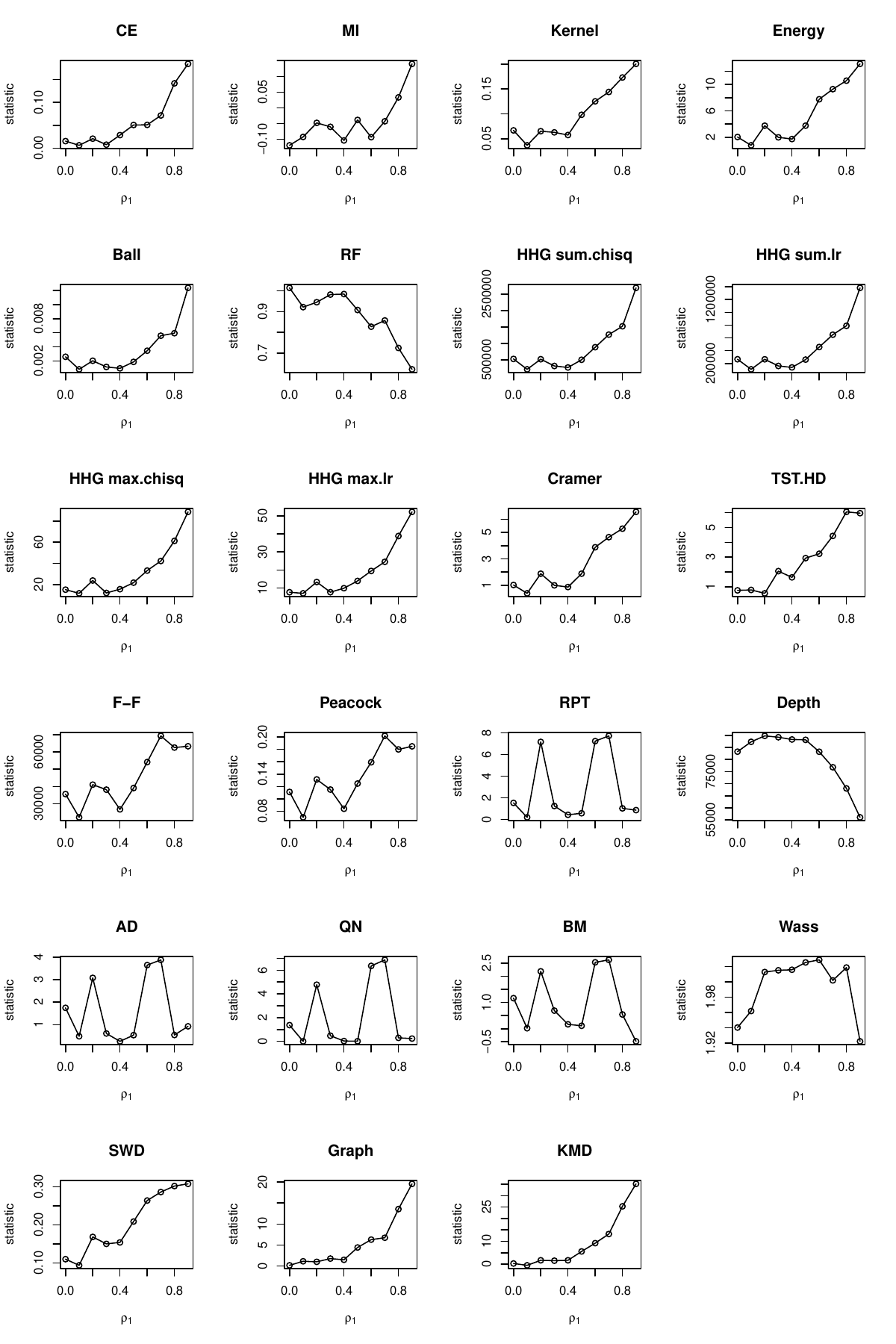}
	\caption{Results of the simulation experiments with changing variance for multivariate two-sample test.}
	\label{fig:benchtst2b}
\end{figure}

\begin{figure}
	\centering
	\includegraphics[width=\textwidth]{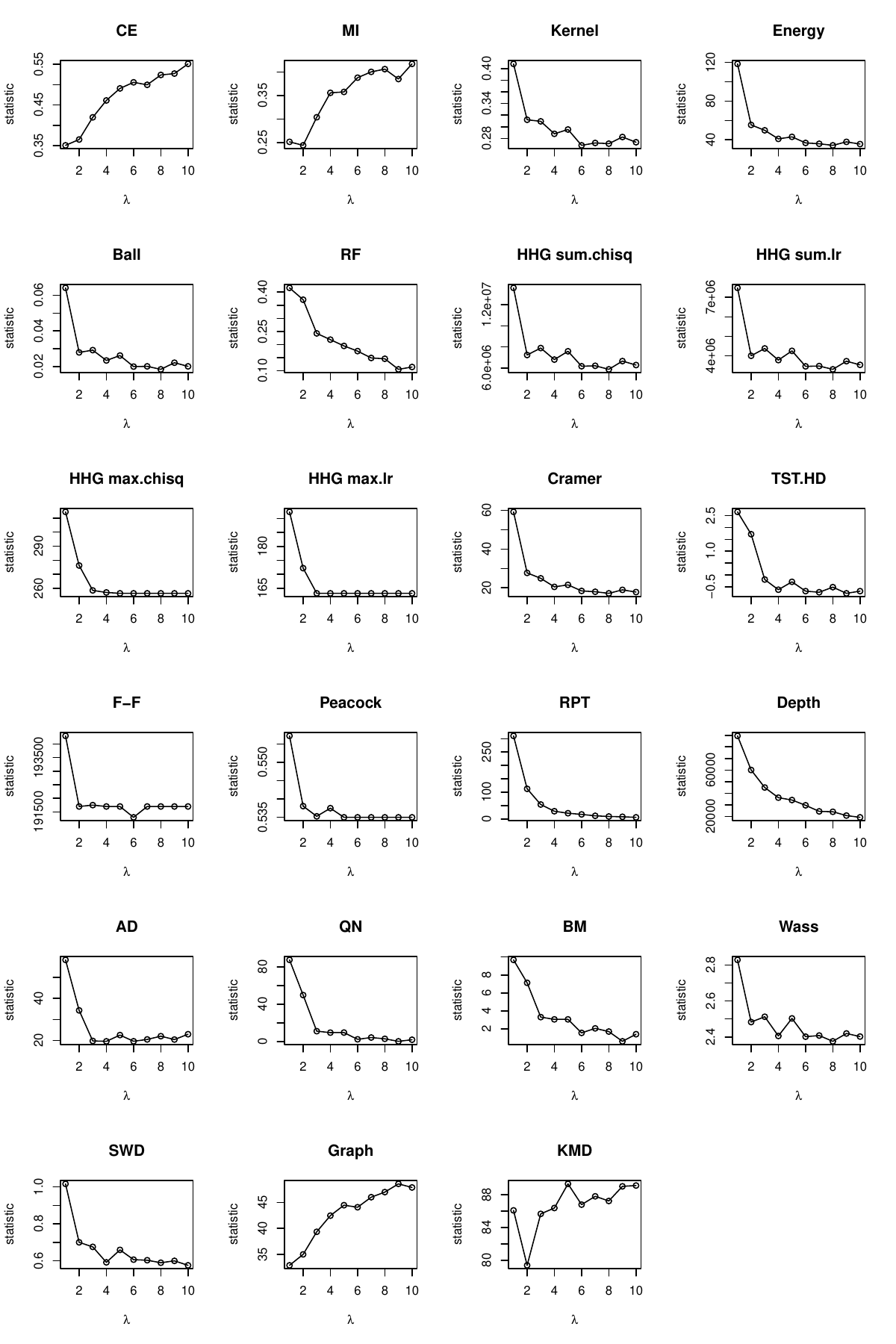}
	\caption{Results of the simulation experiments with changing normal copula for multivariate two-sample tests.}
	\label{fig:benchtst2c}
\end{figure}

\section{Change point detection}
\label{s:cpdbench}
Change point detection is a basic task in time series analysis and there are many existing methods for this \cite{Aminikhanghahi2017,Reeves2007,Truong2020}. We proposed a CE based nonparametric multivariate method for multiple change point detection \cite{Ma2024a}. To evaluate it, we compare it to the methods implemented in \textsf{R} as listed in Table \ref{tb:benchcpd} in simulation experiments\footnote{The code is available at \url{https://github.com/majianthu/cpd}}.

The packages that implements the methods for change point detection in simulation experiments include:
\paragraph{changepoint}
Scott and Knott\cite{Scott1974} proposed a univariate method based on binary segmentation.

\paragraph{ecp}
James and Matteson\cite{James2015} proposed a kernel based multivariate method for change point detection.

\paragraph{rid}
Fan and Wu\cite{Fan2024} proposed a method called random interval distillation (RID).

\paragraph{CptNonPar}
McGonigle and Cho\cite{McGonigle2023} proposed a nonparametric multivariate method for change point detection.

\paragraph{npwbs}
Ross\cite{Ross2021} proposed a univariate method based on Wild Binary Segmentation (WBS).

\paragraph{MFT}
Messer et al.\cite{Messer2018} proposed a univariate method, called multiple filter test (MFT).

\paragraph{jcp}
Messer\cite{Messer2022} proposed a univariate method for change point detection.

\paragraph{InspectChangepoint}
Wang and Samworth\cite{Wang2017} proposed a method for change point detection based sparse projection.

\paragraph{hdbinseg}
Cho and Fryzlewicz\cite{Cho2015} proposed a multivariate method for change point detection based sparsified binary segmentation. 

\paragraph{changepoint.np}
Killick et al.\cite{Killick2012} proposed a univariate method for change point detection with low computational cost.

\paragraph{changepoint.geo}
Grundy et al.\cite{Grundy2020} proposed a method for change point detection based on geometrically inspired mapping.

\paragraph{mosum} 
Eichinger and Kirch\cite{Eichinger2018} proposed a MOSUM method for change point detection.

\paragraph{SNSeg}
Zhao et al.\cite{Zhao2022} proposed a method for change point detection based on self-normalization. 

\paragraph{offlineChange}
Ding et al.\cite{Ding2017} proposed a multivariate method for change point detection. 

\paragraph{IDetect}
Anastasiou and Fryzlewicz\cite{Anastasiou2022} proposed a univariate method for change point detection.

\paragraph{wbs}
Fryzlewicz\cite{Fryzlewicz2014} proposed a univariate method for change point detection based WBS. 

\paragraph{breakfast}
Fryzlewicz\cite{Fryzlewicz2020} proposed a method for change point detection based on WBS2.

\paragraph{mscp}
Levajkovic and Messer\cite{Levajkovic2023} proposed a MOSUM based method for change point detection. 

\paragraph{L2hdchange}
Li et al.\cite{Li2024} proposed a two-way MOSUM based method for change point detection.

\paragraph{gfpop}
Hocking et al.\cite{Hocking2020} proposed a univariate method for change point detection based dynamic programming.

\paragraph{HDCD}
Moen et al.\cite{Moen2023} proposed a sparsity adaptive method for change point detection.

\paragraph{HDDchangepoint}
Drikvandi and Modarres\cite{Drikvandi2025} proposed a nonparametric method for change point detection.

\paragraph{HDcpDetect}
Li et al.\cite{Li2019} proposed a multivariate method for change point detection.

\paragraph{decp}
Ryan and Killick\cite{Ryan2023} proposed a method for covariance change point detection based random matrix theory.

\paragraph{cpss}
Zou et al.\cite{Zou2020} proposed a sample-splitting based univariate method for change point detection.

~\

Simulation experiment for evaluating these methods are designed as follows: first generating samples from 4 univariate or multivariate normal distributions with different means, variances or both, and then combining these samples one after another to simulate a univariate or multivariate sequence with three change points. We applied the methods in Table \ref{tb:benchcpd} to the simulated data.

Six types of change points are simulated, including
\begin{enumerate}
	\item univariate change points on mean, variance, and mean-variance;
	\item multivariate change points on mean, variance, and mean-variance.
\end{enumerate}

In the six simulation experiments, the setting for mean and variance/covariance of the four normal distributions is listed in Table \ref{tb:benchcpdsimpara}. After repeating the simulation for each setting, we derived the mean of the number of change points detected by the methods as their performance measures. Experimental results (see Figure \ref{fig:benchcpd}) show that CE based method can detect the change points in all the six situations and perform better than others. Particularly, CE based method presents much better in the difficult cases with change points on variance/covariance. Note that CE based method is the only one among all that can be applied to all the six situations and also perform well. It is noteworthy that it is tuning-free.

\begin{table}
	\centering
	\caption{Implementations of the methods for change point detection in \textsf{R}.}
	\begin{tabular}{l|c|c|c|c|c|c}
		\toprule
		\multirow{2}{*}{\textbf{Package}}&\multicolumn{3}{c|}{Univariate}&\multicolumn{3}{c}{Multivariate}\\
		\cline{2-7}
		&Mean&Mean-Var&Var&Mean&Mean-Var&Var\\
		\midrule
		\texttt{copent}\cite{Ma2024a}&\checkmark&\checkmark&\checkmark&\checkmark&\checkmark&\checkmark\\
		\texttt{changepoint}\cite{Killick2014}&\checkmark&\checkmark&\checkmark&&&\\
		\texttt{ecp}\cite{James2015}&&&&\checkmark&\checkmark&\checkmark\\
		\texttt{rid}\cite{Fan2024}&\checkmark&\checkmark&\checkmark&\checkmark&\checkmark&\checkmark\\
		\texttt{CptNonPar}\cite{McGonigle2023}&\checkmark&\checkmark&\checkmark&\checkmark&\checkmark&\checkmark\\
		\texttt{npwbs}\cite{Ross2021}&\checkmark&\checkmark&\checkmark&&&\\
		\texttt{MFT}\cite{Messer2018}&\checkmark&&&&&\\
		\texttt{jcp}\cite{Messer2022}&\checkmark&\checkmark&\checkmark&&&\\
		\texttt{InspectChangepoint}\cite{Wang2017}&\checkmark&\checkmark&\checkmark&\checkmark&\checkmark&\checkmark\\
		\texttt{hdbinseg}\cite{Cho2015}&&&&\checkmark&\checkmark&\checkmark\\
		\texttt{changepoint.np}\cite{Killick2012}&\checkmark&\checkmark&\checkmark&&&\\
		\texttt{changepoint.geo}\cite{Grundy2020}&&&&\checkmark&\checkmark&\checkmark\\
		\texttt{mosum}\cite{Meier2021}&\checkmark&\checkmark&\checkmark&&&\\
		\texttt{SNSeg}\cite{Sun2025}&\checkmark&\checkmark&\checkmark&\checkmark&\checkmark&\checkmark\\
		\texttt{offlineChange}\cite{Ding2017}&&&&\checkmark&\checkmark&\checkmark\\
		\texttt{IDetect}\cite{Anastasiou2022}&\checkmark&&&&&\\
		\texttt{wbs}\cite{Fryzlewicz2014}&\checkmark&&&&&\\
		\texttt{breakfast}\cite{breakfast}&\checkmark&&&&&\\
		\texttt{mscp}\cite{Levajkovic2023}&\checkmark&&&&&\\
		\texttt{L2hdchange}\cite{Li2024}&&&&\checkmark&\checkmark&\checkmark\\
		\texttt{gfpop}\cite{Runge2023}&\checkmark&\checkmark&\checkmark&&&\\
		\texttt{HDCD}\cite{Moen2023,Pilliat2023}&\checkmark&\checkmark&\checkmark&\checkmark&\checkmark&\checkmark\\
		\texttt{HDDchangepoint}\cite{Drikvandi2025}&&&&\checkmark&\checkmark&\checkmark\\
		\texttt{HDcpDetect}\cite{Li2019}&&&&\checkmark&\checkmark&\checkmark\\
		\texttt{decp}\cite{Ryan2023}&&&&\checkmark&\checkmark&\checkmark\\
		\texttt{cpss}\cite{Zou2020}&\checkmark&\checkmark&\checkmark&&&\\
		\bottomrule
		
	\end{tabular}
	\label{tb:benchcpd}
\end{table}

\begin{table}
	\centering
	\caption{Setting for the normal distributions in simulation experiments for change point detection.}
	\begin{tabular}{l|c|c|c|c|c|c}
		\toprule
		\multirow{2}{*}{\textbf{Mean, Variance/Covariance}}&\multicolumn{3}{c|}{Univariate}&\multicolumn{3}{c}{Multivariate}\\
		\cline{2-7}
		&Mean&Mean-Var&Var&Mean&Mean-Var&Var\\
		\midrule
		$(\mu_1,\rho_1)$&(0,1)&(0,1)&(0,1)&(\textbf{0},0.2)&(\textbf{0},0.2)&(\textbf{0},0.2)\\
		$(\mu_2,\rho_2)$&(5,1)&(5,3)&(0,10)&(\textbf{10},0.2)&(\textbf{10},0.8)&(\textbf{0},0.8)\\
		$(\mu_3,\rho_3)$&(10,1)&(10,1)&(0,5)&(\textbf{5},0.2)&(\textbf{5},0.1)&(\textbf{0},0.1)\\
		$(\mu_4,\rho_4)$&(3,1)&(3,10)&(0,1)&(\textbf{1},0.2)&(\textbf{1},0.9)&(\textbf{0},0.9)\\
		\bottomrule
		
	\end{tabular}
	\label{tb:benchcpdsimpara}
\end{table}

\begin{figure}
	\centering
	\subfigure[Univariate mean]{\includegraphics[width=\textwidth]{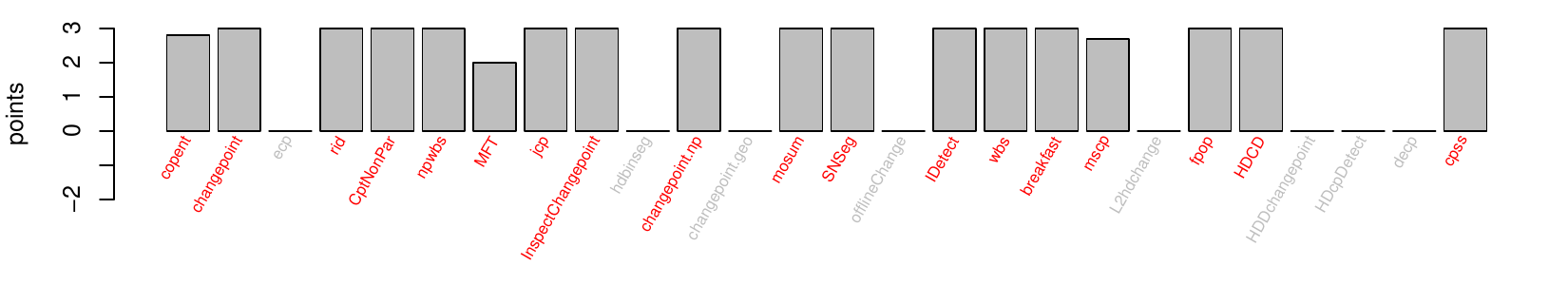}}
	\subfigure[Univariate mean-Var]{\includegraphics[width=\textwidth]{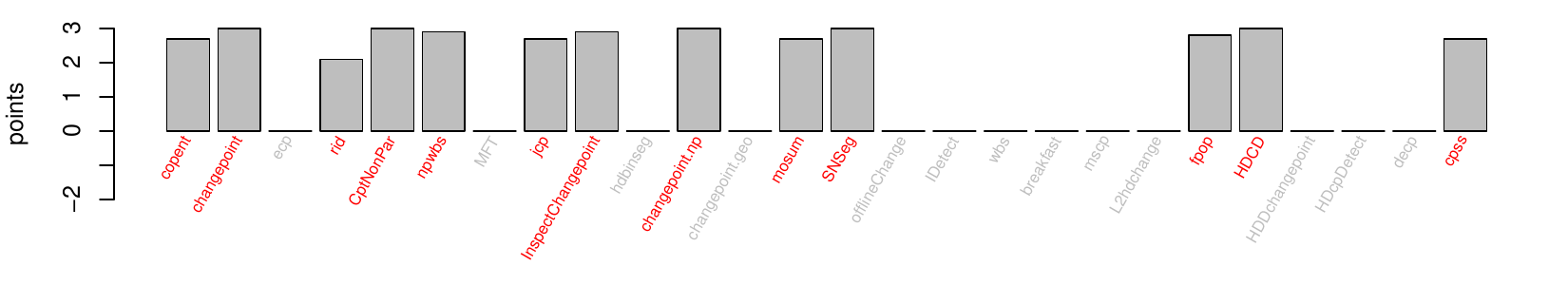}}
	\subfigure[Univariate var]{\includegraphics[width=\textwidth]{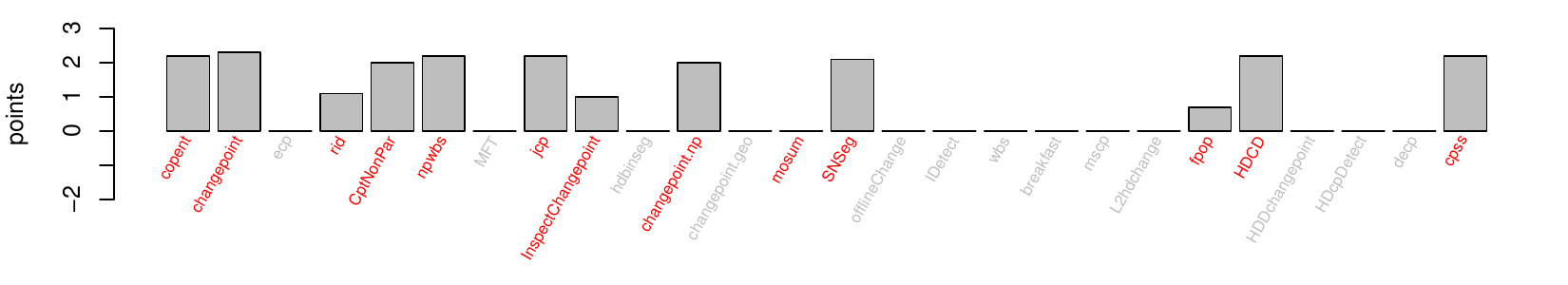}}
	\subfigure[Bivariate mean]{\includegraphics[width=\textwidth]{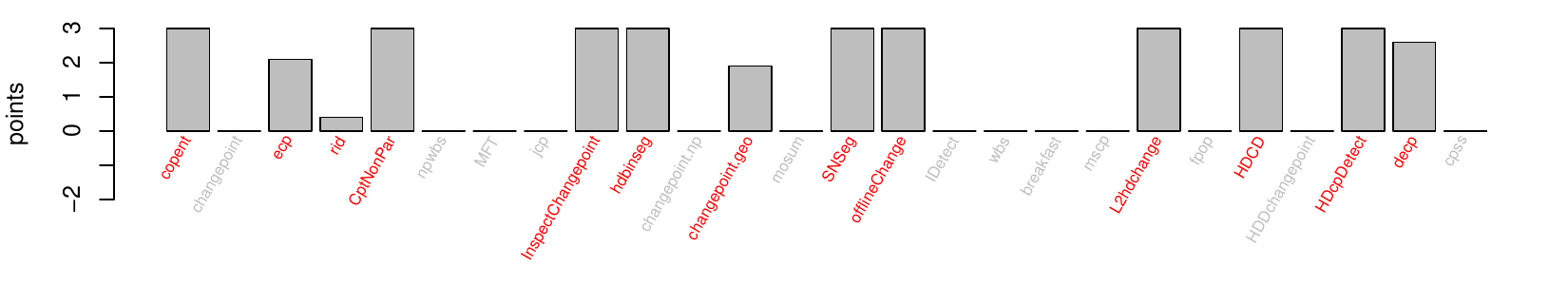}}
	\subfigure[Bivariate mean-var]{\includegraphics[width=\textwidth]{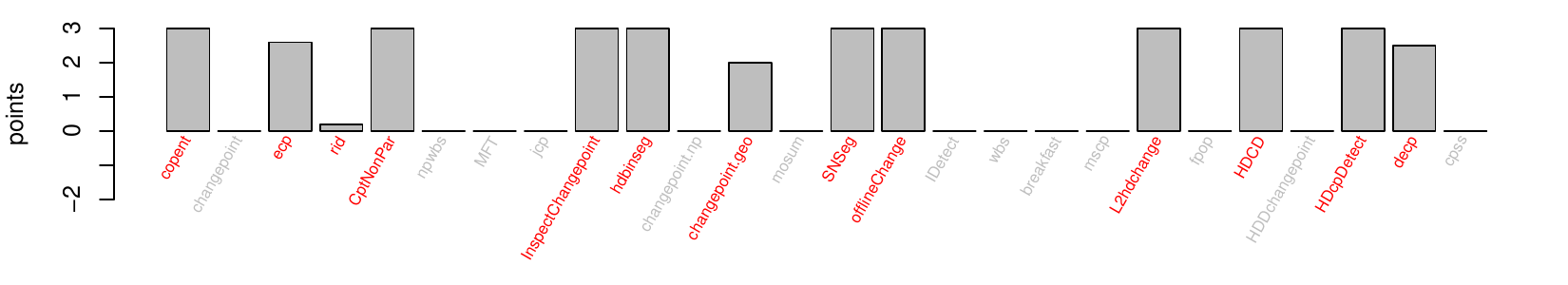}}
	\subfigure[Bivariate Var]{\includegraphics[width=\textwidth]{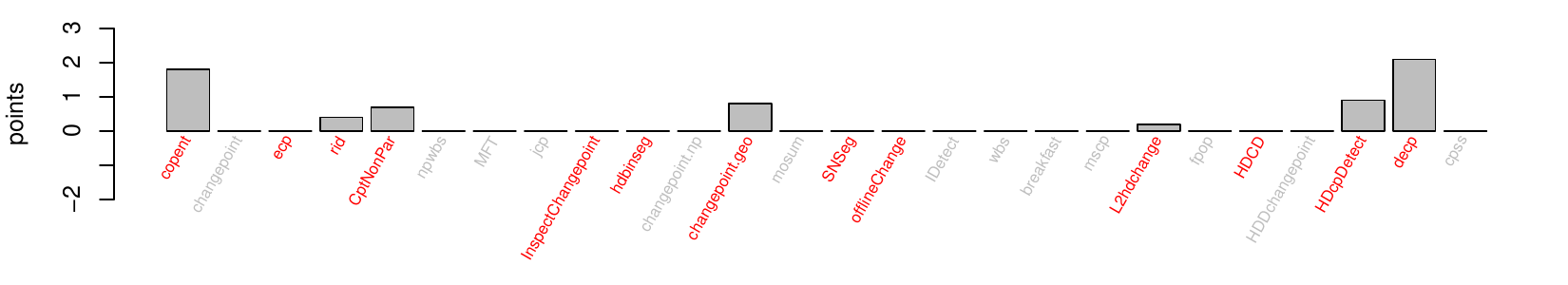}}
	\caption{Average number of detected change points by the methods evaluated in simulation experiments.}
	\label{fig:benchcpd}
\end{figure}

\section{Symmetry tests}
\label{s:symbench}
We designed three simulation experiments to verify the effectiveness of CE base symmetry test and compared it with the other tests\cite{Ma2025a}. The following tests implemented in the \texttt{symmetry} package in \textsf{R} are used for comparison:

\paragraph{MI} Mira test statistic \cite{Mira1999};

\paragraph{CM} Cabilio–Masaro test statistic \cite{Cabilio1996};
	
\paragraph{MGG} Miao-Gel-Gastwirth test statistic \cite{Miao2006};
	
\paragraph{B1} $\sqrt{b_1}$ test statistic \cite{Milosevic2019};
	
\paragraph{KS} Kolmogorov–Smirnov test statistic \cite{Milosevic2019};
	
\paragraph{SGN} sign test statistic\cite{Milosevic2019};
	
\paragraph{WCX} Wilcoxon test statistic \cite{Milosevic2019};
	
\paragraph{FM} the test statistic based on characteristic function\cite{Feuerverger1977};
	
\paragraph{RW} Rothman-Woodroofe test statistic \cite{Gaigall2020};
	
\paragraph{BHI} Litvinova test statistic \cite{Litvinova2001};
	
\paragraph{BHK} Baringhaus-Henze test statistic \cite{Baringhaus1992};
	
\paragraph{BH2} Baringhaus-Henze test statistic \cite{Baringhaus1992};
	
\paragraph{MOI and MOK} Milo\v{s}evi\'c-Obradovi\'c test statistic \cite{Milosevic2016};
	
\paragraph{NAI and NAK} Nikitin-Ahsanullah test statistic \cite{Nikitin2015};
	
\paragraph{K2 and K2U} Bo\v{z}in-Milo\v{s}evi\'c-Nikitin-Obradovi\'c  test statistic \cite{Bozin2020};
	
\paragraph{NAC1, NAC2, BHC1 and BHC2} Allison-Pretorius test statistic \cite{Allison2017}.

~\

Simulation experiments\footnote{The code is available at \url{https://github.com/majianthu/symmetry}} are based on three asymmetric distribution families, including

\paragraph{Beta distribution} Beta distributions are the continuous distribution on $[0,1]$ or $(0,1)$ with two parameters $a,b>0$ for symmetry control, defined as
\begin{equation}
	f(x;a,b)=\frac{\Gamma(a+b)}{\Gamma(a)\Gamma(b)}x^{a-1}(1-x)^{b-1},
	\label{eq:betafun}
\end{equation} 
where $\Gamma$ is Gamma function.

\paragraph{Asymmetric Laplace distribution} Asymmetric Laplace distributions generalize Laplace distribution and are defined as
\begin{equation}
	f(x;\mu,\delta,k)=\frac{1}{\delta(k+k^{-1})}e^{-(x-\mu)k^s s/\delta},
\end{equation}
where $s=sgn(x-\mu)$, $\mu,\delta$ is for location and scale, and $k$ is the parameter for symmetry control.

\paragraph{Bimodal normal distribution} Bimodal normal distribution is derived by mixture of two normal distributions $X_1\sim N(\mu_1,\delta_1)$ and $X_2\sim N(\mu_2,\delta_2)$ with a proportion parameter $p$ for symmetry control.

In the simulation with Beta distribution, nine samples are generated by $a+b=10$ and $b=1,\cdots,9$; in the simulation with asymmetric Laplace distribution, nine samples are generated with $\mu=0,\delta=1$ and $k=0.1,\cdots,0.9$; in the simulation with bimodal normal distribution, nine samples are generated with $\mu_1=0,\mu_2=5,\delta_1=\delta_2=1$ and $p=0.1,\cdots,0.9$ . The sample size is 300 for all. So we simulate the process from asymmetry to symmetry to asymmetry again.

Experimental results are shown in Figure \ref{fig:beta}, Figure \ref{fig:alaplace}, and Figure\ref{fig:binorm} from which one can learn that CE based test measures the change of symmetry of the simulated distributions correctly and performs better than others.

\begin{figure}
	\includegraphics[width=\textwidth]{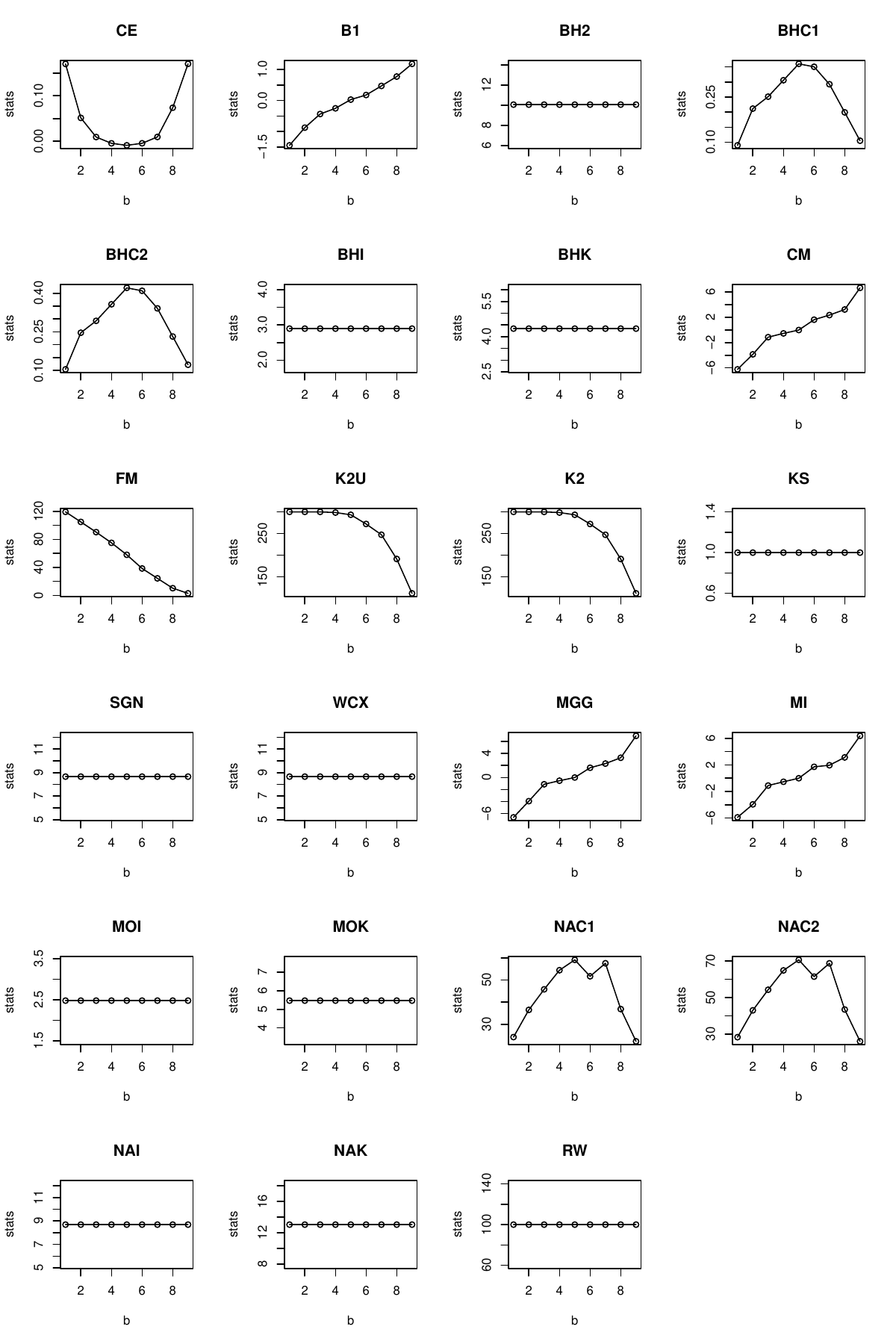}
	\caption{Results of the simulation experiments with Beta distribution.}
	\label{fig:beta}
\end{figure}

\begin{figure}
	\includegraphics[width=\textwidth]{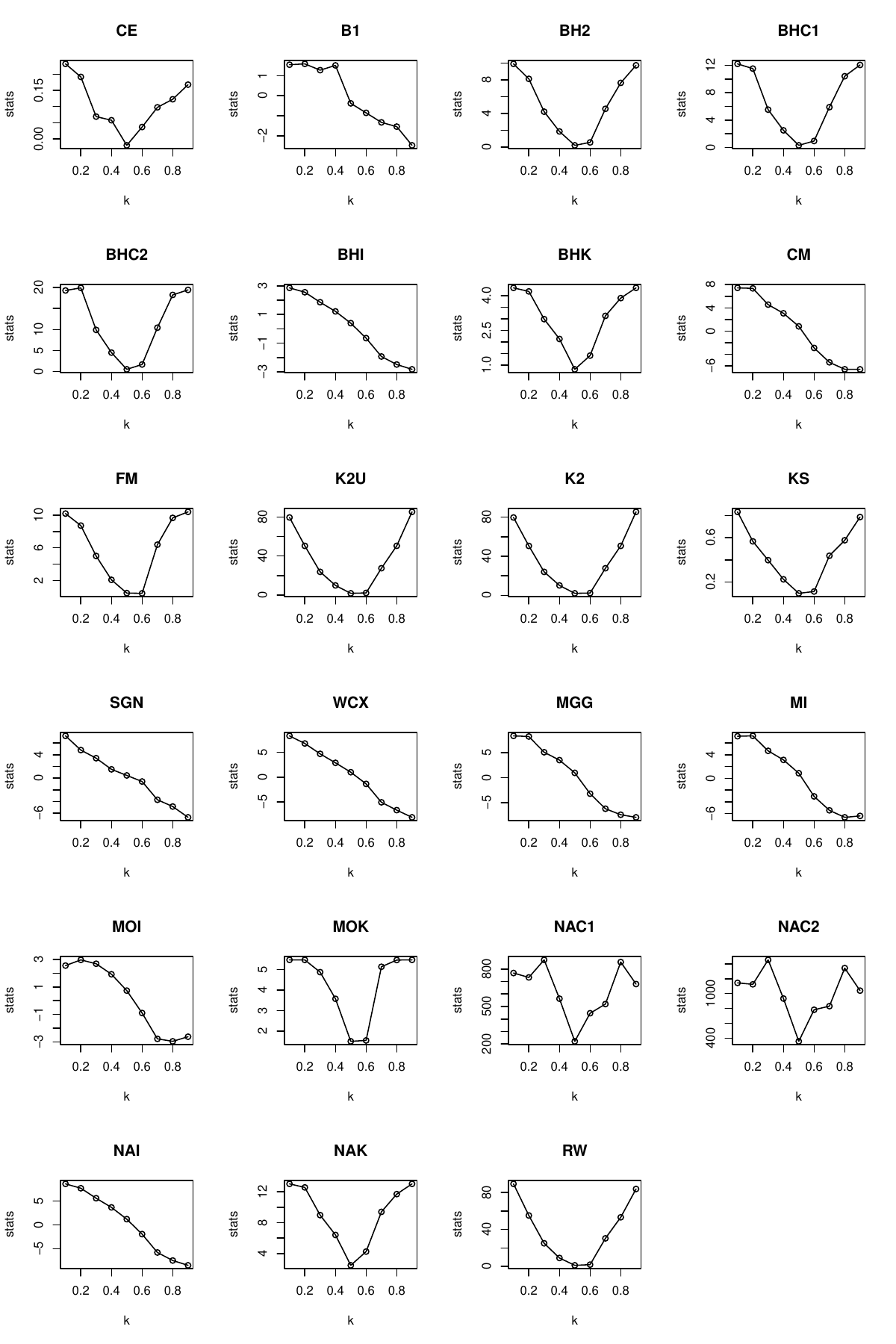}
	\caption{Results of the simulation experiments with asymmetric Laplace distribution.}
	\label{fig:alaplace}
\end{figure}

\begin{figure}
	\includegraphics[width=\textwidth]{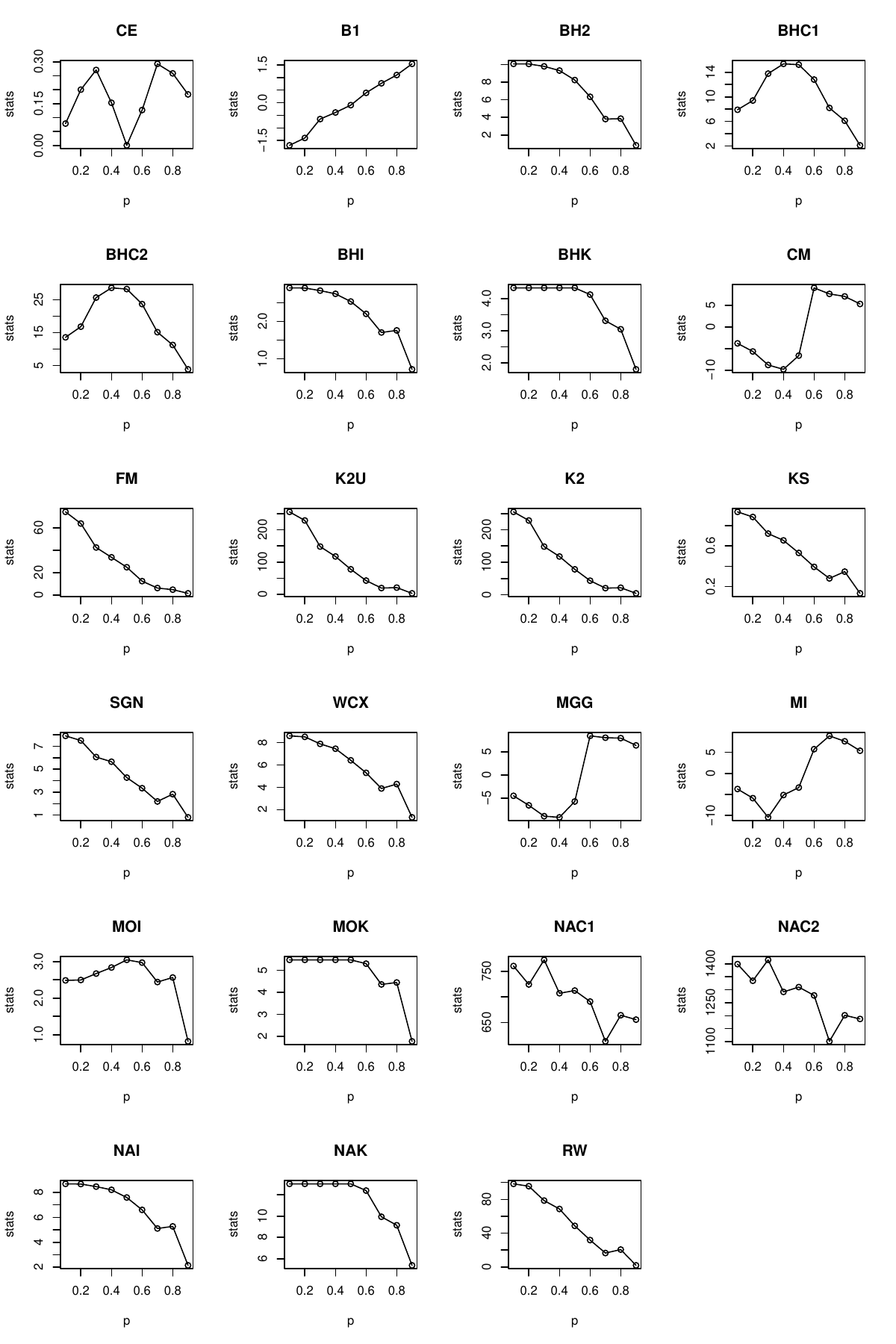}
	\caption{Results of the simulation experiments with bimodal normal distribution.}
	\label{fig:binorm}
\end{figure}

\chapter{Generalizations}
\label{chap:math}

\section{Tsallis copula entropy}
Shannon entropy has the properties, such as continuity, additivity, and symmetry\cite{Khinchin1957,Csiszar2008} and CE, as a special type of entropy, shares these properties as well. There are many generalized entropies, including Tsallis entropy \cite{Tsallis1988,Tsallis2011,Ilic2014} which is derived by extending additivity.

Given random variables $(X,Y)\sim P(x,y)$, Shannon entropy is defined as
\begin{equation}
	H(x,y)=-\int_{x}\int_{y} p(x,y)\log p(x,y)dxdy.
\end{equation}
Accordingly, the definition of Tsallis entropy is:
\begin{definition}[Tsallis Entropy]
	Given random variables $(X,Y)\sim P(x,y)$, Tsallis entropy is defined as
	\begin{equation}
		H_T(x,y)=\frac{1}{q-1}\big(1-\int_{x}\int_{y}p^q(x,y)dxdy\big),
		\label{eq:tsallis}
	\end{equation}
	where $q\in R$.	
\end{definition}

If $q\rightarrow 1$, then Tsallis entropy becomes Shannon entropy:
\begin{equation}
	\lim\limits_{q\rightarrow 1}H_T(x,y)=H(x,y).
\end{equation}

Mortezanejad et al. \cite{Mortezanejad2019} defined Tsallis CE by replacing probability density function $p(x,y)$ with copula density function $c(u,v)$ in \eqref{eq:tsallis}:
\begin{definition}[Tsallis Copula Entropy]
	Given random variables $(X,Y)$ and their copula density function $c(u,v)$, Tsallis copula entropy is defined as
	\begin{equation}
		H_{TC}(x,y)=\frac{1}{q-1}\big(1-\int_{u}\int_{v}c^q(u,v)dudv\big),
		\label{eq:tsallisce}
	\end{equation}
	where $q\in R$.
\end{definition}
If $q\rightarrow 1$, then Tsallis CE becomes CE:
\begin{equation}
	\lim\limits_{q\rightarrow 1}H_{TC}(x,y)=H_C(x,y).
\end{equation}

Based on the definition of Tsallis CE, Mortezanejad et al. proposed maximum Tsallis CE principle\cite{Mortezanejad2019}. Given
moments conditions 
\begin{align}
	\int_{u}\int_{v}c(u,v)dudv&=1,\label{eq:moment0}\\
	\int_{u}\int_{v}u^r c(u,v)dudv&=\frac{1}{r+1},\label{eq:moments1}\\
	\int_{u}\int_{v}v^r c(u,v)dudv&=\frac{1}{r+1}, \label{eq:moments2}
\end{align}
$r=1,\ldots,n$, and conditions of Spearman's $\rho$
\begin{equation}
	\int_{u}\int_{v}uv\ c(u,v)dudv=\frac{\rho+3}{12},
	\label{eq:rhoconstraint}
\end{equation}
Lagrange equation of maximum Tsallis CE is
\begin{align}
	L(c,\Lambda)=&\frac{1}{q-1}\big(1-\int_{u}\int_{v}c^q(u,v)dudv\big)\\
	&-\lambda_0\big(\int_{u}\int_{v}c(u,v)dudv-1\big)\\
	&-\sum_{r=1}^{n}\lambda_r^{u}\big(\int_{u}\int_{v}u^r c(u,v)dudv-\frac{1}{r+1}\big)\\
	&-\sum_{r=1}^{n}\lambda_r^{v}\big(\int_{u}\int_{v}v^r c(u,v)dudv-\frac{1}{r+1}\big)\\
	&-\lambda_{\rho}\big(\int_{u}\int_{v}uv\ c(u,v)dudv-\frac{\rho+3}{12}\big),
	\label{eq:tcelagrange}
\end{align}
where $\Lambda=(\lambda_0,\lambda_1^{u},\ldots,\lambda_n^{u},\lambda_1^{v},\ldots,\lambda_n^{v},\lambda_\rho)$.
By differentiating this equation with respect $c$, one can derive
\begin{equation}
	\frac{\partial L(c,\Lambda)}{\partial c}=-\frac{
	q}{q-1}c^{q-1} - \lambda_0 - \sum_{r=1}^n \lambda_r^u u^r - \sum_{r=1}^n \lambda_r^v v^r - \lambda_{\rho}uv=0.
\end{equation}
Then copula density function $c$ can be found out from it
\begin{equation}
	c(u,v)= \sqrt[q-1]{-\frac{q-1}{q}(\lambda_0+\sum_{r=1}^n(\lambda_r^{u}u^r+\lambda_r^{v}v^r)+\lambda_{\rho}uv)},
	\label{eq:tsalliscopula}
\end{equation}
where the elements of $\Lambda$ can be derived by inserting \eqref{eq:tsalliscopula} into \eqref{eq:moment0}, \eqref{eq:moments1}, \eqref{eq:moments2}, and \eqref{eq:rhoconstraint}.

\section{Survival copula entropy}
Survival function is a special type of distribution corresponding to cumulative distribution function. Given random variables $X,Y$, survival function is defined as
\begin{equation}
	\overline{F}(x,y)=P(X>x,Y>y).
\end{equation}
Accordingly, one can define bivariate survival copula as
\begin{equation}
	\overline{C}(u,v)=C(U>u,V>v).
\end{equation}
where $u,v$ are margins.

With survival function, Rao et al.\cite{Rao2004} defined CUmulative Residual Entropy (CURE):
\begin{definition}[Cumulative Residual Entropy]
	Given random variables $(X,Y)$ and their survival function $\overline{F}(x,y)$, cumulative residual entropy is defined as
	\begin{equation}
		H_{CURE}(x,y)=-\int_{x}\int_{y}\overline{F}(x,y)\log \overline{F}(x,y)dxdy.
	\end{equation}
\end{definition}

In a same way, Nair and Sunoj\cite{Sunoj2023} defined survival CE (SCE) by replacing copula density function in Definition \ref{def:ce} with survival copula:
\begin{definition}[Survival Copula Entropy]
	Given random variables $(X,Y)$ and their survival copula $\overline{C}(u,v)$, survival copula entropy is defined as
	\begin{equation}
		H_{SCE}(x,y)=-\int_{0}^{1}\int_{0}^{1}\overline{C}(u,v)\log \overline{C}(u,v)dudv.
		\label{eq:sce}
	\end{equation}
\end{definition}

\section{Cumulative copula entropy}
Cumulative distribution function and survival function are a pair of related concepts. Given random variable $X$, cumulative distribution function is defined as
\begin{equation}
	F(x)=P(X\leq x).
\end{equation}

With it, Di Crescenzo and Longobardi\cite{DiCrescenzo2009} defined CUmulative Entropy (CUE):
\begin{definition}[Cumulative Entropy]
	Given random variables $\mathbf{X}$ and their cumulative distribution function $F(\mathbf{x})$, cumulative entropy is defined as
	\begin{equation}
	H_{CUE}(\mathbf{x})=-\int_{\mathbf{x}}F(\mathbf{x})\log F(\mathbf{x})d\mathbf{x}.
\end{equation}
\end{definition}

Arshad et al.\cite{Arshad2024} defined Cumulative Copula Entropy (CCE):
\begin{definition}[Cumulative Copula Entropy]
	Given random variables $\mathbf{X}$ and their cumulative copula $C(\mathbf{u})$, cumulative copula entropy is defined as
	\begin{equation}
		H_{CCE}(\mathbf{x})=-\int_{\mathbf{u}}C(\mathbf{u})\log C(\mathbf{u})d\mathbf{u}.
	\end{equation}
\end{definition}

They then defined Fractional Cumulative Copula Entropy (FCCE):
\begin{definition}[Fractional Cumulative Copula Entropy]
	Given random variables $\mathbf{X}$ and their cumulative copula $C(\mathbf{u})$, fractional cumulative copula entropy is defined as
	\begin{equation}
		H_{FCCE}(\mathbf{x})=\int_{\mathbf{u}}C(\mathbf{u})(-\log C(\mathbf{u}))^r d\mathbf{u}, 
	\end{equation}
	where $0<r<1$.
\end{definition}

With empirical copula, they then defined Empirical Beta CCE:
\begin{definition}[Empirical Beta Cumulative Copula Entropy]
	Given random variables $\mathbf{X}$ and their empirical copula $\hat{C}(\mathbf{u})$, empirical beta cumulative copula entropy is defined as
	\begin{equation}
		H_{EBCCE}(\mathbf{x})=-\int_{\mathbf{u}}\hat{C}(\mathbf{u})\log \hat{C}(\mathbf{u}) d\mathbf{u}.
	\end{equation}
\end{definition}
Furthermore, fractional empirical beta CCE can be defined:
\begin{definition}[Fractional Empirical Beta Cumulative Copula Entropy]
	Given random variables $\mathbf{X}$ and their empirical copula $\hat{C}(\mathbf{u})$, fractional empirical beta cumulative copula entropy is defined as
	\begin{equation}
		H_{FEBCCE}(\mathbf{x})=\int_{\mathbf{u}}\hat{C}(\mathbf{u})( -\log \hat{C}(\mathbf{u}) )^r d\mathbf{u}, 
	\end{equation}
	where $0<r<1$.
\end{definition}

They also defined Cumulative Copula Kullback-Leibler Divergence (CCKL) and applied it to goodness of fit of copulas.

\begin{definition}[Cumulative Copula Kullback-Leibler Divergence]
	Given two copulas  $C_1(\mathbf{u}),C_2(\mathbf{u})$, cumulative copula Kullback-Leibler divergence is defined as
	\begin{equation}
		D_{CCKL}(C_1,C_2)=\int_{\mathbf{u}} C_1(\mathbf{u})\log \frac{C_1(\mathbf{u})}{C_2(\mathbf{u})} d\mathbf{u}- \frac{\rho_k(C_1)-\rho_k(C_2)}{2^k n(k)},
	\end{equation}
where $\rho_k(C)=n(k)(2^k\int_{\mathbf{u}}C(\mathbf{u})d\mathbf{u}-1)$ is k dimensional Spearman's $\rho$, $n(k)=\frac{k+1}{2^k-k-1}$.
\end{definition}

\section{Copula extropy}
Extropy is the dual to entropy\cite{Lad2015}. Given random variables $(X,Y)\sim P(x,y)$, Shannon entropy is
\begin{equation}
	H(x,y)=-\int_{x}\int_{y} p(x,y)\log p(x,y)dxdy,
\end{equation}
and its corresponding extropy is defined as\cite{Balakrishnan2022}
\begin{definition}[Extropy]
	Given random variables $(X,Y)\sim P(x,y)$, extropy is defined as
	\begin{equation}
		J(x,y)=\frac{1}{4}\int_{0}^{\infty}\int_{0}^{\infty}p^2(x,y)dxdy.
	\end{equation}
\end{definition}

Saha and Kayal\cite{Saha2023} defined Copula Extropy (CEx):
\begin{definition}[Copula Extropy]
	Given random variables $(X,Y)$ and their copula density function $c(u,v)$, copula extropy is defined as
	\begin{equation}
		J_c(x,y)= \frac{1}{4}\int_{0}^{1}\int_{0}^{1}c^2(u,v)dudv.
	\end{equation}
\end{definition}
Saha and Kayal also defined Cumulative Copula Extropy (CCEx):
\begin{definition}[Cumulative Copula Extropy]
   Given random variables $(X,Y)$ and their copula $C(u,v)$, cumulative copula extropy is defined as
	\begin{equation}
		J_C(x,y)=\frac{1}{4}\int_{0}^{1}\int_{0}^{1}C^2(u,v)dudv.
	\end{equation}
\end{definition}
They also defined Survival Copula  Extropy (SCEx):
\begin{definition}[Survival Copula Extropy]
	Given random variables $(X,Y)$ and their survival copula $\overline{C}(u,v)$, survival copula extropy is defined as
	\begin{equation}
		J_{\overline{C}}(x,y)=\frac{1}{4}\int_{0}^{1}\int_{0}^{1}\overline{C}^2(u,v)dudv.
	\end{equation}
\end{definition}

\section{Cumulative copula Tsallis entropy}
Given random variables $\mathbf{X}$, the definition of Tsallis entropy \eqref{eq:tsallis} can be written as follows\cite{Tsallis1988}:
\begin{equation}
	H_T(\mathbf{x})=\frac{1}{q-1}\int_{\mathbf{x}}\big(p(\mathbf{x})-p(\mathbf{x})^q\big)d\mathbf{x},
	\label{eq:tsallis2}
\end{equation}
where $q\in R$, and as the following form further:
\begin{equation}
	H_T(\mathbf{x})=-\int_{\mathbf{x}}p(\mathbf{x})\log_{q}(p(\mathbf{x}))d\mathbf{x},
\end{equation}
where $\log_q(r)=\frac{r^{q-1}-1}{q-1}, q>0, q\neq 1$.

Zachariah et al.\cite{Zachariah2025} defined Cumulative Copula Tsallis Entropy (CCTE):
\begin{definition}[Cumulative Copula Tsallis Entropy]
	Given random variables $\mathbf{X}$ and their copula $C(\mathbf{u})$, cumulative copula Tsallis entropy is defined as
	\begin{equation}
		H_{CCTE}(\mathbf{x})=-\int_{\mathbf{u}}C(\mathbf{u})\log_q(C(\mathbf{u}))d\mathbf{u},
	\end{equation}
	where $q>0, q\neq 1$.
\end{definition}

They then defined Empirical Cumulative Copula Tsallis Entropy (ECCTE):
\begin{definition}[Empirical Cumulative Copula Tsallis Entropy]
	Given random variables $\mathbf{X}$ and their copula $C(\mathbf{u})$ and empirical cumulative copula $\hat{C}(\mathbf{u})$, cumulative copula Tsallis entropy is defined as
	\begin{equation}
		H_{ECCTE}(\mathbf{x})=-\int_{\mathbf{u}}\hat{C}(\mathbf{u})\log_q(\hat{C}(\mathbf{u}))d\mathbf{u},
	\end{equation}
	where $q>0, q\neq 1$.
\end{definition}

They also defined cumulative copula Tsallis inaccuracy:
\begin{definition}[Cumulative Copula Tsallis Inaccuracy]
	Given random variables $\mathbf{X}_1,\mathbf{X}_2\in R^n$and their copula $C_1(\mathbf{u})$ and $C_2(\mathbf{u})$, cumulative copula Tsallis inaccuracy is defined as
	\begin{equation}
		D_{CCTI}(C_1,C_2)=-\int_{\mathbf{u}}C_1(\mathbf{u})\log_{q}C_2(\mathbf{u})d\mathbf{u}.
	\end{equation}
\end{definition}

Inspired by CCKL\cite{Arshad2024}, they also defined cumulative copula Tsallis divergence:
\begin{definition}[Cumulative Copula Tsallis Divergence]
	Given random variables $\mathbf{X}_1,\mathbf{X}_2\in R^n$ and their copula $C_1(\mathbf{u})$ and $C_2(\mathbf{u})$, cumulative copula Tsallis divergence is defined as
	\begin{equation}
		D_{CCTD}(C_1,C_2)=\int_{\mathbf{u}}C_1(\mathbf{u})\log_{q}\frac{C_1(\mathbf{u})}{C_2(\mathbf{u})} d\mathbf{u}-\frac{\rho_k(C_1)-\rho_k(C_2)}{2^k n(k)},
	\end{equation}
	where $\rho_k(C)=n(k)(2^k\int_{\mathbf{u}}C(\mathbf{u})d\mathbf{u}-1)$ is k dimensional Spearman's $\rho$, $n(k)=\frac{k+1}{2^k-k-1}$.
\end{definition}
It is easy to know that
\begin{equation}
	\lim\limits_{q\rightarrow 1}D_{CCTD}(C_1,C_2)=D_{CCKL}(C_1,C_2).
\end{equation}

Furthermore, they defined cumulative copula Tsallis MI:
\begin{definition}[Cumulative Copula Tsallis Mutual Information]
	Given random variables $\mathbf{X}$ and their copula $C(\mathbf{u})$ and corresponding independence copula $\Pi(\mathbf{u})=\prod_{i=1}^{k}u_i$, cumulative copula Tsallis mutual information is defined as
	\begin{equation}
		I_{CCTMI}(C)=D_{CCTD}(C,\Pi)=\int_{\mathbf{u}}C(\mathbf{u})\log_q \frac{C(\mathbf{u})}{\Pi(\mathbf{u})}d\mathbf{u}-\frac{\rho_k(C)}{2^k n(k)},
	\end{equation}
	where $\rho_k(C)=n(k)(2^k\int_{\mathbf{u}}C(\mathbf{u})d\mathbf{u}-1)$ is k dimensional Spearman's $\rho$, $n(k)=\frac{k+1}{2^k-k-1}$.
\end{definition}

\section{Copula R\'enyi entropy}
R\'enyi Entropy is a typical generalization of Shannon entropy, which was proposed by Alfr\'ed R\'enyi in 1961\cite{Renyi1961}. R\'enyi entropy is the general information measure that satisfies additivity with Hartley entropy Shannon entropy, collision entropy, and min-entropy as its special cases\cite{Ribeiro2021,Ozawa2024}. 

\begin{definition}[R\'enyi Entropy]
	Given random variables $\mathbf{X}$, R\'enyi entropy is defined as
	\begin{equation}
		H_R(\mathbf{x})=\frac{1}{1-\alpha}\log \int_{\mathbf{x}}p^\alpha (\mathbf{x}) d\mathbf{x},
		\label{eq:renyientropy}
	\end{equation}
	where $\alpha>0,\alpha\neq 1$.
\end{definition}

If $\alpha\rightarrow 1$, then R\'enyi entropy degenerates to Shannon entropy:
\begin{equation}
	\lim\limits_{\alpha\rightarrow 1}H_R(\mathbf{x})=H(\mathbf{x}).
\end{equation}

We can define Copula R\'enyi entropy:
\begin{definition}[Copula R\'enyi Entropy]
	Given random variables $\mathbf{X}$ and their copula density function $c(\mathbf{u})$, copula R\'enyi entropy is defined as
	\begin{equation}
		H_{CR}(\mathbf{x})=\frac{1}{1-\alpha}\log\int_{\mathbf{u}}c^\alpha(\mathbf{u})d\mathbf{u},
		\label{eq:copularenyientropy}
	\end{equation}
	where $\alpha>0,\alpha\neq 1$.
\end{definition}

If $\alpha\rightarrow 1$, then copula R\'enyi entropy becomes CE:
\begin{equation}
	\lim\limits_{\alpha\rightarrow 1}H_{CR}(\mathbf{x})=H_c(\mathbf{x}).
\end{equation}

Saha and Kayal\cite{Saha2025} defined Cumulative Copula R\'enyi Entropy (CCRE) and Survival Copula R\'enyi Entropy (SCRE):

\begin{definition}[Cumulative Copula R\'enyi Entropy]
	Given random variables $\mathbf{X}$ and cumulative copula $C(\mathbf{u})$, cumulative copula R\'enyi entropy is defined as
	\begin{equation}
		H_{CCRE}(\mathbf{x})=\frac{1}{1-\alpha}\log \int_{\mathbf{u}}C^\alpha(\mathbf{u})d\mathbf{u},
	\end{equation}
	where $\alpha>0,\alpha\neq 1$.
\end{definition}

\begin{definition}[Survival Copula R\'enyi Entropy]
	Given random variables $\mathbf{X}$ and their survival copula $\overline{C}(\mathbf{u})$, survival copula R\'enyi entropy is defined as
	\begin{equation}
		H_{SCRE}(\mathbf{x})=\frac{1}{1-\alpha}\log \int_{\mathbf{u}}\overline{C}^\alpha(\mathbf{u})d\mathbf{u},
	\end{equation}
	where $\alpha>0,\alpha\neq 1$.
\end{definition}

If $\alpha > 1$, then $H_{CCRE}(\mathbf{x})>0$; if $0<\alpha<1$, then $H_{CCRE}(\mathbf{x})<0$.

If $\alpha > 1$, then $H_{SCRE}(\mathbf{x})\geq 0$.

If $\alpha\rightarrow 1$, then CCRE and SCRE degenerate to CCE and SCE respectively:
\begin{equation}
	\lim\limits_{\alpha\rightarrow 1}H_{CCRE}(\mathbf{x})=H_{CCE}(\mathbf{x})
\end{equation}
and
\begin{equation}
	\lim\limits_{\alpha\rightarrow 1}H_{SCRE}(\mathbf{x})=H_{SCE}(\mathbf{x}).
\end{equation}

They also defined cumulative copula R\'enyi inaccuracy:
\begin{definition}[Cumulative Copula R\'enyi Inaccuracy]
	Given random variables $\mathbf{X},\mathbf{Y}\in R^n$ and their copula $C_\mathbf{X}(\mathbf{u}),C_\mathbf{Y}(\mathbf{u})$ and margins $F_i(x_i),G_i(y_i),i=1,\ldots,n$, cumulative copula R\'enyi inaccuracy is defined as
	\begin{equation}
		D_{CCRI}(C_\mathbf{X},C_\mathbf{Y})=\frac{1}{1-\alpha}\log\int_{\mathbf{u}}C_\mathbf{X}(\mathbf{u})\{C_\mathbf{Y}(G_i(F_i^-(u_i)))\}^{\alpha-1}d\mathbf{u},
	\end{equation}
	where $\alpha>0,\alpha\neq 1$.
\end{definition}
If $C_\mathbf{X}=C_\mathbf{Y}$, then $D_{CCRI}=H_{CCRE}$.

Similarly, they defined survival copula R\'enyi inaccuracy:
\begin{definition}[Survival Copula R\'enyi Inaccuracy]
	Given random variables $\mathbf{X},\mathbf{Y}\in R^n$ and their survival copula $\overline{C}_\mathbf{X}(\mathbf{u}),\overline{C}_\mathbf{Y}(\mathbf{u})$ and margins $\overline{F}_i(x_i),\overline{G}_i(y_i),i=1,\ldots,n$, survival copula R\'enyi inaccuracy is defined as
	\begin{equation}
		D_{SCRI}(\overline{C}_\mathbf{X},\overline{C}_\mathbf{Y})=\frac{1}{1-\alpha}\log\int_{\mathbf{u}}\overline{C}_\mathbf{X}(\mathbf{u})\{\overline{C}_\mathbf{Y}(\overline{G}_i(\overline{F}_i^-(u_i)))\}^{\alpha-1}d\mathbf{u},
	\end{equation}
	where $\alpha>0,\alpha\neq 1$.
\end{definition}
If $\overline{C}_\mathbf{X}=\overline{C}_\mathbf{Y}$, then $D_{SCRI}=H_{SCRE}$.

\section{Copula R\'enyi divergence and copula Tsallis divergence}
Divergence is a fundamental concept in information theory and probability. KL divergence is the most well known one among many divergences\cite{Kullback1951}, which is defined as follows:
\begin{definition}[Kullback-Leibler Divergence]
	Given probability density functions $p_X,p_Y$, KL divergence is defined as
	\begin{equation}
		D_{KL}(p_X,p_Y)=\int_{x}p_X(x)\log \frac{p_X(x)}{p_Y(x)}dx.
	\end{equation}	
\end{definition}
For any $p_X,p_Y$, $D_{KL}(p_X,p_Y)\geq 0$; if $p_X=p_Y$, then $D_{KL}(p_X,p_Y)=0$. $D_{KL}$ is asymmetric which means $D_{KL}(p_X,p_Y)\neq D_{KL}(p_Y,p_X)$.

Scholars have defined many types of divergences, including R\'enyi divergence\cite{Renyi1961} and Tsallis  divergence\cite{Borland1998,Furuichi2004}:
\begin{definition}[R\'enyi Divergence]
	Given probability density functions $p_X,p_Y$, R\'enyi divergence is defined as
	\begin{equation}
		D_R(p_X,p_Y)=\frac{1}{\alpha-1}\log\int_x p_X^\alpha(x)p_Y^{1-\alpha}(x)dx,
	\end{equation}
	where $\alpha>0,\alpha\neq 1$.
\end{definition}

\begin{definition}[Tsallis Divergence]
	Given probability density functions $p_X,p_Y$, Tsallis divergence is defined as
	\begin{equation}
		D_T(p_X,p_Y)=\frac{1}{q-1}\big(\int_x p_X^q(x)p_Y^{1-q}(x)dx-1\big),
	\end{equation}
	where $q>0,q\neq 1$.
\end{definition}

KL divergence is a special case of both R\'enyi divergence and Tsallis divergence\cite{Renyi1961,Borland1998}:
\begin{equation}
	\lim\limits_{\alpha\rightarrow 1}D_R(p_X,p_Y)=\lim\limits_{q\rightarrow 1}D_T(p_X,p_Y)=D_{KL}(p_X,p_Y).
\end{equation}

Given random variables $(X,Y)$, KL divergence is related to MI and CE:
\begin{equation}
	D_{KL}(p_{XY},p_Xp_Y)=I(x,y)=-H_c(x,y).
\end{equation}

Mohammadi and Emadi\cite{Mohammadi2023a} defined copula R\'enyi divergence and copula Tsallis divergence that extend R\'enyi divergence and Tsallis divergence with copulas, which are defined as follows:
\begin{definition}[Copula R\'enyi Divergence]
	Given random variables $(X,Y)$ and their copula density function $c(u,v)$, copula R\'enyi divergence is defined as
	\begin{equation}
		D_{CR}(c_{XY})=\frac{1}{\alpha-1}\log\int_{u}\int_{v}c^\alpha(u,v)dudv,
		\label{eq:crdiv}
	\end{equation}
	where $\alpha>0,\alpha\neq 1$.
\end{definition}
\begin{definition}[Copula Tsallis Divergence]
	Given random variables $(X,Y)$ and their copula density function $c(u,v)$, copula Tsallis divergence is defined as
	\begin{equation}
		D_{CT}(c_{XY})=\frac{1}{q-1}\big(\int_{u}\int_{v}c^q(u,v)dudv-1\big),
		\label{eq:ctdiv}
	\end{equation}
	where $q>0,q\neq 1$.
\end{definition}

They gave the relation between these two divergences:
\begin{equation}
	D_{CT}(c_{XY})=\frac{1}{\alpha-1}\big(e^{(\alpha-1)D_{CR}(c_{XY})}-1\big).
\end{equation}

They also proved that
\begin{enumerate}
	\item $D_{CR}(c_{XY})\geq 0, D_{CT}(c_{XY})\geq 0$;
	\item if $X,Y$ are independent, then $D_{CR}(c_{XY})=D_{CT}(c_{XY})=0$; 
	\item $D_{CR}(c_{XY})$ and $D_{CT}(c_{XY})$ are invariant to monotonic transformations.
\end{enumerate}

These two copula divergences relate to copula R\'enyi entropy and Tsallis copula entropy respectively:
\begin{equation}
	D_{CR}(c_{XY})=-H_{CR}(x,y),
\end{equation}
and 
\begin{equation}
	D_{CT}(c_{XY})=-H_{TC}(x,y).
\end{equation}

\section{Copula Jeffreys divergence and copula Hellinger divergence}
\label{sec:jeffreys-hellinger}
KL divergence is asymmetry while Jeffreys divergence\cite{Jeffreys1946} and Hellinger divergence\cite{Hellinger1909} are two symmetric divergences in information theory, which both belong to $f$-divergence family\cite{Csiszar1963,Csiszar1967,Ali1966}.
\begin{definition}[Jeffreys Divergence]
	Given probability density functions $p_X,p_Y$, Jeffreys divergence is defined as
	\begin{equation}
		D_J(p_X,p_Y)=\int_x (p_X(x)-p_Y(x))(\log p_X(x)-\log p_Y(x))dx.
	\end{equation}	
\end{definition}

\begin{definition}[Hellinger Divergence]
	Given probability density functions $p_X,p_Y$, Hellinger divergence is defined as
	\begin{equation}		
		D_H(p_X,p_Y)=\int_x (\sqrt{p_X(x)}-\sqrt{p_Y(x)})^2 dx.
	\end{equation}	
\end{definition}

Mohammadi et al.\cite{Mohammadi2021} defined copula Jeffreys divergence and copula Hellinger divergence from $D_J(p_{XY},p_X p_Y)$ and $D_H(p_{XY},p_X p_Y)$:
\begin{definition}[Copula Jeffreys Divergence]
	Given random variables $(X,Y)$ and their copula density function $c(u,v)$, copula Jeffreys divergence is defined as
	\begin{equation}
		D_{CJ}(c_{XY})=\int_u \int_v (c(u,v)-1)\log c(u,v)dudv.
	\end{equation}
\end{definition}
\begin{definition}[Copula Hellinger Divergence]
	Given random variables $(X,Y)$ and their copula density function $c(u,v)$, copula Hellinger divergence is defined as
	\begin{equation}
		D_{CH}(c_{XY})=\int_u\int_v (\sqrt{c(u,v)}-1)^2 dudv.
	\end{equation}
\end{definition}

\section{Copula inaccuracy}
Inaccuracy\cite{Kerridge1961} was proposed for measuring the distance between distributions, which is defined as follows:
\begin{definition}[Inaccuracy]
	Given random variables $\mathbf{X},\mathbf{Y}\in R^n$ and their density functions $p_\mathbf{X},p_\mathbf{Y}$, inaccuracy is defined as
	\begin{equation}
		D_I(p_\mathbf{X},p_\mathbf{Y})=-\int_{\mathbf{x}}p_\mathbf{X}(\mathbf{x})\log p_\mathbf{Y}(\mathbf{x})d\mathbf{x}.
		\label{eq:di}
	\end{equation}
\end{definition}
Inaccuracy is related to entropy as $D_I(p_\mathbf{X},p_\mathbf{Y})\geq H(\mathbf{x})$. Obviously, if $p_\mathbf{X}=p_\mathbf{Y}$, then $D_I(p_\mathbf{X},p_\mathbf{Y})= H(\mathbf{x})$.

Inaccuracy is related to KL divergence as
\begin{equation}
	D_{KL}(p_\mathbf{X},p_\mathbf{Y})=D_I(p_\mathbf{X},p_\mathbf{Y})-H(\mathbf{x}).
\end{equation}

By replacing $p_\mathbf{X},p_\mathbf{Y}$ in \eqref{eq:di} with $c_\mathbf{X},c_\mathbf{Y}$, we can define copula inaccuracy:
\begin{definition}[Copula Inaccuracy]
	Given random variables $\mathbf{X},\mathbf{Y}\in R^n$ and their copula density functions $c_\mathbf{X}(\mathbf{u}),c_\mathbf{Y}(\mathbf{u})$, copula inaccuracy is defined as
	\begin{equation}
		D_{CI}(c_\mathbf{X},c_\mathbf{Y})=\int_{\mathbf{u}}c_\mathbf{X}(\mathbf{u})\log c_\mathbf{Y}(\mathbf{u}) d\mathbf{u}.
	\end{equation}
\end{definition}

So, if $c_\mathbf{X}=c_\mathbf{Y}$, then copula inaccuracy becomes CE.

We can further define cumulative copula inaccuracy:
\begin{definition}[Cumulative Copula Inaccuracy]
	Given random variables $\mathbf{X},\mathbf{Y}\in R^n$ and their cumulative copula $C_\mathbf{X}(\mathbf{u}),C_\mathbf{Y}(\mathbf{u})$, cumulative copula inaccuracy is defined as
	\begin{equation}
		D_{CCI}(C_\mathbf{X},C_\mathbf{Y})=\int_{\mathbf{u}}C_\mathbf{X}(\mathbf{u})\log C_\mathbf{Y}(\mathbf{u}) d\mathbf{u}.
	\end{equation}
\end{definition}

So, if $C_\mathbf{X}=C_\mathbf{Y}$, then cumulative copula inaccuracy becomes cumulative CE.

We mentioned CCTI\cite{Zachariah2025}, and CCRI and SCRI\cite{Saha2025} above. It can be easily proved that
\begin{equation}
	\lim\limits_{q\rightarrow 1}D_{CCTI}(C_\mathbf{X},C_\mathbf{Y})=\lim\limits_{\alpha\rightarrow 1}D_{CCRI}(C_\mathbf{X},C_\mathbf{Y})=D_{CCI}(C_\mathbf{X},C_\mathbf{Y}).
\end{equation}

Pandey and Kundu\cite{Pandey2025} defined Cumulative Copula Fractional Inaccuracy (CCFI):
\begin{definition}[Cumulative Copula Fractional Inaccuracy]
	Given random variables $\mathbf{X},\mathbf{Y}\in R^n$ and their cumulative copula $C_\mathbf{X}(\mathbf{u}),C_\mathbf{Y}(\mathbf{u})$ and margins $F_i(x_i),G_i(y_i),i=1,\ldots,n$, cumulative copula fractional inaccuracy is defined as
	\begin{equation}
		D_{CCFI}(C_\mathbf{X},C_\mathbf{Y})=\int_{\mathbf{u}}C_\mathbf{X}(\mathbf{u})\big[-\log C_\mathbf{Y}(G_i(F_i^-(u_i)))\big]^r d\mathbf{u},
	\end{equation}
	where $0<r<1$.
\end{definition}
So, if $C_\mathbf{X}=C_\mathbf{Y}$, then $D_{CCFI}=H_{FCCE}$.

They further defined Survival Copula Fractional Inaccuracy (SCFI):
\begin{definition}[Survival Copula Fractional Inaccuracy]
	Given random variables $\mathbf{X},\mathbf{Y}\in R^n$ and their survival copula $\overline{C}_\mathbf{X}(\mathbf{u}),\overline{C}_\mathbf{Y}(\mathbf{u})$ and margins $\overline{F}_i(x_i),\overline{G}_i(y_i),i=1,\ldots,n$, survival copula fractional inaccuracy is defined as
	\begin{equation}
		D_{SCFI}(\overline{C}_\mathbf{X},\overline{C}_\mathbf{Y})=\int_{\mathbf{u}}\overline{C}_\mathbf{X}(\mathbf{u})\big[-\log \overline{C}_\mathbf{Y}(\overline{G}_i(\overline{F}_i^-(u_i)))\big]^r d\mathbf{u},
	\end{equation}
	where $0<r<1$.
\end{definition}

\section{Copula information measures}
Borzadaran and Amini\cite{Borzadaran2010} extended a groups of information measures with copulas. Besides Jeffreys divergence and Hellinger divergence mentioned in Section \ref{sec:jeffreys-hellinger}, the other measures include:

\begin{definition}[$\chi^2$-Divergence]
	Given probability density function $f(x,y)$ and $g(x,y)$, $\chi^2$ divergence is defined as
	\begin{equation}
		D_{\chi^2}(f,g)=\int_x \int_y \frac{[f(x,y)-g(x,y)]^2}{f(x,y)}dxdy.		
	\end{equation}	
\end{definition}

\begin{definition}[$\alpha$-Divergence]
	Given probability density function $f(x,y)$ and $g(x,y)$, $\alpha$ divergence is defined as
	\begin{equation}
		D_{\alpha}(f,g)=\frac{1}{1-\alpha^2}\int_x \int_y \left\{1-\left[\frac{g(x,y)}{f(x,y)}\right]^{\frac{1+\alpha}{2}}\right\}f(x,y)dxdy.		
	\end{equation}	
\end{definition}

\begin{definition}[Combination of Version of $\alpha$-Divergence]
	Given probability density function $f(x,y)$ and $g(x,y)$, combination of version of $\alpha$-divergence is defined as
	\begin{equation}
		D_{C\alpha}(f,g)=\frac{4}{\beta^2}\int_x \int_y \left\{1-\left[\frac{f(x,y)}{g(x,y)}\right]^{\frac{\beta}{2}}\right\}^2 g(x,y)dxdy.		
	\end{equation}	
\end{definition}

\begin{definition}[Lei-Wang Divergence]
	Given probability density function $f(x,y)$ and $g(x,y)$, Lei-Wang divergence is defined as
	\begin{equation}
		D_{LW}(f,g)=\int_x \int_y f(x,y)\log \frac{2f(x,y)}{f(x,y)+g(x,y)} dxdy.		
	\end{equation}	
\end{definition}

\begin{definition}[Power Divergence]
	Given probability density function $f(x,y)$ and $g(x,y)$, Power divergence is defined as
	\begin{equation}
		D_{Po}(f,g)=\frac{1}{\lambda(\lambda+1)}\int_x \int_y \left\{\left[\frac{f(x,y)}{g(x,y)}\right]^\lambda -1\right\}f(x,y)dxdy.		
	\end{equation}	
\end{definition}

\begin{definition}[Bhattacharyya Distance]
	Given probability density function $f(x,y)$ and $g(x,y)$, Bhattacharyya distance is defined as
	\begin{equation}
		D_{Bh}(f,g)=\int_x \int_y \sqrt{f(x,y)g(x,y)} dxdy.		
	\end{equation}	
\end{definition}

\begin{definition}[Harmonic Distance]
	Given probability density function $f(x,y)$ and $g(x,y)$, harmonic distance is defined as
	\begin{equation}
		D_{Ha}(f,g)=\int_x \int_y \frac{2f(x,y)g(x,y)}{f(x,y)+g(x,y)} dxdy.		
	\end{equation}	
\end{definition}

\begin{definition}[Triangular Discrimination]
	Given probability density function $f(x,y)$ and $g(x,y)$, triangular discrimination distance is defined as
	\begin{equation}
		D_{\Delta}(f,g)=\int_x \int_y \frac{[f(x,y)-g(x,y)]^2}{f(x,y)+g(x,y)} dxdy.		
	\end{equation}	
\end{definition}

\begin{definition}[Relative Information Generating Function]
	Given probability density function $f(x,y)$ and $g(x,y)$, relative information generating function is defined as
	\begin{equation}
		D_{Ri}(f,g)=\int_x \int_y \left[\frac{f(x,y)}{g(x,y)}\right]^{t-1}f(x,y) dxdy,		
	\end{equation}
	where $t\geq 1$.
\end{definition}

Through the divergence between $p_{XY}$ and $p_X p_Y$, they defined the following copula information measures with copula density function $c_{XY}$:

\begin{definition}[Copula $\chi^2$-Divergence]
	Given random variables $(X,Y)$ and their copula density function $c(u,v)$, copula $\chi^2$ divergence is defined as
	\begin{equation}
		D_{\chi^2}(c_{XY})=\int_u \int_v \frac{[c(u,v)-1]^2}{c(u,v)}dudv.		
	\end{equation}	
\end{definition}

\begin{definition}[Copula $\alpha$-Divergence]
	Given random variables $(X,Y)$ and their copula density function $c(u,v)$, copula $\alpha$ divergence is defined as
	\begin{equation}
		D_{\alpha}(c_{XY})=\frac{1}{1-\alpha^2}\int_u \int_v \left\{1-[c(u,v)]^{-\frac{\alpha+1}{2}} \right\}c(u,v)dudv.		
	\end{equation}	
\end{definition}

\begin{definition}[Copula Combination of Version of $\alpha$-Divergence]
	Given random variables $(X,Y)$ and their copula density function $c(u,v)$, copula combination of version of $\alpha$-divergence is defined as
	\begin{equation}
		D_{C\alpha}(c_{XY})=\frac{4}{\beta^2}\int_u \int_v [1-c^\frac{\beta}{2}(u,v)]^2 dudv.			
	\end{equation}	
\end{definition}

\begin{definition}[Copula Lei-Wang Divergence]
	Given random variables $(X,Y)$ and their copula density function $c(u,v)$, copula Lei-Wang divergence is defined as
	\begin{equation}
		D_{LW}(c_{XY})=\int_u \int_v c(u,v)\log\left[\frac{2c(u,v)}{c(u,v)+1}\right]dudv.		
	\end{equation}	
\end{definition}

\begin{definition}[Copula Power Divergence]
	Given random variables $(X,Y)$ and their copula density function $c(u,v)$, copula power divergence is defined as
	\begin{equation}
		D_{Po}(c_{XY})=\frac{1}{\lambda(\lambda+1)}\int_u \int_v [c^\lambda(u,v) -1]c(u,v) dudv.		
	\end{equation}	
\end{definition}

\begin{definition}[Copula Bhattacharyya Distance]
	Given random variables $(X,Y)$ and their copula density function $c(u,v)$, copula Bhattacharyya distance is defined as
	\begin{equation}
		D_{Bh}(c_{XY})=\int_u \int_v \sqrt{c(u,v)} dudv.		
	\end{equation}	
\end{definition}

\begin{definition}[Copula Harmonic Distance]
	Given random variables $(X,Y)$ and their copula density function $c(u,v)$, copula harmonic distance is defined as
	\begin{equation}
		D_{Ha}(c_{XY})=\int_u \int_v \left[\frac{2c(u,v)}{c(u,v)+1} \right] dudv.		
	\end{equation}	
\end{definition}

\begin{definition}[Copula Triangular Discrimination]
	Given random variables $(X,Y)$ and their copula density function $c(u,v)$, copula triangular discrimination distance is defined as
	\begin{equation}
		D_{\Delta}(c_{XY})=\int_u \int_v \frac{[c(u,v)-1]^2}{1+c(u,v)} dudv.		
	\end{equation}	
\end{definition}

\begin{definition}[Copula Relative Information Generating Function]
	Given random variables $(X,Y)$ and their copula density function $c(u,v)$, copula relative information generating function is defined as
	\begin{equation}
		D_{Ri}(c_{XY})=\int_u \int_v c^t(u,v) dudv,		
	\end{equation}
	where $t\geq 1$.
\end{definition}

\chapter{Real Applications}
\label{chap:realapp}
\section{Theoretical physics}
Thermodynamics is an ancient branch of theoretical physics, established in the 19th century by Clausius, Boltzmann, Gibbs, and others. It studies the theoretical connection between the macroscopic state (such as temperature) and the microscopic state of physical systems. Entropy and the second law of thermodynamics are its core theoretical contents. Shannon's information theory was inspired by the concept of entropy in thermodynamics. The theoretical connection between thermodynamics and information theory has always been an important topic in related fields. CE is a mathematical concept proposed from the field of information theory, and its physical meaning and interpretation have remained largely unexplored. Ma \cite{ma2021thermodynamic} applied CE theory to the derivation and calculation of entropy in equilibrium correlated particle systems, providing a thermodynamic interpretation of CE and establishing another theoretical connection between thermodynamics and information theory.

\section{Astrophysics}
Redshift is one of the most important pieces of information about celestial bodies in the universe, reflecting their cosmic distance from Earth and useful for studying galactic evolution and cosmology. Photometric redshift is a method for estimating the redshift of celestial bodies from optical observations. Because optical observations are easier to perform than spectroscopic observations, photometric redshift is one of the main methods in modern astronomical surveys; generally, spectroscopic observations of celestial bodies of interest are performed after obtaining photometric redshift information. Machine learning methods have become one of the main methods for constructing photometric redshift prediction models, but their prediction accuracy still needs improvement. Ma \cite{Ma2023c} proposed using a variable selection method based on CE to construct such estimation models to improve the accuracy of prediction models. This method first estimates the CE between optical observations and redshift as a measure of the importance of the observed variables, and then uses the important observed variables as inputs to the model to predict the redshift. He applied the method to quasar observation data from the Sloan Sky Survey, and the results showed that the model obtained using CE selection was more accurate than the unselected model, especially for high redshift ($z>4$), where the prediction accuracy was significantly improved. \footnote{The code is available at \url{https://github.com/majianthu/quasar}} Simultaneously, the method also selected a set of interpretable optical observation variables, including luminosity magnitude, ultraviolet brightness and standard deviation, and brightness in four other bands, providing scientific evidence for further astrophysical research and the design of optical observation instruments.

\section{Space science}
The ionosphere is the region of the atmosphere from 50 km to 1000 km altitude composed of a large number of ions and free electrons. Its physical properties are influenced by solar radiation activity and also affect microwave propagation through it. Total Electron Content (TEC) is a key physical parameter of the ionosphere, significantly impacting space-based satellite communication systems and GPS. Doppler Frequency Shift (DFS) is a technique for detecting ionospheric variations, estimated from the energy spectrum of transmitted signals, and has been shown to reflect short-term ionospheric changes. Studying the complex relationship between TEC and DFS can deepen our understanding of ionospheric characteristics and is of great significance to the operation of related space systems. However, this relationship is nonlinearly complex, posing a challenge to research. TE, as a mathematical tool for studying nonlinear causal relationships, is an ideal method for investigating such relationships. Akerele et al. \cite{Akerele2024} used a set of analytical tools to study the nonlinear relationship between TEC and DFS in high-frequency microwave signal transmission in the equatorial region, utilizing CE based TE as a causal relationship analysis tool. They analyzed the aforementioned relationship using high-frequency Doppler system measurement data from Abuja and Lagos, Nigeria, in 2020-2021, and NASA's TEC data. The results showed that TE analysis revealed that changes in TEC lead to significant changes in DFS, and the strength of this causal relationship varies with seasonal ionospheric changes, while correlation analysis did not reveal this causal relationship. This conclusion has important reference value for the management of high-frequency microwave communication systems.

\section{Geology}
Lithofacies is a fundamental concept in geology, referring to rock types with similar characteristics formed in a specific sedimentary environment. Lithofacies classification is an important issue in geological surveys and exploration, and is of great significance for tasks such as lithofacies formation evaluation and reserve identification. With the increase in geological data, the use of machine learning methods for lithofacies classification has become an important topic and has received increasing attention. However, most existing studies directly use machine learning methods to construct classification models without selecting the appropriate models, leading to models that need improvement. Furthermore, since the classifiers used are mostly black-box models, they lack interpretability, making the resulting models difficult for geologists to understand. Ma \cite{Ma2025} proposed using CE to construct a lithofacies classification model. First, CE is used to calculate the correlation between geological variables and lithofacies, and then variables are selected based on the correlation strength to obtain the final classification model. He applied this method to well logging lithofacies classification data obtained from the widely studied Kansas Geological Survey in the United States, constructing a lithofacies classification model after selecting geological variables. Experimental results show that this method can select fewer geological variables as input to the classification model without reducing classification performance. In addition, this method selects variables closely related to lithofacies, such as marine facies labeling and well logging depth, as important variables, so that the resulting classification model conforms to geological knowledge and is interpretable.

\section{Geophysics}
Land aridity is an attribute of the dynamic interaction between land surface moisture and energy. Traditional aridity measures mostly use long-term mean values ​​of climatic conditional variables for calculation, which are difficult to reflect short-term surface moisture-energy interactions. Evapotranspiration is a key variable characterizing short-term surface water vapor-energy interactions, including the dissipation of moisture from the land and plant surfaces. Traditionally, it is classified into three conceptual frameworks based on soil moisture and energy supply: water-driven, energy-driven, and transitional. Studies have shown that the evapotranspiration-soil moisture relationship is also influenced by other factors, such as cloud cover, wind speed, and vegetation. Considering how these factors affect evapotranspiration provides the possibility for developing new land aridity classification frameworks. Shan et al. \cite{Shan2024,Shan2025master} proposed a new method to characterize land aridity by considering the short-term coupling effect of land and atmosphere. This method uses conditional mutual information based on CE to calculate the causal relationship strength between evapotranspiration and soil moisture and solar radiation, and then uses the difference between these two causal relationships to classify land aridity into six types, corresponding to three evapotranspiration conceptual frameworks. Based on hourly records of air temperature, dew point temperature, soil moisture, potential heat flow, sensitive heat flow, evapotranspiration, and surface solar radiation in mainland China during the summers of 1990-2020, this method was used to obtain a spatial distribution map of land aridity. The map was compared with the aridity index of the United Nations Environment Programme, revealing that the conditional mutual information distribution map calculated by this method is consistent with the geographical distribution of water and energy. The resulting aridity distribution can more accurately capture short-term surface processes, thus providing valuable supplementary information on short-term land-atmosphere interactions. This method deepens the understanding of climate aridity characteristics and provides a characterization tool sensitive to short-term climate changes such as extreme heat waves and sudden droughts.

Land-atmosphere coupling refers to the turbulent exchange process between the Earth's surface and the atmospheric boundary layer, leading to energy and material cycling at multiple spatiotemporal scales and influencing the occurrence of extreme weather events. Introducing land-atmosphere coupling can improve the predictive power of numerical weather prediction and climate models. Understanding the weakening of land-atmosphere coupling in the stable boundary layer during winter and representing boundary layer subgrid-scale processes in prediction models is a noteworthy issue. The decoupling phenomenon caused by weakened turbulence under low temperatures renders traditional similarity theories for turbulence parameterization in prediction models ineffective, leading to inaccurate near-surface temperature modeling. Two-meter temperature is the most commonly used forecast variable for the near-surface atmosphere, and establishing a two-meter temperature prediction model is considered a key step in the development of surface flux parameterization. Based on meteorological observations in the snow-covered mountainous region of Finse, Norway, Mack et al. \cite{Mack2025} proposed a two-meter temperature interpolation model based on a Copula Bayesian network to replace traditional models in numerical weather prediction systems. They first used CE to examine the relationship between CE and decoupling metric between meteorological observation variables at two Finse stations, and compared the CE values ​​and decoupling metric between actual observations and the traditional 2-meter temperature model for 2-meter and 10-meter temperatures, respectively. The analysis revealed that as the degree of vertical decoupling increases, the CE between the two stations increases, a phenomenon termed ``information decoupling". Furthermore, compared to observations, the CE between 2-meter and 10-meter temperatures in the traditional model increases more significantly with decreasing decoupling metric, indicating that the traditional model fails to fully utilize effective observational information. They then established a coupled surface and atmospheric temperature model using a combination of Vine copula and Copula Bayesian networks to predict near-surface temperatures, achieving superior performance compared to the traditional model and validating the rationale behind the CE analysis. The authors argued that using CE theory for coupling uncertainty quantification is a novel concept, and that calculating the information loss between the model and observations using CE can also serve as a new indicator for evaluating model prediction performance.

\section{Fluid mechanics}
Centrifugal pumps are widely used mechanical devices in various engineering applications, and predicting their performance is a crucial concern during design and use. Traditional methods for predicting centrifugal pump performance based on fluid dynamics numerical simulations require significant computation time, are relatively complex, and suffer from large prediction errors. Utilizing neural networks to predict centrifugal pump performance has become an important new approach, but such methods require large amounts of sample data, which is often limited in practical engineering. Chen et al. \cite{Chen2025} proposed a method combining SMOGN oversampling technology and GAKHO hyperparameter tuning technology for centrifugal pump performance prediction. SMOGN is used for preprocessing small sample resampling, while GAKHO is used to tune the hyperparameters of the neural network. In this method, tools such as CE were used to verify the effectiveness of the SMOGN resampling technology. They validated the method using data from 10 centrifugal pumps at different speeds, finding that it can accurately predict key performance parameters such as head and efficiency of centrifugal pumps even with small sample sizes.

\section{Thermology}
Micro heat pipes are highly efficient phase-change heat transfer devices, characterized by high thermal conductivity, small size, and high heat transfer power, and are widely used as heat dissipation units in modern electronic devices. Micro heat pipes come in various shapes and have diverse core components, including the liquid wick. The device's geometry and manufacturing process parameters are two major factors affecting its heat transfer performance. Therefore, studying the relationship between the geometric parameters and manufacturing process parameters of micro heat pipes and their heat transfer performance is an important issue, as it can guide the structural design and manufacturing process parameter setting of micro heat pipes. Traditional design and manufacturing methods rely on human experience and repeated trials, which are costly and inefficient. Utilizing machine learning models to establish the relationship between design requirements and geometric and process parameters offers a new solution. Li et al. \cite{Li2025sst,Zhang2024master2} proposed a method for predicting micro heat pipe design and manufacturing parameters based on a combination of CE and neural network models. CE measures the correlation strength between micro heat pipe structural design requirements parameters and their heat transfer power, selecting key design parameters as model inputs. A one-dimensional convolutional residual network is then used to establish a prediction model from the selected key design requirements parameters to structural and process parameters. They validated the method based on test data from actual manufactured micro heat pipe samples. The results show that the key design parameters related to heat transfer power selected by CE are verified by actual experimental results. The prediction model based on CE-selected parameters closely matches the actual values ​​on the test design requirements data of actual manufacturing, with a maximum error of only 2.6\%.

Circulating fluidized bed boilers (CFBBs) are clean coal combustion technologies employing fluidized bed combustion. They offer advantages such as adaptability to various fuels and good load regulation capabilities, and are widely used in industrial sectors such as thermal power plants. However, further improving the thermal efficiency of CFBBs and reducing emissions are key industry concerns, as these two objectives are conflicting and cannot be achieved simultaneously. Ma et al.\cite{Ma2025b} categorized this problem as a multi-objective parameter optimization problem and proposed a parameter optimization method combining CE, a weighted fusion tree model, and a multi-objective tree structure estimator. CE is used to analyze the functional relationship between combustion process parameters and optimization objectives (NOx emissions and thermal efficiency). Based on this analysis, a weighted fusion tree model is used to construct an optimization objective prediction model. Finally, the multi-objective tree structure estimator is used to optimize and adjust the input parameters of the prediction model to regulate the optimization objectives. They applied this method to a 330MW CFBB. CE analysis revealed that pollutant emissions and thermal efficiency are mainly related to the supply of coal and air. The final optimization results also showed that reducing coal supply while increasing air supply can balance and regulate boiler thermal efficiency and NOx emission concentration. This technique can be easily extended to other types of industrial boilers, thus providing a consistent solution for combustion control in industrial thermal systems.

\section{Theoretical chemistry}
Allostery, considered the ``second secret of life", is a pervasive biological phenomenon present in almost all proteins. It refers to the regulatory effect where allosteric regulatory molecules bind to proteins, inducing changes at points beyond the binding site. The most common allosteric system model is the two-state allosteric model, which describes the thermodynamic cycle of the allosteric process. This model assumes that receptor activation is a two-state process, which is inconsistent with the multimodal processes revealed by NMR experiments. A deeper understanding of the molecular mechanisms of ligand-induced receptor activation requires constructing new theories to understand the thermodynamic coupling between the ligand binding site and the activation site. Cuendet et al. \cite{Cuendet2016} proposed a new theory called the Allostery Landscape, which defines a thermodynamic coupling function to quantify thermodynamic coupling in biomolecular systems. They pointed out that the new function is closely related to the copula density function and CE, which defines the information transfer property of allosteric systems, i.e., information transfer between the ligand binding site and the activation site. They applied the new theory to the thermodynamic coupling analysis of the N-terminus and C-terminus of the alanine dipeptide.

\section{Cheminformatics}
Cheminformatics is an interdisciplinary field combining chemistry and information science. It uses the characterization of chemical structures as data to solve problems such as molecular design, chemical reaction simulation, and planning. Quantitative structure-activity relationship (QSAR) is a cutting-edge problem in this field, studying the quantitative relationship between molecular structure and its physicochemical properties to guide the design of molecules with specific characteristics. Its applications are wide-ranging. Molecular physicochemical properties can be understood as invariance to certain symmetry transformations of molecular structure, and learning this invariant transformation from data is a key objective of molecular design. Wieser et al. \cite{wieser2020inverse} transformed the symmetry transformation learning problem into an information bottleneck problem, proposing a Symmetry-Transformation Information Bottleneck (STIB) method. This method represents molecular characterization as an implicit representation consisting of two parts, one of which corresponds to the invariant representation. Based on the transformation invariance of MI (CE), a learning algorithm for the problem model is designed. The authors applied the algorithm to the QM9 database\cite{ramakrishnan2014quantum} containing 134,000 organic molecules, using a subset of 6095 molecules with fixed stoichiometry ($\rm C_7 O_2 H_{10}$), and their corresponding band gap energy and polarity as the target invariant properties. Experimental results show that the STIB method provides a symmetric transformation that can learn the invariance of molecular properties, band gap energy, and polarity, validating the effectiveness of the method.

\section{Material science}
Heat-resistant energetic materials refer to special materials with high energy and high thermal stability, which can maintain stable properties in high-temperature environments. Therefore, they are key materials in important fields such as national defense, aerospace, and geological exploration, such as propellant for aerospace and hypersonic weapons, and explosives for deep well drilling. However, such materials are scarce and experimental research is extremely dangerous, so designing such materials has been a challenging problem that materials scientists have been striving to overcome. "Designing" energetic materials "from scratch" requires a process of "design-screening-evaluation", among which using machine learning methods to build material structure-property prediction models to predict the molecular properties of the design is a key step in material molecular screening\cite{Liu2025b}. Traditional energetic molecular property prediction model construction processes only use molecular features that are linearly related to thermal stability, without considering factors that have a nonlinear relationship with the thermal decomposition temperature of energetic materials, such as crystal structure and packing mode. Tian et al.\cite{Tian2023b,Liu2025} proposed a feature selection method combining PCC and CE to select features correlated with thermal decomposition temperature from molecular topological structure and quantum chemical computational features, and then construct a prediction model. The CE method is introduced to screen features with a non-linear relationship to thermal decomposition temperature. They collected 460 energetic compounds and generated a dataset containing 286 features. Using this method, they selected 87 features, which were then used as input to models such as random forest and SVM to predict the thermal decomposition temperature of the compounds. This yielded better prediction results than traditional methods, with the prediction error controlled at 28.5°C in cross-validation experiments. They applied the method to molecules generated by their designed molecular generator, ultimately selecting 16 energetic molecules with good thermal stability potential and strong detonation ability, validating the practical value of the method.

Near \textbeta \ titanium alloys possess characteristics such as low density, high strength, high toughness, and corrosion resistance, making them one of the key materials for manufacturing aerospace components. The properties of alloy components are determined by their microstructure and crystalline material, and the material properties are highly correlated with the deformation mode during processing. Studying the influence of deformation processing on crystalline material is one of the key issues in titanium alloy research. Computational simulation technology is a new method for material processing research besides physical experiments, such as numerical methods for crystal plasticity. Zhang et al. \cite{Zhang2024a} studied the relationship between the deformation mechanism and material evolution of near \textbeta\ titanium alloys in the \textbeta\ region. They experimentally simulated the evolution of eight major titanium alloy components under three sliding deformation modes in the \textbeta\ region and used CE to calculate the nonlinear correlation strength between each deformation mode and the content of each component. The results showed that the three modes were highly correlated with all other components except for the \{001\}<100> component, consistent with previous research findings; \textgamma\ fiber and \{110\}<001> material showed lower correlation; meanwhile, the \{112\} and \{123\} modes were more strongly correlated than the \{110\} mode, indicating that the two play a decisive role in the high-temperature deformation process.

Refractory multi-principal element alloys (RMPEAs) are a class of near-equimolar mixtures of multiple refractory metallic elements. Due to their exceptional high-temperature strength and resistance to high-temperature softening, they are considered highly promising candidate materials for high-temperature applications. Realizing the potential of RMPEAAs in specific applications requires accurate prediction of their physical properties, but their vast combination space presents a challenge. Applying traditional density functional theory (DFT) methods to RMPEAAs requires significant computational resources, increasing the difficulty of material selection. The EMTO (Exact Muffin-Tin Orbitals) method significantly improves computational efficiency. Building machine learning models based on traditional DFT methods to predict alloy properties has become a widely adopted approach, accelerating the material design process. However, analyzing the nonlinear relationship between alloy characteristics and physical properties, and thus determining the appropriate machine learning model, is a crucial issue. Jin et al.\cite{Jin2025} proposed combining EMTO first-principles calculations with a feature selection method based on CE to construct a predictive model for the physical properties of RMPEAAs. They verified the reliability of this method using a V-Nb-Ta system, and the machine learning model predictions showed high consistency with experimental results. They selected eight features as model inputs based on the CE method. The relationship between the melting point and the $\Omega$ criterion obtained by the CE method, as well as the relationship between these two and the V, Nb and Ta contents, are consistent with the theoretical relationship, which shows the rationality of the CE method.

\section{Hydrology}
Floods are one of the major natural disasters, and flood forecasting is an important means of reducing flood losses and managing flood resources. Precipitation-runoff models based on precipitation data can be used to forecast floods a certain period in advance. However, the complexity and nonlinearity of water systems make it very difficult to select the correct model input when building such models. Chen et al.\cite{chen2013measure,chen2014copula,Chen2021} proposed using the CE method to select inputs and build a neural network forecasting model. Compared with traditional methods, the CE based method can build high-dimensional models without making assumptions about the marginal distribution of individual variables, and the error in estimating the quantitative relationship between precipitation and runoff using CE is smaller. Chen et al. applied the method to build a flood forecasting model for the Jinsha River basin, and the results showed that the neural network model using CE to select inputs achieved the best prediction performance. Li et al.\cite{Li2022a} studied the monthly runoff forecasting problem in the upper reaches of the Yangtze River based on CE and machine learning methods. They used 130 global circulation indices, 7 meteorological factors, and monthly runoff data from two hydrological stations, Gaochang and Cuntan, to construct a prediction model by combining three variable selection methods (CE, etc.) and five machine learning models. The results showed that the combination of CE and LSTM achieved the best prediction performance at the Gaochang station, while the combination of random forest and CE achieved satisfactory performance at the Cuntan station. Mo et al. \cite{Mo2023} proposed a long-term runoff forecasting model framework that combines CE, LSTM, and GARCH methods, where CE is used to screen runoff-related forecasting factors. Compared with traditional methods, CE is more suitable for complex situations with interactive correlations between factors. They applied the method to runoff forecasting studies in Hongze Lake and Luoma Lake. The results showed that, compared with traditional methods, the CE method in this framework not only successfully identified the interactive effects between factors but also quantified the contribution of each factor in each forecast period, thus selecting the key driving factors related to forecasting. Ultimately, this framework yielded more accurate, stable, and reliable forecasting results than the comparative methods. Chen et al. \cite{Chen2023b} proposed a runoff forecasting method using spatiotemporal graph convolutional networks. First, they constructed a topological structure map of stations within the basin. Then, they used an adjacency matrix to represent the spatiotemporal dependence between geographically adjacent stations. They then used tools such as CE to analyze the spatiotemporal correlation between adjacency relationships, periodicity, and meteorological elements and runoff volume. Finally, they constructed a corresponding graph convolutional network with an attention mechanism as the runoff forecasting model. They verified the effectiveness of the method using the Jinsha River basin as an example. Wang \cite{Wang2018master} proposed a flood forecasting method based on a combination of CE and wavelet neural networks. He used CE to select the input to the forecasting model and then used a wavelet neural network to construct the forecasting model. He applied this method to the prediction of daily water levels at Dingjiaba Station in Shengze Town, Suzhou City, Jiangsu Province. Experimental results showed that the three forecasting performance indicators of this method were all higher than those of the comparison methods, reaching the level of actual formal forecast releases, providing technical support for the development of regional flood control systems.

Drought is another important type of hydrological event and a major natural disaster. Frequent droughts seriously affect China's economic and social security, especially in the Yellow River Basin, where the threat of drought is particularly severe, urgently requiring research on drought drivers and prediction. Wen et al.\cite{wen2019,Wen2021master} used the CE theory to analyze monthly meteorological data of Henan Province from 1951 to 2014, finding that among many driving factors, precipitation, temperature, water vapor pressure, and relative humidity have the greatest impact on the occurrence of drought in the region. Huang and Zhang\cite{Huang2019} used the CE method to analyze meteorological data of Lanzhou from 1957 to 2010 to construct a drought prediction model for the region, finding that wind speed, temperature, water vapor pressure, and relative humidity are the meteorological factors most relevant to drought in the region. Huang \cite{huang2021phd} studied the relationship between meteorology, hydrology, and drought in the Yellow River Basin, explored the driving mechanism of drought, introduced the concepts of meteorological drought and hydrological drought, and proposed using the CE method to explore the dynamic nonlinear response relationship between the two. By analyzing the meteorological and hydrological drought indices of hydrological stations in different regions of the Yellow River Basin, the lag effect time of hydrological drought on meteorological drought was obtained, providing a reference for responding to drought events. Niu \cite{Niu2023} used tools such as CE to study the spatiotemporal characteristics of drought propagation in nine sub-regions of the Yellow River Basin. Based on meteorological, soil moisture, and runoff data of each sub-region from 1961 to 2020, he used CE to calculate the nonlinear correlation between different types of non-stationary drought indices, and then obtained indicators such as drought response time scale, drought propagation intensity, and drought propagation rate. Finally, he discovered the strength and weakness characteristics of propagation sensitivity and propagation intensity among meteorological drought, agricultural drought, and hydrological drought in each sub-region. Ni et al.\cite{ni2020vine} proposed a vine-copula structure selection method based on MI, utilizing the equivalence relationship between MI and CE, and applied it to the modeling of characteristic variables and the modeling of flow correlation structures at multiple hydrological stations in drought identification in the Yellow River Basin. Kanthavel et al.\cite{Kanthavel2022} proposed a comprehensive drought index using theoretical tools such as CE and vine copula, integrating four indices: standardized precipitation index, drought monitoring index, standardized soil moisture index, and standardized runoff drought index, which can better reflect relevant hydrological and meteorological variables and different types of drought simultaneously. CE theory was used to measure the correlation between the new index and the original index. They applied this index to a monthly and four-month drought study in the Tapti River Basin in central India, validating its effectiveness and revealing the spatiotemporal distribution characteristics of drought in the region. Mohammadi et al.\cite{Mohammadi2021} used three correlation measurement estimation methods based on Copula and CE theories to analyze the dependencies among drought variables (drought intensity, duration, and time interval) in three Iranian cities (Zahedan, Enzeli, and Mashhad) from 1950 to 2017. Xu \cite{Xu2025master} proposed using the TE method based on CE to analyze the direction and time delay of drought propagation. He used this method to analyze actual observation data from meteorological and hydrological stations in the Colorado River Basin, USA, from 1972 to 2019. First, he constructed two non-stationary drought indices (non-stationary standardized precipitation index and non-stationary standardized runoff index). Based on the river system structure, he divided the basin into seven sub-regions. Then, he used GC to test the causal relationship between meteorological and hydrological drought propagation in the same sub-region and the causal relationship between the same drought type propagation in different sub-regions. Finally, he used TE to estimate the corresponding drought propagation time, successfully discovering the characteristics of the direction and time delay of drought propagation within the basin. Liu \cite{Liu2024master} applied the CE method to study the formation and evolution mechanism of drought in the Nemur River Basin of Heilongjiang Province. Based on hydrological and meteorological data of the river basin from June 2003 to December 2014, he first constructed three drought indices: hydrological drought, meteorological drought, and agricultural drought. Then, he used CE and other methods to analyze the driving factors of meteorological and hydrological drought, constructed dynamic response VAR models for the two types of drought, and further analyzed the drought propagation characteristics of the Nemur River Basin with the CE based TE method. He successfully identified the causal relationship direction and time delay between different drought types in the three sub-basins of the river, providing an important reference for drought response in the Songnen Plain downstream of the river basin.

River ice is a unique hydrological phenomenon in mid-to-high latitude rivers, exhibiting distinct seasonal characteristics. River ice phenology studies the seasonal variations in the freezing and melting of river ice, which exhibit spatial heterogeneity and annual variability. River ice phenology is influenced by multiple factors, such as climatic conditions, hydrological conditions, and human activities, with climate change being the primary driving factor. Against the backdrop of global climate change, river ice phenomena have undergone significant changes in the past period. Understanding and quantifying the dynamics and intrinsic driving forces of river ice is crucial for formulating river management strategies. Xing et al. \cite{Xing2024} studied the response of temporal variations in river ice phenology in the Heilongjiang River to climate change, using the CE tool to calculate the nonlinear correlation strength between five key dates in the river ice cycle (new ice day, ice-sealing day, freezing day, thawing day, and ice-melting day) and climatic factors. The results showed that temperature, air pressure, and evaporation pressure are the most important climatic driving factors for river ice phenology. Specifically, the dates of new glaciation and ice cover are correlated with the evaporation pressure of the preceding one or two months, while the correlation between the freezing date and evaporation pressure has a lag of more than four months. Furthermore, these three dates are strongly correlated with the mean maximum temperature, mean temperature, and minimum extreme temperature of the same month, indicating that the freezing date responds to temperature more quickly than other climatic factors. Meanwhile, the thawing date is correlated with the two-month average temperature range, minimum extreme temperature, and maximum extreme temperature, while the ice melt date is correlated with the five-month average temperature range and monthly relative humidity, exhibiting significant lag characteristics. The authors also studied the correlation of 15-year moving averages based on CE, discovering temporal variations in air pressure and temperature with river ice dates. These findings help to incorporate these identified factors into river models, establishing more accurate ice age prediction models, thereby enabling river management to better adapt to the impacts of climate change.

Hydrological and meteorological observation networks are the infrastructure for acquiring hydrological information. Designing and optimizing these networks is a comprehensive scientific and engineering problem. A fundamental design principle is to ensure statistical independence among observation stations to maximize the acquisition of information about the hydrological system. MI is a primary tool for measuring statistical independence, but its computation is challenging. Xu et al.\cite{xu2017a,Xu2018phd,Wang2019,Wang2022book} proposed a multi-objective optimization CE based hydrological observation network design method, comprising two steps: 1) grouping observation stations based on CE; and 2) selecting the optimal station combination for each group. The CE based computational method not only handles the non-Gaussianity of hydrological variables but also demonstrates greater reliability and efficiency in computational performance. The authors applied the method to the design of the I-Luo River hydrological observation network in the Yellow River Basin and the Shanghai rainfall observation network. The results show that the CE method achieves higher computational accuracy and can be applied to high-dimensional multivariate estimation scenarios. Based on the principle of least overlapping information, Li et al.\cite{li2020developing,li2021developing} proposed a network optimization objective consisting of two sub-objectives, one of which is designed based on CE to measure redundant information. The authors applied this method to the design and optimization of the Fen River runoff observation network, the Beijing urban area, and the Taihu Basin rainfall observation network, demonstrating its reliability and effectiveness. Xu et al.\cite{Xu2018phd,Xu2022,Xu2022a,Wang2022book} proposed using the vine copula to construct a station relationship network, then calculating the CE value between stations based on the estimated vine copula. Based on this, they proposed a station optimization objective combining CE and the Kriging index, using a sliding window method to select the optimal stations. Based on daily precipitation observation data from the Huai River Basin from 1992 to 2018, they used this method to optimize the network of 43 rainfall observation stations in the basin. The results showed that the network obtained by this method can more effectively acquire precipitation-related information than networks obtained by traditional similar methods. They also applied this method to the optimization design of the rain gauge network in Shanghai, resulting in a network with low redundancy, high total information volume, and minimal internal estimation error. Yang \cite{Yang2019master} proposed a network optimization criterion combining joint entropy ratio, redundancy ratio, and NSE efficiency coefficient, and based on CE theory, proposed a new MI calculation method, improving the accuracy of the calculation. He applied the method to 14 hydrological stations in the Choctohatch River basin in the United States for station optimization research, ultimately obtaining a network containing only 5 stations, improving the monitoring efficiency of the network.

Analyzing the correlation between the main stream and tributaries of rivers is crucial for hydraulic engineering design, flood prevention, and risk control. As a large-scale hydraulic engineering project in the upper reaches of the Yangtze River, the Three Gorges Dam plays a vital role in flood control. Therefore, studying the correlation between the main rivers in this section is of significant reference value for engineering design and safe operation. Chen and Guo\cite{chen2019copulas,Chen2013book} proposed using CE to calculate the strength of river correlations. They applied this method to the upper reaches of the Yangtze River, which includes five main streams and tributaries, calculating the correlation between rivers based on flood records from 1951 to 2007. They found that the overall correlation between rivers is not high, consistent with the climatic characteristics of the region. The strongest correlation is between the Min and Tuo rivers, due to their proximity and shared precipitation area. The Jinsha River shows some correlation with both the Min River and the Tuo River, posing a certain threat to the flood control function of the Three Gorges Dam. The Jinsha River, Jialing River, Min River, and Tuo River all have a significant impact on flood occurrence in the Yangtze River Basin.

Superimposed flood events in different rivers and regions easily form complex flood events, but the spatial relationships between different flood processes are difficult to accurately describe and assess using existing correlation analysis methods. Wang and Shen\cite{Wang2023a} proposed a framework integrating the vine copula and correlation assessment, which utilizes CE theory to estimate correlation strengths such as MI, CMI, and R statistics from the vine copula. They applied the method to assess the relationship between two extreme runoff series variables (peak flow and flood flow) in 102 identified complex flood events in the upper Yangtze River. The results show that the multidimensional R-vine copula model of this framework can better depict the complex and diverse hydrological correlations, especially the vine structure, which represents the sequence of tributary floods entering the main stream and the spatial relationships between hydrological stations; the three correlation strengths estimated by this framework better reflect the nonlinear relationships in complex spatiotemporal hydrological systems during complex flood events than traditional correlation strengths.

Water and sediment regulation in the Yellow River is crucial to the formulation of Yellow River management strategies. Scientifically understanding and assessing the characteristics of water and sediment flux changes in the Yellow River is a fundamental scientific issue and is of great significance for judging the sediment situation in the Yellow River. Especially in recent decades, due to the combined effects of climate change and human activities, the water and sediment content of the Yellow River has changed significantly, requiring accurate estimation of the distribution changes in runoff and sediment transport. The copula function is a basic mathematical tool for analyzing such distributions; however, such problems often involve small observational samples, making it difficult to accurately estimate the parameters of the copula function. Qian et al. \cite{Qian2022} proposed a copula parameter estimation method based on CE and total correlation to solve the copula parameter estimation problem under conditions of limited sample size. They applied the method to the analysis of runoff and sediment transport data from 1960 to 2016 in the Xiliugou River Basin of the Yellow River. The water-sediment relationship in this basin changed significantly around 1999, but data was scarce. The analysis results show that for both periods around 1999, the new method yields more accurate copula parameter estimates and a better fit to the data than the two traditional methods.

River scouring refers to the continuous erosion of riverbed sediment by flowing water, resulting in a lowering of the riverbed or changes in river channel width. Contraction scouring is a special type of river scouring, caused by a reduction in the area of ​​the flowing water. It typically occurs when artificial bridge foundations or other structures restrict water flow, or when natural factors cause the river channel to narrow, posing a safety risk to bridges and other structures. Accurate estimation of contraction scouring depth is crucial for risk assessment and structural design of bridges and other structures. Traditional empirical formulas often have low prediction accuracy; therefore, using machine learning methods to build scouring depth prediction models is a novel approach. Wang et al.\cite{Wang2024a} proposed a Principal Component Analysis (PCA)-enhanced Support Vector Regression (SVR) method for constructing scouring depth prediction models. Based on experimental data of clear water scouring in converging channels, they first employed five dimensionality reduction methods (including PCA, tSNE, NMF, LDA, and KPCA) to transform the original variables related to scouring, such as fluid parameters, sedimentary characteristics, and geometric features, to obtain new input variables. Then, they used CE to select new variables strongly correlated with scouring depth. The results showed that the PCA method yielded the new principal component vector with the strongest correlation to scouring depth; therefore, PCA was adopted as the method for generating model input variables. Subsequent SVR model prediction experiments showed that the prediction model using the PCA method's principal component variables as input achieved higher prediction accuracy than traditional methods ($R^2$ = 0.971, MAPE = 7.54\%), validating the superiority of this method.

Watershed partitioning is an important method in hydrological research. Dividing watersheds into similar regions based on hydrological similarity characteristics can solve difficult problems such as hydrological calculations in areas without hydrological observations. Runoff response is an important hydrological characteristic of watersheds, and partitioning watersheds based on the similarity between observations from hydrological stations is a fundamental research approach. Traditional watershed partitioning methods, based on correlation evaluation, often fail to reflect the complex relationships inherent in the hydrological system. Liu et al.\cite{liu2022water} proposed using the R-statistic based on CE to measure runoff similarity between nodes, and then using a community detection algorithm to partition the watershed. They applied this method to the Poyang Lake system, partitioning the watershed using hydrological station observations, and compared their method with the traditional K-means clustering method. The results show that this method can effectively capture the regulating effect of lakes and reservoirs on runoff within the watershed, thus obtaining a more reasonable watershed partitioning than traditional methods.

Multi-site streamflow generation is one of the major problems in stochastic hydrology, and the generated flow information is essential for any water resource management. Given the limited number of streamflow records, generating multi-site runoff data is crucial, requiring the design of appropriate data generation models. Porto et al. \cite{porto2021a,Porto2023} proposed a multi-site annual runoff generation model combining a generalized linear model (GLM) and a copula function. The former represents the time-series structure, while the latter models the spatial correlation among multiple stations. In evaluating model performance, the authors used several statistical descriptive indices, including CE, which measures nonlinear full correlation. The authors applied this model to generate multi-site streamflow time-series data for the Jaguaribe-metropolitan reservoir system in Brazil. The results show that the model outperforms the current state-of-the-art in terms of performance, particularly in the CE index, which measures multi-site correlation and is closer to historical observation data than other models.

The South-to-North Water Diversion Project is the world's largest water conservancy project, undertaking the strategic task of diverting water from the Danjiangkou Reservoir in the Han River basin of the Yangtze River to cities in northern China. Accurate inflow forecasting is a prerequisite for scientific and rational water supply scheduling, enabling the project to utilize natural water resources more fully and efficiently. However, forecasting models constructed using traditional methods often fail to meet the accuracy requirements of water diversion forecasts. This is because traditional analytical methods cannot handle the nonlinear characteristics of hydrological systems, leading to unreasonable inflow forecasting models and thus poor predictive performance. Huang et al.\cite{huang2021} constructed a monthly inflow forecasting model for the Danjiangkou Reservoir, using CE to select a set of meteorological and hydrological factors as model inputs. The resulting model showed significantly better forecasting performance than traditional models. The model's success lies in the fact that the forecasting factors selected using CE are closely related to medium- and long-term inflow runoff, confirming the intrinsic connection between the Indian Ocean Dipole event, the activity of the South China Sea subtropical high, and summer heavy rainfall in the Han River basin, which aligns with the operational laws of natural hydrological systems.

Climate change and human activities directly affect the hydrological system cycle, causing varying degrees of spatiotemporal changes in hydrological factors such as runoff, precipitation, and evaporation. Therefore, studying the spatial relationships between hydrological factors such as precipitation and runoff, and further analyzing the underlying causes of these spatiotemporal changes—climate change and human activities—is an important topic in hydrology, attracting attention from scholars worldwide, and providing scientific reference value for economic and social activities such as water resources planning and management. Jiang \cite{Jiang2023} used tools such as CE to analyze watershed raster data of precipitation, evaporation, potential evapotranspiration, runoff, and the vegetation index NDVI in the Yangtze River Basin. From the obtained spatial correlations, he discovered the spatial distribution characteristics of these factors and provided a qualitative geographical explanation. In particular, based on CE estimates, he found that actual evapotranspiration and precipitation have a significant impact on annual runoff, and the spatial correlation between annual runoff and the above factors exhibits spatial heterogeneity.

\section{Climatology}
Climate change is a key topic in climatology research, manifesting not only in changes in the magnitude of hydroclimate variables but also in the distribution of their seasonal and periodic variations. These changes impact the intensity and frequency of precipitation and temperature, leading to an increase in extreme weather events such as floods, droughts, and heat waves. The correlation between precipitation and temperature exacerbates the occurrence and intensity of combined extreme weather events. Studying the impact of climate change on the correlation structure between precipitation and temperature is crucial. Hao and Singh \cite{Hao2015} used the CE metric to investigate the impact of climate change on this correlation structure. Their study used daily precipitation and temperature data from Fort Worth, Dallas, Texas, from 1948 to 2010, calculating the negative CE value between temperature and precipitation every five years as the correlation structure strength. They found that the correlation structure strength (negative CE value) between temperature and precipitation in the region decreased from 0.18 from 1948-1980 to 0.06 from 1948-2005, demonstrating that climate change has impacted the relationships between hydroclimate variables in the region.

Climate assessment is a fundamental task in scientifically addressing climate change. Its goal is to monitor and analyze global and regional climate and its changes, with a particular focus on trends and extreme climate risks. Climate classification refers to categorizing regions based on similar climate characteristics. The most common K\"oppen classification uses temperature patterns and seasonal precipitation as climate characteristics. Condino \cite{Condino2009} proposed a dynamic classification algorithm based on the Jensen-Shannon divergence, where the JS divergence-based classification criterion is estimated using a representation method based on CE theory. He applied the method to the European climate assessment problem, classifying the climate of 25 major European cities based on daily temperature and precipitation data from European meteorological stations from 1951 to 2008. The results show that the proposed algorithm successfully distinguishes urban clusters belonging to the southern and northern climate zones of Europe. When the transition zone between the north and south climate zones is further considered, the algorithm also provides reasonable classification results for cities in central Europe that are consistent with actual climate conditions.

\section{Meteorology}
Environmental pollution is one of the major problems in modern society. Analyzing the causes of air pollution from a meteorological perspective and clarifying its underlying mechanisms helps to better understand pollution problems and thus predict, intervene in, and manage pollution. Understanding the causal relationships in the atmospheric system is key to this problem. Based on observations of meteorological factors and environmental pollutants, the TE method in statistics can be used to analyze the causal relationship between meteorological factors and environmental pollution. Ma \cite{jian2019estimating} used his proposed CE based TE estimation method to analyze continuous meteorological and PM2.5 observation data in Beijing, obtaining a causal intensity variation diagram of the four meteorological factors on PM2.5 concentration over a 24-hour time lag. The variation diagram shows that the causal intensity of the four meteorological factors on PM2.5 concentration roughly goes through two stages: a rapid increase and a slow increase. The authors also specifically discuss and verify the applicability of the stationarity assumption and the Markovianity assumption of this method to this mesoscale numerical analysis problem. The causal variation diagram obtained in this paper reflects the inherent dynamic characteristics of atmospheric system motion, enhancing our understanding of the meteorological causes of PM2.5 pollution. Furthermore, the temporal causal relationships obtained provide a reference for integrating meteorological factors to construct better-performing pollution forecasting models.

Effective air pollution forecasting plays a fundamental role in pollution control and is beneficial for protecting public health. However, current air pollution (such as PM2.5 concentration) forecasts still fall short of requirements in terms of accuracy and stability. Developing higher-performance forecasting models has attracted widespread attention. Considering the shortcomings of traditional methods, Wang et al. \cite{Wang2022} proposed a new air pollution forecasting and early warning method, using a combination of CE and multiple machine learning models. The CE method was used to select factors influencing PM2.5 concentration fluctuations for constructing the final model. They applied their method to actual air pollution forecasting and early warning systems in Shanghai and Guangzhou, demonstrating that the new method achieves better forecast accuracy and stability than other comparative methods. Wu et al.\cite{Wu2022b} proposed a CE based PM2.5 forecasting method, utilizing CE to calculate the correlation between meteorological factors and air pollutant concentrations to select model input features, and establishing a forecasting model based on a combination of LSTM and evolutionary algorithms. This method achieved good forecasting performance on historical data from Beijing in 2016. Chen et al.\cite{Chen2023a} used CE to select factors influencing PM2.5 from multiple factors, and then used a Temporal Convolutional Network (TCNA) with self-attention mechanism enhancement to construct a model for predicting PM2.5 concentration. He applied this method to hourly meteorological and pollution observation data from 12 areas in Beijing from 2013 to 2017, and the resulting prediction model showed high interpretability and prediction accuracy. Guo et al.\cite{Guo2023} proposed a PM2.5 concentration prediction method that combines empirical decomposition model combination with neural networks. This method uses CE and other methods to select empirical model decomposition factors obtained from different methods. He validated this method on hourly PM2.5 observation data from a Beijing observation station from 2020 to 2022, demonstrating its effectiveness and superiority.

Global warming has led to increasingly stronger typhoons in South China, causing severe damage to the region. Predicting the severity of typhoon disasters based on observational data is crucial for assessment and response. However, the numerous influencing factors of typhoon disasters and their non-linear relationships with disaster severity make predictive model construction challenging. Chen et al. \cite{chen2019} proposed a method for constructing typhoon disaster prediction models based on tools such as CE. Using data from 44 typhoons that made landfall or affected Guangxi between 1985 and 2014, along with concurrent disaster statistics related to disaster initiation, impact, and prevention/mitigation, they constructed 21 disaster influencing factors. CE was then used to screen the factors most correlated with the disaster index, revealing that maximum wind speed, minimum air pressure, duration of heavy rainfall, and extreme rainfall values ​​were most correlated with the disaster index, objectively reflecting the actual situation. Experiments also showed that the model constructed using the factors screened by CE had higher prediction accuracy than models constructed using similar comparative methods, providing a reference for typhoon disaster prediction in Guangxi.

\section{Environmental science}
Air pollution is one of the major environmental problems facing modern cities, seriously affecting urban operations and residents' lives. Analyzing the diffusion patterns of air pollution is a crucial issue in environmental science, providing fundamental guidance for environmental regulatory departments to better understand pollution patterns and effectively respond. Data generated from numerous urban meteorological observation networks helps analyze diffusion patterns and predict pollution spread. Wu \cite{Wu2022d} proposed a node-feature-free network link prediction algorithm and applied it to the problem of modeling and predicting urban air pollution propagation paths. He applied the method to PM2.5 observation data from environmental monitoring stations in Lanzhou in 2017, constructed a propagation network using TE based on CE, and then applied the proposed network link prediction algorithm to predict pollution propagation paths. Experimental results show that this method can accurately identify pollution propagation paths, providing theoretical support for the formulation of pollution control strategies in Lanzhou.

Nitrogen oxides (NOx) are one of the main pollutants emitted by thermal power plants, requiring strict control of their emission concentrations through monitoring. Power plants typically use neutralization methods in SCR (Selective Catalytic Reduction) reactors to control NOx emission concentrations, but this method suffers from significant delays and cannot achieve precise control. Therefore, soft sensor models are generally used in conjunction with SCR controllers to achieve control targets. Jin et al.\cite{Jin2023} proposed a soft sensor algorithm framework combining VMD-Bayes-Lasso to predict NOx emission concentrations. This framework first uses CE to screen system variables related to NOx concentrations to predict the decomposed NOx concentration modal variables, then superimposes these variables to obtain the final prediction result. Finally, a model error prediction model based on the Lasso algorithm is designed to correct the prediction results. They validated this algorithm framework on data from a 660MW coal-fired power plant in Ningxia, achieving better prediction accuracy than comparative methods. Specifically, they analyzed the correlation between system variables and target variables using the CE method, achieving the goal of simplifying the model and improving prediction accuracy. Yang \cite{Yang2024master} proposed a method for predicting NOx emission concentrations from power plant boilers. First, the system delays under different operating conditions are identified. Then, an adaptive feature selection method based on CE is used to select appropriate system variables. Finally, an ensemble learning algorithm is used to construct a prediction model. The authors validated the method on a 1000MW supercritical boiler system in a power plant. The results show that the proposed method can predict NOx emission concentrations more accurately than comparative methods. Specifically, the CE based feature selection method can effectively filter out the set of system variables that significantly affect NOx emission concentrations; moreover, compared with similar methods, this feature selection method provides better prediction results.

$SO_2$ is another major pollutant emitted by coal-fired power plants, posing a serious threat to human production and living environments, necessitating its monitoring and control. Limestone-gypsum desulfurization systems are widely used wet desulfurization processes in coal-fired power plants; however, due to data acquisition issues and system latency, the $SO_2$ emission concentration at the system outlet is difficult to measure accurately, making real-time control of the emission concentration impossible by adjusting system operating parameters. Qiao \cite{Qiao2024master} proposed a soft-sensing technique for $SO_2$ prediction to precisely control system emission concentrations. This technique first uses feature selection, CE, and a prediction model to determine auxiliary variables and their time delays and orders, with CE used to determine the time delay of key variables. Then, a regression model is established using a CatBoost-Bayes error compensation model. He validated the method based on historical operating data from a 550MW unit in Shanxi Province, showing that the model established by this method has good soft-sensing accuracy. Simultaneously, comparative experiments also show that optimizing the time delay and order of auxiliary variables using methods such as CE significantly improves model performance, verifying the rationality and effectiveness of the method.

Ammonia ($NH_3$) is an important alkaline gas in the atmosphere, playing a crucial role in the atmospheric nitrogen cycle and thus closely related to numerous environmental problems. Ammonia-containing aerosol particles are a significant source of PM2.5 in the air; changes in ammonia levels in nature can also lead to soil acidification, eutrophication of water bodies, and reduced biodiversity. Therefore, studying the spatiotemporal variations of ammonia concentration and its influencing factors has significant scientific value and practical implications. Atmospheric ammonia levels mainly originate from human agricultural, industrial, and urban transportation activities. As a developed economic region and densely populated area in China, the Yangtze River Delta region faces particularly serious environmental problems related to ammonia. Xue et al.\cite{Xue2025} used ammonia column density data from the European meteorological satellite infrared atmospheric interferometer, NASA $NO_2$ column density data, and European ERA5 meteorological reanalysis data to study the long-term spatiotemporal variations of ammonia column density in the Yangtze River Delta region from 2014 to 2020 and the driving factors behind these variations. Among their findings, they used the CE method to analyze the factors influencing the spatial variation of ammonia concentration. By calculating the correlation strength between ammonia column density and spatial variables such as meteorological factors, pH value, population density, and cultivated land ratio, they found that ammonia concentration is closely related to factors such as surface air pressure, precipitation, pH value, and cultivated land ratio, indicating that the distribution of ammonia in the Yangtze River Delta region is affected by both natural and human activities.

\section{Ecology}
In ecology, the study of animal movement trajectories is a crucial fundamental problem, revealing basic ecological processes such as population activity patterns, inter-population competition, and interactions between populations and environmental resources. The application of information technology in ecology has generated a vast amount of animal trajectory data, requiring appropriate modeling methods for analysis. Circular-linear data is a common time-series data type in ecology, describing discretized animal movement processes, including two variables: direction of movement and distance of movement. These two variables are usually correlated; that is, in straight-line movement, the direction of movement is smaller while the distance is larger, and in turning movement, the direction of movement is larger while the distance is smaller. Furthermore, the distribution of the direction of movement variable is generally symmetrical. Therefore, the angularly symmetric circular copula function is commonly used as a tool to model this type of data, and a copula-based correlation metric is used to measure the correlation between the two variables. Hodel and Fleberg\cite{Hodel2021} implemented the algorithmic toolkit \texttt{Cylcop} for modeling and analyzing circular copulas, which includes a CE based mutual information estimation algorithm as a correlation metric for analyzing animal trajectory data.

Soil is the carrier of terrestrial ecosystems, and changes in the climate environment have a direct impact on the evolution of its internal structure and function. Studying the relationship between climate change and soil ecosystems is one of the main problems in soil ecology. Soil carbon and nitrogen are key factors in its ecosystem, and the soil carbon-to-nitrogen ratio is a key indicator for assessing its system function. Soil carbon and nitrogen transformation is mainly regulated by microbial processes, and microorganisms are highly sensitive to changes in humidity. Existing research mainly focuses on the response of soil carbon and nitrogen dynamics to short-term humidity changes, lacking understanding of the response to long-term humidity changes, especially those caused by climate change. Li et al. \cite{Li2025a} used a microbial-enzyme decomposition (MEND) model to study the long-term effects of soil moisture changes on soil carbon and nitrogen dynamics. They used a broad-leaved forest ecosystem in the Dinghushan National Nature Reserve in Guangdong Province as their research site, collecting climate environment observation data and soil experimental data from 2009 to 2012. The experiment first used 10 global climate models (GCMs) from CMIP6 to generate soil water content data under four future socioeconomic pathway scenarios. Then, interpolation methods and bias corrections were used to obtain future (2021-2100) soil moisture changes in the study area, leading to the standard soil moisture index (SSI) at different time scales (1-72 months). Next, the parameters of the MEND model were corrected using field-collected soil analysis data. Soil moisture data generated by GCMs and the corrected MEND model were then used to simulate soil carbon and nitrogen dynamics. Finally, the correlation between soil carbon and nitrogen variables and SSI at different time scales was estimated using the CE based MIR statistic. The study found that soil carbon and nitrogen dynamics respond to moisture changes with lag and accumulation, particularly the long-term (72-month) correlation between soil organic carbon and total nitrogen variables. Soil organic carbon tends to accumulate under drought conditions, and total nitrogen exhibits a similar characteristic. The experiment compared the MIR with the linear correlation coefficient, revealing a non-linear relationship between the experimental variables; therefore, the CE based MIR is more suitable for the experimental analysis.

Since the Industrial Revolution, human activities have significantly increased their impact on the natural environment, and the resulting ecological and environmental changes, in turn, affect the sustainability of human societal development. Socio-ecological systems are non-linear, dynamically interacting hierarchical systems composed of humans and the environment, and system resilience refers to its ability to withstand internal and external natural and social shocks while maintaining its core functions. The system resilience of ecologically vulnerable areas is of great significance to the sustainable economic and social development of these areas, and understanding the evolution of socio-ecological resilience in these areas is crucial for improving regional ecological preservation capabilities. Karst landforms, with their slow soil formation, thin soil layers, and difficulty in recovery after damage, are typical ecologically vulnerable systems, widely distributed in southwestern China. Li et al.\cite{Li2025b} proposed a method to study the evolution of socio-ecological system vulnerability using tools such as CE based TE and network analysis. Taking ecologically vulnerable Karst areas in Guizhou and Guangxi in southwestern China as their research object, they used this method to study the evolution of ecosystem characteristics in the southwestern Karst region from 1990 to 2022, based on regional economic and social data, meteorological observations, and land remote sensing data. This study first uses data to calculate five dimensions of the "Driving-Stress-State-Impact-Response" (DPSIR) ecological environment system analysis framework. Then, based on the changes in these indicators, the target time period is divided into three distinct periods. Next, TE is used to analyze the dynamic relationships between the five dimensions of DPSIR within each period. Finally, network analysis is employed to identify key factors influencing system function. The TE-based analysis of DPSIR factors shows that the TE from driving force to stress exhibits a long-term downward trend across the three periods, indicating a reduced impact of ecological degradation on economic growth and a shift towards a high-quality development model that balances development and ecology. Simultaneously, the TE from response to driving force gradually increases across the three periods, demonstrating a strengthening of positive feedback from humans to driving factors and enhanced synergy between social development and environmental protection. The authors argue that the analysis reveals the evolutionary process of socio-ecological vulnerability in the Southwest Karst region, demonstrating that this evolution is closely related to the government's consistent and continuous ecological protection policies.

\section{Animal morphology}
Animal morphology, the oldest branch of zoology, studies the morphology and anatomical structure of animals and their changes during development and evolution. As a fundamental discipline of zoology, morphological research forms the basis of animal classification, such as the morphological classification of fish. Because fish often have similar appearances, identification of their species can be biased, necessitating research into the measurement of similarity between fish structural morphologies. Escolano et al. \cite{Escolano2017} proposed an estimation method for graphic similarity measurement, converting graphics into multidimensional manifold embedding vectors and then using the CE estimation method to estimate the MI between vectors as a graphic similarity measure. They applied this method to the \texttt{GatorBait} marine fish image database, which contains 100 triangular mesh images of fish in 30 categories. Since each category corresponds to a genus rather than a species, morphological differences exist within the same category, making classification difficult. They used the new measurement method to classify the fish images in the database, and experiments showed that the new method achieved better classification performance than traditional methods on the dataset.

Abalone is an important marine mollusk with high nutritional and economic value. Morphological studies of abalone, involving the measurement of morphological variables, are crucial for understanding its growth process and population distribution, and are essential for the management of this marine resource. Purkayastha and Song \cite{Purkayastha2022} proposed a novel concept for measuring causality, called Asymmetric MI (AMI), to determine the direction of causal predictivity between variables, and provided a fast and robust estimation method based on CE theory. They applied the AMI method to the UCI abalone dataset, analyzing measurements of morphological parameters such as length, diameter, height, and weight, clarifying the causal relationships between age and these variables during abalone growth.

\section{Botany}
Sap flow refers to the process by which plants fundamentally transport water to their leaves to compensate for water loss through transpiration. It serves as a key indicator for assessing the amount of water lost through transpiration in trees. Because it is closely related to both the plant's intrinsic growth mechanisms and external environmental factors, these factors can be used to predict sap flow. Traditional methods cannot handle the nonlinear relationship between environmental factors and sap flow, resulting in less than ideal prediction results. Wang \cite{Wang2023master2} and Li et al.\cite{Li2025d} proposed a method for predicting sap flow based on historical environmental factors. This method first decomposes sap flow data using the EMD method, then uses CE to select the modal components that best represent the characteristics of the sap flow data to reconstruct the sap flow. Next, CE is used to select the environmental factors most relevant to the sap flow. Finally, based on the selected factors and the reconstructed sap flow, a prediction model is constructed using traditional machine learning and deep learning methods. He validated the method using sap flow data of kauri trees throughout 2012 and data on nine environmental factors from the SAPFLUXNET project's public dataset. The experimental results showed that CE selected five environmental factors (including saturated vapor pressure deficit, wind speed, deep soil moisture content, atmospheric relative humidity, and net radiation). The deep learning model obtained from this selection performed better than the comparison method (MSE=0.0229, MAPE=0.6759\%, R2=0.9755). This indicates that after introducing the CE method, the model obtained conforms to the correlation between environmental factors and sap flow, thereby improving the prediction performance.

\section{Agronomy}
Global warming and its resulting environmental changes directly impact food production, thus exacerbating global food security issues. Rice is one of the most important cereal crops, accounting for about 40\% of cereal production, and is crucial to food security in China. Studying how climate change affects rice yield and proposing countermeasures is a vital issue related to food security. Zhang et al. \cite{Zhang2023,Zhang2023c} used crop models and atmospheric circulation models to study the impacts of climate change on the growth and yield of double-cropping rice in southern China (Jiangnan and South China) and proposed countermeasures. The study employed the CERES-rice submodule of the DSSAT crop model to simulate rice growth and yield, and four atmospheric circulation models (GCMs) from CMIP6. They used CE and random forest analysis to examine the nonlinear relationships between meteorological factors and crop yield in each month. They used 27 sets of data from each GCM to drive rice crop models at 54 locations in southern China to obtain the final yield, and also studied the impact of sowing date. The study found that rising meteorological factors can advance rice maturity and reduce yield. If the effects of CO2 are also considered, early rice yields will increase, while late rice yields will still decrease. According to CE calculations, the relationship between yield and CO2 concentration for both rice seasons is the strongest among meteorological factors. Early planting of early rice and delayed planting of late rice may increase yields to some extent. The findings of this study point to a path for governments and farmers to address future climate change and provide important reference for adopting appropriate adaptation strategies.

As one of the world's three major food crops, rice plays a vital role in agricultural production. Accurately predicting rice yield is crucial for ensuring food security and guiding agricultural production, making it an important issue in the agricultural field. Rice yield is not only related to the characteristics of the rice itself but is also influenced by environmental factors such as weather. This influence is non-linear, posing a challenge to accurate yield prediction. Zhang et al. \cite{Zhang2024} proposed a rice yield prediction method based on deep learning technology. This method utilizes CE to select environmental factor variables with a non-linear relationship to yield and employs CNN and GRU to construct a prediction model. They validated the method using real data from Lin'an District, Zhejiang Province. The results show that CE can capture the non-linear relationship between rice yield and environmental variables, and the combination of CE and CGRU yields the best prediction results.

\section{Immunology}
As a versatile second messenger, intercellular calcium ions are used in studies monitoring immune system activation. The immune system functions through interactions between immune cells, such as lymphocytes, but our understanding of inter-lymphocyte communication mechanisms remains limited. Quantifying these mechanisms is essential for understanding the immune system's operation. Coordination within the lymphatic system occurs at the micrometer scale, requiring definitive empirical studies to elucidate the principles behind this communication mechanism. Abe et al. \cite{Abe2025} proposed using TE to analyze in vivo calcium ion imaging of mouse B cells. By preprocessing autotracking imaging data, they calculated the TE between paired cells using a CE based TE estimation method. They defined the information transfer rate as the ratio of TE to the Euclidean distance between moving cells. Their analysis showed a negative correlation between the information transfer rate and intercellular distance, providing supporting evidence for a possible paracrine communication mechanism in lymphocyte interactions. This CE based TE framework also provides a reliable metric for studying immune cell interactions at different spatiotemporal scales.

\section{Neurology}
Establishing causal relationships between neural signals is crucial for understanding brain connectivity. Causal connections reflect the direction of information transmission between different regions within the brain network during cognitive processes, characterizing the dynamic relationships between brain regions in cognitive processes. Compared to traditional Granger causality tests, model-free TE is more suitable for such causal analysis tasks. Redondo et al. \cite{Redondo2025} proposed a new TE concept based on CE theory, called STE (Spectral Transfer Entropy), to calculate the TE between time-domain signals after frequency domain filtering. Compared to calculating TE directly on the original signal, STE calculated in a specific frequency domain has greater neurological interpretability. They applied the method to the analysis of EEG signals from patients with attention deficit hyperactivity disorder (ADHD), using STE to construct a causal brain connectivity network, and discovered differences in attention-related brain connectivity between ADHD patients and healthy individuals. Experimental results show that healthy individuals exhibit a clear causal relationship in the $\theta$ and $\alpha$ frequency bands associated with attention and controlled memory access, while the brain network connections of ADHD patients are mainly in the $\delta$ oscillation, which can be explained as being related to attention deficits.

During task execution, the information processing activities of various brain regions are coordinated as a whole. Phase synchronization refers to the temporal structural relationship of neural signals in different brain regions during brain activity. Traditional phase synchronization studies focus on describing the relationship between paired brain regions, while global phase synchronization focuses on the synchronization strength between neural signals from multiple brain regions. Li et al.\cite{Li2025} proposed using CE to calculate multivariate MI to estimate GPS and detect brain network structure. They first compared this method with other similar methods using the R\"ossler model, demonstrating that it can better estimate the GPS strength in the network. Then, they applied this method to the analysis of SEEG signal data from patients with temporal lobe epilepsy, finding that the termination of epilepsy was accompanied by an increase in GPS strength and brain region integration, consistent with existing research findings.

Stroke is one of the leading causes of death worldwide, and accurate diagnosis is crucial for effective treatment and management. Magnetic resonance imaging (MRI) is a powerful tool for stroke diagnosis, but its complexity poses a challenge to MRI data analysis. Stroke causes local damage to the nervous system and its functional connectivity, but analysis of changes in MRI functional connectivity cannot reveal the impact of the stroke on dynamic brain function. Effective Connectivity (EC) reflects the dynamic causal relationships between brain regions, providing an effective pathway for analyzing and diagnosing the impact of stroke on dynamic brain function. Ciezobka et al. \cite{Ciezobka2025} proposed a stroke classification and diagnostic method combining directed EC graph analysis and graph neural network classification. They used three methods—Reservoir Computing (RC), Granger Causality (GC), and TE—to generate directed EC graphs, and then used GNNs to classify the EC graphs for diagnosis. This method utilizes the CE based TE estimation algorithm to generate EC graphs. They compared the above methods on MRI data from stroke patients and healthy individuals collected at Washington University in St. Louis. The experimental results showed that EC maps generated based on the RC, GC, and TE methods achieved roughly equivalent classification performance even with small sample heterogeneity. This diagnostic method allows for the visualization and interpretation of EC maps, facilitating the translation of research findings into clinical practice.

\section{Cognitive neuroscience}
Cognitive neuroscience analyzes observational data of various modalities of brain activity to understand the mechanisms by which the brain, as an information processing organ, represents, processes, and communicates with external stimuli. As a nonlinear statistical measure, MI is considered an ideal statistical tool for analyzing the correlations between brain signals. However, the difficulty in estimating MI hinders its widespread application. Ince et al.\cite{ince2017a} proposed an MI estimation method based on the equivalence relationship between MI and CE, called Gaussian Copula Mutual Information (GCMI). The GCMI method utilizes the property that CE is independent of marginal functions. First, the marginal functions of each variable are transformed into Gaussian function, resulting in a joint Gaussian distribution. Then, the MI is calculated based on the relationship between the obtained Gaussian distribution correlation matrix and the MI. This method is simple, convenient, and distribution-independent. However, since the calculation of MI from Gaussian distribution data is biased, this method still requires correction operations. Ince et al. compared GCMI with other MI estimation methods and applied it to analyze EEG data from a face detection task\cite{ince2016the} and MEG data from an auditory-speech stimulation task\cite{kayser2015irregular}. In the face detection task experiment, GCMI was used to calculate the correlation strength between image content and cognitive response, and successfully selected the cognitive response-sensitive region (the eye area in the image). In the auditory stimulation experiment, Ince et al. investigated the effect of rhythmic features in speech on the rhythmic synchronization of the brain's hearing. Through analysis of EEG response data of speech stimulation, the authors found that changing the pauses between syllables and words leads to a decrease in auditory delta synchronization. In this experiment, GCMI was the main tool for data analysis.

Building upon the GCMI algorithm, Combrisson et al. \cite{Combrisson2022} proposed a group-level method for analyzing brain cognitive networks based on information theory. This method combines nonparametric permutation operations with information metrics to analyze fixed-effect or random-effect models, adapting to variations among multiple individuals and across multiple tasks. They applied this method to data from two existing studies: the first analyzed high-gamma activity in MEG data of individuals performing a cognitive-behavioral mapping task, identifying task-related brain networks involving multiple motor, somatosensory, and visual cortex areas; the second analyzed anterior insula SEEG data from a reward-punishment learning task, revealing response delays in reward-punishment tasks and significant differences between reward and punishment responses. Wang et al. \cite{Wang2023} proposed a method for classifying cognitive levels in older adults. This method first constructs a brain cognitive network using GCMI, then uses GCMI for feature selection, and finally uses SVM to classify cognitive levels based on selected brain network connections. They applied the method to resting-state fMRI data from 98 elderly Portuguese individuals and found that the proposed method could capture nonlinear relationships between brain regions in the data and ultimately achieve a higher classification accuracy than similar methods.

Speech comprehension is a primary cognitive function of the human brain, and studying the encoding and parsing of speech information by neural activity in the brain is a crucial issue in cognitive neuroscience. The speech envelope contains low-frequency temporal information in the speech signal, and research shows it can explain most of the changes in neural responses. Speech envelope tracking studies the relationship between the speech envelope and its neural responses using methods such as electroencephalography (EEG). Due to the nonlinear characteristics of the brain, commonly used linear models cannot adequately represent this relationship. MI, as a nonlinear relationship metric, is considered capable of capturing the nonlinear relationship between the speech envelope and neural responses. De Clercq et al. \cite{DeClercq2022} used GCMI to compare the ability of linear models and MI analysis to characterize the nonlinear components of the brain, based on two sets of storytelling speech and corresponding collected EEG data. Experimental results showed that MI analysis detected significant nonlinear components beyond those of linear models, proving that GCMI is a more suitable tool for studying neural envelope tracking than linear models. The authors also experimentally verified that compared with traditional MI estimation methods, the GCMI method based on CE principles has many advantages, including robustness, unbiasedness, and suitability for multivariate analysis.

Neuron specification refers to the property of a neuron to perform a specific function, which can be identified by studying the relationship between external environmental stimuli and neural response signals. GCMI, as a measure of nonlinear correlation, is an ideal tool for studying this problem. Pospelov et al. \cite{Pospelov2023} used the GCMI method to calculate the correlation strength between calcium fluorescence signals and environmental variables and animal behavior. They analyzed calcium signals recorded in the CA1 region of the mouse hippocampus, revealing specialized neurons related to the animal's external environment, such as place neurons, and specialized neurons related to their behavioral activities, such as neurons active during running, standing, and resting. The study also identified neurons that responded to discrete variables, such as the animal's place location (central, near a wall, and in a corner) and its speed (resting, slow, and fast). He detected a total of 781 specializations in 472 neurons across four sets of experiments.

Neurological research indicates a close correlation between the level of consciousness and the complexity of brain activity; changes in brain activity complexity lead to changes in the level of consciousness. Therefore, measuring brain complexity to assess the level of consciousness is an important research direction. Information theory provides various tools for analyzing brain complexity, among which $\Omega$ information is unique in assessing the level of internal brain interactions. It can measure not only overall correlation but also the significance of synergistic or redundant effects; however, its estimation algorithm faces the problem of variable combinatorial explosion. Belloli et al. \cite{Belloli2025,Belloli2024} proposed a method based on GCMI to estimate $\Omega$ information, significantly improving the computational efficiency of the estimation algorithm. They applied this estimation algorithm to human awake and deeply anesthetized fMRI datasets to analyze brain complexity at different levels of consciousness. The experiment divided the activity of 55 brain regions in the data into 11 different brain region networks for analysis, with each network containing 5 brain regions. Experimental results show that deep anesthesia reduces both the maximum and minimum values ​​of$\Omega$ information in brain region interactions, weakens the synergistic interaction between different brain region networks, and also reduces redundancy within the same brain region network.

Pupil dilation reflects cognitive and behavioral processes and is controlled by cortical structures connected to the brainstem pupillary innervation system. Pupil measurement provides a powerful tool for understanding the brain mechanisms of cognitive activity. Modern recording technology can record pulse data of neuronal activity, making the analysis of the relationship between cognitive behavior and neuronal activity an important issue in the field. However, how to analyze the cognitive-neuronal correlation from a large amount of experimental data on neuronal activity has become a bottleneck problem. Waldencite{Walden2024,Walden2024a} proposed a method combining Copula-GP and Gaussian process factor analysis (GPFA) based on CE theory to study the interaction between visual cortical neural processes and pupillary dilation in live mice. He applied this method to neuronal pulse signal data recorded by mice during drift grating stimulation. First, GPFA was used to reduce the dimensionality of the data, then Copula-GP was used to estimate the Copula, and finally, CE was estimated to measure the correlation between cognition and neuronal activity. Experimental results show that the Copula-GPFA method can effectively estimate the CE value from the data, confirming the information interaction process between neuronal activity trajectory and pupillary dilation, which is consistent with existing research findings.

Visually evoked potentials (VEPs) are neurological analysis techniques used to study the brain's visual cognitive function through the electrophysiological response to visual stimuli in the visual cortex. Analyzing information flow during brain activity using information theory methods is an important neurological data analysis approach that can improve our understanding of the brain's dynamic response characteristics. Directed information (DI) and TE based on the GC concept are two main tools for measuring information flow. Kaufmann \cite{Kaufmann2017} proposed using these two tools to analyze VEP EEG recordings. He presented two decomposition forms of DI and TE based on two different causal definitions, and further provided estimation methods for these two information flow decomposition forms under the Gaussian Copula assumption based on CE. He applied these DI and TE information flow estimation methods to a set of VEP EEG data from visual stimulation experiments in healthy individuals, successfully discovering the information flow between the occipital and frontal electrodes.

\section{Motor neuroscience}
Muscle synergy is fundamental to movement, referring to the spatiotemporal coordination between muscle groups when performing various actions. The human motor control system is a system with redundant degrees of freedom; it is generally believed that the nervous system completes an action through the combination and coordination strategies of motor primitives. A crucial fundamental problem in motor control research is identifying simplified basic muscle synergy strategies in motor control. Understanding the underlying basic synergistic mechanisms of motor control by decomposing electromyographic (EMG) signal data during the movement process is a basic research method, but handling the nonlinearity in the signal is one of the major challenges. MI estimation based on CE is a powerful tool for addressing this challenge. Wu et al.\cite{wu2021she,Wu2022} combined multivariate variational mode decomposition with CE based MI to construct a muscle coupling network model. Based on surface EMG data, they analyzed the spatiotemporal synergy between upper limb muscles during the reaching motion of healthy individuals, successfully characterizing the strength of muscle coupling relationships. Wang \cite{Wang2021master} proposed using CE based MI estimation for intermuscular coupling strength analysis and validated the method using sEMG datasets from the Ninapro 4 database. He found that this method can effectively analyze potential neuromotor control mechanisms and describe the differences in intermuscular coupling relationships across different frequency bands. He also proposed a TE estimation method based on CE based TE representation, called GCTE, and used it to construct intermuscular coupling networks. He applied this analysis method to sEMG signal data from motor tasks in healthy individuals and stroke patients, constructing intermuscular coupling networks in different frequency bands. The results show that the GCTE method can describe the asymmetric information flow between muscle nodes in the coupling network, revealing nonlinear intermuscular coupling relationships, and has application value for upper limb motor function rehabilitation evaluation. He also proposed MI and CMI estimation methods based on R-vine copula, called RVCMI and RVCCMI, and demonstrated the advantages of RVCMI and RVCCMI over traditional methods based on simulation experiments, providing new tools for intermuscular coupling network analysis. Zhu et al.\cite{Zhu2022} proposed a representation of TE based on CE, and then used R-vine Copula to estimate CE and thus TE. They applied this method to the study of upper limb muscle coupling networks, constructing a muscle coupling network based on sEMG data of upper limb muscle movements in fatigue/non-fatigue states, and found that the coupling relationship between muscle groups in the fatigue state gradually deepens compared to the non-fatigue state. Jin et al. \cite{Jin2022} proposed a method combining wavelet analysis and CE estimation to analyze sEMG signal data of muscle fatigue state under voluntary movement in healthy individuals, and found that in elbow flexion movement, the intermuscular coupling strength is most significant in the Beta and gamma frequency bands, and the coupling strength of synergistic muscle pairs is greater than that of antagonistic muscle pairs; the coupling strength after fatigue is stronger than that before fatigue. Reilly and Delis \cite{Reilly2022,OReilly2024} proposed using CE based GCMI to measure the spatiotemporal correlation between EMG signals, and then using matrix factorization dimensionality reduction methods to discover the basic muscle synergistic patterns in the spatiotemporal correlation of EMG signals. They collected EMG data of point-to-point movements performed by humans, applied their methods to the data, and obtained physiologically significant spatiotemporal patterns of muscle coordination.

\section{Computational neuroscience}
Computational neuroscience is a discipline that uses computational theories and methods to study and understand the functions and mechanisms of the nervous system, investigating how to describe the individual and group responses of biological neurons to signal stimuli. Neural plasticity refers to the adaptive structural changes of a neuronal network in response to external stimuli; constructing theoretical models of plasticity is one of the main focuses of computational neuroscience. Leugering and Pipa \cite{leugering2018a,Leugering2021phd} proposed a theoretical framework for neuronal population plasticity based on Copula theory, constructing an adaptive network model that can maintain the invariance of model outputs even when the model input changes unknown. In this framework, CE is used to measure the statistical properties of neuronal populations, measuring the amount of information between input and output. Analysis of information transmission between neurons is another important problem in computational neuroscience. Analyzing the information transmission relationships between computational neurons requires the decomposition of the information transfer mechanism (MI) between multiple neurons. Partial information decomposition is the theory that decomposes MI into three parts: Synergy, Redundancy, and Unique Information. Based on CE theory and method, Pakman et al.\cite{pakman2021estimating} proposed a method for estimating unique information and applied it to analyze information processing in multiple neuron models. Coroian et al.\cite{Coroian2024} proposed a method for estimating the coordination of continuous distributions based on CE theory to improve the computational efficiency of existing coordination estimation algorithms and applied it to brain neural activity data to analyze co-emergent phenomena in brain functional connectivity maps and changes in brain region activity patterns.

\section{Psychology}
The brain is a distributed network system. It not only controls the body and alters its internal physiological state but also influences multiple higher-level processes. Simultaneously, visceral information is constantly monitored by the brain, meaning that visceral processes are reflected in cortical activity. Research on brain activity related to visceral events is an important topic. Processes in the autonomic nervous system are interconnected, and information theory provides tools for studying these relationships. Ravijts \cite{Ravijts2019} investigated the approximate estimation of the temporal interaction of heartbeat-evoked potentials (HEP) under four emotional stimulus characteristics (valence, arousal, dominance, and liking). He used the DEAP dataset of physiological signals used for emotion analysis and employed the GCMI method based on CE to estimate statistics such as MI, synergy, and redundancy to measure the temporal interaction on HEP ​​under different emotional stimuli. The experiment revealed the temporal interaction phenomenon on HEP ​​under dominance and liking stimuli, revealing for the first time the temporal characteristics of HEP modulated by emotion perception.

\section{System biology}
A major task of systems biology is to study the interactions between regulatory, signal transduction, and metabolic processes using biochemical kinematic models. Building such models requires selecting appropriate input variables, and MI is one tool for variable selection. However, the commonly used kNN MI estimation is often biased and needs correction. Charzy\'nska and Gambin\cite{charzyńska2015improvement} proposed a bias correction method and found that the correction effect is significant when MI is estimated using the relationship between MI and CE. The authors applied the method to a widely studied negative feedback loop problem model between the p53 protein and the Mdm2 ligase, and the results showed that this method can obtain more accurate analytical results reflecting the model input-output relationship that reflects system behavior than traditional local sensitivity analysis methods.

One of the main objectives of systems biology in analyzing molecular biology data is to establish networks and dynamic mechanisms of complex biological phenomena to analyze the function and behavior of living tissues. MI plays a fundamental role in constructing gene pathway networks. Farhangmehr et al. \cite{charzyńska2015improvement} first proposed using CE to estimate MI during network construction. They applied this method to yeast cell cycle data and compared the resulting dynamic network with the Kyoto Genomics Encyclopedia database. The experimental results showed that using CE to estimate MI improved computational efficiency.

\section{Bioinformatics}
Bioinformatics is an emerging discipline that studies the mechanisms of life and disease by analyzing gene data (including gene expression profiling) using algorithms. Gene expression profiling is data obtained by observing the dynamics of a living organism at the gene molecular level using DNA microarray technology, thus reflecting various phenomena and mechanisms of living systems at the genome level. Wieczorek and Roth\cite{wieczorek2016causal} proposed an analytical method for studying the interactions between time series data, called Causal Compression. Unlike traditional analyses of causal relationships between full time series, this method studies the sparse expression of causal interactions between time series based on Directed Information decomposition, and provides solutions for two types of problems: time-series causal segmentation and causal bipartite graph discovery. Based on the equivalence between CE and MI, the authors proved that this method is only related to the copula density function of the data distribution, and designed a solution method accordingly. The authors applied this method to human hepatitis C virus infection data in the NCBI database (NCBI/GEO query number: GSE7123) to study the time-series gene expression profiles of recombinant hepatitis C virus core protein genotype 1 infection patients treated with pegylated interferon and ribavirin. They focused on two genes playing important interactive roles in interferon signaling: the transcript STAT1 and the interferon-induced antiviral gene IFIT3, generating different interactions between the two genes in effectively treated and untreated patients. The study found that, according to the analysis, interferon therapy eliminated the association between the two genes in most effectively treated patients, while the association was unaffected in untreated patients. Furthermore, the analysis showed a causal interaction between the two genes before and after treatment in both types of patients, but for effectively treated patients, the early effect of IFIT3 on the later effect of STAT1 was more significant, consistent with previous research findings.

Many diseases are related to gene structural variations. Copy number variations (CNVs) refer to variations in DNA segments longer than 1 kb, and are abundant in the human genome. As important gene variations, CNVs contain a large number of DNA sequences, disease points, and functional units, providing clues for disease research. Studies have shown that the formation and development of various cancers are related to different CNVs. Therefore, discovering the relationship between CNVs of different genes and different cancers is helpful for studying cancer etiology and diagnostic methods. Selecting cancer-related features from the gene characteristics of a large number of CNVs is an important problem in bioinformatics. Wu and Li\cite{Wu2022a,Wu2022c} proposed a gene selection method called Correlation Redundancy and Interaction Analysis (CRIA), which selects cancer-related genes based on CNVs for cancer classification. The CRIA method utilizes the multivariate correlation characteristics of CE and designs a gene feature interaction strength measure to screen genes with strong correlation to cancer type. They applied this method to the cBioPortal cancer genomics data, utilizing data from six cancer types to select 200 cancer-related genes. To verify the algorithm's effectiveness, they compared the method with eight other gene selection algorithms based on data from Arizona State University. The results showed that the genes selected by the CRIA method could more accurately predict cancer types. Shang et al.\cite{Shang2024,Shang2024master,Yan2025} proposed a feature selection method based on CE, called CEFS+, for high-dimensional gene data analysis. They applied this method to three publicly available high-dimensional cancer gene datasets and compared the proposed method with similar methods. The results showed that the proposed method gave the best prediction results in the experiments, demonstrating its superiority. Pan et al.\cite{Pan2023b} proposed a CE based improved gray wolf optimization algorithm for feature selection of high-dimensional genomic data. This method uses the CE values ​​corresponding to genes to initialize the gray wolf optimization algorithm population before performing feature selection optimization. They validated the improved algorithm on 10 sets of high-dimensional small-sample gene expression datasets. The results showed that, compared with other feature selection methods, the improved algorithm improved classification performance by selecting fewer features. The CE based improved algorithm achieved the best classification results on 7 out of the 10 datasets, demonstrating the effectiveness and superiority of the algorithm.

In the life process, complex diseases are often caused by multiple genes, which is reflected in genomic data as disease type data being simultaneously correlated with multiple gene features. Therefore, when building bioinformatics models, the interaction between these multiple genes must be considered during feature selection to obtain a better-performing disease classification model. Zhong \cite{Zhong2024master} proposed a feature selection method based on Copula and normalization metrics, called IFSMRMR, which considers both the correlation between features and disease labels and the redundancy between features. CE is used to measure the strength of interactions during feature selection. He compared this method with six common feature selection methods on eight public gene expression datasets. The results show that this method improves classification performance by considering feature interactions and also achieves better feature clustering results, outperforming similar comparative methods and validating the effectiveness and superiority of the proposed method.

Constructing gene regulatory networks based on gene sequencing data is one of the main problems in bioinformatics, aiming to understand gene function and identify the dynamic processes of gene expression. Single-cell sequencing technology can simultaneously measure the whole-genome expression of a large number of individual cells, while time-series single-cell sequencing data reflects the dynamic processes of gene regulation within cells. Therefore, nonlinear time-series causal analysis tools such as TE can be used to discover gene regulatory networks. Zhu \cite{Zhu2023} proposed a gene regulatory network construction method based on TE analysis, called GRN-PAGATE, which employs the CE based TE estimation method. He validated the method on Ecoli data from the DREAM3 challenge and single-cell sequencing data from early mouse embryonic blood development, and compared it with similar methods. Experimental results show that the method has performance equivalent to GRNTSTE on Ecoli data and slightly better than similar methods such as DynGENI3 and SCRIBE; on mouse embryonic data, the method can effectively discover key gene regulatory relationships that other methods have failed to discover, and its performance is superior to similar comparative methods.

Bayesian networks (BNs), as a tool for causal relationship analysis, are increasingly used in omics data analysis in systems biology. However, their linear assumptions make them unsuitable for analyzing the dynamics of complex regulatory networks. CE, as a model-free nonlinear relational measure, provides an efficient algorithmic tool for analyzing and processing complex high-dimensional gene data. Li\cite{Li2024phd} proposed a novel method combining BNs and CEs to construct gene regulatory networks. In this method, CEs are used to further filter and identify the sub-network structures generated by BNs to obtain the final gene regulatory network. Based on six publicly available hepatocellular carcinoma (HCC) sequencing datasets, the authors applied this method to the problem of identifying gene regulatory sub-networks in HCC and compared it with five similar methods. Experiments based on this method yielded a gene sub-network containing 12 genes and 28 edges, which is the smallest sub-network among all methods. The authors applied the sub-networks obtained by each method to three GEO datasets to conduct HCC tumor classification experiments to verify the effectiveness of the sub-networks obtained by each method. Experimental results showed that the subnetwork obtained by this method performed best in all three classification experiments, with an average AUC significantly higher than 90\%. Subsequent differential gene expression analysis further confirmed that the genes in the obtained subnetwork were significantly downregulated at both the mRNA and protein levels, representing potential biomarkers for disease diagnosis. The results obtained by this method are interpretable, can enhance the understanding of the underlying molecular mechanisms of diseases, and help identify potential therapeutic targets and develop corresponding treatments. It is expected to be applied in similar analyses of other diseases.

Human genome sequencing data has opened a door to understanding the biological mechanisms of diseases. However, the dimensionality of genomic data far exceeds the number of people sequenced, leading to the curse of dimensionality in genomic data analysis. Understanding the data and transforming it into biomedical technologies requires more new high-dimensional statistical analysis tools, and information theory provides a theoretical tool for solving this problem. The unique concepts of information, redundancy, and synergy in traditional information theory and partial information decomposition theory can help us analyze high-dimensional data, but estimating these concepts from the data is a practical challenge. Lacalmita\cite{Lacalmita2025} proposed using a CE based GCMI method to estimate MI and synergy concepts, and based on this, proposed a genomic data analysis workflow for disease prediction. This method first uses complex network community detection technology to cluster genes to obtain gene clusters. Then, for each cluster, feature selection technology is used to construct a disease prediction model, thereby obtaining gene clusters with better prediction performance. Then, CE based MI and synergy metrics are used to further cluster the obtained gene clusters to obtain gene subclusters. Finally, the genes in the subclusters are used for disease prediction. He validated the method using publicly available genomic data from NCBI on hepatocellular carcinoma and autism spectrum disorder, finding that the method can effectively screen high-dimensional genes. Furthermore, genes selected using the concept of synergistic clustering were better than those selected using MI clustering, as the former selected genes were more consistent with the conclusions of biomedical disease mechanism research and had better disease prediction performance.

\section{Clinical diagnostic}
Heart disease is one of the most common clinical diseases. Physicians have accumulated rich clinical diagnostic experience in heart disease and can make diagnostic decisions based on various physiological measurements. Developing intelligent clinical diagnostic models based on this experience has been a long-term goal in the field. The key to developing such models lies in selecting a set of physiological measurement variables to construct the predictive diagnostic model. Based on the well-known UCI Heart Disease Dataset\cite{asuncion2007uci}, Ma\cite{jian2019variable} proposed using CE as a variable selection method to select a set of physiological variables to construct a diagnostic model. This dataset contains real clinical heart disease physiological measurements and diagnostic data from four locations around the world, among which 13 physiological measurement variables were identified as clinically relevant by medical experts. Experimental results show that the CE method selected 11 of the 13 clinician-identified variables, the most among the compared methods, thus achieving the best predictive accuracy. Simultaneously, the CE method also discovered other diagnosis-related variables beyond the identified variables, providing new references for further clinical testing.

Diabetes is another common clinical disease. The management of diabetic patients' conditions is closely related to clinical outcomes (morbidity and mortality). Therefore, establishing a strict inpatient management process for diabetic patients is crucial for their safety, which necessitates analysis and research on management standards. To assess the treatment effectiveness of hospitalized patients, the US medical community established the Health Facts dataset\cite{Strack2014}, which includes data on diabetic patients from 130 US hospitals and treatment networks. Based on data from 101,721 hospitalized patients over a 10-year period (1999-2008) using this dataset, Mesiar and Sheikhi\cite{Mesiar2021} used the CE variable selection method to build a predictive model to predict the ``medication status" variable from 49 other variables. This model achieved good predictive results, reaching 97.2\% accuracy with only 20 variables selected, improving our understanding of medication-related variables and constructing a rational drug use evaluation model.

Cancer prognosis refers to the assessment of the future development of cancer based on its clinical manifestations and diagnostic results, in order to aid further clinical decision-making. Prognostic factors considered in clinical assessment are crucial, but often numerous and require analysis and selection. For example, there are as many as a hundred prognostic factors for lung cancer. Prognostic models are patient risk prediction models built upon prognostic factors and are important clinical tools in cancer treatment. Ma\cite{Ma2022a} proposed a survival analysis variable selection method based on CE and applied it to the problem of prognostic factor selection to build a prognostic model predicting patient survival time. He validated the method based on two publicly available lung cancer datasets, finding that it can select prognostic factors that meet clinical criteria and obtain a better predictive model than similar methods, exhibiting better predictive performance while ensuring model interpretability.

Breast cancer is one of the most common malignant tumors in women, and its incidence and mortality rates in China are showing an increasing trend year by year, seriously threatening women's health and family happiness. Utilizing statistical methods to analyze clinical data and construct diagnostic models to assist clinical diagnostic decisions can improve doctors' work efficiency and reduce misdiagnosis rates, thereby promoting patient health improvement. Fu \cite{Fu2023b} proposed using a feature selection method to construct a prognostic model for breast cancer patients. She employed three feature selection methods: Lasso, CE, and RFREF. Clinical diagnostic data of breast cancer patients from 2010 to 2014 in the SEER database were analyzed. Four models—logistic regression, random forest, XGBoost, and Stacking—were constructed using the features selected by the three methods to predict patients' 5-year survival. The results showed that the logistic regression model constructed using features selected by CE achieved the highest predictive accuracy (96.84\%).

Cataracts are a common eye disease and the leading cause of blindness. Phacoemulsification is the preferred surgical treatment for cataracts worldwide. Although this surgery is highly advanced, postoperative complications such as corneal edema can still occur, affecting visual recovery and causing patient discomfort. Therefore, constructing a risk factor-based corneal edema prediction model is essential in clinical practice. Luo et al.\cite{Luo2023,Luo2023master} proposed using the CE method to construct a postoperative corneal edema risk prediction model. They applied this method to data from 178 patients, selecting predictive variables from 17 variables in the data. Ultimately, the four variables used in the clinical prediction model (diabetes, best corrected visual acuity, lens thickness, and cumulative dissipated energy) were reduced to two (best corrected visual acuity and cumulative dissipated energy) without affecting prediction accuracy. The results showed that the prediction model obtained using CE has clinical application value, reducing the amount of clinical information needed for prediction while maintaining predictive performance.

Aortic regurgitation is a common valvular heart disease, primarily characterized by the backflow of blood from the aorta into the left ventricle during diastole. Aortic valve replacement surgery is one of the traditional treatments for aortic regurgitation. Left ventricular ejection fraction (LVEF) is an important indicator of cardiac function, and studying its improvement before and after surgery can provide evidence for the timing and outcome prediction of valve replacement surgery. Sunoj and Nair \cite{Sunoj2023} extended the concept of CE using Survival Copula, proposing a new concept called Survival Copula Entropy (SCE) to measure the dependencies between variables related to the survival function. They applied SCE to clinical data from aortic valve replacement surgery and found a positive correlation between LVEF before and after surgery.

Brain tumors are a high-mortality tumor, accounting for approximately 5\% of all cancers, and their incidence has been on the rise in recent years in China. Brain tumor lesions are characterized by diverse morphologies and unpredictable locations, making diagnosis challenging. Classification and identification based on non-invasive medical images is the primary clinical diagnostic method. Utilizing deep learning methods to extract quantitative features from tumor medical images and construct diagnostic models can assist physicians in clinical diagnosis, thus attracting extensive research. How to extract and select quantitative features from images is a key issue in constructing auxiliary diagnostic models. Pan \cite{Pan2023} proposed such a feature selection method, first initializing the feature set using correlation metrics such as CE, and then optimizing the feature set using the Grey Wolf optimization algorithm with classification performance as the objective. Using imaging data from 102 patients with low-grade gliomas with ATRX mutations from the First Affiliated Hospital of Chongqing Medical University, Southwest Hospital, and Sichuan Cancer Hospital, he extracted 5,530 radiomics features across five categories. The results showed that, compared to the comparative methods, the proposed method achieved the best classification performance with the fewest (13) features selected, and the selected features were correlated with the ATRX mutation status, showing potential as biomarkers.

Pulse wave diagnosis is a primary diagnostic method in traditional Chinese medicine (TCM) because it carries complex and diverse pathological information, reflecting the physiological state of the cardiovascular system to a certain extent. TCM pulse diagnosis relies heavily on the personal experience of renowned physicians. Researching analytical algorithms for pulse wave data is crucial for the non-invasive diagnosis of common diseases such as diabetes and hypertension, contributing to the scientific development of TCM. Tang\cite{Tang2023} proposed a multimodal pulse wave diagnostic algorithm based on graph convolutional neural networks. This algorithm converts the pulse wave into a three-channel image containing complementary pathological information, extracts image features using ResNet, and finally uses correlation metrics such as CE to obtain an adjacency matrix reflecting the temporal correlation between pulse wave signals to construct a graph convolutional neural network for disease classification and diagnosis. Based on actual wrist and fingertip pulse wave data, he classified the health status of patients with hypertension and diabetes, demonstrating a prediction accuracy of over 99\%.

\section{Geriatrics}
Alzheimer's disease is one of the major neurodegenerative diseases affecting the elderly, clinically manifested as excessive cognitive decline. Early screening and diagnosis can help dementia patients and their families intervene and manage disease progression at an early stage, effectively improving patients' quality of life and reducing family and social costs and burdens. The Mini-Mental State Examination (MMSE) is one of the widely used cognitive ability screening tools in clinical practice. Ma \cite{jian2019predicting} analyzed the correlation between finger tapping characteristics and the MMSE using CE, discovering a set of characteristics associated with the MMSE, including tapping frequency (or number of taps or average tapping interval). Based on this correlation, they constructed a predictive model from finger tapping characteristics to the MMSE, achieving good predictive results. This predictive model holds promise for use in cognitive ability screening for diseases such as dementia.

Parkinson's disease (PD) is another common neurodegenerative disease, clinically manifested as bradykinesia and motor dysfunction. Repetitive transcranial magnetic stimulation (rTMS) is a clinical treatment technique that uses pulsed magnetic fields to act on the central nervous system to improve physiological function. It is widely used in the treatment of neurological and psychiatric diseases and has recently been applied to PD rehabilitation research to alleviate patient symptoms and improve motor function. Li et al.\cite{Li2023} studied the neuromodulation mechanism of rTMS in the adjunctive treatment of motor symptoms in PD patients. They analyzed EEG data before and after rTMS treatment using CE based GCMI and other methods, constructed a brain functional network connectivity matrix, and obtained three network characteristic parameters. The experimental results showed that rTMS mainly altered the beta and gamma oscillations in PD patients, and the corresponding changes in the motor cortex may be related to the improvement of motor function.

Falls are a significant health risk for older adults, requiring scientific management and early intervention. Fall prediction is a crucial tool for managing fall risk. The Timed Up and Go (TUG) test is a primary tool for assessing fall risk. Ma \cite{ma2020predicting} proposed a fall risk prediction method combining video analysis and machine learning. This method first analyzes 3D posture information from videos of older adults performing the TUG test. Then, it calculates a set of gait features from a sequence of posture information over a period of time. By using CE to examine the correlation between gait features and fall risk indices, a set of risk-associated gait features (including stride length, gait speed, and gait speed variance) is selected. Finally, these features are used as input to construct a fall risk prediction model. Experiments on real data demonstrate good predictive performance. This analysis also shows an intrinsic link between mobility reflected by gait features and fall risk, making the model clinically interpretable.

\section{Psychiatrics}
Depression is a common mood-related mental disorder affecting approximately 350 million people worldwide, making its research crucial for human health. Electroencephalography (EEG) is a non-invasive method for measuring brain activity and is widely used in brain disease research. Brain functional networks (BCNs) are functional indicators reflecting brain activity constructed based on EEG signals, and these networks can be constructed using various methods such as MI and coherence analysis. Zhang et al.\cite{Zhang2022,zhang2022master} proposed using a brain network connectivity index based on the imaginary part of coherence to study the identification of patients with depression. They used feature selection methods such as CE and relief filtering to select connectivity features from the EEG network, finding that the coherence online feedback indicator feature set obtained by combining CE and relief filtering can effectively distinguish between patients with depression and healthy individuals.

Mental illnesses are brain dysfunctions caused by psychological, physiological, and social environmental factors, such as schizophrenia, bipolar disorder, and attention deficit hyperactivity disorder. Numerous studies have shown that these diseases may share a common psychopathological pattern, making cross-disease pathological analysis crucial. Functional magnetic resonance imaging (fMRI), as a non-invasive method for detecting brain activity, can provide multi-band brain activity signals from different brain regions, allowing for the inference of brain functional connectivity networks for research on disease mechanisms and diagnostic methods. However, traditional brain functional networks are mostly based on linear relationships between brain regions, while brain neural activity exhibits non-linear dynamic characteristics, necessitating the design of brain functional network generation methods capable of detecting non-linear relationships. Han et al.\cite{Han2025} proposed a multi-frequency decomposition entropy (MDE) method based on variational mode decomposition (VMD) and CE to infer non-linear brain functional connectivity networks from fMRI time-series data. In this study, VMD is used to perform frequency mode decomposition on the time-series fMRI signals of brain regions, while CE is used to calculate the nonlinear correlation strength matrix between brain regions in different modalities. Finally, based on this matrix, paired brain regions with strong correlations are selected. They applied this method to the open fMRI brain imaging dataset for mental illnesses, which includes sample data from schizophrenia, bipolar disorder, attention deficit hyperactivity disorder, and healthy individuals. The study found that the network relationships obtained by this method for patients and healthy individuals in the three diseases are significantly different, and the networks for each disease have their own unique characteristics, thus possessing the potential to serve as disease biomarkers. The authors compared the brain functional connectivity networks generated by the MDE method with those generated by linear methods, finding that the MDE method based on the CE nonlinearity metric can detect differences in the networks between patients and healthy individuals, while traditional linear methods cannot detect these differences, validating the superiority of this method.

\section{Forensics}
Forensic DNA analysis is one of the hottest and cutting-edge technologies in the field of forensic medicine. It refers to the analysis of the correspondence between DNA genetic information and human phenotypic characteristics, and the characterization and prediction of the phenotypic characteristics (such as height and age) of the DNA source from the DNA test information of biological samples, providing clues and evidence for case analysis and identity identification. Height is an important biological characteristic of interest to forensic medicine, and inferring human height based on DNA analysis is one of the important issues of interest in forensic medicine. DNA methylation is an important epigenetic phenomenon, and many studies have shown a link between DNA methylation and height. Therefore, predicting height characteristics using DNA methylation sites is a problem worthy of research. Wang \cite{Wang2024phd} studied the problem of predicting height based on DNA methylation sites. He first used publicly available data to confirm the link between DNA methylation and height, screened a set of relevant sites, and then used machine learning methods to construct a height prediction model based on the selected sites based on the sequencing data of the recruited population, obtaining relatively accurate prediction performance (prediction error of about 4.5 cm). In his study, he used methods such as CE to analyze and estimate the correlation between selected loci and height for different sexes. He found that the correlation between methylation features and height differed between sexes, which is consistent with the sex heterogeneity in human growth and development. This difference was also reflected in the model prediction experimental results, thus demonstrating the necessity of constructing different height prediction models based on sex. In contrast, the analysis based on the linear PCC did not find this sex difference, illustrating the superiority of CE as a nonlinear correlation measure.

\section{Pharmacy}
Drug-target interaction (DTI) prediction refers to the computational prediction of the binding between small drug molecules and protein targets. It is one of the most crucial steps in the drug discovery process, improving the efficiency of traditional experimental drug candidate screening. The use of machine learning methods to predict DTIs from drug and protein feature sets has been extensively studied and yielded significant results. Among these, selecting an appropriate feature set to build a DTI prediction model is a key issue, as the selection result is crucial to the final model's predictive performance. Gao et al. \cite{Gao2024} proposed an interactive multi-feature fusion DTI prediction algorithm. They designed a multi-feature fusion algorithm called RCI using tools such as CE, which can select the optimal feature set with high interactivity and low redundancy without sacrificing prediction accuracy. They compared their proposed method with similar DTI baseline algorithms using four benchmark datasets (Enzyme, IC, GPCR, and NR). Experimental results show that the RCI method outperforms similar feature selection methods, achieving the best prediction accuracy when the selected feature set is minimized. Furthermore, the prediction performance of the RCI-based prediction algorithm is better than that of similar DTI baseline algorithms, thus verifying the algorithm's superiority.

\section{Public health}
Epidemics are a crucial topic in public health, and timely diagnosis of patients with epidemics is essential for controlling their spread. Patients infected with epidemic viruses often present with symptoms such as fever, making them difficult to distinguish from patients with normal fever. The currently prevalent novel coronavirus patients exhibit such fever symptoms, making the development of technologies based on clinical data to differentiate between virus-infected individuals and those with normal influenza an urgent issue. However, with over a dozen related symptoms, selecting an appropriate set of variables is key to successful research. Mesiar and Sheikhi\cite{Mesiar2021} analyzed 19 symptom variables related to the diagnosis of COVID-19 patients using real clinical data based on the CE based variable selection method. They found that age, fatigue, and nausea and vomiting were the most important diagnostic variables, achieving a diagnostic accuracy of 85\%. Increasing the number of diagnostic variables to 15 improved the accuracy to 91.4

Hypertension is the leading cause of death worldwide, posing a serious threat to public health. Genome-wide association studies have shown that multiple genes are closely related to hypertension. Several studies have reported that the type I cell membrane calcium transporter gene (ATP2B1) is associated with systolic and diastolic blood pressure. This gene has 21 CpG sites. Investigating the relationship between this gene and its CpG sites and hypertension is a new and important question. Purkayastha and Song \cite{Purkayastha2022} proposed a new concept of asymmetric predictability, called asymmetric MI (AMI), and provided an estimation method using CE theory. They applied this method to the ELEMENT dataset, analyzing data from 525 children aged 10-18 years, and found that ATP2B1 is associated with diastolic blood pressure, confirming previous findings. They also found that the CpG site CG17564205 of this gene is associated with diastolic blood pressure, and based on AMI, diastolic blood pressure is predictive of this site. This new finding indicates that blood pressure can be altered at this site.

\section{Economics}
The evaluation of economic policies requires quantitative analysis, which can scientifically and objectively assess policy effectiveness. Shan and Liu \cite{Shan2020,luo2022} proposed a decision tree construction method for quantitatively analyzing the effects of policy combinations. CE is used to measure nonlinear correlations and construct decision trees. The method utilizes information gain based on the definition of CE to build policy decision trees that distinguish different policy target groups, with leaf nodes representing the group divisions corresponding to different policy combinations. They applied this method to the field of development economics to evaluate the effectiveness of China's poverty reduction policies, analyzing data from Sichuan Province in the 2018 government-led census of impoverished households. The analysis found that employment policies, new sources of income, and whether or not a household has a mortgage are the main policy factors affecting household income, revealing different characteristics of the income structure of different target impoverished groups corresponding to these policy combinations. This method, without historical data, evaluated and verified the effectiveness of poverty reduction policies and identified more effective policy combinations. 

The core objective of economics is to discover causal relationships. Traditional economics relies on inferential modeling and experimental design based on it. Causal discovery is a method of finding causal relationships from data, and combining it with economic theoretical models is a new path for designing economic experiments. Bossemeyer\cite{Bossemeyer2021} proposed a conditional independence test algorithm based on the relationship between CE and MI, and applied it to the PC algorithm for causal structure discovery. The authors used the new PC algorithm to study bargaining theory in economics, investigating the role of reciprocity in bargaining behavior and the role of response time in this process. The authors applied the algorithm to data from eBay's Best Offer platform, finding a correlation between the price concessions of both parties, confirming the reciprocity theory; at the same time, they discovered a causal effect of the opponent's counter-offer response time on the next asking price.

An industrial chain refers to a chain-like relationship formed between industrial sectors based on economic relations. The industrial chain is based on various factors such as resource allocation and specialization, forming upstream and downstream relationships for value exchange. Upstream enterprises provide products and services to downstream enterprises while receiving feedback, thus establishing an interactive relationship. Correlation analysis between the various links of the industrial chain is of significant reference value for industrial layout management and portfolio design. Based on the concept of CE, Wei \cite{wei2021} proposed the concept of Pair-Copula entropy to measure the pairwise correlation within multiple variables. She applied this concept to a study on the correlation between various links in the domestic livestock and poultry farming industrial chain. Based on stock price data of nine major listed companies in the upstream, midstream, and downstream sectors of this field, she used Pair-Copula entropy to measure the correlation between upstream, midstream, and downstream sectors within the industrial chain. The study found that the upstream sector had strong correlations, while the downstream sector had weak correlations; unconditional correlations were strong, while conditional correlations were weak; and there were strong correlations between upstream and midstream sectors.

Investor sentiment has a wide-ranging and multifaceted impact on financial markets, and investor sentiment analysis is one of the important issues in economic research. Due to the integration of social media and market relationships, investor sentiment spreads between groups and countries, forming a transmission network that allows localized sentiment fluctuations to rapidly spread and cause systemic effects. Han and Zhou\cite{Han2022} proposed a method based on a combination of wavelet analysis, TE, and network analysis to study the patterns of investor sentiment transmission among companies, employing the TE estimation method based on CE. They used Baidu search index data of 137 listed new energy vehicle companies in China from 2015 to 2021 to represent investor sentiment, decomposing it into multi-scale information using wavelet analysis, constructing a sentiment transmission network using TE, and finally analyzing short-term and long-term transmission characteristics using network analysis. They found that investor sentiment exhibits short-term localized activity and a continuous and gradually increasing evolutionary pattern.

Inflation expectations directly influence the economic behavior of market participants and are one of the causes of inflation. Studying the relationship between inflation and expectations is an important topic, particularly valuable for central bank policymakers. Ardakani\cite{Ardakani2024} proposed using the CE method to analyze the information content of expectations on inflation, proving that negative Fisher Information (FI) is the lower bound of CE and can serve as a minimum measure of the relationship between inflation and expectations. Using tools such as CE, he analyzed monthly inflation indices (CPI and PPI) and inflation expectation indices (University of Michigan Survey Index, Cleveland Federal Reserve Bank 2-year, 10-year, and 30-year expectation indices) in the United States from 1982 to 2022, finding that the CE is smallest between 30-year expectations and inflation, indicating that it provides more information for predicting inflation. This research provides a powerful tool for central banks to manage expectations to achieve inflation targets, helping to understand the predictive power of different expectations on inflation, thereby enabling more effective inflation control.

\section{Management}
Accurate prediction of agricultural futures prices helps provide a reference for scientific decision-making by relevant government departments, and is therefore of great significance to ensuring national food security. However, price prediction is affected by a variety of complex factors, such as the international situation and market sentiment. Therefore, identifying the influencing factors of prices is crucial for building accurate price prediction models. An et al.\cite{An2023} proposed a hybrid prediction framework based on historical data and text data, which integrates multiple methods. Empirical Mode Decomposition (EMD) is used to preprocess historical data, Dynamic Topic Model (DTM) and sentiment analysis are used to extract information from Weibo text. Then, methods such as CE are used to screen the extracted factors for constructing the prediction model. The authors verified the proposed framework on two real-world datasets: pork price data from the National Bureau of Statistics and soybean futures price data from the Dalian Commodity Exchange, and collected Weibo text data within the corresponding time periods. In the experiment, the authors compared the CE method with similar methods such as dCor and HSIC. The results showed that the CE based prediction model gave the best prediction performance on both datasets.

Brazil is the world's largest producer and exporter of sugar, with a sugarcane cultivation history dating back to 1532. Sugarcane is cultivated almost throughout the entire country, with the central-southern region and the northeast being the main production areas. S\~ao Paulo state leads in both planted area and sugarcane derivative production. Simultaneously, Brazil uses over half of its sugarcane to produce anhydrous and hydrous ethanol for vehicle fuel, aiming to reduce the country's dependence on imported oil. It is the world's second-largest producer and third-largest consumer of fuel ethanol. Therefore, Brazil's fuel market, consisting of gasoline and ethanol, is closely related to the agricultural market of sugarcane production and is influenced by various natural, economic, and social factors. The price relationships among these three commodities are complex, and analyzing these relationships is of significant reference value to Brazilian market managers and participants. Flores\cite{Flores2025master} used the CE based TE method to analyze time-series data on the prices, returns, and volatility of these three commodities from May 2004 to November 2023, analyzing the dynamic changes in the relationship between the three commodities before, during, and after the COVID-19 pandemic. Analysis revealed distinct dynamic relationships across the three time periods: pre-pandemic, the market was stable and predictable, with moderate interaction among commodity prices; during the pandemic, commodity price volatility and interrelationships significantly increased; and post-pandemic, the market stabilized in a new equilibrium, with the dynamics of the relationship influenced by geopolitics and energy policies. The study also showed that gasoline and ethanol influence each other, with gasoline having a greater impact on ethanol, consistent with ethanol's status as a gasoline substitute. Simultaneously, ethanol's impact on sugar was stronger than gasoline's, as gasoline indirectly affects sugar through ethanol. The TE method based on CE successfully revealed the relationship between the prices of sugar, ethanol, and gasoline in Brazil, as well as the impact of the pandemic on this relationship, demonstrating its power as a dynamic economic relationship analysis tool.

Inventory management is a crucial aspect of enterprise operations and management, and a significant issue in management science. The newsboy problem, a typical single-cycle inventory management model, has long been a focus of research in this field. In recent years, research on the newsboy problem using data-driven models and methods has demonstrated superiority over traditional approaches, thus becoming a hot topic. Tian and Zhang\cite{Tian2023} proposed an end-to-end algorithmic framework that uses a deep learning model to predict order quantities from feature data such as online product reviews, employing methods including CE to select input features for the model. They applied their method to the automotive inventory management problem, building a model based on historical sales data of Volkswagen Lavida vehicles from 2016 to 2022, reviews from a certain website, a search engine index, and macroeconomic indices. The results show that this method can significantly reduce the sum of excess costs and shortage costs, reducing costs by 31.8\% compared to similar methods.

Chinese enterprises face both opportunities and challenges in overseas mergers and acquisitions (M\&A). Exploring the various domestic and international factors influencing Chinese enterprises' overseas M\&A and analyzing the short-term and medium-to-long-term performance of M\&A is of significant theoretical and practical importance. Wang et al. \cite{Wang2022b,Wang2025b} proposed using the Copula VECM model to analyze the impact of economic variables strongly correlated with the number of overseas M\&A transactions, specifically considering the dynamic impact of macroeconomic variables neglected by other researchers. Because there are many such economic variables, the complexity of the constructed VAR model can easily increase, leading to inaccuracies in the estimation model. Therefore, they proposed using CE to select economic variables before building the model. They selected quarterly data on the number of overseas M\&A transactions and seven other macroeconomic variables potentially correlated with the number of M\&A transactions from the Wind database. Through CE correlation analysis, they concluded that macroeconomic leverage ratio, GDP, money supply growth rate, and exchange rate are four macroeconomic factors that cannot be ignored in influencing the M\&A activities of Chinese enterprises. They further analyzed and discussed the inherent economic logic of the selected variables' impact on the number of M\&A transactions, enhancing the model's rationality.

\section{Sociology}
Gender inequality is a key issue in sociological research. From a gender perspective, we can identify many inequalities, such as those between the sexes in income, education, and occupation. Analyzing and identifying the sociological factors that lead to inequality is a concern for scholars, and using quantitative methods to analyze relevant sociological data is one research approach. However, the causal chains between various social factors are highly complex, requiring the use of scientific data analysis tools. Ma\cite{ma2022causal} proposed a multi-domain causality identification method, treating gender as an external social variable, transforming the inequality problem into a domain transfer problem in data analysis, and using conditional independence tests based on CE to discover causal relationships between social variables. He applied this method to data from the U.S. National Adult Income Survey, analyzing the causal chain between gender, education, and income, and finding scientific evidence that gender leads to educational inequality, which in turn causes income inequality.

\section{Pedagogy}
There are inherent connections between various subjects in high school education. The curriculum emphasizes the fundamental role of mathematics in subjects such as physics, chemistry, and biology. Mathematical knowledge, mathematical thinking, and methodologies profoundly influence the teaching of other subjects. Therefore, mathematics scores are considered to be correlated with scores in other subjects. Utilizing empirical methods to study the relationship between mathematics and other subjects, and analyzing the correlation between mathematics scores and other scores, is an important fundamental issue with universal reference value for teaching reform and the selection of learning methods. Based on the final exam scores of science students in the first and second years of high school and two mock exam scores in the third year of high school in a certain city in 2013, Liu\cite{liu2018master} studies the correlation between mathematics scores and scores in other subjects. The author compares three correlation measurement methods: classical linear correlation coefficient, rank correlation coefficient, and MI. From the perspective of the theoretical relationship between CE and MI, the paper analyzes and demonstrates the superiority of the MI measure, and experimentally proves that the MI measure can better characterize and reveal the influence mechanism of mathematics on other subjects (Chinese, English, physics, chemistry, and biology, etc.).

\section{Computational Linguistics}
City service hotlines are an important component of the government's public management system, promoting communication between the government and citizens and improving public services. However, traditional manual dispatching methods cannot meet the ever-increasing demand for hotlines. How to efficiently and quickly process a large volume of citizen hotline requests is a crucial issue for improving the service quality of city service hotlines. The accumulation of large amounts of hotline text data makes it possible to quickly filter and process hotline requests. Natural language processing (NLP) methods can be used to process hotline text data, thereby constructing an intelligent dispatching system. Chen et al.\cite{Chen2023patent} proposed a city hotline dispatching method based on knowledge graph technology. This method constructs a hotline knowledge graph based on city hotline data, and then dispatches requests based on the search results of the constructed knowledge graph, significantly improving the efficiency of hotline services. In this intelligent dispatching system, CE is used as a feature selection method to preprocess city hotline data to construct and update the knowledge graph. The results show that CE outperforms other similar methods. The authors applied this method to the Jinan citizen service hotline system, and by continuously updating the knowledge graph, ultimately achieved a dispatching accuracy rate of over 90\%.

Word embedding is a fundamental technique in NLP. It maps words to a semantic vector space, ensuring that semantically similar vectors are also close in distance within the vector space. High-quality word embeddings are the goal of NLP model training and directly impact the quality of downstream NLP tasks. Therefore, measuring word embedding quality is a crucial NLP problem. As a model-independent metric, MI is used to evaluate word embedding quality, calculating how much original semantic input information is preserved in the embedding space. Estimating MI becomes key to solving this problem. Chen et al. \cite{Chen2025b} proposed a Vector Copula-based MI estimation method by extending CE theory. This method first estimates the Vector Copula and then calculates MI based on the equivalence between CE and MI. They applied their method to the problem of evaluating the quality of word embeddings in language models. They composed paired text evaluation data from the IMDB movie review dataset containing both positive and negative reviews, then calculated the word embeddings of two models, Llama-3 and BERT, on this dataset. Next, they used an Autoencoder model to losslessly map the original word embedding vectors to a 16-dimensional vector space. Finally, they used their proposed method to calculate the MI between these 16-dimensional paired vector sets. They compared their proposed MI estimation method with similar methods, and experimental results showed that their method significantly outperformed the comparison methods in MI estimation of word embeddings from both language models, demonstrating the effectiveness and superiority of their proposed method.

\section{Media science}
How public health emergencies affect public sentiment is a crucial question with both theoretical and practical significance, offering valuable insights for government information dissemination and public opinion management. Particularly in the new media environment, the spread and evolution of public sentiment are influenced by multiple factors, making it even more complex. The COVID-19 pandemic has provided an opportunity to study this issue. Zhang et al. \cite{Zhang2022a} investigated the characteristics and mechanisms of the impact of the COVID-19 pandemic on public sentiment during the outbreak in Shanghai. Based on data from the Weibo platform related to the ``Shanghai epidemic", they studied the influencing factors, temporal evolution, and causal relationship between the pandemic and public sentiment. The study utilized the CE based TE method to analyze the causal relationship between the pandemic and public sentiment, empirically finding that the causal effect of the pandemic on negative public sentiment was greater than that on positive sentiment, and that positive sentiment had an inhibitory effect on negative sentiment.

\section{Law}
Communities are basic units of social life, and community security management is closely related to everyone's life. There is an inherent link between community attributes and community crime. Analyzing the relationship between community economic, social, and demographic attributes and various types of crime can deepen our understanding of criminal behavior and provide important reference for law enforcement agencies to rationally allocate and deploy resources. Wieser\cite{Wieser2020} proposed a new information bottleneck estimation method based on the equivalence relationship between CE and MI. Due to the utilization of the transformation invariance of CE, this method has better estimation performance than traditional methods of the same kind. He applied this method to the US community and crime dataset, analyzing the relationship between 125 economic and social factors and 18 crime attributes (including 8 types of criminal behavior, per capita crime rate, and per capita (non-)violent crime rate), and learned a latent variable model that can represent this relationship, providing a reference for constructing crime prediction models.

\section{Political science}
Political security is crucial to national security. Political science research focuses on the relationship between leadership factors and regime crises, and uses this information to allocate resources for intelligence gathering, stabilization, or overthrowing actions. Based on the International Political Leadership Dataset from Syracuse University's Moynihan Institute for Global Affairs, Card\cite{Card2011} studied the nonlinear relationships between 37 leadership factors and political security, using CE (MI) as a nonlinear analysis tool. The study focused on the relationship between two leadership variables (the reasons for regime establishment and the reasons for regime termination) and other factors. The analysis results corroborate existing sociological theories, confirm known relationships, and uncover previously unknown relationships and phenomena.

\section{Military science}
Timely and accurate identification of target intent is a crucial aspect of battlefield situational awareness and a foundation and prerequisite for command and decision-making. Identifying the intent of aerial targets faces numerous uncertainties, such as uncertainties in behavioral and physical characteristics, flight rules, and operational capabilities, making timely and accurate intent identification extremely difficult. Zhang et al.\cite{Zhang2022patent} proposed a target intent identification method based on dynamic Bayesian networks to identify intent from time-series data of targets in complex situations. The method utilizes a CE based MI estimation algorithm to generate a Bayesian network structure from target attributes and intent data, then uses an adaptive genetic algorithm to iteratively optimize the network structure, and finally uses the optimized network to identify the intent of unknown targets. They applied this method to the processing of aerial targets, using the target's position information, flight information, and radar and communication system information to identify six different intents (patrol, early warning/command, electronic reconnaissance, electronic jamming, attack, and strike). This method is not limited to aerial targets and can be easily extended to other types of targets.

\section{Informatics}
Disruptive technologies are original and innovative technologies that can transform existing mainstream technologies and industries, driving transformative progress in the economy and society. Conducting forward-looking identification and prediction research on disruptive technologies is a crucial issue in the field of science and technology intelligence analysis, providing guidance for science and technology policy formulation, science and technology industry layout, and the cultivation of a science and technology innovation ecosystem. Research on science, technology, and industry interaction patterns based on knowledge network analysis is one approach to solving the identification and judgment problem. Xu et al.\cite{Xu2023b} proposed a disruptive technology research process framework. Using incremental technologies as a reference, they acquired textual data from scientific, patent, and industry literature. They constructed knowledge networks for each of the three using natural language processing techniques. Then, using three overall network attributes and network community similarity attributes of the knowledge networks, they divided the knowledge network interaction patterns into five pre-defined modes, including a science-technology-industry linkage mode. Among these, CE was used to measure the correlation between the overall network attributes of the three knowledge networks to characterize the interaction modes. They conducted empirical research using regenerative medicine (stem cells) as a disruptive technology and leukemia treatment as an incremental technology reference. They obtained relevant textual data from authoritative databases up to the end of 2020 and used this process framework to study the commonalities and differences in the science-technology-industry interaction models of the two comparative fields, deepening their understanding of the knowledge flow and diffusion patterns of disruptive technology innovation ecosystem elements.

\section{Energy}
Weather is a crucial influencing factor in energy systems, directly impacting both energy production and consumption. Especially when renewable energy is integrated into the energy system, weather factors such as wind speed and solar radiation determine the production capacity of wind and solar power, while temperature changes affect residents' energy consumption demands. However, natural systems possess significant uncertainty, posing challenges to the stable and efficient operation of renewable energy systems. Therefore, renewable energy network management systems need to establish reasonable models to integrate renewable energy sources into the network. Information theory provides a tool for managing the uncertainty of weather systems. Fu et al.\cite{fu2017uncertainty} studied a method for establishing weather models in integrated energy systems based on information theory. The authors used Copula functions to establish a joint distribution model of weather variables and employed the MI calculated by CE as an evaluation index for model accuracy to guide the modeling process. MI was also used to measure the correlation strength between various energy outputs. The authors used the obtained integrated energy system model to simulate the operation of an energy system in a region of northern China and compared it with actual data. The results show that the simulation of the system model basically matches the actual situation, indicating that the constructed weather model can meet the operational needs of the energy management system.

Photovoltaic power generation technology is highly uncertain due to environmental factors such as weather, which impacts the safe and stable operation of the power grid. Forecasting the active power of photovoltaic power plants based on meteorological conditions and other factors helps grid dispatchers better formulate dispatch strategies and address the impact threats posed by the uncertainty of photovoltaic power generation. Zhu and Zhang \cite{Zhu2022a} proposed a method combining optimization algorithms, mode decomposition, CE, and deep learning models to improve the accuracy of power generation prediction. They compared their method with several similar methods on photovoltaic power plant data from Yulara, Australia, demonstrating that the model obtained by this method better adapts to the impact of weather changes and achieves the best prediction results.

Under conditions of significant weather changes, photovoltaic (PV) power output can fluctuate dramatically, leading to a decrease in the accuracy of day-ahead forecasts. To address this, Yang et al. \cite{Yang2025b} proposed a day-ahead PV power forecasting method based on weather type classification. First, weather is categorized into transitional weather days and various stable weather day types. Then, TE based on CE is used to obtain meteorological factors with causal relationships to PV power under different weather types. Finally, a fusion model of multiple neural networks is used for point and interval forecasts. They conducted a forecasting experiment using this method based on power and meteorological data from a PV power plant in Songjiang District, Shanghai, in 2021. The experiment first categorized weather into transitional weather days and four types of stable weather days (corresponding to sunny winter weather, cloudy and rainy winter weather, sunny summer weather, and cloudy and rainy summer weather, respectively). Then, TE was used to select meteorological factors related to different weather types. Finally, the forecasting model of this method was compared with three similar models. Experimental results show that, based on TE analysis, meteorological factors that have a causal relationship between transitional weather days and photovoltaic power include wind level, precipitation, and surface solar radiation intensity. Meteorological factors related to sunny weather include cloud cover and surface solar radiation intensity, while meteorological factors related to cloudy and rainy weather include wind level, precipitation, surface solar radiation intensity, and air pressure. The model built on this causal relationship performs better than the comparative model under different weather types.

Photovoltaic modules are the core equipment of solar power plants, and their operating status directly affects the plant's efficiency and stability, such as power generation, cleanliness, and fault conditions. Therefore, predicting the status of photovoltaic modules is one of the core issues in photovoltaic power plant management. Traditional status prediction methods mostly rely on information from individual modules, neglecting the connections between modules, leading to low prediction accuracy. Wang and Zou \cite{Wang2024p2} proposed a photovoltaic module status method that considers both the influencing factors of individual module status and the connections between modules to construct a state graph neural network prediction model. Specifically, they used CE to estimate the correlations between different influencing factors and between status and influencing factors, constructing the topological relationships between variables in the prediction model. They validated the effectiveness of the method based on a sample set of photovoltaic modules.

Wind energy, as a major clean energy source, is characterized by intermittency and uncertainty, making power prediction and control of wind turbines highly complex. Analyzing the correlation characteristics between various variables within the wind turbine based on monitoring data helps in monitoring the turbine's health status and predicting wind power, thereby better utilizing wind energy resources. Cui and Sun \cite{Cui2022} proposed using CE to analyze the correlation between state variables of wind turbines, and then performing clustering based on CE correlations to obtain the turbine's operating conditions. They applied the method to data from a SCADA system of an offshore wind farm in Guangdong, finding that the CE method better describes the correlations in the data than traditional methods, and using the K-means method to obtain an accurate classification of operating conditions reflecting the wind turbine's operating characteristics and status, which has significant practical implications.

Electricity load forecasting, which predicts electricity consumption over a future period based on historical data, is crucial for smart grid dispatching and power transmission planning. Electricity load is influenced by various factors, exhibiting periodicity and seasonality, and is particularly significantly affected by weather. Therefore, constructing accurate electricity load forecasting models requires considering multiple factors, including weather, and analyzing the characteristics of weather's impact on load. Ma\cite{Ma2023} proposed using the CE based TE method to analyze the time delay characteristics of dynamic systems. The method is applied to electricity consumption data in T\'etouan, Morocco, analyzing the impact of five weather factors on the load of the city's three power supply networks from a time delay perspective, revealing the daily time delay variation characteristics of these impacts. Yan et al. \cite{Yan2024} proposed a comprehensive short-term energy load forecasting method combining clustering algorithms, prediction algorithms, and ensemble learning. First, the data is clustered based on load characteristics. Then, for each data cluster, the CE based TE algorithm is used to analyze and select external factors (including weather and time) that influence the load. Finally, an ensemble learning algorithm is used to forecast the load. They applied the method to 2018 residential building energy load data from Arizona, USA, to predict four types of loads: electricity, gas, cooling, and heating. Experimental results show that the external factors selected using the CE based TE algorithm achieve the best predictive performance in the forecasting model, significantly outperforming other methods that select relevant variables. This is because TE can accurately measure the time-series nonlinear relationship between external factors and load. Kan \cite{Kan2023} proposed a deep learning-based method for short-term forecasting of multi-element energy loads. First, VMD (Virtual Dynamics Decomposition) is used to decompose the multi-element load. Then, CE is used to calculate the connection strength between the decomposed IMF components and load influencing factors, which is used as the adjacency matrix weights of the graph convolutional network. The resulting temporal coupling features are then input into an LSTM model. The output of this model is then multiplied by the output of another Transformer model to obtain the final prediction result. He validated the effectiveness of the method on data from Arizona State University's Tempe campus, finding that CE can effectively calculate the coupling strength between meteorological and temporal factors and the various components of cold, heat, and electricity loads, increasing the model's interpretability. Hu \cite{Hu2022master} proposed a multi-node coincidence prediction method based on multi-task learning. He first analyzed the correlation and causal relationships between node loads, especially the relationship between node loads and total load. Specifically, he used the TE method based on CE to explore the nonlinear causal relationship between node loads and total load. Based on a 24-hour causal analysis of substation load datasets from nine regions of the New Zealand distribution system, he found that the causal relationship between node loads and total load is a fluctuating causal process, indicating that the total load contains information about changes in node loads. Wu et al.\cite{Wu2024p} proposed a distribution network load prediction method based on hybrid machine learning. This method combines multiple methods such as CE and graph neural networks, where CE is used to calculate the correlation strength between grid loads and their influencing factors to construct the graph neural network structure. Wang \cite{Wang2024master} proposed a parallel CNN-GRU attention model for power load prediction based on multivariate time series data. The method uses multiple methods, including CE, to screen the input factors of the model. He validated the method based on grid data and corresponding meteorological data from a certain region in China, obtaining better prediction performance than the comparison methods. Among them, the CE value not only measures the statistical correlation between input variables and loads, but also reflects information transmission and energy exchange in the underlying system, providing richer information than the traditional correlation coefficient. Tang \cite{Tang2024master} proposed a source-load prediction method based on a combination of feature selection, mode decomposition, and neural networks, utilizing methods such as CE to select nonlinear correlation features. He validated the effectiveness and superiority of the proposed method on the Australian photovoltaic dataset and the Panama load dataset.

Renewable wind and solar energy are increasingly becoming an important component of the power energy supply. Ensuring the economic benefits and reliable safety of wind and solar power integration is a major concern in renewable energy utilization. Rational planning is crucial for addressing this concern, guaranteeing a return on investment and ensuring the system's proper operation, while preventing the curtailment of wind and solar energy. Energy storage systems can mitigate the instability and volatility of wind and solar energy and are an integral part of wind and solar system planning. Dong et al. \cite{Dong2022} proposed a wind-solar-storage collaborative planning and configuration method that considers the temporal similarity between source and load. This method uses CE to measure the similarity between wind and solar energy and the load, thereby improving the system's wind and solar energy utilization efficiency. They applied this method to the planning and configuration of a combined wind, solar, thermal, and storage power generation system in an industrial park. The results show that this method can effectively reduce the installed capacity of the energy storage system, improve the absorption capacity of new energy sources, and achieve significant economic and emission reduction benefits.

Frequency is one of the most important physical quantities in a power system, and frequency stability is a fundamental requirement for ensuring the stability of power supply. The unpredictability of renewable energy sources and their large-scale integration into the grid pose challenges to grid frequency stability. To stabilize and control frequency fluctuations caused by renewable energy sources, accurate and rapid prediction of system frequency stability is necessary to help system operators formulate control strategies in advance. Traditional frequency stability prediction is model-driven, but online prediction is impossible due to the time-consuming solution process. Machine learning-based model methods, by simplifying the model to improve computational efficiency, can meet the needs of online prediction. Liu et al.\cite{Liu2022a,Liu2023m} proposed a frequency stability prediction method combining deep learning and CE. CE is used to select model input variables, reducing redundant information and improving computational efficiency. The authors applied the method to two systems: one is the New England 39-node system, integrating the dynamic wind farm model of the Western Electric Dispatch Council; the other is the ACTIVSg500 system based on the grid system in western Southern California. Experiments show that the model built using this method achieves the best performance compared to similar models, meeting practical requirements. The CE method not only simplifies the model and significantly reduces computation time, but also identifies power grid variables related to frequency stability, making the model interpretable.

Broadband oscillations in power systems are triggered by the dynamic interactions of power electronic devices. Their propagation in the power grid can cause chain reactions, seriously endangering the safe operation of the grid. The excitation mechanism of broadband oscillations is complex, exhibiting significant time-varying, nonlinear, and wide-area propagation characteristics, making effective modeling and analysis difficult. Feng et al.\cite{Feng2022,Feng2023,Lu2024} proposed a method for analyzing the influencing factors and propagation paths of broadband oscillations, utilizing the model-free property of CE. This method uses the system's operating state parameters as random variables and selects key factors influencing oscillations by calculating the CE between these parameters and the oscillation damping in each frequency range. Simultaneously, it uses data from when the system oscillates to calculate the copula transfer entropy network between system variables, used to analyze the oscillation propagation process and source location. This analysis method is data-driven and can obtain corresponding analytical results even when the system model is unknown. The authors simulated a direct-drive wind turbine grid-connected system and a four-unit, two-zone system containing a wind farm, analyzing the causal relationships of oscillations between various components within the controller and between different buses in the complex system. Simulation results show that this method can accurately determine the propagation path and source location of broadband oscillations at both the device and network levels, providing support for the study of oscillation propagation mechanisms and a reference for further oscillation suppression measures. Sun et al.\cite{Sun2023,Wang2024p} also proposed a method for identifying broadband oscillation risks in AC/DC hybrid systems using CE. This method identifies risks by analyzing the CE between oscillation influencing factors and the damping variables of oscillation modes within each sub-frequency range. They used this method to analyze the oscillation risk of a provincial power grid system under small disturbances, and used the LCC model to identify key influencing factors such as rectifier control parameters and DC transmission power, providing an accurate and reliable basis for subsequent design of targeted adjustment schemes to suppress oscillations.

The integration of renewable energy sources introduces dynamism and uncertainty into the power grid, increasing the complexity of voltage coordination control and posing new challenges to grid operation. Voltage control requires distributed coordination control of multiple transformers, but communication networks suffer from information transmission problems such as delay, jitter, and packet loss, adversely affecting this distributed coordination control, leading to performance degradation and even system instability risks. Therefore, studying voltage coordination control under communication uncertainty is an important issue. Due to network channel characteristics and congestion, there is a nonlinear correlation between delay, jitter, and packet loss, which has not been fully considered in previous studies, leading to possible suboptimal control or even system instability. Yang et al.\cite{Yang2024} proposed an event-triggered sliding mode control strategy under multivariate correlated communication uncertainty events for voltage coordination control of the power grid. The authors proposed using CE to measure the uncertainty correlation among delay, jitter, and packet loss events, and then obtaining the joint entropy to measure the total communication event uncertainty. This total uncertainty is then used for the design of the sliding mode control state observer, leading to the control law of the sliding mode controller. They conducted simulations and physical experiments on a 7-bus distribution system and an IEEE 69-bus distribution system in the Jiangsu power grid. The proposed sliding mode control strategy was compared with three similar methods. The results showed that the proposed algorithm exhibited good control performance and robustness, and all indicators were better than the comparison methods, making it more suitable for coordinated voltage control of the power grid in complex communication environments.

Line loss rate is a crucial economic and technical indicator for power energy enterprises, measuring their economic efficiency. Therefore, line loss management and abnormal line loss investigation are essential tasks for the power sector. Line loss analysis utilizes scientific calculation methods to analyze the distribution patterns of line losses in the power grid, providing efficient and accurate decision support for management. Hu et al. \cite{Hu2022} proposed a line loss analysis method based on TE , which estimates each user's contribution to the overall line loss by calculating the TE value based on CE. Based on daily power supply and line loss data, they ranked users according to their line loss contribution, applying this method to practical line loss management to reduce the overall line loss rate.

Distribution network topology identification is a crucial issue in power grid system analysis, providing a foundation for distribution network management functions such as power flow calculation, grid state estimation, reactive power optimization and regulation, and network reconfiguration. With the large-scale integration of distributed energy resources into distribution networks, their volatility and uncertainty lead to more varied system topology reconfigurations, posing new challenges to topology identification. Qin and Pan \cite{Qin2023} proposed a novel distribution network topology identification method, transforming the identification problem into a sub-problem of identifying the states of multiple switching nodes based on spatiotemporal correlation. This method first utilizes CE and Markov chains to extract the spatial and temporal nonlinear correlation features between node voltage sequences, respectively. Based on this, a model capable of identifying individual switch state change sequences is obtained. Finally, the network topology structure identification is completed by combining the results of multiple such switch state identifications over a certain period. They simulated a distribution network with dynamically changing topologies connected to wind turbines and photovoltaics, generating a 120-day simulation of household loads in the distribution network. The proposed method was tested based on the network node measurement data. The results show that CE can effectively analyze the correlation between node voltages, enabling the method to effectively identify the network topology structure in a short time.

Electricity price forecasting is crucial for decision-making in the electricity market, helping participants develop trading strategies and allocate resources effectively. However, the widespread use of renewable energy sources introduces uncertainty into electricity supply, complicating electricity price forecasting and making predictive model construction more challenging. Xiong and Qing\cite{Xiong2022} proposed a hybrid electricity price forecasting framework based on time-series data, combining a feature selection method based on CE with signal decomposition, Bayesian optimization, and an LSTM model to construct the predictive model. They apply the method to 2017 data from the Pennsylvania-New Jersey-Maryland Interconnected (PJM) electricity market, demonstrating its effectiveness and practicality.

Lithium-ion batteries are the most widely used green and clean energy source. However, the capacity of lithium-ion batteries degrades with repeated use, making battery health monitoring one of the main problems in battery management systems. Traditional health monitoring models are mostly derived under single load conditions, making them unsuitable for various real-world scenarios and resulting in models based on raw data failing to adapt to new situations. To address this issue, Hu and Wu \cite{He2023} proposed a battery capacity estimation method based on transfer learning, combining causal analysis, attention mechanisms, and LSTM. CE based TE is used to select health state indicators related to capacity degradation, ensuring the transferability of the model under different conditions. The authors applied the method to lithium-ion battery degradation data under three NASA load conditions. The results show that the model based on causal analysis improves the cross-condition prediction accuracy by 8.6\% and 12.4\% compared to models based on two traditional methods, respectively, enhancing the model's robustness. Zhao\cite{Zhao2025master} proposed a method for predicting battery health status. This method includes a multi-dimensional feature selection framework combining multiple methods such as CE, a model construction and optimization method combining deep learning time-series prediction models and ensemble learning methods, and a model transfer learning method adapted to different batteries and operating conditions. He validated the method based on the publicly available CALCE and NASA lithium battery degradation datasets, and the results show that the prediction error of the model obtained by this method is within 2\%. He specifically conducted experimental analysis on the redundancy and sufficiency of the proposed feature selection method, and the results show that the model given by the feature selection framework has good generalization and robustness to adapt to different operating conditions and different batteries, and still has excellent prediction performance even with limited data.

Energy efficiency is one of the main goals of Industry 4.0, and the digitalization of production systems provides a significant opportunity to improve the energy efficiency of industrial equipment. Energy efficiency anomalies are a key to improving energy efficiency; identifying anomalies and determining their causes is an effective way to improve energy efficiency. However, industrial systems generally have complex structures and operating mechanisms, making it difficult to analyze the root causes of energy efficiency anomalies using traditional modeling methods. Ma \cite{Ma2024} proposed using TE for root cause analysis of energy efficiency anomalies. Addressing the non-stationarity of industrial systems, he presented a diagnostic method for energy efficiency anomalies called TE flow, which employs the TE estimation method based on CE. Since TE is model-free, this method can perform root cause analysis of energy efficiency anomalies for various devices without considering their underlying mechanisms. He applied this method to an air compressor system, successfully describing the causal relationships of the system's operation and identifying the air compressor subsystem that caused the system's abnormal energy efficiency state.

\section{Textile engineering}
Cotton gauze, as a basic textile product and a commonly used medical device, is widely used in medical services such as aseptic isolation and hemostasis during surgery. Its production quality directly affects the effectiveness of medical services and patient safety. Therefore, quality control in the gauze production process is subject to strict supervision by government departments. Sealing force refers to the sealing strength during the packaging process of gauze production. It is crucial for maintaining the sterility and preventing contamination of gauze before opening and use, and is therefore an important quality control parameter in the gauze production process. Mortezanejad et al. \cite{Mortezanejad2025} proposed a nonparametric multivariate quality control chart method based on CE for quality control of multidimensional non-normally distributed variables. This method first uses the maximum CE method to estimate the multidimensional distribution function, and then uses the Hotelling $T^2$ statistic to generate a control chart. They applied this method to the sealing force data of the absorbent cotton gauze production process at Mega Company, generating a Copula-based quality control chart. This chart can detect subtle quality changes that are difficult to detect using traditional quality control chart methods, thus demonstrating that this method can control the quality parameters of cotton gauze production.

\section{Food engineering}
As an agricultural product, wine is increasingly becoming accessible to ordinary consumers. Wine quality assessment is crucial for both production and sales, with the wine industry investing heavily in quality evaluation to improve brewing techniques and promote consumption. Traditional quality assessment relies primarily on physicochemical testing and expert opinions; however, expert taste perception is highly subjective, and its underlying mechanisms are difficult to understand. Therefore, it is necessary to study the intrinsic relationship between wine components and expert evaluations to enhance understanding of wine quality and improve the objectivity of quality assessments. Lasserre et al.\cite{Lasserre2021a,Lasserre2022} proposed a causal relationship network learning algorithm, called CMIIC, using CE based (conditional) independence measurement estimation. They applied this algorithm to the quality evaluation data of the famous Portuguese Vinho Verde, identifying physicochemical components associated with the quality of both red and white wines.

\section{Civil engineering}
Building energy consumption accounts for about 40\% of total energy consumption, and building energy-saving technologies are important green energy technologies, which are of great significance to achieving the United Nations' carbon neutrality goal. Heating, ventilation, and air conditioning (HVAC) systems contribute more than 40\% of the energy consumption of commercial buildings and are one of the main research objects of building energy conservation. The operation of HVAC systems has time delay characteristics, which come from the hysteresis of medium conduction and thermal inertia. Understanding and utilizing this characteristic is beneficial to designing appropriate control strategies to achieve energy saving. Li et al.\cite{Li2022} introduced the TE method based on CE to HVAC, and developed a model-free time delay estimation method based on the information theory framework for time series prediction of HVAC systems. They improved the multivariate TE estimator of kNN and designed a time delay estimation algorithm by combining optimization methods. They applied the algorithm to the heating monitoring system of a four-story teaching building in Dalian, analyzed the data of indoor temperature and weather parameters (such as outdoor temperature, relative humidity, solar radiation, wind speed, etc.) and heating parameters (such as hot water supply and return temperature, etc.), identified time delay parameters, and then used the latter two sets of parameters to predict the room temperature for the next period of time. The results show that the TE method can identify the time delay relationship between parameters, thereby improving the room temperature prediction performance.

Reinforced concrete structures are a widely used building structure, and earthquakes can cause significant damage to them. Therefore, assessing the seismic performance of such structures is crucial to the safety of people's lives and property. Plastic rotation and deflection are generally used for seismic damage assessment and seismic performance classification of reinforced concrete beams. Therefore, evaluating the seismic performance limit of reinforced concrete beams is an important issue. However, different countries have different evaluation methods, making them unsuitable for describing structural leakage caused by earthquakes, especially leakage in underground structures. Ma et al.\cite{Ma2024b,Chi2024master} proposed a method for predicting the seismic performance level of reinforced concrete beams based on machine learning, considering crack development. They collected test results from 452 reinforced concrete beams, analyzed the results using PCC to obtain a set of predicted mechanical parameters, selected a set of mechanical parameters with a nonlinear relationship to the performance limit using a CE based MI method, and finally established a predictive model for the performance limit using seven machine learning methods. Experimental results show that PCC selects 22 out of 27 mechanical and dimensional parameters that have a linear relationship with the performance limit, while the MI method selects 5 parameters that have a nonlinear relationship with the performance limit. Based on the selected parameters, the authors constructed a limit prediction and hierarchical classification model for four performance limits, and obtained good prediction and classification results, verifying the effectiveness of the method.

Shear capacity is a crucial parameter in structural design, referring to a structure's ability to withstand shear forces. Its prediction is essential for assessing structural safety and stability. Existing shear capacity predictions for reinforced concrete columns are mostly based on models combining mechanistic and empirical methods. However, the actual shear performance mechanism is complex, with numerous nonlinear influencing factors, leading to low accuracy in existing models. Utilizing machine learning methods to construct shear performance prediction models for reinforced concrete members is an effective approach to address this problem. Chang et al. \cite{Chang2025} proposed a shear capacity prediction model construction method based on CE based feature selection. CE is used to measure the nonlinear relationship between shear force and its influencing factors, allowing for the selection of appropriate input variables for the machine learning prediction model. They collected shear test data for reinforced concrete columns containing 441 samples and validated the proposed prediction model construction method. Experimental results show that the CE method makes reasonable selections of shear capacity influencing factors, and the prediction accuracy of the model obtained based on this method is better than five traditional semi-empirical and semi-theoretical formulas, verifying the effectiveness and superiority of the proposed method.

As a basic building unit, the room plays a crucial role in realizing the building's functions, which stems from the coordinated operation of the equipment and facilities within the room. Understanding the mechanisms of earthquake damage, especially the damage consequences caused by the synergistic effects of elements within a room during an earthquake, is essential for post-earthquake assessment of building and room functionality. Copula theory, as a representation of correlations between variables, particularly the Vine Copula method, is a powerful tool for modeling the correlations between elements within a room. However, traditional Vine Copula construction methods rely on tools such as correlation coefficients, leading to an accumulation of uncertainties during the structural construction process. Furthermore, the implicit Gaussian assumptions made in this way do not reflect real-world conditions. Liu et al.\cite{Liu2025a,Liu2024phd} proposed a Vine Copula construction method based on CE and model selection, utilizing CE theory, to model the correlations between elements within a room, thereby assessing the post-earthquake functionality of rooms. Specifically, CE is used to construct the basic structure of the Vine Copula, separating structural learning and function estimation, thus increasing the reliability and accuracy of the Vine Copula model. Using hospital operating rooms as the research subject, they constructed a physical simulation system and conducted shaking table tests under four indoor scenarios with different vibration intensities. They collected visual, acceleration, and displacement data of medical instruments and facilities in the operating room during the tests. Using the experimental data, they established a damage state model for room elements, and then derived a functional failure state model for the room based on Vine Copula. They found that the simulation results of room system vulnerability estimated using this method were basically consistent with the shaking table test results, demonstrating the effectiveness of using the Copula function to model the synergistic effects of room elements for room system vulnerability assessment. This method is of great significance for the functional assessment of hospital buildings and can also be extended to room systems with more complex functional settings.

Tunnels are transportation structures built underground, underwater, or in mountains, and shield tunneling is a key technology in tunnel construction. However, shield tunneling is affected by factors such as geological conditions and mechanical equipment, causing deviations between the tunneling path and the designed route, thus affecting the progress and safety of tunnel construction. Current shield control relies mainly on human experience, which is insufficient to effectively cope with complex construction environments. Therefore, it is necessary to study design methods to predict shield axis deviation in advance. Lin \cite{Lin2023master} proposed a shield axis deviation prediction method based on CE, utilizing CE to select effective shield tunneling parameters for prediction. He verified the method based on tunneling data from a section of the Nanchang Metro Line 3 rail transit line, using CE and other methods to select 40 shield tunneling parameters and constructing four prediction models for cut-out horizontal deviation, cut-out elevation deviation, shield tail horizontal deviation, and shield tail elevation deviation. Analysis results show that compared with similar methods, the feature set selected by CE is more concise and consistent with the orientation of the shield deviation. Results on the training, test, and validation datasets all demonstrate that the CE based method delivers optimal prediction performance with the fewest possible features selected.

\section{Transportation}
Oversized cargo transportation refers to the specialized transportation operations of large, non-separable objects using multiple modes of transport. It plays a vital role in the national economy, providing crucial support and guarantees for infrastructure construction in key sectors related to national welfare and people's livelihoods, and is also related to national defense, military affairs, and national security. Oversized cargo transportation largely requires multimodal transport methods such as rail and air transport, necessitating the development of an integrated plan that links various local transportation modules. With the digitalization of transportation systems, a large amount of relevant plan data has accumulated, making data-driven oversized cargo transportation plan development an important issue. Research in this area can help improve the scientific rigor and applicability of plan development. Huang \cite{Huang2021a} proposed a method for developing multimodal transport plans for oversized cargo based on module chain construction, utilizing mathematical tools such as CE. This method first decomposes the transportation plan into multiple local module components, then uses relevance measurement tools such as CE to select a set of module attributes for calculating the similarity between plans, and finally retrieves cases with high similarity to the target transportation task from an existing transportation case database as preliminary transportation plans. Due to the diversity of large-item transportation solutions, some case module attributes may exhibit non-Gaussianity, rendering traditional correlation coefficient tools unsuitable for calculating the correlation between attributes. However, CE remains applicable due to its universality. The authors validated the method on data from over 600 real-world cases and constructed a prototype system for solution development.

Air travel and high-speed rail are the two main modes of passenger transport in China. Compared to air travel, the marketization level of high-speed rail ticket prices is lagging behind, lacking flexibility and dynamism. Therefore, studying the factors influencing ticket prices in order to improve the pricing mechanism of high-speed rail tickets is a matter of great concern to the academic community. Xu et al.\cite{Xu2023,Ji2022phd}, based on data on air and high-speed rail ticket prices between Beijing and Shanghai, used tools such as CE and decision trees to study the impact of four types of factors—travel demand, passenger choice, travel efficiency, and travel route—on air and high-speed rail ticket prices. They found that the advance booking period has different degrees of impact on the two types of ticket prices, but the impact of travel time is relatively similar. These research findings have certain reference value for high-speed rail pricing.

Urban rail transit has become one of the main modes of transportation in major cities in China. Improving the management level and operational efficiency of urban rail transit systems is one of the important issues facing the transportation system. Urban traffic passenger flow analysis and forecasting can provide a basis for guiding normal passenger flow, managing abnormal passenger flow, and scheduling rail trains. Analyzing the interaction between rail transit and other modes of transportation such as buses and taxis based on travel record data helps to improve the effectiveness of rail transit passenger flow forecasting. Wang \cite{Wang2022a} proposed using correlation analysis and causal analysis to analyze passenger flow time series data to enhance the understanding of the relationship between passenger flows of different modes of transportation. Among them, the TE method based on CE was used for causal relationship analysis between passenger flows. He applied the method to the time series data of rail transit, bus, and taxi passenger flows at four stations in the Suzhou rail transit system from August 6 to 12, 2018. The causal analysis results showed that the impact of taxi passenger flow on rail transit entry passenger flow at Sanyuanfang and Donghuanlu stations had a 1-hour lag effect, while the lag effect at Dongfangzhimen station was 5 hours. This analysis provides important guidance for predicting passenger flow at rail transit stations.

Railway passenger flow forecasting is fundamental to railway passenger service management. Accurate forecasting can improve the unified scheduling of railway capacity, coordinate network resources, and enhance economic efficiency. However, passenger flow is influenced by both natural and social factors, making accurate forecasting challenging. As a typical time series forecasting problem, time series models are generally used for prediction. A key issue here is how to handle the nonlinear relationship between passenger flow and its external influencing factors. Chang and Song\cite{Chang2024,Chang2024master} proposed an improved Prophet passenger flow forecasting model, which utilizes CE to analyze the nonlinear relationship between weather and holiday factors and passenger flow. They conducted an experimental study using real railway passenger flow data from January 2015 to March 2016. Correlation analysis using CE revealed that the impact of weather factors on passenger flow is negligible. They further constructed and selected new holiday time series features based on the CE tool to improve forecasting performance. Experimental results show that using this improved Prophet model can improve the accuracy of passenger flow forecasting.

Fatigue driving is one of the main causes of traffic accidents. Fatigue detection technology can effectively improve driving safety by monitoring and warning of driver fatigue, and has important practical significance and social value. Detecting driver fatigue using electroencephalogram (EEG) signals is a major technical direction and has received extensive research. However, EEG signals have dynamic and nonlinear characteristics, and improving the accuracy of fatigue detection is a major challenge in this field. Zhou \cite{Zhou2024master} proposed a fatigue detection method based on topological EEG feature selection and fusion. First, the modal components of the EEG signal are obtained through empirical mode decomposition. Then, the functional brain network is calculated and topological features are extracted. The extracted features are further selected using the CE method, and the selected features are fused and input into a machine learning fatigue detection model. He validated the method based on publicly available driver fatigue EEG data, demonstrating its effectiveness with a fatigue detection recognition rate of 93.98$\pm$3.36\%. In the experiment, the CE feature selection method improved the recognition accuracy of all experimental models, proving the effectiveness and universality of the method.

\section{Manufacturing}
Product quality is the lifeblood of manufacturing. Injection molding is a rapidly developing industrial manufacturing technology with wide applications in aerospace, construction, and communications. The injection molding process involves multiple complex physical and chemical reactions, making it highly susceptible to external factors, thus ensuring the stability of plastic product quality is a challenge. Building product quality prediction models based on historical manufacturing process data is one way to improve product quality. However, building a model requires first selecting relevant process parameters as input to achieve better predictive performance. Sun et al.\cite{Sun2021} proposed using the CE method to select process parameter variables for constructing a quality prediction model and applied the method to real injection molding production process data from Foxconn, significantly improving the performance of quality prediction. Cai and Rong \cite{Cai2023} proposed a method to identify key factors affecting quality. First, they used CE to establish a correlation matrix between factors, and then used a network deconvolution method to eliminate indirect influences between factors, thereby identifying the key factors affecting quality. They applied the method to three datasets in the UCI machine learning library, and the results showed that this method can identify key factors more efficiently and achieve the highest prediction accuracy compared to similar methods. They then applied the method to actual data from the production of a thin-film transistor liquid crystal display, and the results showed that the method selected 154 factors from 1540 factors and achieved the best quality prediction accuracy.

The manufacturing of complex mechanical products comprises three stages: design, manufacturing, and assembly. As the final stage of product production, assembly builds upon the manufacturing processes of individual components to create high-precision products. Assembly quality control, based on the manufacturing quality of the components, ensures the overall quality of the finished product. Complex mechanical products involve numerous and interconnected components, making the assembly process intricate. Errors in the assembly quality of upstream stages can impact the quality of downstream stages. Wang \cite{wang2015phd} considered the correlation between upstream and downstream processes and quality control points in her assembly quality control approach. She used Copula to model the correlation between control points and used CE to measure this correlation, thus proposing an optimization method for assembly quality control point control valves. She applied this method to the assembly process of the cylinder head, a key component of a gasoline engine for JAC Motors, validating its effectiveness.

Modern industrial systems are becoming increasingly complex and automated, making industrial process monitoring more and more difficult. How to monitor system anomalies and discover their causes is a crucial problem with wide applications. Utilizing causal analysis to obtain complex causal relationship graphs within industrial systems helps to accurately identify the propagation paths of anomalies, enabling timely intervention. Dong et al. \cite{Dong2023} proposed a fault analysis framework combining dynamic PCA, TE, and LSTM, where CE based TE is used to analyze causal relationships within the system. The authors applied this method to the analysis of hot-rolled strip steel process data at Ansteel in Liaoning Province, successfully analyzing two faults and their causes in the process. The authors also compared the TE-based causal graph analysis method with similar Granger causal analysis methods, showing that the TE method can more accurately perform root cause analysis of faults. Liu et al.\cite{Liu2023,Wang2024} proposed a CE-DR-SVDD method for dynamic distributed process monitoring. First, the Louvain algorithm based on CE is used to group system variables; then, a dynamic recursive support vector data description algorithm is used to construct local monitoring modules; finally, Bayesian inference is used to fuse the local monitoring results to obtain the global monitoring results. He applied the method to experimental data from the Tennessee Eastman process and compared it with similar methods, finding that the method achieved the best detection results on 19 out of 21 simulated faults.

The sintering process (SP) is crucial in the steel industry, but it also consumes a significant amount of energy. Dynamically predicting the carbon consumption of SP helps conserve energy and reduce carbon emissions. Traditional SP modeling, based on certain assumptions, cannot adapt to the dynamic characteristics of the SP system. Data-driven machine learning models can overcome the shortcomings of traditional models. Hu et al.\cite{Hu2023} proposed a dynamic modeling framework that can automatically identify process operating conditions to predict carbon consumption. This framework combines the AKFCM clustering algorithm, CE based model selection, and a width-based learning model method. The authors validated the effectiveness of the method using real data from a steel company, demonstrating that CE can quickly capture complex correlation patterns in SP under different operating conditions, enabling this method to predict sintering carbon consumption more accurately than traditional methods.

Aero-engines are core components of aircraft, and their manufacturing processes are a concentrated reflection of the technological level of the aviation manufacturing industry. Turbine disks are one of the key core components of aero-engines, playing a decisive role in engine performance. Because they operate under extreme conditions of high pressure, high temperature, and high speed, they are prone to fatigue degradation, requiring high reliability. Therefore, their manufacturing process places extremely high demands on them. Improving the manufacturing quality of turbine disks through optimization of manufacturing process parameters is an effective way to obtain highly reliable components. Die forging is the core process in turbine disk manufacturing, determining the quality and performance of the turbine disk. Therefore, predicting and optimizing die forging process parameters has become an important issue. Li et al.\cite{Li2023master2} proposed a complete set of methods for predicting and optimizing turbine disk manufacturing process parameters, used to analyze manufacturing process data, predict manufacturing quality, and thus optimize and improve the manufacturing process. Among them, CE is used to calculate the correlation between process parameters and to cluster and group the process parameters to facilitate subsequent predictive modeling and parameter optimization. They verified the effectiveness of this method on the die forging process in the manufacturing of GH4169 alloy turbine disks, and the CE based clustering method successfully classified and reduced the process parameters.

\section{Reliability}
Degradation processes are prevalent in various engineering systems, leading to reduced system reliability and even failure, such as fatigue and corrosion of metallic materials and parameter drift in semiconductor devices. Degradation process modeling is one of the main technical means to evaluate the effectiveness and lifespan of systems and products. Due to the complexity of modern systems, many factors influence degradation processes, and these factors themselves have nonlinear characteristics and are interconnected, making degradation process modeling a fundamental challenge in reliability engineering. Ignoring the correlation between factors during modeling leads to model errors and reliability estimation errors. Traditional methods for measuring the correlation between factors mainly use linear correlation coefficients, which are difficult to handle complex relationships. Sun et al. \cite{sun2019a} proposed using Copula to model the relationships between process factors and using CE to measure the correlation between degradation process factors. They presented a parametric CE estimation method and successfully applied it to the degradation process analysis of microwave electronic components. The results show that this method can analyze degradation processes at different stages.

Grinding wheels are a key component of CNC grinding machines, used for grinding workpiece surfaces. Their physical wear directly affects machining quality and efficiency. Therefore, grinding wheel maintenance is crucial, and predictive maintenance is a critical issue. Cheng et al.\cite{Cheng2023,Cheng2025} proposed a feature selection method based on CE and maximum correlation minimum redundancy to construct a grinding wheel remaining life prediction model. Based on SCADA data of 55 parameters from five grinding machines on the Weifu High-Tech CPM2.2 camshaft production line, they compared various correlation feature selection methods and found that the CE based method can effectively calculate nonlinear characteristic relationships that traditional correlation methods cannot discover. The obtained 15 parameters are closely related to the grinding wheel's remaining life, consistent with the grinding machine's operating mechanism.

Mechanical equipment is a core functional module in transportation, manufacturing, and other fields, and rolling bearings are fundamental and critical components of various mechanical equipment. Bearing aging leads to a sharp decline in its precision, causing equipment downtime, significant losses, and even serious accidents. Therefore, predicting the remaining service life of rolling bearings is crucial for the safety and maintenance of mechanical equipment, and using machine learning to build such predictive models is an important method. However, traditional prediction methods suffer from problems such as inability to handle the nonlinear relationship between bearing features and remaining service life, failure to consider temporal information in measurement signals, and inability to handle data distribution migration. Meng et al.\cite{Meng2025} proposed a rolling bearing remaining service life prediction method that combines an RSA-BAFT model with CE based feature selection. CE is used to measure the correlation strength between time-frequency domain measurement features and remaining service life for feature selection. They validated the method on the XJTU-SY rolling bearing dataset, showing that the CE method can select features with a nonlinear relationship to remaining service life, and the model built based on CE-selected features has better predictive performance than models built based on similar feature selection methods. The RSA-BAFT model obtained based on CE feature selection also outperforms comparable methods.

Wind turbines are key equipment for developing and utilizing wind energy. However, due to harsh operating environments, heavy operating loads, and variable operating conditions, they are prone to failure, resulting in high maintenance costs and difficulties. Therefore, condition monitoring and fault early warning have become the main technical means to solve these problems. The SCADA system of wind turbines collects historical operating data, and developing fault early warning technology based on this data is a current research focus. Geng \cite{Geng2024master} proposed a wind turbine fault early warning method combining clustering algorithms, CE, and machine learning techniques. First, the DBSCAN clustering algorithm is used for data cleaning. Then, CE and other methods are used to select SCADA operating parameters that are highly correlated with gearbox oil sump temperature. Finally, nine machine learning algorithms are used to establish a gearbox oil sump temperature prediction model. They validated the method based on historical data from a wind farm in Hebei Province from 2018 to 2019. Experimental results show that the method can achieve high model prediction accuracy and issue accurate fault warnings 2-3 hours in advance, proving its effectiveness.

\section{Petroleum engineering}
Coalbed methane (CBM) is a type of natural gas found in coal seams, typically extracted through drilling. It has been developed in China for over 30 years. CBM production is influenced by various geological and technological factors during extraction, making the identification of key factors closely related to these factors crucial for making decisions regarding coalbed fracturing. However, the relationship between the effectiveness of coalbed fracturing and influencing factors exhibits non-linear characteristics, and these factors interact with each other, posing challenges to key factor analysis and production prediction. Luo and Xi\cite{Luo2024} proposed a method combining a feature selection method based on CE, a hybrid optimization algorithm, and a neural network approach to construct a CBM well production prediction model. They validated this method using well data from a block in the Ordos Basin, identifying four key factors from 36 candidate factors using the CE based method: gas content and gas saturation related to the gas layer, and pre-fracturing fluid volume and sand-containing fluid volume related to the fracturing process. The resulting prediction model was able to predict up to 83\% of the production based on experimental data. The authors applied this method to a typical well in the same block to select high-gas-producing areas in the formation, which significantly increased the well's daily production from 322 cubic meters to 950 cubic meters per day.

Onshore crude oil production includes onshore exploration, crude oil extraction and processing, crude oil storage and transportation. To ensure crude oil quality, the production process consumes a large amount of energy, resulting in significant greenhouse gas emissions. Therefore, the oil production sector needs to accurately quantify the carbon emissions of the production process in order to implement emission reduction and carbon mitigation throughout the entire process. Yuan et al.\cite{Yuan2025} analyzed the carbon footprint of the entire crude oil production process based on actual production data from the Shengli Oilfield, clarifying the carbon emission contribution and emission type of each stage. They also used a combination of XGBoost and CE methods to analyze the main influencing factors of electricity and fuel consumption during the process, and then provided suggestions for emission reduction and carbon mitigation for the key influencing factors. The study compared this analytical method with three other similar methods and found that this method most accurately ranked the importance of the influencing factors.

\section{Mining engineering}
Coal is the primary energy source in China, and by 2030, its share of total consumption is projected to reach approximately 55\%. China's annual coal production is close to 4 billion tons, with 90\% mined underground, facing significant challenges such as frequent geological disasters, deep mining operations, and difficult excavation. Unmanned intelligent mining is a crucial technological direction for coal mining, capable of reducing safety accidents and improving mining efficiency. Coal-rock identification, an international challenge in coal mining, is of great significance for unmanned intelligent mining. However, the complex conditions and diverse geological features of underground coal mine working faces, coupled with low accuracy and quality of on-site signals, increase the difficulty of coal-rock identification. Existing research has focused on multiple signals to accurately perceive the properties of the coal-rock medium at the working face. Gao \cite{Gao2024phd} proposed a multimodal coal-rock identification technology solution that integrates visible light, near-infrared spectroscopy, and physical signals from coal mining machine cutting. He collected coal mining machine cutting signal data, including vibration, torque, and pressure, through simulated tunneling tests, extracted 689 time-series signal features, and then selected 45 features using a feature selection method combining XGBoost and CE to perform coal-rock identification under cutting and feed conditions. He developed a working face coal-rock identification solution for cutting trajectory planning, which utilizes the above-mentioned cutting signal-based identification method for real-time lithology sensing. He conducted field technical verification tests at the roadway tunneling working face of the Dafosi Coal Mine of Shaanxi Binchang Mining Group. The verification accuracy rate of the coal-rock identification method based on cutting signals was 91.14\%, and all identification errors occurred near the coal-rock interface, achieving good engineering application results.

Heavy media coal preparation is a crucial step in coal production, primarily using a medium of a certain density to separate clean coal and gangue. Due to the multiple processes involved and the inherent nonlinearity, dynamism, and multivariate coupling characteristics of the process, particularly the time delays between process variables, real-time monitoring, analysis, and control of these processes are extremely challenging. Traditional methods, such as correlation coefficients, cannot accurately estimate these time delays. Jian and Dai \cite{Jian2024ccc} proposed an evolutionary algorithm combined with CE for estimating time delay parameters in heavy media coal preparation. They applied this method to heavy media coal preparation process data from a coal preparation plant in Shanxi province during 2024, estimating the time delay parameters of four key process variables relative to the ash content of clean coal. Their method was then compared with four similar methods, including CE. The analysis results show that the EVO-CE and CE methods can accurately estimate time delay parameter values ​​that conform to actual production conditions and prior process knowledge, while the comparative methods cannot achieve this. Meanwhile, the calculation time for EVO-CE is much shorter than that for CE, thus better meeting the actual needs of production.

\section{Metallurgy}
High-purity metals are special materials with very high purity, possessing physical properties such as high electrical conductivity and stability, and excellent optical properties. They are essential materials for manufacturing various precision scientific instruments and high-tech products. Preparing high-purity metals requires precise processes to ensure high purity, but traditional methods generally suffer from low purity. Vacuum distillation can purify metals in a green and efficient way, but its process parameters need to be manually adjusted, relying on human experience. Tian et al. \cite{Tian2024} proposed an optimized method for preparing high-purity metals by vacuum distillation. They used machine learning techniques such as CE to screen out a set of process parameters that can guarantee high purity and low impurities, established a predictive model with purity and impurity content as target variables, and then used parameter optimization methods based on this model to obtain the optimal process parameters for high-purity metal preparation. They used this method to conduct process parameter optimization experiments for the vacuum distillation preparation of high-purity selenium and tellurium. Based on feature selection methods such as CE, they found that distillation temperature, holding time, condensation temperature, and vacuum degree are important for the prepared purity, while holding time, distillation temperature, heating rate, and condensation temperature are important for the impurity content. Through continuous iterative experiments, the process parameters obtained by this method can achieve good preparation results, and the process parameters can be automatically controlled and optimized according to different product requirements.

The permeability index is an important parameter reflecting the furnace condition during blast furnace smelting, measuring the furnace's ability to receive blast. Predicting the blast furnace permeability index can better control the smelting process and avoid abnormal furnace conditions. Lin et al.\cite{Lin2024} proposed a permeability index prediction method based on wavelet denoising and a nonlinear Transformer model, which fully considers the multi-scale, nonlinear, and large-volume characteristics of blast furnace smelting process data. This method uses PCC and CE to select key variables affecting the permeability index as model inputs. Experimental results show that this method can provide high prediction accuracy and fast inference speed, providing theoretical support for controlling the permeability index and having practical significance for improving the stability and automation of the smelting process.

\section{Chemical engineering}
Fault diagnosis is crucial for the safe and efficient operation of chemical processes, and data-driven fault diagnosis methods are one of the main approaches in actual production operations. Constructing reasonable process representations of normal and fault states is a key step in building a diagnostic model. Yin et al. \cite{Yin2022} proposed a fault diagnosis method based on the gray-scale correlation space of CE, which characterizes the normal and fault states of a process through the CE correlation matrix between variables, and then uses the matrix as input to a convolutional neural network to construct a fault classification model. They applied the method to fault diagnosis data of the Tennessee Eastman process, and the results showed that the method achieved a diagnostic accuracy of over 95\%, validating its effectiveness. Principal Component Analysis (PCA) is a commonly used multivariate process detection method. Its principle is based on the maximum variance criterion to construct process detection statistics from a set of process variables, but it is only applicable to linear cases. Wei and Wang \cite{Wei2022,Wei2023} proposed a nonlinear PCA method based on CE (CEPCA), which obtains process detection statistics from the CE matrix, which has nonlinear characteristics. They applied their method to data from the Tennessee Eastman process and compared it with the PCA method. The results showed that the CEPCA method achieved a better fault detection rate. Pan et al.\cite{Pan2023a,Pan2024master} proposed a fault propagation and root cause analysis method based on a correlation fault causal graph, called DTMTE. This method constructs a causal relationship graph between sensor variables using time delay estimation based on TE, which is used to identify the root cause and propagation path of faults. CE based TE is used to analyze causal relationships. They validated the method on the Tennessee Eastman process and a real dimethyl carbonate chemical production process, respectively. The results showed that the DTMTE method can accurately find the fault causes and propagation paths, outperforming traditional comparative methods.

Understanding the causal relationships between variables in chemical processes is crucial for process control, contributing to better process monitoring and fault diagnosis. Constructing causal relationship diagrams for chemical processes using causal discovery methods allows for root cause analysis of faults, making it an important method for fault diagnosis. Bi et al. \cite{Bi2023} proposed a deep learning-based causal discovery method, CGTST, and compared it with several other methods, including CE based TE. Experimental results show that on a 5-variable continuous stirred tank reactor, the reaction diagram obtained by the TE method closely approximates the actual situation; on the Tennessee Eastman process data, the TE method also achieved estimation results close to reality, demonstrating strong practicality.

Soft sensing technology is one of the important methods for chemical process modeling, referring to the estimation and inference of process variables that are difficult to measure directly through easily measurable process variables. However, due to the influence of various factors such as equipment failure, environmental interference, and signal transmission in actual production processes, process variable data often contains a large number of missing values, thus requiring missing value completion. Generative Adversarial Imputation Nets (GAIN) are a data completion method based on the GAIN algorithm framework, but when the number of missing values ​​is large, the algorithm's performance is difficult to meet practical needs. Wu \cite{Wu2023} proposed an improved GAIN algorithm framework called Information Augmentation GAIN (IEGAIN), in which CE is used to calculate the weight matrix as input to the generator in the new algorithm. He compared IEGAIN with other classic algorithms such as GAIN on the UCI Spam and Letter datasets, publicly available thermal power plant datasets and butane desulfurization tower process datasets, and actual polypropylene production process data. The results show that IEGAIN can complete missing data values ​​with the lowest error.

Bio-fermentation is a green and energy-saving chemical production method, and has become the preferred choice for many chemical industries (such as food processing and pharmaceuticals), offering significant economic benefits and broad market prospects. Fed-batch fermentation is a typical intermittent production method. Its microbial metabolic processes require high precision in process parameter control, necessitating intelligent optimization and control technologies. This necessitates research into real-time measurement methods for key fermentation process parameters. Soft sensing technology, characterized by its economy, ease of implementation and maintenance, overcomes the challenge of accurately modeling the mechanisms of complex fermentation processes, providing an effective solution for monitoring fermentation process parameters. Establishing parameter prediction models based on data is a key core issue in soft sensing technology. Addressing the characteristics of intermittent fermentation processes, Shi\cite{Shi2025master} proposed a soft sensing modeling method for fermentation processes based on temporal difference neural networks. This method utilizes CE based MI for model input variable selection and then employs a temporal neural network with difference operators to establish a soft sensing model for predicting fermentation quality parameters. He applied this method to the IndPensim dataset, simulating industrial-scale penicillin fermentation, to establish a penicillin concentration prediction model, achieving high prediction accuracy and robustness, thus validating the effectiveness of the proposed method.

\section{Medical engineering}
Brain-computer interfaces (BCIs) are systems that generate control commands by analyzing and processing brain signals. Motion imaging-based BCIs can help people generate movement commands through their brains and have many important applications. Magnetoencephalography (MEG) has the advantages of high signal-to-noise ratio and spatiotemporal resolution, and has potential applications in the field of BCIs. However, improving the accuracy of MEG-based BCI systems is a challenge. Tang et al.\cite{Tang2024,Tang2025a} proposed a method for a motion imaging BCI system based on MEG, which uses the TE estimation method based on CE to select the channel set of the BCI system. They tested the method on a publicly available MEG dataset and found that the method can significantly reduce the number of channels selected while maintaining prediction accuracy. The prediction performance is better than similar channel selection methods based on random forests, providing a technical guarantee for the practical application of MEG-based BCIs.

\section{Aeronautics and astronautics}
As aircraft systems become increasingly complex, aircraft design first requires a deeper understanding of their overall design parameters. Theoretical analysis of the coupling relationships between various design parameters helps in analyzing the feasibility of design schemes or optimizing the overall design. Krishnankutty et al. \cite{Krishnankutty2020} proposed two Copula-based MI estimation methods based on the equivalence relationship between CE and MI, and applied these methods to the analysis of technical parameter data of 22 US jet fighters, estimating the coupling relationship between flight range and tolerable load, thus verifying the effectiveness of the analysis methods.

Satellites are a major type of spacecraft in the space age, with wide-ranging civilian and military applications in the information age. As a complex system operating in extreme environments, on-orbit health monitoring of satellites is crucial. Satellite telemetry data is the encoding of various sensor parameters, containing information on the interaction relationships of physical parameters within the satellite's internal operating system. Anomalies in satellites propagate internally due to these interactions; therefore, analyzing the fault propagation chain caused by these internal interactions helps in the timely detection of satellite anomalies and ensures normal satellite operation. Analyzing the causal relationships between telemetry parameters is one approach to solving this problem. Liu et al.\cite{Liu2022,Liu2023master} proposed directly applying CE based TE to analyze real satellite telemetry data, obtaining fault propagation diagrams between telemetry parameters, with results superior to traditional TE methods. Zeng et al.\cite{Zeng2022} proposed an improved TE metric, called NMCTE, for analyzing causal relationship networks between telemetry parameters. This metric utilizes CE based TE representation and estimation methods. They further proposed a CN-FA-LSTM method for anomaly detection based on the obtained causal network. They applied the NMCTE method to real satellite telemetry data, obtaining causal networks with good interpretability. They then compared the CN-FA-LSTM method with six other methods on NASA's publicly available SMAP and MSL datasets, validating the method's superiority.

Turbofan engines are the most commonly used engines in jet aircraft, characterized by high efficiency, reliability, and energy saving, making them one of the key pieces of equipment in the modern aviation industry. Due to their complex structure and long-term operation in extreme environments, turbofan engines are prone to wear and aging. Therefore, monitoring their health status, and subsequently conducting fault prediction and maintenance, is crucial for ensuring aviation safety and improving the reliability and service life of turbofan engines. Thus, how to assess the health status of engines is a fundamental and critical issue. Jia \cite{Jia2023} proposed a health index for turbofan engines, employing an evidence-based reasoning method to fuse engine sensor monitoring data to measure engine health status. In this method, CE is used to calculate the reliability of engine sensor variables during the reasoning process. He applied the method to an engine performance degradation simulation dataset provided by NASA's Green Center and compared it with two traditional methods. The results show that the new method performs better in assessing engine health status, thanks to the integration of nonlinear correlation information between sensor variables based on CE. He further used the obtained one-dimensional composite health index to establish engine fault prediction models and remaining life prediction models, both achieving more accurate prediction results than the compared methods. Sun \cite{Sun2024master} proposed an engine remaining life prediction method based on CE feature selection. First, CE is used to select sensor variables with a nonlinear correlation to engine health factors. Then, these variables are used to reconstruct the health factors and establish an exponential degradation model for engine failure. Next, similarity distance is used to select a group of engines from the engine history database with high similarity to the exponential degradation model prediction. Finally, the median health factors of these engines are used as the predicted remaining life value. She validated the method based on C-MAPSS data. The results show that, based on the nonlinear measurement capability of CE, this method selects 8 out of 21 sensors to construct the exponential degradation model. At 50\%, 70\%, and 90\% operating cycles, the prediction error of this method is reduced by 39.25\%, 41.69\%, and 50.53\% respectively compared to traditional methods, significantly improving the accuracy of remaining life prediction while reducing model complexity and thus improving computational efficiency. She also proposed a remaining lifetime prediction method based on Copula similarity. First, using CE to select pairs or multiple similar sensor combinations, she then constructs a Copula function or a vine-like Copula structure. Next, she calculates the Copula similarity metric matrix between these combinations. Then, based on these similarity matrices, she selects a group of similar engines from a historical database. Finally, she uses the average remaining lifetime of these engines as the predicted value. Experimental results based on the C-MAPSS dataset show that, compared with traditional methods, this method significantly improves prediction error at 50\%, 70\%, and 90\% of the operating cycle. Her proposed methods have significant engineering application value and are expected to be applied to the health management of military and civilian aero-engines.

Flight delays are one of the major problems affecting the normal and effective operation of the international civil aviation industry, causing not only inconvenience to passengers but also huge economic losses to the aviation industry. The aviation system is an organic whole, with upstream and downstream sharing of flight resources, leading to system coupling. This causes delays in upstream flights to propagate downstream, thus flight delay management first requires analyzing this causal relationship. Wu et al.\cite{Wu2020} proposed a method using the CE based TE estimator to analyze the strength of causal relationships between flight delay time series at airports. This enables the civil aviation information system to analyze whether there is a causal relationship between two flights, thereby allowing for a deeper understanding and utilization of the inherent relationships of flight delays between nodes in the aviation system.

Airport Collaborative Decision Making (A-CDM) is a standard operational framework supported by the International Civil Aviation Organization (ICAO). It optimizes flight support decision-making processes by sharing data in real time among various units of the aviation system, and is a core tool for enhancing airport operational efficiency, predictability, and punctuality. In practice, airports and line support departments allocate resources based on the Estimated Incoming Wheel Time (EIBT), which includes both taxiing time and Actual Landing Time (ALDT). Current A-CDM systems primarily focus on ALDT, with limited consideration for EIBT. Airports like Beijing Capital International Airport and Shanghai Pudong International Airport still rely heavily on air traffic controller experience for taxiing time estimation, lacking accurate predictions based on real-time data, or only providing EIBT after landing, impacting service timeliness and efficiency. Tang et al.\cite{Tang2025} proposed a two-stage flight arrival time prediction method, dividing the time after an aircraft enters the terminal maneuvering area into two stages: air flight and ground taxiing, and constructing an arrival time prediction model based on machine learning. Based on data from the Pudong Airport A-CDM system in October 2022, they extracted 18 features across 5 categories (including aircraft and flight features, airport ground operation features, airport terminal maneuvering area features, arrival/departure flow features, and weather features). They then used CE to select features associated with air travel time, taxiing time, and the total time for both stages. Finally, they trained a LightGBM model for prediction. Experimental results show that the CE method selects reasonable arrival time-related features. The most relevant features for both stages are flight distance and taxiing distance. Features related to flight time also include flight altitude, speed, and angle, while the number of operational hotspots is related to taxiing time. These findings are consistent with existing research. Prediction experiments show that the two-stage arrival time prediction model based on CE-selected features outperforms models using all features, achieving approximately 70\% accuracy within 3 minutes and approximately 90\% accuracy within 5 minutes, providing support for flexible choices in actual airport operations. Experiments also show that the CE method selects features that conform to the operation mode of Pudong Airport, demonstrating the method's generalization ability for airports with different characteristics, thus ensuring the predictive performance of the constructed model.

\section{Arms}
Weapon and equipment effectiveness assessment refers to a comprehensive, systematic, and scientific analysis and evaluation of the technical indicators and combat performance of a weapon. Due to the complexity of weapon and equipment systems and their application, the assessment needs to consider multiple factors, thus requiring a comprehensive indicator system to complete the evaluation. Effectiveness indicator systems often contain a large number of different types of indicators, resulting in correlations between indicators and a high dimensionality of the indicator system. Therefore, it is necessary to reduce the system to facilitate subsequent assessment processes. Traditional reduction methods generally use mathematical tools such as correlation coefficients, but their linear assumptions are often not met in practical problems. Chen et al.\cite{Chen2024p} proposed an indicator system reduction method that uses CE to measure the correlation between indicators and reduces the indicators by comparing the average CE of each indicator with other indicators. They verified the method using simulation data of the assessment object, demonstrating that the method has the advantage of handling nonlinear correlations between indicators and is more scientific and accurate than traditional methods.

\section{Automobile}
Modern automobiles integrate electronic systems via in-vehicle networks, enhancing passenger comfort, safety, and versatility. However, with the development of intelligent vehicle technology, in-vehicle devices have become targets for hackers, posing a threat to vehicle security. The CAN bus, a data communication protocol connecting and controlling various electronic components within intelligent vehicles, has become a de facto mainstream standard in the automotive field. However, due to the lack of encryption and authentication mechanisms, it is highly vulnerable to network attacks. Therefore, researching CAN bus intrusion detection technology has become one of the main technical means to improve its security. Gao et al.\cite{Gao2023} proposed a lightweight neural network design method for detecting CAN bus intrusion events. This method first analyzes the attribute set of abnormal CAN data frames, then uses CE to select a few attributes related to intrusion attacks, and finally uses these attributes to construct a CanNet neural network detector to detect intrusions. They validated the CanNet method using CAN bus data from a Hyundai Sonata YF, demonstrating that this method has advantages over similar methods in terms of high detection rate, high real-time performance, and low memory consumption.

High-speed trains are the main mode of transportation for China's high-speed rail system. With their ever-increasing speeds, online early warning of abnormal operating conditions is becoming increasingly important. Rolling bearings are key components of high-speed trains, characterized by high operating speeds and heavy loads. Ensuring their health and reliability is crucial, and fault monitoring and early warning technology is one of the key technologies for ensuring their safe operation. Machine learning-based health monitoring using vibration signals is a traditional method for assessing the health of rotating mechanisms, but its application in electric multiple units (EMUs) faces several challenges: 1) Directly acquired time-varying vibration signals are characterized by high noise disturbances and non-stationarity; 2) High-speed trains operate under multiple different conditions, making it impossible to directly use a single offline-trained model for monitoring; 3) The limited sample data of actual fault conditions makes model training difficult; 4) Traditional monitoring and early warning technology frameworks based on cloud computing and deep learning cannot meet the requirements of efficiency and real-time performance under conditions of limited computing resources at the edge, due to long computation times and large latency. Xu et al.\cite{Xu2025} proposed an adaptive real-time rolling bearing fault diagnosis technology framework, successfully solving the above problems by utilizing technologies such as CE, transfer learning, wide learning algorithms, and edge-cloud collaboration. They utilized the CE method to discover the correlation between vibration and force signals under fault conditions, making fault diagnosis at the end-point using indirect signals possible. They validated the method using real data from normal high-speed rail operation scenarios, finding that the force signal-based diagnostic model can accurately detect damage faults. Compared to traditional methods, this technical framework significantly improves accuracy and real-time performance, achieving an accuracy rate of over 99.98\%, and can effectively perform real-time fault diagnosis of rolling bearings in high-speed trains.

\section{Control engineering}
Airport terminals are transportation hubs in the aviation system, responsible for the transfer of passengers between air and ground, and play a vital role in enhancing urban sustainability and promoting environmentally friendly aviation. HVAC systems consume over 60\% of the energy in airport terminals, indicating significant room for improvement in terminal energy efficiency. Flight schedules result in dynamic patterns in passenger flow at terminals, making it crucial to consider these dynamic characteristics to improve the energy efficiency of HVAC system control. The key lies in predicting indoor temperature to meet the dual objectives of comfort and energy efficiency desired by the system control. Li et al.\cite{Li2025c} proposed a passenger-centered model predictive control (MPC) framework for energy-efficient terminal indoor temperature control. This framework includes an indoor temperature target dynamic setter, an indoor temperature predictor, and a control command optimizer. The MPC predictor consists of a causal constraint neural network model, with inputs composed of factors causally related to indoor temperature selected using the CE based TE method. They validated this framework using Guangzhou Baiyun International Airport Terminal 1, selecting four international and domestic departure areas as targets. TE analysis experiments based on actual on-site data revealed that historical room temperature, adjacent area room temperature, outdoor temperature, solar radiation intensity, passenger density, supply airflow temperature, and fan frequency are the dependent variables for indoor temperature. The analysis also identified time-delay parameters of causal relationships, while traditional linear causal analysis methods such as GC did not provide equally reasonable results. Based on this, the neural network prediction model derived from the TE analysis results showed more accurate predictive performance than models constructed using comparative methods. The CE based TE method demonstrated a stronger ability to analyze nonlinear causal relationships than traditional linear causal analysis methods, giving the prediction model interpretability. On-site tests in the target area of ​​the terminal building showed that energy consumption based on this MPC framework was 17\% to 46\% lower than that of the comparative methods.

\section{Electronics engineering}
The increasing integration of semiconductor chips places ever higher demands on microelectronic packaging. Microelectronic packaging plays a crucial role in isolating the external environment and dissipating internal heat, ensuring the stable operation of integrated circuits. This requires packaging materials to possess good stability, high strength, and other desirable physical properties. Liu \cite{Liu2022b}, using Cu-based materials, established a CuNi binary alloy system. Utilizing a combination of first-principles calculations and machine learning, based on cluster correlation function characteristics, he predicted the configuration energy and Young's modulus, which are related to material strength and stability, respectively. The authors analyzed the rationality of the prediction model using CE. By calculating the correlations between features, and between features and configuration energy and Young's modulus, they found a higher correlation between model features and Young's modulus, while the correlation between configuration energy and Young's modulus was lower. This improved the model's interpretability and helps in designing more reasonable material property prediction models.

\section{Communication}
Communication security is one of the main concerns in mobile communications, and it is generally addressed through encryption techniques at the communication layer. In resource-constrained emerging networks (such as IoT and WSN), key distribution is a challenge. The reciprocity of wireless channels provides a mechanism for sharing keys between communicating parties, who can obtain the key by measuring the wireless channel. The concept of key capacity provides a theoretical upper limit for key extraction from wireless channels. However, in reality, key capacity is often limited by many practical physical conditions (such as terminal movement and channel noise), requiring quantitative analysis. Wang et al.\cite{wang2016physical,Wang2016} investigated the influence of physical factors on key capacity under uniform scattering conditions, transforming it into a random variable MI calculation problem. They verified the correctness of their theoretical derivation based on a simulated physical environment, and the simulation experiment used a CE based MI estimation algorithm to estimate the key capacity. Simulation results show that the theoretical derivation is verified and can guide practical applications.

One of the main challenges in developing 6G communication network technology is achieving higher data transmission rates to meet the demands of scenarios such as more immersive experiences, 3D vision, and industrial intelligence. Traditional communication theory does not consider semantic information in transmitted data, but 6G technology can leverage AI-based semantic communication to achieve higher network transmission performance. Fu et al.\cite{Fu2023} proposed an end-to-end service framework based on semantic communication for 6G networks, integrating semantic communication with AI's semantic analysis capabilities and utilizing a Transformer-based encoder-decoder to compress semantic information. The semantic encoder's loss function consists of a semantic loss function based on Euclidean distance and an information content loss function based on CE. They validated the service framework using image data, training it on the ImageNet-1K dataset and then performing simulations on the VOC2012 dataset. The results show that compared to traditional communication schemes, this service framework achieves optimal performance in both object detection and image semantic reconstruction, and achieves performance similar to full semantic feature transmission schemes, making it a promising technological component of 6G networks.

\section{High performance computing}
Improving energy efficiency is a key goal of high-performance computing research. Optimal energy efficiency settings, such as processor frequency, can reduce energy consumption during program execution. However, determining the optimal configuration is a time-consuming process, requiring reconfiguration after any program modification. Using machine learning methods to automatically determine the optimal configuration based on performance events is a new research direction, but it requires identifying which events are energy-efficient to determine the optimal configuration. Gocht-Zech \cite{Gocht-Zech2022} proposed using feature selection to select energy-efficient events. He selected six feature selection methods and provided corresponding estimation methods based on CE theory. Real-world experiments show that this Copula-based method can identify energy-efficient performance events, thereby improving program execution energy efficiency and saving 24\% of energy consumption with a 7\% increase in runtime cost.

\section{Information security}
Adversarial attacks and defenses are hot topics in information security, referring to attacks launched by attackers leveraging their understanding of system and algorithm characteristics, as well as corresponding defensive measures. Deep neural networks are an important class of algorithms in machine learning with wide applications; researching their attack and defense algorithms is crucial for the security of such artificial intelligence systems. Liu et al.\cite{Liu2024} proposed a MI estimation algorithm based on CE, called ${CE}^2$, and used this algorithm to propose a neural network adversarial training algorithm. This algorithm fully utilizes the reliability of CE based MI estimation against adversarial attacks, designing a network training algorithm to guide the neural network prediction model to minimize the attack of adversarial examples. The authors first demonstrated the performance advantage of ${CE}^2$ over traditional MI estimation algorithms through simulation experiments, and then verified the superiority of the ${CE}^2$-based neural network defense algorithm over other similar classic defense algorithms in defending against typical deep neural network adversarial attacks on the CIFAR-10 and CIFAR-100 datasets.

The Internet of Things (IoT), as an information technology that connects everything, has been widely applied to various fields of social life, such as industrial production, healthcare, and smart homes, becoming an important type of information infrastructure. IoT security is a prerequisite for ensuring the normal operation of these facilities; therefore, IoT intrusion detection is a significant security issue that has been extensively studied. However, due to data limitations, existing intrusion detection methods cannot be effectively applied across domains. Wang et al.\cite{Wang2024p3,Wang2025} proposed a novel cross-domain IoT intrusion detection method. This method utilizes incremental clustering to generate a graph structure from the intrusion dataset, then uses a graph neural network to extract domain features, and finally uses distance based on CE to align these features across domains. Finally, a conditional domain adversarial neural network is used for intrusion detection. They validated the effectiveness of their method on four publicly available IoT intrusion detection datasets. The CDO-based cross-domain data alignment method effectively reduces inter-domain differences in extracted features, thereby improving detection and classification performance.

The widespread adoption of cloud computing has made cloud environment security a significant challenge for technology users. In particular, malware exploits vulnerabilities in cloud infrastructure, leading to data leaks, unauthorized system access, and identity theft. Utilizing machine learning algorithms to analyze network traffic data and detect malware attacks is a widely adopted cloud security technique, effectively improving the automation of security protection and achieving rapid, accurate, and adaptive responses. To address the shortcomings of traditional machine learning malware detection algorithms, Baawi et al.\cite{Baawi2025}提 proposed a novel malware detection classifier algorithm to improve detection performance. They validated their new algorithm on the publicly available malware detection dataset Meraz'18 and compared it with other machine learning algorithms. The comparative experiments included two classification scenarios: with and without feature selection. Feature selection employed a CE based method. Experimental results show that the new algorithm outperforms similar traditional machine learning algorithms. With CE for feature selection, the new algorithm achieves almost the same detection performance using only 20 selected features as it does with all 53 features. Wang \cite{Wang2025master} proposed a DDoS attack detection method for Software-Defined Networking (SDN) security. He utilizes a feature selection method combining CE and feature importance index to filter features, and then employs an optimization algorithm combined with random forest to build an attack detection classification model. He validated the effectiveness of the proposed method on the CIC-IDS2017 and InSDN datasets published by the Canadian Cyber ​​Security Research Centre, as well as on a simulation platform he built. The results show that this feature selection method reduces the number of features while maintaining model detection accuracy, thus reducing model training and runtime, and effectively lowering the CPU utilization of the network controller, demonstrating good practical value.

\section{Geomatics}
Hyperspectral remote sensing is a widely used cutting-edge mapping technology that can acquire diagnostic spectral information of different ground features through remote sensing spectral imaging. Due to the large number of bands, large data volume, and significant redundancy in hyperspectral images, feature extraction techniques are needed to select effective bands to characterize the imaged objects. Therefore, hyperspectral image band selection is one of the important problems in this field. The main idea is to select a subset of bands that maximizes the imaging evaluation criterion function. Among these, information theory-based criteria are one of the main methods for band selection. Zeng and Durrani \cite{Zeng2009} proposed a band selection method based on CE and applied it to real hyperspectral data collected at Indian Pine in northwestern Indiana, USA. The results show that CE provides a robust MI band selection method.

Deformation refers to the changes in position, size, and morphological characteristics of a deformable body under the influence of external factors. It is often an extremely slow physical process, such as the deformation of mountains or buildings. Deformation reaching a certain level can trigger safety risks, and if not effectively prevented, it can cause serious damage in natural disasters such as earthquakes and landslides. Deformation monitoring involves monitoring and measuring deformable bodies, and using monitoring data to construct deformation models to predict deformation trends. Engineering deformation monitoring is one of the important issues in the field of engineering surveying, requiring guaranteed monitoring accuracy and reliability, and is of great significance to the construction and operation safety of large-scale projects. Common deformation monitoring and analysis methods generally only model and predict individual monitoring points. However, monitoring points within a deformable body are not isolated but have inherent correlations, which can be used to improve the prediction accuracy of single-point monitoring. Zhang et al.\cite{Cao2023,Zhang2022master2} proposed a machine learning deformation analysis and prediction method that considers the correlation between neighboring deformation points, using CE to select neighboring points that are correlated with the prediction point. He validated the method using data collected by a surveying robot on the eighth cofferdam during the construction of the Taihu Tunnel from 10 December 2020 to 8 October 2021. The data was compared with two traditional correlation metrics, showing that the model based on CE-selected neighboring monitoring points achieved higher prediction accuracy, proving that CE is more suitable for nonlinear time series prediction problems in deformation measurement. This method has significant application value for long-term deformation prediction problems such as early warning of cofferdams in practical engineering.

\section{Ocean engineering}
Human exploration of the ocean is fundamental to marine engineering construction, marine resource development and management, and marine military operations. Seabed sediment information detection is a prerequisite for many of these activities, making it a crucial issue in marine surveying. Multibeam sonar systems are among the main surveying equipment in marine surveying, used to acquire and classify seabed sediment information through acoustic detection. Zhao \cite{Zhao2022phd} proposed a comprehensive multibeam sonar seabed sediment classification technique. This technique extracts a set of spatial, frequency, and scale features from multibeam backscattered images, then uses correlation tools such as CE to remove redundant features, and finally constructs a sediment classification model using the selected features. He experimentally validated the proposed feature selection and model construction method on the Oostende Harbor dataset in Belgium. The results show that using tools such as CE can reveal nonlinear correlations between features, and removing redundant features significantly improves the model's classification performance.

\section{Finance}
Quantitative finance is an emerging financial discipline that guides financial decision-making through the analysis of quantitative relationships in financial data. Based on the vast amounts of financial market trading data generated by financial trading systems, mathematical tools are used to analyze the quantitative relationships between financial products, clarifying market patterns and dynamics, and thus managing financial assets. Analyzing the correlations between market financial variables is a crucial problem in financial engineering, helping traders understand their dynamic relationships and adjust portfolios to manage risk. Due to the nonlinear and non-Gaussian characteristics of financial market variables, MI has become an ideal metrics for correlation, and MI estimation algorithms have become an important tool in the quantitative finance toolbox. MI estimation algorithms based on CE has been implemented in the quantitative finance algorithm libraries \texttt{MLFinLab}\cite{mlfinlab} and \texttt{ArbitrageLab}\cite{arbitragelab}, and are widely used in the industry.

Based on real data from the Chinese stock market (Shanghai A-share Index, Shenzhen A-share Index, and CSI 300 Index), Wang \cite{Wang2015} studied a method for optimizing investment portfolios by utilizing the correlation network between stock assets. The method employed linear and nonlinear correlation measures, including CE, to construct a relationship network between stock assets based on correlation strength, and then build an investment portfolio. The study estimated CE for different families of copula parameter functions. Liao \cite{Liao2023} studied the problem of investment target selection. Based on three indicators—return on equity, three-year compound annual growth rate of net profit, and price-to-earnings ratio—he initially screened 10 A-share listed companies from over 4,000 listed companies. He then used tools such as CE to conduct statistical analysis of the target stock price data to assess the portfolio's risk resistance.

Stock market investors always hope to invest in well-performing listed companies, making the ability to distinguish between good and bad stocks crucial. The ST stock system is a stock risk warning mechanism implemented in Chinese A-share market, helping investors choose portfolios and mitigate risks. Stock classification is an important issue in stock analysis and has reference value for financial market investors. Zhu \cite{Zhu2022b,Zhu2024} proposed a machine learning-based ST stock classification method, employing the Boruta algorithm and the CE method for feature selection, then using six regression models for prediction, and optimizing the model's hyperparameters using the Optuna framework. He selected data from 2076 stocks (including 351 ST stocks) on the Shanghai and Shenzhen Stock Exchanges from 2016 onwards in the tushare database, containing 139 stock feature variables, and finally used the Boruta and CE methods to screen out seven interpretable variables. Model prediction results show that this method achieves the best prediction accuracy in feature selection and XGBoost model combination.

The Belt and Road Initiative (BRI) is an international cooperation initiative proposed by the countries along the Silk Road, playing a significant role in promoting the economic and social development of China and related countries. The BRI is a financial market indicator for industries and companies related to this development initiative, reflecting the development trends and changes in the countries and regions involved. It has important reference value for government and investor decision-making; therefore, the predictive analysis of the index is an important issue in this field. Xu \cite{Xu2024} proposed a BRI return prediction method combining the GAS model, CE, and lightGBM, where CE is used to select input features for the prediction model. He validated the method using data from 58 BRI constituent stocks from 2020 to 2023. The results show that, compared with similar methods, the GAS-CE-LGBM method performs best on all four prediction performance evaluation indicators. In particular, using CE for feature selection significantly improves the model's predictive performance, indicating that CE can capture the nonlinear dynamic relationships between variables in the problem.

Analyzing financial data requires modeling it mathematically, but the non-Gaussian nature of financial variables and their joint distribution presents challenges for data modeling. Calsaverini and Vicente\cite{Calsaverini2009,Calsaverini2013} presented a method for selecting copula function models. This method utilizes the marginal distribution independence property of CE (MI) to separate the objective of the Copula discrimination problem from the marginal functions, and then uses the definition of CE to transform the problem into a model selection problem with MI as the upper bound. The authors also defined the concept of information excess. They applied the modeling method to the daily logarithmic return data of 150 stocks in the S\&P 500 index from 1990 to 2008, using information excess to verify the effectiveness of the method when applied to the T-Copula function family.

R-vine Copula is a flexible tool for constructing multivariate Copula distributions. Determining the vine structure is a key step in building such models. Alanazi\cite{Alanazi2021} proposed a method for constructing R-vine Copula based on the relationship between CE, MI, and CMI. This method involves building a minimum spanning tree based on MI, calculating the CMI on each pair of edges of the previous subtree, constructing new subtrees based on the CMI, and determining the hierarchical structure of the vine Copula. He applied this R-vine Copula construction method to the modeling problem of inter-stock correlation structures, constructing an R-vine Copula model of asset relationship structures based on data from 15 major stocks of the German DAX index (January 2005 to August 2009). Compared with traditional methods, the Copula correlation structure model established by this method can better fit the data. Wang \cite{Wang2023b} proposed a similar vine Copula structure selection algorithm based on the relationship between CE, MI, and CMI. The authors used this algorithm to analyze the correlation structure among the five major industry indices of the CSI 2000 Index. Using data from March 1, 2019 to March 1, 2022, they constructed a vine Copula structure based on Kendall correlation coefficient and a vine Copula structure based on MI. The results show that, in terms of goodness of fit, the latter is better than the former; and from the perspective of interpretability, the latter more reasonably describes the dependence between the assets of the five major industries.

The financial crisis brought the issue of systemic risk in the financial system to the attention of regulatory authorities worldwide. The liberalization of Chinese stock market has deepened the integration of the economy and finance, leading to coupling between various industries and thus increasing the degree of systemic risk. Therefore, it is necessary to study the cross-industry risk spillover effects in order to prevent and mitigate them. Entropy, as a mathematical tool for quantifying uncertainty, is well-suited for measuring financial risk portfolios. Xiong \cite{Xiong2020} used tools such as CE to analyze the daily logarithmic return data of 11 industries in Chinese stock market from January 5, 2005 to July 3, 2020, studying the dynamic evolution of individual industry risk and cross-industry risk spillover characteristics, with a particular focus on the risk characteristics of the 2008 financial crisis, the 2013 liquidity crunch, and the 2015 stock market crash. The study found that the dynamic changes in industry-linked CE lag behind the occurrence of accumulated independent entropy, indicating that inter-industry linkages lead to enhanced systemic risk; the 2008 financial crisis had stronger internal market contagion and greater destructive power; and the recent internal correlation among the 11 industries is relatively strong. Ding \cite{Ding2024master} analyzed daily return data of 116 listed financial institutions from October 27, 2006 to December 31, 2023 using CE data to study the characteristics of risk linkages within the financial system. The study found that systemic shocks lead to a sharp increase in the degree of risk linkages within the financial system. Risk linkages within the banking sector are stronger than those between other sectors, while inter-sectoral risk linkages are stronger than intra-sectoral linkages. Diversified financial sectors can amplify the risks caused by systemic shocks.

Financial vulnerability stems from the inherent instability caused by the high leverage of the financial sector. Financial vulnerability measurement tools enable timely national responses and interventions during crises, thus attracting extensive research. Increasingly mature network analysis theory provides methodological tools for measuring financial vulnerability from the perspective of financial networks; however, traditional network construction methods, based on linear relationship metrics such as the PCC, cannot reflect the nonlinear relationship characteristics within financial systems. Chen et al.\cite{Chen2024} proposed a network curvature-based financial vulnerability measurement method using CE. This method first constructs a financial network using CE, then calculates four discrete Ricci curvatures of the network as a measure of market vulnerability. They applied this method to stock data from the CSI 300 Index between April 2006 and April 2022 to analyze market vulnerability before and after the financial crisis. The results show that this method more clearly describes post-financial crisis market vulnerability than methods based on PCC, while possessing the same risk measurement capabilities as traditional risk measures.

Recent studies have shown that financial portfolios are vulnerable to shocks and significant financial risks during extreme financial events. Households and portfolio managers urgently need to understand the impact of financial shocks and extreme events on investments. Traditional risk measurement tools struggle to detect this tail dependence, while Copula-based methods are increasingly demonstrating superiority in this area. Ardakani and Ajina\cite{Ardakani2024a} proposed using a CE based MI to measure tail risk in extreme event zones, clarifying the significance of diversification strategies in mitigating tail risk. They applied this risk measurement tool to data from the 2022 Consumer Finance Survey and found that certain portfolios exhibited strong correlations, thus increasing tail risk. This finding provides important insights for households managing tail risk.

Credit risk is one of the main fundamental risks faced by the financial and banking industry, and effective management of credit risk is essential to ensuring financial security. Credit scoring card models are a method for assessing customer credit risk and serve as a decision-making tool for managing financial risk. These models classify customers into credit ratings based on their historical credit data to determine their financial privileges. Traditional methods for building credit scoring card models rely on expert experience, resulting in low efficiency and imperfect models. Kong et al. \cite{Kong2021} proposed an automated credit risk model construction method based on CE, which significantly improves modeling efficiency while ensuring high predictive performance and interpretability. They compared this method with expert modeling on real credit card data. Experimental results show that the method significantly shortens the modeling time and achieves predictive performance and interpretable customer credit characteristics comparable to expert models.

Peer-to-peer (P2P) lending is a financial model that facilitates fundraising and lending online. The credit risk of this model primarily stems from borrowers' failure to fulfill their repayment obligations, posing a significant risk to the funds of creditors. Therefore, accurately assessing the credit risk of lenders is crucial, and constructing a personal credit risk model using lending data is a primary solution. Peng \cite{Peng2022m} proposed using CE to measure the nonlinear correlation between risk variables and high-dimensional features of personal data, thus selecting input features for a personal credit risk prediction model. He conducted an empirical study using loan data from the US P2P lending platform Lending Club, comparing CE and PCC, two commonly used feature selection methods. He found that the nonlinear features selected by CE achieved better predictive results on the XGBoost model.

Green credit is a financing tool provided by financial institutions to listed companies with the goal of ecological and environmental protection. Researching green credit risk assessment can improve financial institutions' risk control in the use of this tool. Wang \cite{Wang2023master} proposed a green credit risk assessment method based on a combination of CE and machine learning models. Using A-share listed companies in 2021 as a case study, she selected 67 indicators from three aspects: company status, innovation and development capabilities, and green evaluation. Based on CE, she selected 18 of these indicators to form a risk assessment indicator system, and then used four machine learning methods to construct an assessment model. Experimental results show that the accuracy of the obtained model reaches 95.01\%, providing a reliable assessment tool for green credit risk.

Accurately predicting financial product prices can help investors manage risk and make investment decisions; therefore, establishing relevant predictive models is one of the important issues of concern to researchers. Because financial products have inherent market logic, their prices also exhibit corresponding causal linkages. Therefore, this causal relationship between prices can be used to build more accurate price prediction models than traditional methods. Zhang et al. \cite{Zhang2023a} proposed a transfer learning framework based on causal relationships between prices. They used the CE based TE method to calculate the causal relationships between the prices of different financial products, selecting the dependent variable price to predict the result variable price. Based on this selection, they proposed a learning algorithm for training a deep learning model to obtain the prediction model. They applied the algorithm to daily price data of major international financial indices, energy futures prices, and agricultural futures prices from 2010 to 2021. The results showed that the CE based TE method discovered causal relationships between similar prices. Based on this, the model obtained using this transfer learning framework gave better prediction results than similar comparative algorithms on all three types of price data.

Epidemics pose a serious threat to public health, prompting societies and individuals to take countermeasures. These responses, in turn, generate significant economic and social impacts, particularly on financial markets. Studying the impact of epidemics on financial markets is an important topic with practical significance for market stakeholders. Gurgul and Syrek\cite{Gurgul2024} used the CE method to study the correlation characteristics of the Polish stock market index during the 2019 COVID-19 pandemic, specifically examining the correlation between the WIG index and its 14 sector indices on March 13, 2020, the day the outbreak occurred in Poland. They found that this correlation significantly strengthened after the announcement of the epidemic. They\cite{Gurgul2024a} also used the same method to study the stock markets of four countries (France, Germany, UK, and US), using CE to calculate the correlation between the closing prices of stock sectors and the stock market index in each country. They found that this correlation also significantly strengthened after the epidemic. This finding is consistent with the experience gained during the 2008 financial crisis. They also found that the conclusions obtained by the CE method are consistent with experience, while the conclusions obtained by the traditional PCC method are inconsistent with past experience, such as the latter underestimating the post-epidemic correlation of the German stock market. The author argues that this is because CE can measure the nonlinear correlation between financial market variables without making any assumptions, thus validating the superiority of CE.

After more than 40 years of rapid development, China's insurance industry is undergoing a digital transformation. The application of insurance technology is profoundly impacting companies within the industry, and solutions addressing industry pain points are being implemented at an accelerated pace. Therefore, studying the extent to which Chinese insurance companies utilize technology and the impact of these applications on their operations is an important research topic. Li \cite{Li2023master} proposed theoretical hypotheses regarding the impact of insurance technology on four aspects of insurance company performance (including development capabilities, profitability, operational capabilities, and risk management capabilities), and conducted an empirical study based on relevant data from 114 insurance companies nationwide from 2018 to 2020. He used a regression benchmark model to analyze the impact of insurance technology levels on companies, and then used CE calculations to verify the nonlinear correlation strength between insurance technology and model variables. Both studies showed that insurance technology levels have a significant impact on insurance companies' business expense ratio, expedited reinsurance ratio, return on total assets, and comprehensive investment return rate. The empirical findings are consistent with the theoretical analysis. Based on this theory and empirical research, he offered valuable suggestions for insurance companies and industry regulators.

Premium setting is a crucial aspect of insurance services provided by insurance companies, impacting the efficiency and security of insurance operations. Insurance companies typically determine premium amounts based on risk assessments of customers' basic social information. Traditional premium setting relies on experience, lacking scientific rigor and rationality. Predicting premiums using data analysis models represents a new business model. However, customer information items are often numerous, requiring careful filtering to achieve accurate predictive capabilities. Uddin\cite{Uddin2025} proposes using a variable selection method based on CE to build a premium prediction model. This method was applied to publicly available auto insurance business data, selecting 5 items from 16 customer information items to build the prediction model, and comparing the method with similar modeling approaches. Experimental results show that the model built using the CE method delivers optimal predictive performance.

In recent years, research interest in machine learning methods for financial market forecasting has been increasing, mainly due to their nonlinear analysis capabilities and high asset prediction accuracy. However, practical deployment in cryptocurrency market forecasting is limited because traditional machine learning methods cannot select predictive variables correlated with the target financial asset in dynamic market environments and extreme market conditions. The fundamental reason lies in the unreasonable efficient market assumptions underlying these methods. The CE method can analyze nonlinear, non-Gaussian, and asymmetric correlations without distributional assumptions, providing a tool for solving this problem. Based on the adaptive market assumption, Mahmutovic\cite{Mahmutovic2025} proposed a method for effective and interpretable forecasting under real market dynamics. This method uses a CE based approach to select time-varying and tail-correlated indicator variables, while employing the copula divergence error function to guide the predictive model's learning. He validated the method using real-world historical data from four cryptocurrencies (Bitcoin, Ethereum, Ripple, and Dogecoin) over many years. The results showed that the CE based method improved prediction accuracy while increasing model interpretability, and the copula divergence error function reduced accumulated error. The success of this method demonstrates the rationality of the adaptive market assumption in the cryptocurrency market.

\appendix
\chapter{Softwares}
\label{chap:impl}
The methods for estimating CE, TE, the statistics of multivariate normality test, two-sample test, and change point detection in this monograph have been implemented in the \texttt{copent} package in \textsf{R} and \textsf{Python}\cite{ma2021copent} which are avaiable on  CRAN and PyPI respectively:
\begin{itemize}
	\item CRAN  \url{https://cran.r-project.org/package=copent};
	\item PyPI  \url{https://pypi.org/project/copent/}.
\end{itemize}
The source codes of the packages are available in the repositories of the author's GitHub: \url{https://github.com/majianthu/}.

The third-party software implementations of the CE based methods include:
\begin{itemize}
	\item \textsf{R}: \texttt{cylcop}\cite{Hodel2021,Hodel2022};
	\item \textsf{Python}: \texttt{MLFinLab}\cite{mlfinlab}, \texttt{ArbitrageLab}\cite{arbitragelab}, \texttt{gcmi}\cite{ince2017a,Ince2020}, \texttt{pytorch-mighty}\cite{Ulianych2023}, \texttt{HOI}\cite{BraiNets2024}, \texttt{THOI}\cite{Belloli2024}, \texttt{Frites}\cite{Combrisson2022a}, \texttt{Tensorpac}\cite{Combrisson2020}, \texttt{driada}\cite{AdvancedBrainStudies2024}, \texttt{CopulaGP}\cite{Kudryashova2022,CopulaGP}, \texttt{Polars-ds}\cite{Qin2024}, and \texttt{effconnpy}\cite{Crimi2025,Ciezobka2025};
	\item \textsf{Julia}: \texttt{CopEnt.jl}\cite{Xu2021}, \texttt{CausalityTools.jl}\cite{causalitytools}, and \texttt{Copulas.jl}\cite{Laverny2024};
	\item \textsf{Matlab}: \texttt{gcmi}\cite{ince2017a,Ince2020}, and \texttt{FieldTrip}\cite{Oostenveld2011};
	\item \textsf{C++}: \texttt{NPStat} \footnote{The CE estimation method is decomposed into two functions for empirical copula estimation and entropy estimation.}\cite{npstat}.
\end{itemize}

\bibliographystyle{unsrt}
\bibliography{ce-survey}

\end{document}